\documentclass[10pt,a4paper]{article}
\usepackage[utf8]{inputenc}
\usepackage{amsmath, amsthm,mathrsfs}
\usepackage{amsfonts}
\usepackage{amssymb}
\usepackage[normalem]{ulem}
\usepackage{xcolor}

\usepackage{mathtools}
\mathtoolsset{showonlyrefs}
\newcommand{\Par}{\lambda}
\newcommand{\zt}{\zeta_{2,B}(\omega)}
\newcommand{\kk}{\pi \tilde{k}}
\newcommand{\omegac}{\breve{\omega}}
\newcommand{\Vs}{V^{*}}
\newcommand{\MM}{\mathcal{M}}
\newcommand{\kkp}{\pi \tilde{k}-\phi_3(\omega)}
\newcommand{\kkq}{\pi \breve{k}}
\usepackage{bbm}
\newcommand{\yy}{Y}
\usepackage{hyperref}
\newcommand{\K}{\kappa_{+}}
\newcommand{\AC}{\mathcal{A}}

\newcommand{\DD}{\mathbf{e}}

\newcommand{\xx}{X}

\newcommand{\p}{\partial}
\newcommand{\ep}{\epsilon}
\newcommand{\VV}{V}

\newcommand{\LL}{\mathbb{L}}

\newcommand{\uAi}{\tilde{K}_{-}}

\newcommand{\uBi}{\tilde{K}_+}
\newcommand{\uAii}{u_{Ai_2}}
\newcommand{\uBii}{u_{Bi_2}}

\newtheorem{theo}{Theorem}

\usepackage{slashed}
\usepackage[left=2cm,right=2cm,top=2cm,bottom=2cm]{geometry}

\usepackage{authblk}
\usepackage{enumerate}

\usepackage{enumerate}

\theoremstyle{plain}
\newtheorem{thm}{Theorem}[section]

\newtheorem{lemma}[thm]{Lemma}

\newtheorem{prop}[thm]{Proposition}

\newtheorem{cor}[thm]{Corollary}

\newtheorem{conjecture}[thm]{Conjecture}

\theoremstyle{remark}
\newtheorem{rmk}{Remark}[section]

\theoremstyle{definition}

\newtheorem{defn}{Definition}[section]

\newcommand{\R}{\mathbb{R}}

\newcommand{\rd}{\partial}
\newcommand{\ls}{\lesssim}

\theoremstyle{plain}

\theoremstyle{remark}

\theoremstyle{definition}

\usepackage{bm}

\usepackage{color}

\usepackage{float}

\addtocounter{tocdepth}{-2}
\usepackage{graphicx}
\usepackage{authblk}
 
\numberwithin{equation}{section}

\begin{document}
	
	\title{Polynomial time decay for 
		solutions of the Klein--Gordon equation\\ on a subextremal Reissner--Nordstr\"{o}m black hole}

	\author[1]{Yakov Shlapentokh-Rothman}
	\author[2]{Maxime Van de Moortel}

	\affil[1]{\small  Department of Mathematics, University of Toronto,  40 St.~George~Street, Toronto, ON, Canada} 
	\affil[1]{\small Department of Mathematical and Computational Sciences, 
		University of Toronto Mississauga, 3359 Mississauga Road, Mississauga, ON, Canada}
	\affil[2]{\small Department of Mathematics, Rutgers University, Hill Center, 110 Frelinghuysen Road, Piscataway, NJ, USA}

	\maketitle
	
	\abstract{We consider the massive scalar field equation $\Box_{g_{RN}} \phi = m^2 \phi$ on any subextremal Reissner--Nordstr\"{o}m exterior metric $g_{RN}$. 
		We prove that solutions  with  localized initial data  decay pointwise-in-time  at the polynomial rate  $t^{-\frac{5}{6}+\delta}$ in any spatially compact region (including the event horizon), for some small  $
		\delta\leq \frac{1}{23} $. \\ Moreover, assuming the validity of the  Exponent Pair Conjecture on exponential  sums in Number Theory,  our result implies that  decay upper bounds  hold at the rate $t^{-\frac{5}{6}+\epsilon}$, for any arbitrarily small $\epsilon>0$.

		In our previous work, we  proved that  each fixed angular mode  decays at the exact rate $t^{-\frac{5}{6}}$, thus the upper bound $t^{-\frac{5}{6}+\epsilon}$ is sharp, up to a $t^{\ep}$ loss. Without  the restriction to a fixed angular mode,  	 the solution turns out to have an unbounded Fourier transform  due to discrete  frequencies associated to quasimodes, and caused by the occurrence of stable timelike trapping. Our analysis nonetheless  shows that   inverse-polynomial  asymptotics in $t$ still hold after  summing over all angular modes.

		\section{Introduction}
		
		We are interested in the large time asymptotic behavior of  the massive Klein--Gordon equation on the exterior region of a Schwarzschild or a Reissner--Nordstr\"{o}m  black hole in the full sub-extremal range:  \begin{equation}\label{KG.intro}\begin{split}
				&\Box_{g_{RN}} \phi = m^2 \phi, \\ & g_{RN}=-\Big(1-\frac{2M}{r}+\frac{\DD^2}{r^2}\Big) dt^2 + \Big(1-\frac{2M}{r}+\frac{\DD^2}{r^2}\Big)^{-1} dr^2 + r^2\left( d\theta^2+ \sin^2(\theta) d\varphi^2\right),
			\end{split}
		\end{equation} where $0\leq |\DD|< M$ (sub-extremality condition) and $r\in[ r_+(M,\DD),+\infty)$, where $r=r_+(M,\DD)>0$ is the largest root of $\Big(1-\frac{2M}{r}+\frac{\DD^2}{r^2}\Big)$  corresponding to the event horizon of the Reissner--Nordstr\"{o}m black hole. The Klein--Gordon equation \eqref{KG.intro} describes the evolution of a massive scalar field on a curved spacetime,  arguably  
		the simplest model for  massive matter in the relativistic setting. In the black hole case, the late-time asymptotics  of \eqref{KG.intro}  have been the subject  of many heuristic and numerical works in the  Physics literature \cite{massive1,massive2,Konoplya.Zhidenko.num,KonoplyaZhidenko,Dilaton1,Dilaton2,BurkoKhanna,Burko,HodPiran.mass,massive0}. However, the problem of decay-in-time for solutions of \eqref{KG.intro} has remained open.
		
		\paragraph{Decay and late-time asymptotics for fixed angular mode} 	In our previous work \cite{KGSchw1}, we proved that solutions $\phi$ of \eqref{KG.intro} (with sufficiently regular and localized data)  which  are supported on a fixed spherical harmonic $Y_L(\theta,\varphi)$ decay at the sharp rate $(t^{*})^{-\frac{5}{6}}$ as $t^{*}\rightarrow+\infty$ on any compact $r$-region, where $t^{*}(t,r)$ is a horizon-penetrating time coordinate (comparable to $t$ for any fixed $r>r_+(M,e)$, but also capturing advanced-time on the event horizon $\mathcal{H}^+=\{r=r_+(M,e)\}$, see \cite{KGSchw1}, Remark 1.1 and Section~\ref{Setup.section}). \begin{thm}[Late-time asymptotics for fixed angular momentum $L$, \cite{KGSchw1}]\label{thm.old}
			For any $L\in \mathbb{N}\cup\{0\}$ and smooth compactly supported  $(\phi_0,\phi_1)$, there exists   an explicit profile $\mathcal{F}_L[\phi_0,\phi_1](t^{*},r)$  such that the solution $\phi_L= r^{-1}\psi_{L}(t^{*},r) Y_{L}(\theta,\varphi)$ of \eqref{KG.intro}  with  initial data $(\phi_0,\phi_1)$ at $\{t^{*}=0\}$ satisfies for all $r_+\leq r \leq R_0$, $t^{*}\geq 1$:
			\begin{align}\label{decay.intro.old}\begin{split}
					&\psi_{L}(t^{*},r)= \mathcal{F}_L[\phi_0,\phi_1](t^{*},r) \cdot (t^{*})^{-\frac{5}{6}}+ O_{R_0,L,\phi_0,\phi_1}(  \left(t^{*}\right)^{-1+\ep}) 
					\text{ for any } \ep>0,\\ &  \sup_{t^{*} \geq 0} \bigl|\mathcal{F}_L[\phi_0,\phi_1]\bigr|(t^{*},r) <\infty \text{ for all } r \geq r_+. \end{split}
			\end{align} 
		\end{thm}  In what follows, we will generalize Theorem~\ref{thm.old} to solutions of \eqref{KG.intro} that are not supported on a fixed spherical harmonic $Y_L$: In particular, we will examine whether \eqref{decay.intro.old} also holds uniformly in $L$.
		We actually proved in \cite{KGSchw1} that the explicit profile $\mathcal{F}_L[\phi_0,\phi_1](t^{*},r)$ is proportional to the following exponential sum: \begin{equation}\label{first.osc.sum}\sum_{q=1}^{+\infty} (-\gamma_m)^{q-1}q^{\frac{1}{3}}\cos(mt^{*}  -\frac{3}{2} [2\pi M]^{\frac{2}{3}} m q^{\frac{2}{3}} (t^{*})^{\frac{1}{3}}     +O_{q}((t^{*})^{-\frac{1}{3}})+ \varphi(r)),\end{equation} with $0<|\gamma_m|(M,e,m^2,L)<1$. The specific dependence in $L$ of exponential sums of the form \eqref{first.osc.sum}  will be crucial in proving quantitative decay in time \emph{uniformly} in $L$,
		as we will explain in the paragraphs below.
		
		\paragraph{Decay without a rate for general localized solutions   and absence of bound states}
		From Theorem~\ref{thm.old}, one can also softly infer \emph{decay-without-a-rate for general solutions} of \eqref{KG.intro} which are not necessarily supported on a fixed spherical harmonic, i.e.\ $\underset{t^{*} \rightarrow +\infty}{\lim} \phi(t^{*},r,\theta,\varphi)=0$, for any $r\geq r_+(M,e)$. In particular, Theorem~\ref{thm.old} shows that (localized, regular) \emph{bound states solutions of \eqref{KG.intro} do not exist}. Such a statement may come as a surprise, since \eqref{KG.intro}  has an effective \emph{attractive long-range} (Coulomb-like) potential\footnote{In contrast, for the massless wave equation (\eqref{KG.intro} with $m^2=0$), the effective potential $V(r)=\frac{L(L+1)}{r^2} + O(r^{-3})$ is a short-range perturbation of the  wave equation on Minkowski spacetime.}
		$V(r)$, similar to the well-known hydrogen atom: \begin{equation} \label{V.intro}
			V(r)=\Big(1-\frac{2M}{r}+\frac{\DD^2}{r^2}\Big)\Big(m^2 + \frac{L (L+1)}{r^2} +\frac{2rM -2\DD^2}{r^4}\Big)= m^2 -\frac{2Mm^2}{r} +...
		\end{equation}  Equations $\Box \phi= V(r) \phi$ on flat Minkowski space are known to admit bound states \cite{Ethan2}: however, in our case \eqref{KG.intro}, the presence of a coercive energy at the event horizon $\mathcal{H}^+=\{r=r_+(M,e)\}$ is crucial in proving that \emph{bound states do not exist}, allowing for time decay (see \cite{KGSchw1} and the discussion in Section~\ref{intro.Kerr.section}).

		\paragraph{Quasimodes and stable trapping}	The absence of bound states and the decay-in-$t^{*}$ without a rate is however compatible with the existence of so-called \emph{quasimodes}, namely approximate solutions of \eqref{KG.intro} corresponding to ``almost''-bound states, which typically arise in the presence of \emph{stable trapping}. 
		On the Schwarzschild/Kerr-Anti-de-Sitter black holes, 
		the stable trapping phenomenon  indeed gives rise to a slow $[\log(t^{*})]^{-p}$ decay \cite{decayads,Gannot}. 
		Stable trapping may also arise for wave equations in the presence of obstacles, giving   rise to quasimodes and a slow decay in time at a logarithmic rate, see e.g.\ \cite{Burq,cardosopopov}.
		
		Returning to the Klein--Gordon equation \eqref{KG.intro} on a Reissner--Nordstr\"{o}m  black hole, we note that a  hydrogen-like potential as  in \eqref{V.intro} admits  a minimum
		(see Figure~\ref{Potential}), which implies the existence of \emph{stable timelike trapping}
		in the high angular frequency  (large $L$)  regime\footnote{In close connection to stable trapping, we  show that the  local integrated decay estimate with a source \emph{fails}, see Theorem~\ref{quasimode.thm}.}. This suggests  the existence of quasimodes localized near $r\approx L^2$ approximating more accurately the solution as $L \rightarrow+\infty$. In fact, the following related scenario has  been proposed in the Physics literature \cite{SantosWCC}:  the decay in time of solutions of \eqref{KG.intro} is of order $[\log(\log(t^{*}))]^{-1}$.

		\paragraph{Quantitative inverse-polynomial decay for general solutions}	
		We prove, however, that solutions of \eqref{KG.intro} with localized and regular initial data \emph{decay at an inverse polynomial rate} $(t^{*})^{-\frac{5}{6}+\delta}$ for some $\delta>0$ in any compact $r$-region \emph{despite the presence of stable timelike trapping}. The following Theorem~\ref{thm.intro} can be seen as the analogue of Theorem~\ref{thm.old} with no restriction  to finitely many spherical harmonics, which is related to obtaining  uniform-in-$L$ decay estimates. Theorem~\ref{thm.intro} consists of two statements: an unconditional decay $O((t^{*})^{-\frac{5}{6}+\delta})$ result for $\delta=\frac{1}{23}$, and a conditional decay result for $\delta>0$  arbitrarily small,  assuming the validity of the so-called Exponent Pair Conjecture  in Number Theory (see the discussion below), which predicts upper bounds on exponential sums such as \eqref{first.osc.sum} (see also Remark~\ref{rmk.general} for extensions of Theorem~\ref{thm.intro}).

		\begin{thm}[Late-time polynomial decay for general solutions]\label{thm.intro}
			
			Let $\phi$ solve the Klein-Gordon equation~\eqref{KG.intro} with smooth compactly supported initial data $\left(\phi_0,\phi_1\right) = \left(\phi|_{\left\{t^* = 0\right\}},n_{\{t^*=0\}}\phi|_{\left\{t^* = 0\right\}}\right)$ along $\{t^* = 0\}$. Then, there exists $N > 0$ so that for every $R > 0$, $\delta > 0$, and $t^* \geq 1$ we have
			\begin{equation} \label{cor.eq}
				\sup_{r \in [r_+,R]}\left|\phi\left(t^*,r,\theta,\varphi\right)\right| \lesssim_{R,\delta} [t^{*}]^{-\frac{5}{6}+\delta}\left(\left\vert\left\vert \phi_0\right\vert\right\vert_{H^N\left(\{t^* = 0\}\right)}+\left\vert\left\vert \phi_1\right\vert\right\vert_{H^{N-1}\left(\{t^* = 0\}\right)} \right),
			\end{equation}
			where \begin{enumerate}[I.]
				\item \label{thm.intro.I} (Unconditional decay result) 	we can unconditionally choose $\delta = \frac{1}{23}$.

				\item  \label{thm.intro.II} (Conditional decay result) assuming that the Exponent Pair Conjecture is true, we may take $\delta=\ep$ for any arbitrarily small $\ep>0$. 
			\end{enumerate}

			Moreover, \eqref{cor.eq} also holds  replacing  $\phi$ by any derivative $\rd^{\alpha} \phi$ in $(t^{*},r,\theta,\varphi)$ coordinates.
		\end{thm} 
		\begin{rmk}\label{rmk.general}
			For the simplicity of the exposition, we have focused on proving time decay in the region $r \in [r_+,R]$ under the assumption of compactly supported initial data. Our methods, however, can also be extended to address    more general situations, in which we either relax the compact support assumption of the initial data, or we consider time-decay upper bounds  which are uniform in $r$.  \begin{enumerate} [A.]
				\item \label{exp.statement} The assumption of compact support of the initial data $(\phi_0,\phi_1)$ can easily be replaced by a decay assumption with suitable spatial exponential weights in $r$ and Theorem~\ref{thm.intro} still holds identically.
				\item   \label{polydecaytime} The method of proof behind Theorem~\ref{thm.intro} moreover allows to establish polynomial time decay in the compact region  $r \in [r_+,R]$ for any  solution whose initial data only decay polynomially in $r$. 
				\item  \label{uniformrdecaytime} Moreover, the same methods also allow one to establish polynomial time decay uniformly in $r$ for solutions with compactly supported (or  exponentially decaying) 
				initial data. We leave the systematic study of this to a future work.
				
				\item \label{logdecaypoly}It is interesting to consider the question of determining decay rates in time which are uniform in $r$ in the situation when the initial data decay only polynomially. In this case, in view of the presence of quasimodes which can concentrate in regions where $r\sim \log^p(t^*)$,  we can only hope to achieve logarithmic decay in time: see Theorem~\ref{quasimode.thm}  and the associated  discussion of quasimodes. 
				
				\item 	We note, however,  that $r\phi$ is uniformly bounded over the spacetime, as a straightforward consequence of Sobolev inequalities, the boundedness of energy, and commutation with angular momentum operators. Polynomial time decay for $\phi$ is therefore trivially retrieved in regions of the form $\{ r\gtrsim (t^*)^{q}\}$ for $q>0$.
			\end{enumerate}
		\end{rmk}

		\begin{rmk}	Statement~\ref{thm.intro.II} of Theorem~\ref{thm.intro} in fact holds under the weaker assumption that $\left(3\epsilon,\frac{2}{3}+\epsilon\right)$ is an exponent pair for any $\epsilon>0$ (a statement that follows from the Exponent Pair Conjecture), see Appendix~\ref{appendix.conj}.\end{rmk}

		\paragraph{Oscillating sums and connections to Number Theory} Our analysis leading to the proof of Theorem~\ref{thm.intro} first shows that    the solution $\phi=r^{-1}\underset{L\in \mathbb{N}}{\sum} \psi_L Y_L$ of \eqref{KG.intro} takes the following schematic form involving an exponential sum, where $|\Gamma_L| \approx 1-e^{-4L \log L + O(L)}$ (compare with \eqref{first.osc.sum}): for any $\ep>0$: \begin{equation}\label{first.osc.sum2}\psi_L(t^{*},r) \approx (t^{*})^{-\frac{5}{6}}\left|\left|\Gamma_L\right|-1\right|\sum_{q=1}^{\lceil (t^{*})^{1/2} \rceil} (-\Gamma_L)^{q-1}q^{\frac{1}{3}}\cos(mt^{*}  -\frac{3}{2} [2\pi M]^{\frac{2}{3}} m q^{\frac{2}{3}} (t^{*})^{\frac{1}{3}}     +...)+ \overbrace{\text{errors,}}^{=O([t^{*}]^{-\frac{5}{6}+\ep})}\end{equation}
		
		\begin{rmk}
			With slightly more work, we expect that the errors in \eqref{first.osc.sum2}  can be upgraded to $O(  (t^{*})^{-\frac{5}{6}-\eta})$ for  $\eta>0$, so  the main contribution to the decay of $\phi$ in Theorem~\ref{thm.intro} is  given by the exponential sum in \eqref{first.osc.sum2}. We, however, do not  pursue this, in view of the difficulties in estimating exponential sums (see below).
		\end{rmk}

		To prove Theorem~\ref{thm.intro}, 	we must also control the exponential sum of \eqref{first.osc.sum2} uniformly in $L$, which turns out to be challenging, since  $|\Gamma_L|\rightarrow 1$ as $L\rightarrow \infty$.  \eqref{first.osc.sum2} immediately implies  a  \emph{$O([t^{*}])^{-\frac{2}{3}}$ pointwise upper bound} (uniformly in $L$), which is clearly not sharp. To prove the sharper decay rates claimed in Theorem~\ref{thm.intro}, we use results obtained in the context  of Number Theory to estimate exponential sums such as the one  in \eqref{first.osc.sum2}.

		It turns out that the current techniques in Analytic Number Theory do not allow us to prove the $(t^{*})^{-\frac{5}{6}+\epsilon}$ conjectured decay from the exponential sum in \eqref{first.osc.sum2}. Nonetheless, such a decay estimate follows from the celebrated \emph{Exponent Pair Conjecture} on exponential sums with a monomial phase (see also Appendix~\ref{appendix.conj}): 
		\begin{conjecture}[Exponent Pair Conjecture \cite{NT3,GrahamK,Montgomery}]\label{exp.conj.intro}
			Let  $f:[1,+\infty)\rightarrow \R$ be a  smooth function behaving approximately like  the monomial $Tx^{-\sigma}$ at any order, namely:  there exists constants $\sigma>0$, $T>0$ such that for any $\delta \in (0,1)$ and $M \in \mathbb{N}$, there exists $\eta(\delta,M)>0$ such that for all $x \leq \eta(\delta,M) T^{1/\sigma}$:  \begin{equation}\label{f.condition.intro}
				\bigl| f^{(j+1)}(x)-  T  \frac{d^j}{dx^j}[ x^{-\sigma}] \bigr|  \leq  \delta \cdot T \bigl| \frac{d^j}{dx^j}[ x^{-\sigma}] \bigr| \text{ for all } j \in[1,M].
			\end{equation}
			Then, for any $\epsilon>0$,  any function $f$ satisfying \eqref{f.condition.intro} and for any $0 \leq N \leq N' \leq 2N$, we have for all $T\gtrsim N^{\sigma}$: \begin{equation}\label{exp.pair.est}
				\bigl|	\sum_{n=N}^{N'} \exp( i f(n)) \bigr| \ls (\frac{T}{N^{\sigma}})^{\epsilon} N^{\frac{1}{2}+\epsilon}. 
			\end{equation} 
		\end{conjecture}
		
		Exponent pairs  play an important role in Number Theory: 
		for example,   Conjecture~\ref{exp.conj.intro} implies the validity of the celebrated Lindel\"{o}f Hypothesis \cite{lindelof,lindelof2}, a statement postulating  $O(T^{\ep})$ upper bounds on $\zeta(\frac{1}{2}+iT)$ as $T\rightarrow +\infty$, and  which is closely related to the Riemann Hypothesis.  An abundant literature has since been developed on exponent pairs, see e.g.\ \cite{NT3,GrahamK,Montgomery} and references therein, and \cite{NT5} for a recent survey. The sharpest result is due to Bourgain \cite{NT2}  and shows that \eqref{exp.pair.est} holds for any $\epsilon>\frac{13}{84}$. 
		Combining our analysis and these partial results on Conjecture~\ref{exp.conj.intro}, we derive the (unconditional) Statement~\ref{thm.intro.I} of Theorem~\ref{thm.intro}.  Our second (conditional) Statement~\ref{thm.intro.II} in contrast follows from assuming the validity of Conjecture~\ref{exp.conj.intro}.

		\paragraph{Unbounded Fourier transform of the solution and failure of integrated decay} We now briefly discuss the quantitative behavior in both physical and frequency space of the solution of \eqref{KG.intro} with localized data as in Theorem~\ref{thm.intro}. Writing $\phi$ in temporal Fourier space and using spherical harmonics $Y_L$ as \begin{equation}\label{FT.intro}
			\phi(t^{*},r,\theta,\varphi) = \sum_{L \in \mathbb{N}} Y_L(\theta,\varphi) \int_{\omega=-\infty}^{+\infty} e^{i\omega t^{*}} \hat{\phi}_L(\omega,r)d\omega
		\end{equation}  allows to recast  \eqref{KG.intro} into the following Schr\"{o}dinger equation, with a potential $V(r)$ given in  \eqref{V.intro}: \begin{equation}\label{eq.intro}
			(1-\frac{2M}{r}+\frac{\DD^2}{r^2})\frac{d}{dr}\left( (1-\frac{2M}{r}+\frac{\DD^2}{r^2})\frac{d}{dr}(r \hat{\phi}_L) \right)=  \left[-\omega^2+ V(r) \right] r\hat{\phi}_L+\hat{H}_L(\omega,r),
		\end{equation} where $\hat{H}_L$ is a source term (which comes from the initial data $\phi_0,\phi_1$).
		In this paragraph, we discuss   the \emph{obstructions to decay} due of the shape of the potential $-\omega^2+V(r)$  (see Figure~\ref{Potential}). In the regime, $m-L^{-2}\ll|\omega|<m$ (with $L\gg 1$), note the presence of two turning points  delimiting a classically allowed region $\{ L^2 \lesssim r \lesssim [m^2-\omega^2]^{-1}\}$  containing a  minimum of the potential, which leads to \emph{stable trapping} occurring in the low frequencies $|\omega| \approx m$.  We show that schematically, $\hat{\phi}_L$ takes the following form  for $|\omega|<m$:  \begin{equation}\label{resolvent.intro}\begin{split}
				&	\hat{\phi}_L(\omega,r)=  \exp\big( L \cdot f_{data}(r)\big) \cdot  \ep_L^2\cdot \frac{	\cos(\pi \tilde{k} )+  e^{O(L)}\cdot \sin(\pi \tilde{k})}{ 	\ep_{L}^2 \cdot \cos(\pi \tilde{k} )+ \sin(\pi \tilde{k})}	 + ...\\ &  \kk  \approx (m^2-\omega^2)^{-1/2}+O(1),  \\ & \ep_L \approx \exp(-2L \log L + O(L)).\end{split}
		\end{equation}
		
		\eqref{resolvent.intro} schematically reflects an amplification of size $\exp\left( L \cdot f_{data}(r)\right)$ (the rate $f_{data}(r)=e^{O_{R_0}(1)}$ depends on the  initial data and $r$, but is bounded away from $0$ in any compact $r$-region $\{r\leq R_0\}$) when $ \pi \tilde{k}$ is an integer\footnote{Note:  for $\kk$ away from integers,  there is no amplification in \eqref{resolvent.intro} since $\ep_L^2 \cdot \exp\left( L \cdot f_{data}(r)\right)=  \exp\left( -4L \log L+O(L)\right) \ll 1$.}, which occurs infinitely often in the limit $|\omega|\rightarrow_{|\omega|<m} m$  (see  Theorem~\ref{TPsection.mainprop} for the complete and rigorous expression). This exponential amplification is the main quantitative consequence of stable trapping and leads to a pointwise-unbounded Fourier transform and the corresponding failure of standard Morawetz (namely: integrated local decay) estimates  with a source. It is also related to the other main consequence of stable timelike trapping: namely, the existence of quasimodes in large-$r$ regions, as follows.
		
		\begin{figure}\label{Potential}
			
			\begin{center}
				
				\includegraphics[width=100 mm, height=50 mm]{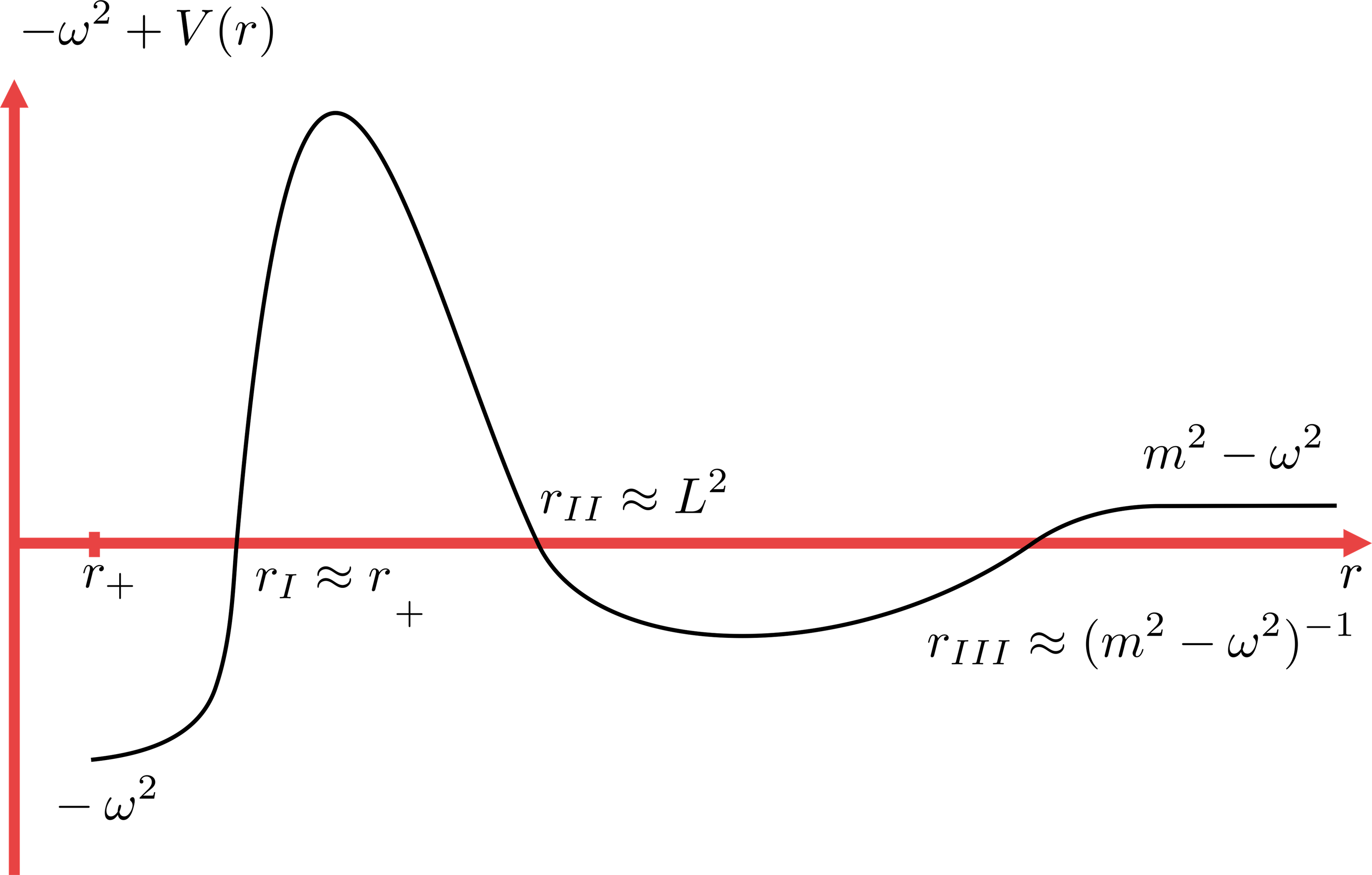}
				
			\end{center}
			\caption{The  effective potential $-\omega^2 + V(r)$, when $m-L^{-2}\ll|\omega|<m$, where $V(r)$ is as in \eqref{V.intro}.}

		\end{figure}

		\begin{thm} \label{quasimode.thm}~
			\begin{enumerate}[I.]
				\item \label{blowup}(Blow-up of the Fourier transform in $L^{\infty}$ norm).  Let $R>r_+$. There exists  smooth $\left(\phi_0,\phi_1\right)$ with support in $[r_+,R]$ such that the solution $\phi$ of \eqref{KG.intro} with  initial data   	$ \left(\phi|_{\left\{t^* = 0\right\}},n_{\{t^*=0\}}\phi|_{\left\{t^* = 0\right\}}\right)=\left(\phi_0,\phi_1\right)$ along $\{t^* = 0\}$ satisfies the following estimates:  \begin{equation}\label{fourier.blow-up.intro}\begin{split}
						&\text{for all } L\in\mathbb{N}\cup\{0\}\text{: }\sup_{r\in[r_+,R]}\ \|\hat{\phi}_L(\cdot,r)\|_{L^{\infty}_{\omega}(m-L^{-2},m)}\ \gtrsim_{R}\  e^{C L},\\&	\text{ therefore: }	\sup_{r\in[r_+,R]}\ \|	\hat{\phi}\left(\cdot,r\right)\|_{L^{\infty}\left((m-L^{-2},m)\times \mathbb{S}^2 \right)} =\infty.\end{split}
				\end{equation} 
				\item \label{failure} (Failure of Morawetz estimate  with a source $F$). We  now consider solutions of the inhomogeneous Klein--Gordon equation with zero initial data: \begin{equation}\label{KG.inh.intro}
					\Box_{g_{RN}} \phi = m^2 \phi + F.
				\end{equation}
				
				For any $p\in\mathbb{N},\ q\in\mathbb{N},\  N \in\mathbb{N}$, $R_0>r_+$, for all large $\Delta>0$, there exists a smooth compactly source function $F_{\Delta}$   such that the solution $\phi$ of \eqref{KG.inh.intro} satisfies: \begin{equation}\label{morawetz.fail}
					\int_{t^{*}=0}^{+\infty}     \sum_{|\alpha|\leq 1 }\| r^{-p}\rd_{t}^{\alpha_t} \rd_{i}^{\alpha_i} \phi\|^2_{L^2(B(0,R_0))}(t^{*})dt^{*} \geq  \Delta  \int_{t^{*}=0}^{+\infty}   \sum_{|\alpha|\leq N } \| r^{q}\rd_{t}^{\alpha_t} \rd_{i}^{\alpha_i} F_{\Delta}\|^2_{L^2(\R^3)}(t^{*})dt^{*}.
				\end{equation}
				
				\item\label{quasimode.statement}  (Quasimode construction and logarithmic lower bounds for initial data with polynomial decay).  Let $p>2$ be large enough. There exists a sequence $\Phi_L(t^{*},r,\theta,\varphi)$ for $L \in \mathbb{N}$  and $L\gg 1$  which 
				
				\begin{enumerate}[a.]
					\item \label{quasi0}  are time-periodic:  $\Phi_L(t^{*},r,\theta,\varphi) = e^{-i \omega_L t^{*}} F_L(r,\theta)$, where $F_L$ is smooth and $\omega_L =m - O(L^{-p}) \in \R$.
					\item \label{quasi1} are supported in a region $L^2 \lesssim r \lesssim L^p$ and satisfy $\left\vert\left\vert F_L\right\vert\right\vert_{L^2(\R^3)} = 1$.
					\item  \label{quasi2} are solutions  of \eqref{KG.intro} up to small errors: $\left|\mathcal{E}_L\right|=  \left|\left[\Box_{g_{RN}} - m^2 \right] \Phi_L \right| \lesssim e^{-DL}$, for some constant $D>0$ independent of $L$.

				\end{enumerate}
				Therefore, 	 let $N>0$, $M>0$ be a large enough numbers and let $\mathcal{E}(M,N)$ be the set of solutions $\phi$ of \eqref{KG.intro} with initial data $ \left(\phi|_{\left\{t^* = 0\right\}},n_{\{t^*=0\}}\phi|_{\left\{t^* = 0\right\}}\right)=\left(\phi_0,\phi_1\right)$ along $\{t^* = 0\}$, with regularity $r^{M}\left(\phi_0,\phi_1\right) \in H^N(\{t^{*}=0\}) \times H^{N-1}(\{t^{*}=0\}) $. There exists constants $P\left(M,N\right) > 0$ and  $C(M,N)>0$ such that \begin{equation*}
					\limsup_{t^{*}\rightarrow +\infty}\sup_{\phi\in \mathcal{E}(M,N)}\  [\log(1+t^{*})]^{P}\frac{ \underset{ L^2 \ls r \ls L^p,\ \omega \in \mathbb{S}^2}{\sup}|\phi|(t^{*},r,\omega)}{ \| r^{M} \phi_0 \|_{H^N(\{t^{*}=0\})}+ \| r^{M}\phi_1 \|_{H^{N-1}(\{t^{*}=0\})}}  \geq C(M,N) >0.
				\end{equation*}
				
			\end{enumerate}
			
		\end{thm}
		
		\begin{rmk} We emphasize that a  key point in the failure of the Morawetz estimate with a source is that we are not allowed for these estimates to measure the source in a norm that contains growing weights in $t^{*}$. Moreover, we note in contrast that the analogue of   Statement~\ref{failure} for the initial value problem (with no source) is not true: indeed,  as a consequence of  the decay estimates of   Theorem~\ref{thm.intro}, solutions arising from compactly supported initial data are square integrable in time uniformly along any fixed compact set in $r$:

			\begin{equation}\label{log.lower}
				\int_{t^{*}=0}^{+\infty}     \sum_{|\alpha|\leq 1 }\| \rd_{t}^{\alpha_t} \rd_{i}^{\alpha_i} \phi\|^2_{L^2(\left\{r \leq R_0\right\})}(t^{*})dt^{*} \leq  C_{R_0}  \sum_{|\alpha|\leq N } \| \rd_{t}^{\alpha_t} \rd_{i}^{\alpha_i} (\phi_0,\phi_1)\|^2_{L^2(\R^3)}. 
			\end{equation}

		\end{rmk}
		
		\begin{rmk}
			We emphasize that Statement~\ref{quasimode.statement} of Theorem~\ref{quasimode.thm} showing logarithmic-in-time lower bounds  in large $r$-regions only applies to solutions with initial data that merely decay polynomially. Indeed, Theorem~\ref{thm.intro} (and Remark~\ref{rmk.general}, Statement~\ref{uniformrdecaytime}) show, in sharp contrast, polynomial decay in time   for solutions with compactly supported (or exponentially decaying) initial data.
		\end{rmk}
		
		We would like to point out that analogous phenomena to Statement~\ref{blowup} and Statement~\ref{failure}  in Theorem~\ref{quasimode.thm} also occur in the other known settings where stable trapping manifests itself \cite{Benomio,Gannot,decayads,quasimodeads} (see Section~\ref{intro.trapping.section} and the discussion immediately below). We will  also discuss the  comparison between Statement~\ref{quasimode.statement} and other stable trapping situations in the next paragraph below. However, we emphasize that the analogue of Theorem~\ref{quasimode.thm} is false for the (massless) wave equations on a Schwarzschild or Kerr black hole on which both integrated local energy decay and polynomial decay in $t^{*}$ hold when the initial data lies in a suitable polynomially weighted Sobolev space, see Section~\ref{intro.massless.section}.

		\paragraph{The localization of our quasimodes in far-away regions}
		
		We recall that so-called \emph{classically forbidden regions} (at fixed $\omega$) are  regions  $\{r,\ - \omega^2 + V(r)<0\}$, where the celebrated Agmon estimates \cite{Agmon} lead us to expect the solution to be of size $\exp(-O(L))$ ($L\gg1$ being the angular momentum, recalling that $L^{-1}:=\hbar$ plays the role of a semi-classical parameter): From Figure~\ref{Potential}, we see the existence of a forbidden region on which $ r \ls L^2$. However,  $ L^2 \ls r \ls (m^2-\omega^2)^{-1}$ is an allowed region in which the quasimodes of Theorem~\ref{quasimode.thm}, Statement~\ref{quasimode.statement} are supported, in the regime where  $(m^2-\omega^2)^{-1} \approx L^p$ for a large $p>2$.

		A key conceptual point is these quasimodes do not end up causing an obstruction to polynomial decay for compactly supported (or exponentially decaying, recalling Remark~\ref{rmk.general}, Statement~\ref{exp.statement}) initial data because their  support moves towards $r = \infty$ as $L \to \infty$.   The concrete manifestation of this is that in the formula  \eqref{resolvent.intro}   the \emph{amplification when $\tilde{k}$ is an integer only occurs on a very small $\tilde{k}$-interval of size $O(\exp(-4L\log L+ O(L))$}, which we exploit when computing the inverse Fourier transform of $\hat{\phi}_L$.
		
		On the other hand, if one considers initial data with only (sufficiently fast) polynomial decay towards $r = \infty$, then the construction of quasimodes
		in  Theorem~\ref{quasimode.thm}, Statement~\ref{quasimode.statement} allows one to infer   that there \emph{cannot} exist a uniform polynomial decay in time bound in the region of spacetime $ \log^2(t^{*})\ls r \ls \log^p(t^{*})$ for any sufficiently large $p>2$,   see also Remark~\ref{rmk.general}, statement~\ref{logdecaypoly}. (Note, however, that polynomial decay in time still holds in any compact region $r\in [r_+,R]$, as mentioned in Remark~\ref{rmk.general}, statement~\ref{polydecaytime}).

		Quasimodes also occur on Schwarzshild/Kerr-AdS black holes \cite{Gannot,decayads,quasimodeads}, but in contrast are supported on  a fixed compact region (with respect to $L$) and thus lead to slow logarithmic decay lower bounds in $t^{*}$, even when restricting the solution to a spatially compact region} (in contrast to  Theorem~\ref{thm.intro}). A similar phenomenon occurs on black strings/rings \cite{Benomio} (see  Section~\ref{intro.trapping.section}). \begin{rmk}We also note that the Klein--Gordon quasimodes for \eqref{KG.intro} are related to stable timelike trapping, while the Anti-de-Sitter black holes quasimodes arise from null   trapping. Null trapping, however, does not create any major obstruction to decay in $t^{*}$ for \eqref{KG.intro}, as we show in (the proof of) Theorem~\ref{thm.intro} (see the discussion in Section~\ref{intro.proof.section}).	\end{rmk}
	\paragraph{Outline of the rest of the introduction}
	
	In Section~\ref{intro.massless.section},  we discuss the massless wave equation (i.e. \eqref{KG.intro} with $m^2=0$) on black hole spacetimes, and compare with the massive case $m^2>0$. In Section~\ref{intro.Kerr.section}, we compare \eqref{KG.intro} to the Klein--Gordon equation on Minkowski-like spacetimes, and on the Kerr black hole. In Section~\ref{intro.trapping.section}, we discuss the phenomenon of stable trapping and the presence of quasimodes. In Section~\ref{numerics.intro}, we discuss previous numerical and heuristic works on late-time asymptotics for \eqref{KG.intro}, and related Physics literature. Finally in Section~\ref{intro.proof.section}, we give a short summary of the proof of Theorem~\ref{thm.intro}.
	
	\subsection{The wave equation on black holes and local integrated  decay}\label{intro.massless.section}
	The (massless) wave equation corresponds to \eqref{KG.intro} with $m^2=0$, and behaves qualitatively differently from the Klein--Gordon case $m^2>0$. Notably, in contrast to \eqref{KG.intro} where \eqref{V.intro} is a long-range potential, the effective potential for the wave equation  on a Reissner--Nordstr\"{o}m black hole is  repulsive and short-range:
	\begin{equation} \label{V.intro.shortrange}
		V(r)=\Big(1-\frac{2M}{r}+\frac{\DD^2}{r^2}\Big)\Big(\frac{2M }{r^3}
		-\frac{2\DD^2}{r^4}\Big) \sim  \frac{2M}{r^3}  \text{ as  r } \rightarrow +\infty.
	\end{equation} 
	
	\paragraph{Price's law} Comparing to \eqref{KG.intro} with $m^2\neq 0$ (Theorem~\ref{thm.old}): In the wave case, the decay is faster $O( (t^{*})^{-3})$  and higher spherical harmonic  $ \phi_L$ decay even faster and do not oscillate in $t^{*}$, following Price's law \cite{Pricepaper}: \begin{equation}\label{Price.law}
		\sum_{L \geq L_0}\phi_L(t^{*},r)Y_L(\theta,\phi) \sim F_{L_0}(r)Y_{L_0}(\theta,\phi)\cdot (t^{*})^{-3-2L_0},
	\end{equation} see \cite{AAG1,Hintz} for the proof of \eqref{Price.law} on a Schwarzschild/Reissner--Nordstr\"{o}m black hole,   \cite{PriceLaw,Schlag2,Schlag1,AAG1} for earlier works on  versions of \eqref{Price.law}, \cite{Tataru,Tataru2,Hintz,AAG2} for analogues of \eqref{Price.law} on a Kerr black hole, and \cite{millet,Ma1,Ma3,Ma4,Ma5} for non-zero spin wave equations.  
	
	\paragraph{Integrated  local energy decay} Stepping away from specific asymptotics such as \eqref{Price.law}, we discuss    local integrated decay estimates. We note that for both the wave equation and~\eqref{KG.intro} the null geodesic flow is related to the high-frequency behavior of solutions, but only in the massive case~\eqref{KG.intro} is the timelike geodesic flow also relevant for high frequencies. In sharp contrast to Theorem~\ref{quasimode.thm}, Morawetz estimates with a source are true for the wave equation on Schwarzschild   \cite{BlueSoffer1,Red}, i.e.\ 
	for solutions of \eqref{KG.inh.intro} with $m^2=0$: for  $C>0$, \begin{equation}\label{morawetz}
		\int_{t^{*}=0}^{+\infty}     \sum_{|\alpha|\leq 1 }\| r^{-1}\rd_{t}^{\alpha_t} \rd_{i}^{\alpha_i} \phi\|^2_{L^2(\R^3)}(t^{*})dt^{*} \leq    C \int_{t^{*}=0}^{+\infty}   \sum_{|\alpha|\leq 2 } \| r\ \rd_{t}^{\alpha_t} \rd_{i}^{\alpha_i} F\|^2_{L^2(\R^3)}(t^{*})dt^{*}.
	\end{equation} Analogues of \eqref{morawetz} also hold  for  spin-weighted wave equations  \cite{pasqualotto2019spin,MihalisStabExt,Maxwell3,Maxwell5}.  We note that various techniques have been developed for wave equations which take in as a ``black box'' assumption that \eqref{morawetz} holds, and then  derive polynomial point-wise decay estimates via physical space methods, see e.g.\ \cite{rp,AAG1,TataruMaxwell}.

	On a Kerr black hole, the presence of \emph{superradiance} (non-positivity of the usual energy) creates an important obstruction to local decay: nonetheless analogues of
	\eqref{morawetz} hold for  (spin-weighted) wave equations on Kerr \cite{Blue,Red,claylecturenotes,Tataru,KerrDaf,hidden,SRTdC2020boundedness,SRTdC2023boundedness2,GKS,Giorgi.RN,BlueMaxwell,TeukolskyDHR,Maspin2,Maspin1}. These integrated local energy decay estimates play a crucial role in nonlinear stability results~\cite{rp,blackbox,SchwarzschildStab,KS.polarized,klainerman2021kerr,stabilitykerrformalism,GCM,effectiveuniform,ShenGCM}.

	\subsection{Bound states and growing modes for Klein--Gordon equations}\label{intro.Kerr.section}

	It is well-known that sufficiently regular solutions of \eqref{KG.intro}  are uniformly bounded in time: This follows from the energy identity, the redshift effect~\cite{Red}, commuted versions thereof, and Sobolev inequalities. The main asymptotic question is thus whether  decay-in-time holds, or not. In this section, we will compare \eqref{KG.intro} (in which we show quantitative polynomial time decay) to other situations.

	\paragraph{Bound states for the Klein--Gordon equation on static non-black-hole spacetimes}
	It is well-known that equations of the form $\Box \phi = V(r) \phi$ for an attractive long-range potential $V$ admit \emph{bound state solutions}, see e.g.\ \cite{reedsimonIV}: a celebrated example is when $V(r)=-\frac{1}{r}$ is  the hydrogen atom potential. Klein--Gordon equations of the form $\Box_g \phi = m^2 \phi$, where $g$ is a static, asymptotically Schwarzschild spacetime, but \emph{with no event horizon} (excluding the black hole case) can also be addressed with a similar formalism and bounds states are known to exist  \cite{Ethan2}. We note that such (linear) bounds states are also  related to the existence of non-trivial stationary non-black hole spacetimes  (solving the nonlinear Einstein--Klein--Gordon equations), so-called ``boson stars'' \cite{Boson,boson2,boson3}.

	\paragraph{Absence of bound states for the Klein--Gordon equation on the Reissner--Nordstr\"{o}m black hole} In contrast, we proved in \cite{KGSchw1} that solutions of \eqref{KG.intro}   on   the Reissner--Nordstr\"{o}m black hole \emph{admit no bound states}, and in fact decay-in-time: at a specific $[t^{*}]^{-5/6}$ rate for each angular mode, and without a rate for the general solution (Theorem~\ref{thm.old}). One of the key features we exploited in \cite{KGSchw1} is the presence of a \emph{positive flux of energy} on the event horizon (which is absent on a non-black-hole spacetime) -- a phenomenon which plays an important quantitative role in our proof of Theorem~\ref{thm.intro}.  
	\paragraph{Existence of bound states and growing modes for the Klein--Gordon equation on the Kerr black hole} On the Kerr black hole, the flux of energy on the event horizon is no longer necessarily positive, due to \emph{superradiance}. In \cite{Yakov}, the first author showed the existence of exponentially growing-in-time solutions (therefore, boundedness fails, in contrast with \eqref{KG.intro}) which have a negative energy flux along the horizon, and  of  exactly time-periodic solutions which have a vanishing flux along the horizon. 
	These bound states lead to the construction of hairy black holes \cite{chodosh-sr1,chodosh-sr2} for the Einstein--Klein--Gordon equations.  Theorem~\ref{thm.intro} of course cannot hold on Kerr spacetimes which possess these non-decaying solutions; however,  it is natural to speculate that our analysis leading to expressions of the form \eqref{first.osc.sum2} remains relevant, after projecting away from the bound states/exponentially-growing modes. We   leave this question to future works.

	\subsection{Stable trapping, quasimodes and quasinormal modes}   \label{intro.trapping.section}

	\paragraph{Stable trapping and logarithmic decay for asymptotically AdS black holes}
	It is well-known that the Schwarzschild-Anti-de-Sitter spacetime $g_{SAdS}$ --the analogue of the Schwarzschild black hole in the presence of a negative cosmological constant $\Lambda<0$ -- admits stably trapped null geodesics. Exploiting this stable trapping, Holzegel--Smulevici \cite{quasimodeads} constructed quasimode solutions of the analogue of \eqref{KG.intro}:
	
	\begin{equation}\label{wave.ADS}\begin{split}
			&\Box_{g_{SAdS}} \phi = -\alpha\ \phi , \text{ with } \alpha \text{ satisfying the Breitenlohner-Freedman condition } \frac{3\alpha}{|\Lambda|} \in (\frac{3}{2},\frac{9}{4}) , \\ & g_{SAdS}=-\Big(1-\frac{2M}{r}+ \frac{-\Lambda}{3} r^2\Big) dt^2 + \Big(1-\frac{2M}{r}+ \frac{-\Lambda}{3} r^2\Big)^{-1} dr^2 + r^2\left( d\theta^2+ \sin^2(\theta) d\varphi^2\right),
		\end{split}
	\end{equation}   satisfying the quasimode conditions \ref{quasi0}, \ref{quasi2} from Theorem~\ref{quasimode.thm} but with compact support in tortoise coordinate $r^{*}$ (thus, not satisfying the analogue of condition~\ref{quasi1}); they deduced that   solutions with Sobolev initial data (even after restricting them to a compact region in $r$) in general decay no faster than \begin{equation}\label{log.decay}
		\frac{1}{\log(t^{*})} \text{ as } t^{*} \rightarrow+\infty,
	\end{equation}
	The slow logarithmic decay \eqref{log.decay} is due to \emph{the superposition of angular  modes}: Indeed,  fixed angular mode solutions of  \eqref{wave.ADS}  decay exponentially \cite{SchwAds,decayads}		 in $t^{*}$. The drastically different asymptotic behavior in time for \eqref{wave.ADS} between fixed angular modes and their superposition is to be contrasted with the situation for \eqref{KG.intro} (compare Theorem~\ref{thm.old} for the fixed angular mode and Theorem~\ref{thm.intro} for their superposition). The key difference is that $g_{SAdS}$ can be seen as a manifold with a (conformal) timelike boundary, thus the quasimodes are  contained in a \emph{fixed compact set}, whereas, in contrast, for \eqref{KG.intro} the analogous quasimodes are contained in $\{r \gtrsim L^2\}$ sets escaping to infinity as $L \rightarrow+\infty$ (where $L$ is the angular frequency), recalling Theorem~\ref{quasimode.thm}.

	\paragraph{Stable trapping and logarithmic decay on black strings/rings} The Schwarzschild black string  is a 5-dimensional analogue of the Schwarzschild black hole, obtained by adding a trivial $\mathbb{S}^1_z$ factor \cite{strings}, and   taking the  following form: 
	
	\begin{equation}\label{wave.string}\begin{split}
			g_{BS}=-\Big(1-\frac{2M}{r}\Big) dt^2 + \Big(1-\frac{2M}{r}\Big)^{-1} dr^2 + r^2\left( d\theta^2+ \sin^2(\theta) d\varphi^2\right)+ dz^2.
		\end{split}
	\end{equation}
	A function $\Phi = e^{i m z}\phi (t,r,\theta,\phi)$ will satisfy $\Box_{g_{BS}}\Phi  = 0$ if and only if $\phi$ solves the Klein--Gordon equation \eqref{KG.intro} of mass $m^2$ on the Schwarzschild black hole. 
	Otherwise said, the (4+1)-dimensional wave equation $\Box_{g_{BS}}\Phi = 0$  can be viewed as a (3+1)-dimensional Klein--Gordon equation \eqref{KG.intro}, but with an arbitrarily large Klein--Gordon mass $m^2$. 
	
	By, in particular, exploiting a relation between ``infinitely boosted'' black strings (and their Kerr analogues) and so-called black rings (see \cite{blackrings}),
	in \cite{Benomio}, Benomio showed that a large class of $5$-dimensional spacetimes comprising of   black strings \eqref{wave.string}, the Kerr analogues of black strings, and  black rings, all possess quasimodes at specific  frequencies $\omega=\omega_m \approx m $ for $ m \gg1$, which are all supported in a fixed compact set  independent of $m$: these quasimodes, in turn, lead to  logarithmic decay lower bounds \eqref{log.decay} (even after restricting  the solution to a spatially compact region) for  solutions with Sobolev initial data. Comparing  the large parameter $m$ of \cite{Benomio} to our large parameter $L$, we note that in contrast, our quasimodes $\phi_L$ are supported on a set $\{r\gtrsim L^2\}$ escaping to infinity for $L \gg 1$. We also emphasize that in the black string/ring context, the $m^2$ parameter can be  arbitrarily large, whereas the Klein--Gordon mass $m^2$ in \eqref{KG.intro} is fixed.
	
	\paragraph{Connection between quasimodes and quasinormal modes/resonances} 
	
	\eqref{log.decay}  was also obtained by Gannot   \cite{Gannot} using the  framework of resonances  (see  \cite{res_final} and references therein for precise definitions). For any resonance $\lambda \in \mathbb{C}$, one can  construct a solution $\phi_{\lambda}$  called \emph{quasinormal mode}, which will play a role in the discussion of Section~\ref{numerics.intro}.  We finally mention the close connection between quasimodes and   the existence of  infinitely  many resonances approaching the real axis, see e.g.\ \cite{quasimode4,quasimode5,quasimode1,quasimode2,quasimode3}.

	\subsection{Numerics, heuristics, and other massive matter models}\label{numerics.intro}

	\paragraph{Decay for the Klein--Gordon equation and other massive matter models on a black hole for fixed angular mode}
	The late-time behavior of \eqref{KG.intro} is a classical subject in the Physics literature, starting with the study of its quasinormal modes \cite{massive1,massive2,QNM4,QNM5} for $L=0$, or $L=1$. Late-time tails were first studied heuristically in \cite{KoyamaTomimatsu,KoyamaTomimatsu2,KoyamaTomimatsu3} claiming a $O( [t^*]^{-5/6})$ rate  \emph{for each fixed angular mode}, consistently with Theorem~\ref{thm.old}. Numerical studies \cite{BurkoKhanna,HodPiran.mass} have later confirmed these asymptotics. We  mention analogous works for dilaton massive fields, \cite{Dilaton1,Dilaton2},  Dirac massive fields \cite{massiveDirac,massivedirac2}, and vector-valued massive (a.k.a.\ Proca) fields \cite{QNM1,Proca2}, where late-time tails were also found to be   $O( [t^*]^{-5/6})$. For \eqref{KG.intro} on a  Kerr  black hole, there is no decay for general solutions due to exponentially growing modes 	(see \cite{massive0,super1,super2} for numerics); However,  $O( [t^*]^{-5/6})$ late-time were found numerically  \cite{Burko,Konoplya.Zhidenko.num,KonoplyaZhidenko},  for solutions (presumably) belonging to a sub-space orthogonal to the exponentially growing solutions.
	It would be interesting to apply the techniques behind the proof of Theorem~\ref{thm.intro}  in the context of these other massive matter fields, a problem we leave to future works.

	\paragraph{Stable trapping}   We, however, emphasize that the above works focus on fixed angular modes. To the best of our knowledge, the only previous work discussing the full solution (i.e.\ the superposition of angular modes) is \cite{SantosWCC} and its announced companion paper \cite{Santos.unpublished}. The first part of the scenario put forth in \cite{SantosWCC} asserts the presence of  quasinormal modes $\phi_L$ approaching the real axis at the rate $O(e^{-S_0 L})$ for  $S_0>0$, as $L\rightarrow +\infty$ (due to stable trapping), which is consistent with  Theorem~\ref{quasimode.thm}. The second part of the \cite{SantosWCC} scenario leads to the claim of very slow decay of the solution of \eqref{KG.intro} at a rate \begin{equation}\label{loglog}
		\frac{1}{\log( \log(t^{*}))} \text{ as } t^{*} \rightarrow+\infty.
	\end{equation}
	
	Theorem~\ref{thm.intro} (with Remark~\ref{rmk.general}, Statement~\ref{uniformrdecaytime}) disproves decay at the rate \eqref{loglog} for solutions of \eqref{KG.intro} with compactly supported initial data (and also exponentially decaying data, recalling Remark~\ref{rmk.general}, Statement~\ref{exp.statement}). However, asymptotics of the form \eqref{loglog} (in a large $r$ region) for solutions arising from initial data with polynomial decay are in principle consistent with  the aforementioned quasimode construction from Theorem~\ref{quasimode.thm}, Statement~\ref{quasimode.statement}, which provides logarithmic lower bounds of the following form $$ \frac{1}{[\log(t^{*})]^q}  \text{ as } t^{*} \rightarrow+\infty \text{ for some } q>0$$ in a far-away $\{ \log^2(t^{*}) \ls r \ls \log^{p}(t^{*})\}$ region, for some $p>2$.

	\paragraph{The Schwarzschild black hole as a possible end-point of the superradiant instability}
	Given the existence of  superradiant instabilities  for the Klein--Gordon equation on (certain) rotating Kerr spacetimes, it is natural to ask what is the nonlinear endpoint of the instability for the Einstein--Klein--Gordon system. In the early works~\cite{Superradiance} and~\cite{Detweiler} it was speculated that the massive scalar field would cause the black hole to eventually eject its angular momentum via superradiance and settle down to a non-rotating Schwarzschild black hole. Other possible endstates could include the later discovered hairy black hole solutions~\cite{numerics.Kerr,chodosh-sr1,chodosh-sr2}, or, in principle, certain special rotating Kerr black holes which are immune to the superradiant instability associated to a given Klein--Gordon mass. However, the works~\cite{hair.unstable,SantosWCC} present numerical and heuristic arguments that both the hairy black hole solutions and \emph{all} rotating Kerr spacetimes are superradiantly unstable. Since these works also argued that the scalar field decayed  very slowly on a Schwarzschild black hole (see the discussion above~\eqref{loglog}) they speculated that nonlinear evolution may drive the solution towards a naked singularity.
	
	Our main result Theorem~\ref{thm.intro} is, of course,  insufficient to address the nonlinear problem. However, insofar as it (and Remark~\ref{rmk.general}, Statement~\ref{uniformrdecaytime}) invalidates \eqref{loglog} for compactly supported (and exponentially decaying) initial data, it suggests that, in the case of sufficiently localized initial data, one should not discard so lightly   the possibility that the Schwarzschild black hole plays an important role in the long-term  nonlinear evolution of the superradiant instability.\footnote{Of course, in vacuum, the Schwarzschild black hole is ``unstable'' to perturbations to a nearby Kerr, so any sense in which Schwarzschild is to be an ``endstate'' for the Einstein--Klein--Gordon system would have to involve nonlinear considerations and a genericity assumption which leads to a non-vanishing scalar field.}

	\subsection{Short summary of the proof}\label{intro.proof.section} As already discussed above, our proof will exploit the complete separability  of \eqref{KG.intro}, allowing to  reduce the problem to the  equation \eqref{eq.intro} for any fixed time frequency $\omega \in \R$ and angular frequency $L\in \mathbb{N}$. Denoting $u_L= r\hat{\phi}_L$, \eqref{eq.intro} takes the following  Schr\"{o}dinger form, introducing the coordinate $\frac{ds}{dr} = (1-\frac{2M}{r}+ \frac{\DD^2}{r^2})$:\begin{align}\label{Schrodinger2}
		\frac{d^2  u_L}{ds^2} =  \left(-\omega^2  +  \underbrace{ \Big(1-\frac{2M}{r}+\frac{\DD^2}{r^2}\Big)\Big(m^2 + \frac{L (L+1)}{r^2} +\frac{2rM -2\DD^2}{r^4}\Big)}_{=V(r)}\right) u_L+ H_L.
	\end{align} Our   strategy is to analyze \eqref{Schrodinger2} at fixed $\omega \in \R$, and characterize its singularities in $\omega$ to then  use \eqref{FT.intro} (with a stationary phase argument) to obtain time-decay for the solution $\phi$ of \eqref{KG.intro}.
	\paragraph{The geometry of the Reissner--Nordstr\"{o}m metric and the potential $V(r)$} Note that $s(r) \sim r $ as $r\rightarrow+\infty$. It is also well-known that  $s \in \R$ behaves in the following way: there exists a constant $\kappa_+(M,e)>0$ (the surface gravity of $\mathcal{H}^+$) such that near the black hole event horizon $\mathcal{H}^+=\{r=r_+(M,e)\}$: \begin{equation}
		e^{ 2\kappa_+ s } \approx (r-r_+) \text{ if r } \approx r_+. \text{ In particular, } \mathcal{H}^+=\{s=-\infty\} .
	\end{equation} Thus,  the coordinate $s$ turns \eqref{eq.intro} 
	into the scattering problem \eqref{Schrodinger2} on the real-line $s\in (-\infty,+\infty)$. It is easy to see that the potential $V(r)$ has the following asymptotics near $s=-\infty$ and $s=+\infty$ respectively \begin{align}\label{pot.asymp}
		&V(s) \approx e^{2\kappa_+ s} r_+^{-2} L^2 \text{ as } s \rightarrow -\infty,\\   &V(s) \approx m^2 -\frac{2Mm^2}{r(s)}+\frac{L^2}{r^2(s)} +O(s^{-2}) \text{ as } s \rightarrow +\infty, 
	\end{align} where in the above (as in the rest of this section) we treat $L \gg 1$ as a large parameter. \paragraph{The three frequency-regimes} \eqref{Schrodinger2} indicates that the frequencies at which $u_L$ is the most singular are $\omega =\pm m$: most of our analysis indeed focuses on the vicinity of $\omega =\pm m$. We distinguish  three regimes: \begin{enumerate}[A.]
		\item \label{A}$|\omega|>m$: This is the most favorable regime, in which it is in principle possible to establish uniform bounds for an arbitrary number of $\omega$ derivatives of $u$, see \cite{Ethan1}. The key realization is that the main part of the effective potential $m^2-\omega^2-\frac{2Mm^2}{r}$ (for large $r$) keeps a  strictly negative sign as $\omega \underset{|\omega|>m}{\to} m$. This allows one to introduce a  ``WKB-coordinate'' and  use a multiplier  in the new coordinate system. 
		
		\item \label{B}$|\omega|<m$: In the low-frequency regime, we note by \eqref{pot.asymp} that $-\omega^2 +V(r)$ admits 3 \emph{turning points}: \begin{itemize}
			\item (Horizon turning point) At  $s\approx\frac{\log(\frac{ |\omega|}{L})}{2\kappa_+}\ll -1 $: As can be seen from its location   near the event horizon $\mathcal{H}^+=\{s=-\infty\}$, this turning point is unaffected by how close $|\omega|$ is to $m$, and can be addressed using the same techniques as for the massless wave equation $m^2=0$: we indeed rely on the previous literature \cite{Schlag_exp,Schlag2} (using classical analysis) and \cite{Red,KerrDaf} (with microlocal multipliers) in this region, exploiting the exponentially decaying character (red-shift) of   $V(s)$ as $s\rightarrow -\infty$.
			
			\item  (``Compact region'' turning point, if $0<m^2 - \omega^2 \ls  L^{-2}$) At $s\approx L^2 \gg 1$ (assuming $L\gg 1$): Note from Figure~\ref{Potential} that to the left of this turning point ($s\ls L^2$) lies a classically forbidden region.
			
			\item (Near-infinity turning point, if $0<m^2 - \omega^2 \ls  L^{-2}$) At $s\approx (m^2-\omega^2)^{-1}\gtrsim L^2 $:  Note from Figure~\ref{Potential} that to the right of this turning point ($s \gtrsim (m^2-\omega^2)^{-1}$) lies another classically-forbidden region. 
		\end{itemize}

		Now, we discuss the two sub-regimes of the $|\omega|<m$ regime.
		\begin{enumerate}[1.]
			\item \label{i}  $|\omega| < m - L^{-p}$, for a large $p>2$. 
			We show this frequency regime is not the main contributor to the $O( [t^{*}]^{-5/6})$ decay as it produces a term decaying as $O( [t^{*}]^{-1+})$. Note  that  this regime includes a subregime when $m^2 - \omega^2 \approx L^{-2}$ corresponding to the coalescence of the two ``large $s$'' turning points, and also the regime  $m^2 - \omega^2 \gtrsim  L^{-2}$, where the two ``large $s$'' turning points disappear.
			
			\item \label{ii}  $m - L^{-p}<|\omega| < m$. This regime turns out to be the most singular, and the slowest contributor to the  $O( [t^{*}]^{-\frac{5}{6}})$ decay.
		\end{enumerate}

	\end{enumerate}  In what follows, we briefly discuss our methods of proof for each regime separately.
	
	\paragraph{Large frequencies (Regime~\ref{A})} On the technical level, our proof in the $|\omega|> m$ regime relies on the introduction of  a WKB coordinate $\zeta(s)$ (see \cite{olver}), recalling that $|V(r)-\omega^2| \approx \omega^2-m^2 + \frac{2Mm^2}{r}$ for $r\gg1$: \begin{align}\label{WKB.intro}
		&\frac{d\zeta}{ds} = \sqrt{ \omega^2-m^2+\frac{2Mm^2}{r(s)}} \text{ for large  } s\gg 1,\\ & v_L= | \omega^2-m^2+\frac{2Mm^2}{r(s)}|^{-1/4} u_L \text{ for large  } s\gg 1, \label{WKB.intro2}
	\end{align} The terminology comes from the (approximately) $\sqrt{|V-\omega^2|}$ factor in \eqref{WKB.intro} (and the $|V-\omega^2|^{-1/4}$ factor in  \eqref{WKB.intro2}) reminiscent of the classical WKB analysis. 
	The ODE satisfied by $v_L$ in $\zeta$ coordinates yields an effective potential $W(\zeta)\approx 1$   near infinity (see \cite{olver}), which gets rid of  singular behavior near infinity.
	Exploiting the simple form of the ODE, we use a microlocal multiplier method to prove that $v_L \in W^{1,2-}_{\omega}((m,2m])$, which  gives  $O([t^{*}]^{-1})$ decay for its  inverse-Fourier transform. It is likely that smoothness in $\omega$ of $v_L$ can also be obtained using our microlocal multiplier method (cf.~\cite{Ethan1} who obtained smoothness for the cutoff resolvent $R\left(E\pm i0\right)$ in the $E \to 0^+$ limit for Schr\"{o}dinger operators  with Coulomb-like potentials on general asymptotically conic manifolds using different methods), though we do not pursue this as it would not improve our final result as stated in Theorem~\ref{thm.intro}.

	\paragraph{Low-frequencies away from $m$ (Regime~\ref{i})} The main delicate analysis in the $|\omega |< m - L^{-p}$  regime  takes place when   $L^{-p}< m-|\omega |\ls  L^{-2}$ (presence of two ``large $s$'' turning points). Our strategy relies on the analysis near infinity of a hydrogen-atom like elliptic operator of the following form \begin{equation}\label{elliptic.intro}
		-\frac{d^2}{dr^2} -\frac{2M m^2}{r^2} + \frac{L(L+1)}{r^2} + O(r^{-2}), \text{ with eigenvalues } \lambda_n = \frac{(Mm^2)^2}{n+L^2} + O\left(e^{-DL}\right),\ n\in \mathbb{N}.
	\end{equation} 
	
	We show point-wise  regularity estimates in $\omega$ for $v_L$ \emph{away from these eigenvalues}, more precisely if \begin{equation}\label{close.intro}
		|m^2-\omega^2 -\lambda_n |\gtrsim e^{-\frac{L}{2} \log L},
	\end{equation} using microlocal multipliers. When \eqref{close.intro} is violated, we obtain weaker point-wise regularity estimates in $\omega$; however, because the $\omega$-size of these intervals  is $O( e^{-\frac{L}{2} \log L} )$ small, we get \emph{improved integrated} regularity in $\omega$, which allows us to show  $O([ t^{*}]^{-1})$ decay of the corresponding inverse-Fourier transform (after applying a  stationary phase argument).
	However, in the regime $m-|\omega |\ls e^{-C L}$, \eqref{elliptic.intro} is a poor approximation of the actual elliptic operator in \eqref{Schrodinger2}, due to $\exp( O(L))$ errors (see already Remark~\ref{approx.poor}). Therefore, our multiplier method breaks down, which will warrant a more tailored analysis in the regime $m-|\omega|\ll L^{-p}$.
	
	\paragraph{Low-frequencies close to  $m$ (Regime~\ref{ii})} To address the regime $0<m-|\omega| \ls L^{-p}$, we revert to classical analysis, based on the WKB method and its extensions (see e.g. \cite{olver}). Our goal is to show, up to errors, that $u_L$ has the schematic form given by \eqref{resolvent.intro}. The main ideas of the WKB-analysis are as follows: \begin{itemize}
		\item the oscillating terms $\cos(\pi k)$, $\sin(\pi k)$ in \eqref{resolvent.intro} come from the transition between the second turning point  ($s\approx L^2$), and the third turning point at ($s\approx k^2$)  occurring in a classically allowed region (with  oscillatory behavior of the WKB solutions, see \cite{olver}). The key quantity in this WKB-transition is \begin{equation}
			\pi k(\omega) \approx \pi (m^2-\omega^2)^{-1/2} \approx \int_{L^2}^{k^2} |V|^{1/2}(s )ds \gg 1.
		\end{equation}
		\item The small constant  $\ep_L \approx \exp(-2L \log L + O(L)) $ in \eqref{resolvent.intro} comes from the transition between the second turning point ($s\approx L^2$) and the compact region ($s\ls 1$) on which the initial data is supported  occurring  in a classically forbidden region. This context is that of the wave penetration through a barrier (see \cite{olver}), with exponential behavior of the WKB solutions.  The key quantity in this WKB-transition is \begin{equation} 
			\int_{1}^{L^2} |V|^{1/2}(s )ds \approx \int_{-\infty}^{L^2} |V|^{1/2}(s )ds \approx L \log  L + O(L), \text{ with } L \gg 1.
		\end{equation}We also emphasize that the \emph{energy identity} is crucial to get the lower bound  $\ep_L\gtrsim   \exp(-2L \log L + O(L)) $. 
		
		\item The  $\exp(L \cdot f_{data}(r))$ term in \eqref{resolvent.intro} comes from the fundamental solutions of \eqref{Schrodinger2}  behaving as $$ r^{-L}=\exp(- L \log r)\ \& \ r^{L+1}=\exp( (L+1) \log r) \text{ in the compact } r \text{ region.}$$

	\end{itemize}
	The goal is  to take the inverse-Fourier transform of \eqref{resolvent.intro} to deduce time-decay,  using stationary phase. \begin{itemize}
		\item In \eqref{resolvent.intro}, if $\pi k$ is  far from an integer, \eqref{resolvent.intro} is of size $\exp(-4L \log L + O(L))$, which is summable in $L$.
		\item If $\pi k\in \mathbb{N}$,  \eqref{resolvent.intro} is of size $\exp(L \cdot f_{data}(r))$, which blows-up as $L\rightarrow +\infty$. However, \eqref{resolvent.intro} is pointwise large only if $|\pi k(\omega) - N| \ls \exp( -  L \log L)$, which is a small $\omega$-size region: thus, we show that the integration-in-$\omega$ from inverse-Fourier transform \emph{yields an estimate that it is summable in $L$}.
	\end{itemize}
	Our main result in the range $0<m-|\omega|\ls L^{-p}$ (Regime~\ref{ii}) is obtained in Theorem~\ref{TPsection.mainprop}, see   Section~\ref{regimeB2.section}. 	A detailed guide to the proof of Theorem~\ref{TPsection.mainprop}  is also given in subsection~\ref{strategy.section}.
	\subsection{Outline of the rest of the paper} In Section~\ref{Setup.section}, we provide the equations/transformations we will be using throughout the paper.
	The Fourier transform at fixed $\omega$ is estimated using ODEs: in Section~\ref{regimeA.section} (high frequencies $|\omega|>m$, Regime~\ref{A});  in Section~\ref{regimeB1.section} (low frequencies away from $m$: $ m-|\omega|\gtrsim L^{-p}$, Regime~\ref{i});  in Section~\ref{regimeB2.section} (low frequencies near $m$: $ m-|\omega|\ls L^{-p}$, Regime~\ref{ii}). In particular, the proof of Theorem~\ref{quasimode.thm} follows from the analysis in Section~\ref{regimeB2.section}.  In Section~\ref{fourier.section}, we take the inverse-Fourier transform of the solution previously analyzed in Sections~\ref{regimeA.section}, ~\ref{regimeB1.section}, and \ref{regimeB2.section} and deduce decay-in-time, concluding the proof of Theorem~\ref{thm.intro}.  Appendix~\ref{bessel.section} contains  a discussion of Bessel functions,  Appendix~\ref{volterra.section} of Volterra equations   and  Appendix~\ref{appendix.conj} of the Exponent Pair Conjecture.
	
	\subsection{Acknowledgments} MVdM thanks Sagun Chanillo and Henryk Iwaniec for insightful conversations on exponential sums and the exponent pair conjecture. 
	YS acknowledges support from an Alfred P. Sloan Fellowship in Mathematics and from NSERC discovery grants RGPIN-2021-02562 and DGECR-2021-00093.	MVdM gratefully acknowledges support from the NSF Grant	 DMS-2247376.
	\section{Preliminaries}\label{Setup.section}
	\paragraph{Equations in physical space}	To be as general as possible, we consider \eqref{KG.intro} with a  source $F$, i.e.\ \begin{equation}\label{KG.prelim}
		\Box_{g_{RN}} \phi =m^2 \phi  + F.
	\end{equation} \eqref{KG.prelim} in  $(t,r,\theta,\varphi)$ coordinates gives, defining $\psi=r\phi=  \sum_{L \in \mathbb{N}} \psi_L(t,r) Y_L(\theta,\varphi)$, where $Y_L(\theta,\varphi)$ is the spherical harmonics on $\mathbb{S}^2$: writing also  $F=\sum_{L \in \mathbb{N}} F_L(t,r) Y_L(\theta,\varphi)$, we  obtain
	\begin{equation}\label{eq}\begin{split}
			&	- \p_t^2 \psi_L + 		(1-\frac{2M}{r}+\frac{\DD}{r^2})\frac{d}{dr}\left( (1-\frac{2M}{r}+\frac{\DD}{r^2})^{-1}\frac{d}{dr} \psi_L \right)  \\ = & \Big(1-\frac{2M}{r}+\frac{\DD^2}{r^2}\Big)\Big(m^2 + \frac{L (L+1)}{r^2} +\frac{2rM -2\DD^2}{r^4}\Big)\psi + \underbrace{H_L}_{:= 	(1-\frac{2M}{r}+\frac{\DD}{r^2}) F_L}.\end{split}
	\end{equation} Introducing a coordinate $s \in (-\infty,+\infty)$ satisfying \begin{equation}\label{s.definition}
		\frac{ds}{dr}=  1-\frac{2M}{r}+\frac{\DD}{r^2},
	\end{equation} \eqref{eq} becomes

	\begin{equation}\label{eq:main}\begin{split}
			- \p_t^2 \psi_L + \p^2_{s} \psi_L   =  \Big(1-\frac{2M}{r}+\frac{\DD^2}{r^2}\Big)\Big(m^2 + \frac{L (L+1)}{r^2} +\frac{2rM -2\DD^2}{r^4}\Big)\psi_L + H_L.\end{split}		\end{equation}
	
	\paragraph{Fourier transform}	Now, letting $u_L(\omega,s) = \int_{\R} e^{i\omega t} \psi_L(t,s) dt$, we obtain that $u_L$ satisfies the equation

	\begin{equation}\label{eq:mainVrhs}
		\p^2_{s} u -V(s) u = H, 
	\end{equation} where $H= H_L$, or 	\begin{equation}\label{eq:mainV}
		\p^2_{s} u_L =V(s) u_L 
	\end{equation} without the source $H_L$, where \begin{equation}
		V(s) = -\omega^2+ \Big(1-\frac{2M}{r}+\frac{\DD^2}{r^2}\Big)\Big(m^2 + \frac{L (L+1)}{r^2} +\frac{2rM -2\DD^2}{r^4}\Big).
	\end{equation} 
	When $\omega > m$, we say a solution $u$ to~\eqref{eq:mainVrhs} satisfies outgoing boundary conditions if 
	\begin{equation}\label{outgoingboundary}
		u \sim e^{-i\omega s}\text{ as }s\to -\infty,\qquad u\sim r^{\frac{iMm^2}{\sqrt{\omega^2-m^2}}}e^{i\sqrt{\omega^2-m^2}s}\text{ as }s\to \infty,
	\end{equation}
	where, as usual, the boundary condition as $s\to\infty$ must be interpreted as an asymptotic expansion in $r$.
	
	When $m > \omega$,  we say a solution $u$ to~\eqref{eq:mainVrhs} satisfies outgoing boundary conditions if 
	\begin{equation}\label{outgoingboundary2}
		u \sim e^{-i\omega s}\text{ as }s\to -\infty,\qquad u\sim r^{\frac{Mm^2}{\sqrt{m^2-\omega^2}}}e^{-\sqrt{m^2-\omega^2}s}\text{ as }s\to \infty,
	\end{equation}
	
	For two solutions $f(s)$ and $h(s)$ to the equation~\eqref{eq:mainV}, we define the corresponding Wronskian by
	\[W\left(f,h\right) \doteq \partial_sf h - f\partial_sh.\]
	Since $\partial_sW = 0$, the Wronskian is a constant.
	\paragraph{Reissner--Nordstr\"{o}m black hole geometry} We define $r_{\pm}(M,e)$ as the two roots of the polynomial $r^2(1-\frac{2M}{r}+\frac{\DD^2}{r^2})$ under the assumption $0\leq |e|<M$: they are both positive, distinct, and the largest $r_+$ correpsonds to the radius of the event horizon $\mathcal{H}^+=\{r=r_+(M,e)\}$:
	\begin{equation}\label{thisisrplus}
		r_{\pm}+ \doteq M \pm+ \sqrt{M^2 - \mathbf{e}^2}.
	\end{equation}

	Noting that \eqref{s.definition} only defines $s$ up to a constant, we fix this constant and define define \begin{equation}\label{s.def}
		s= r+ (2\K)^{-1} \log( [r-r_-][r- r_+]),
	\end{equation}
	where $\kappa_+>0$ is the surface gravity of the event horizon defined as:
	\[\kappa_+ \doteq \frac{1}{2}\frac{d}{dr}\left(1-\frac{2M}{r} + \frac{e^2}{r^2}\right)|_{r=r_+}.\]
	Note that for any $q\in \mathbb{N}$: \begin{equation}\label{s.r.diff}
		r^{-q} - s^{-q} = \frac{ q}{\K} s^{-q-1} \log(s) + O( s^{-q-2}\log(s)) 
	\end{equation}
	We also note that \begin{equation}\label{C+.def}
		1-\frac{2M}{r}+\frac{\DD^2}{r^2} \sim  \underbrace{\frac{e^{-2\K r_+}}{r_+^2}}_{:= C_+>0} e^{2\K s} \text{ as } s \rightarrow -\infty.
	\end{equation}
	
	\paragraph{Horizon-penetrating coordinates} Since the $t$ coordinate degenerates at the event horizon, we introduce the $\left(t^*,r,\theta,\varphi\right)$ coordinate system just as in Section 2.1 of \cite{KGSchw1}. We will have
	\begin{equation}\label{thisisps}		
		t^* = t + p(s),
	\end{equation}
	for a suitable function $p(s)$,
	See~\cite{KGSchw1} for the specific formulas.
	
	\paragraph{Conventions}	We introduce the following two conventions throughout the paper:
	\begin{enumerate}
		\item Unless said otherwise, all implied constants may depend on $m$, $M$, or $\mathbf{e}$. In particular, we use the notation $A \ls B$ to imply that there exists a constant $C(m,M,\DD)>0$ such that $A \leq C B$.
		
		\item We also use the notation $A\approx B$ if there exist constants $C^{\pm}(m,M,\DD)>0$ such that $ C^- B \leq A \leq C^+ B$.
		\item Unless said otherwise, the spherical harmonic number $L$ may be assumed sufficiently large relative to any other constants which are introduced.
	\end{enumerate}

	\section{The high-frequencies}\label{regimeA.section}
	In this section we will study solutions to~\eqref{eq:mainVrhs} satisfying outgoing boundary conditions~\eqref{outgoingboundary} under the additional assumptions that

	\begin{equation}\label{omegabiggerthanm}
		\omega - m > 0,
	\end{equation}
	\begin{equation}\label{Llarge}
		L \gg 1.
	\end{equation}

	We have two main results for this section. The first is the following:
	\begin{prop}\label{almosthereexcepthorizon}
		Let $u$ be a solution to \eqref{eq:mainVrhs} satisfying outgoing boundary conditions~\eqref{outgoingboundary} under the additional assumptions that~\eqref{omegabiggerthanm} and~\eqref{Llarge} hold. Furthermore, assume that there exists a large positive constant $R_0$ so that $H$ vanishes for $r \geq R_0$. Then, for any $\mathring{s} > 0$, we have that 
		\begin{align}\label{toprovewhenomegaissortofbig}
			&\sum_{j=0}^1\sup_{s \in [-\mathring{s},\mathring{s}]}\left|\partial^j_{\omega}u\right|^2 
			\\ \nonumber &\qquad \lesssim_{\mathring{s},R_0}
			L^6 \left(\omega^8+\left(\omega^2-m^2\right)^{-1/2-5\delta}\right)\int_{-\infty}^{\infty}\left(1+|s|\right)^{1+\delta}\left[\left|H\right|^2 +   [1+(-s)_+^2]\left|\partial_sH\right|^2 + \left|\partial_{\omega}H\right|^2\right]\, ds
		\end{align}
	\end{prop}
	
	Our second estimate establishes an estimate which is uniformly up to $r=r_+$.
	\begin{prop}\label{highfreqest} Let $u$ be a solution to \eqref{eq:mainVrhs} satisfying outgoing boundary conditions~\eqref{outgoingboundary} under the additional assumptions that~\eqref{omegabiggerthanm} and~\eqref{Llarge} hold. Furthermore, assume that there exists a large positive constant $R_0$ so that $H$ vanishes for $r \geq R_0$. Then 
		
		\begin{align*}
			&\sum_{j=0}^1\sup_{s \leq 0}\left|\partial^j_{\omega}\left(e^{i\omega s}u\right)\right|
			\\ \nonumber &\qquad \lesssim_{R_0}  L^6\left(\omega^8+\left(\omega^2-m^2\right)^{-1/2-5\delta}\right)\int_{-\infty}^{\infty}\left(1+|s|\right)^{1+\delta}\left[\left|H\right|^2 + [1+(-s)_+^2]\left|\partial_sH\right|^2 + \left|\partial_{\omega}H\right|^2\right]\, ds
			\\ \nonumber &\qquad \qquad +   L^4 \omega^8\int_{-\infty}^{\infty}\left[\left|\partial_{\omega}\left(e^{i\omega s}{H}\right)\right|^2 + \left|{H}\right|^2\right](r-r_+)^{-1}\, ds.
		\end{align*}
	\end{prop}

	\subsection{Critical points of the potential}\label{acriticalsection}
	It is useful to split the potential $V$ into three parts
	\[V = -\left(\omega^2-m^2\right) + L(L+1)V_{\rm main} + V_{\rm junk},\]
	\[V_{\rm main} \doteq r^{-2}\left(1-\frac{2M}{r} + \frac{e^2}{r^2}\right),\]
	\[V_{\rm junk} \doteq \left(-\frac{2M}{r}+ \frac{\mathbf{e}^2}{r^2}\right)m^2 + \left(1-\frac{2M}{r} + \frac{e^2}{r^2}\right)\left(\frac{2rM - 2\mathbf{e}^2}{r^4}\right).\]
	The point of this decomposition is that $V_{\rm junk}$ does not contain any parameters which may be large. (For large $r \gg 1$ however, we have $V_{\rm junk} \sim -\frac{2Mm^2}{r}$ which,  when $r \gg L$m will dominate $V_{\rm main}$.)
	
	The key properties of $V_{\rm main}$ are captured by the following lemma whose straightforward proof we omit.
	\begin{lemma}\label{critv}The function $V_{\rm main}(r) : [r_+,\infty) \to (0,\infty)$ has a unique critical point at 
		\[r_{\rm crit} \doteq  \frac{3M + \sqrt{9M^2-8Q^2}}{2},\]
		where it has a non-degenerate maximum.
		
		For $r \in [r_+,r_{\rm crit}]$ we have the lower bound
		\[\frac{dV_{\rm main}}{dr} \gtrsim \left|r-r_{\rm crit}\right|,\]
		and for $r \in [r_{\rm crit},\infty)$ we have the lower bound
		\[-\frac{dV_{\rm main}}{dr} \gtrsim r^{-4}\left|r-r_{\rm crit}\right|\]
	\end{lemma}
	As it well known, the structure of the critical points of $V_{\rm main}$ is  connected to the fact that all trapped null geodesics orbits are unstable.
	\subsection{The WKB coordinate $\xi$}
	It will be convenient to modify the $s$ coordinate when $r$ is large via a standard WKB approximation. We start by defining a preliminary function $P(\omega,s)$
	\begin{defn}\label{thisisP}Let $S_{\rm \xi}$ be a sufficiently large constant to be fixed later. Then we define the function $P\left(\omega,s\right)$ by
		\begin{enumerate}
			\item For $s \leq S_{\xi}$ we set $P = 1$.
			\item For $s \in [2 S_{\xi},\infty)$ we set $P = \omega^2-m^2+\frac{2Mm^2}{r}$. 
			\item For $s \in [S_{\xi},2  S_{\xi}]$ we just require that $P \geq \frac{1}{2}{\rm min}\left(1,\omega^2-m^2+\frac{2Mm^2}{r}\right)$, $P$ extends to a smooth function of $(s,\omega) \in [S_{\xi},2 S_{\xi}] \times (m-\delta_0,\infty)$ for some $\delta_0 > 0$ (possibly depending on $S_{\xi}$), $\left|\partial_s\log P\right| \lesssim S_{\xi}^{-1}\log(\omega)\omega$, and, $-\partial_s \log P \lesssim \omega^{-2}r^{-2}$. We also require that in this region $\left|\partial_{\omega}P\right| \lesssim \omega$.
		\end{enumerate}
	\end{defn}
	
	Next we use the function $P$ to define our modified coordinate $\xi$.
	\begin{defn}\label{thisisXi}Let $P$ be as in Definition~\ref{thisisP}. We then define $\xi\left(\omega,s\right)$ by
		\[\frac{\partial \xi}{\partial s} = \sqrt{P},\qquad \xi|_{s=0} = 0.\]
	\end{defn}
	
	In the next lemma we transform the equation~\eqref{eq:mainVrhs} into a convenient form in the $\xi$-coordinate.
	\begin{lemma}\label{neweqninxicoord}Let $u$ solve~\eqref{eq:mainVrhs} and that~\eqref{omegabiggerthanm} holds, and $P$ and $\xi$ be as in Definitions~\ref{thisisP} and~\ref{thisisXi}. Then, if we define
		\[v \doteq P^{1/4}u,\]
		we have that
		\[\partial_{\xi}^2v -Wv = P^{-3/4}H,\]
		where
		\[W \doteq W_{\rm main} + W_{\rm error},\]
		\[W_{\rm main} \doteq P^{-1}V,\]
		\[W_{\rm error} \doteq - P^{-3/4}\partial_s^2\left(P^{-1/4}\right).\]
	\end{lemma}
	\begin{proof}This is a straightforward calculation.
	\end{proof}

	It is convenient to introduce a new definition of outgoing boundary conditions for the $\xi$-variable which simply translates the original outgoing condition~\eqref{outgoingboundary} into the $\xi$-coordinate.
	\begin{defn}\label{outxi} We say that $v\left(\xi,\omega\right)$ is an outgoing solution if
		\begin{equation}\label{outgoingeqnforv}
			\partial_{\xi}^2v - Wv = \check{H},
		\end{equation}
		for $W$ as in Lemma~\ref{neweqninxicoord} and any function $\check{H}$, and if $v \sim e^{i\xi}$ as $\xi \to \infty$ and if $v \sim e^{-i\omega \xi}$ as $\xi \to -\infty$.
	\end{defn}
	
	In the following lemma, we provide a useful estimate for $\xi$.
	\begin{lemma}\label{somestuffaboutthexicoor}When $s \geq 2S_{\xi}$ there exists constants $c_0 > 0$ and $c_1 > 0$ and functions $d_0(\omega)$ and $d_1(\omega)$ so that 
		\[c_0r\left(\omega^2-m^2 + \frac{2Mm^2}{r}\right)^{1/2} + S_{\xi}d_0(\omega) \leq \xi \leq c_1r\left(\omega^2-m^2 + \frac{2Mm^2}{r}\right)^{1/2} + S_{\xi}d_1(\omega),\]
		\[\left|d_0\right| + \left|d_1\right| \lesssim \omega.\]
		The bounds for the constants $c_0$ and $c_1$ and the functions $d_0$ and $d_1$ are independent of $S_{\xi}$.
	\end{lemma}
	\begin{proof}We have
		\[\xi\left(s,\omega\right) = \int_{2S_{\xi}}^s\left(\omega^2-m^2 + \frac{2Mm^2}{R}\right)^{1/2}\, dS + D\left(\omega\right),\]
		for a function $D\left(\omega\right)$ which satisfies
		\[\left|D\left(\omega\right)\right| \lesssim S_{\xi}\omega.\]
		We then have
		\[\left(2Mm^2\right)^{1/2}\int_{2S_{\xi}}^sR^{-1/2}\, dS + D(\omega) \leq \xi\left(s,\omega\right) \leq  (2Mm^2)^{1/2}\left(\frac{r(\omega^2-m^2)}{2Mm^2}+1\right)^{1/2}\int_{2 S_{\xi}}^s R^{-1/2}\, dS + D(\omega) \Rightarrow\]
		\[\left(2Mm^2\right)^{1/2}\left(s^{1/2} + O\left(S_{\xi}^{1/2}\right)\right)+ D(\omega) \leq \xi\left(s,\omega\right) \leq  (2Mm^2)^{1/2}\left(\frac{r(\omega^2-m^2)}{2Mm^2}+1\right)^{1/2}\left(s^{1/2} + O\left(S_{\xi}^{1/2}\right)\right) + D(\omega).\]
		This establishes the lemma in the case when $r\left(\omega^2-m^2\right) \lesssim 1$.
		
		Now let's consider the case when $r\left(\omega^2-m^2\right) \gg 1$. In this case we start by observing that 
		\begin{align*}
			&\frac{1}{2}\left(\omega^2-m^2\right)^{1/2}\int_{r\left(2S_{\xi}\right)}^r\left(1+\frac{2Mm^2}{R\left(\omega^2-m^2\right)}\right)^{1/2}\, dR + D(\omega)
			\\ \nonumber &\qquad \leq \xi\left(s,\omega\right) \leq 2\left(\omega^2-m^2\right)^{1/2}\int_{r\left(2 S_{\xi}\right)}^r \left(1+\frac{2Mm^2}{R\left(\omega^2-m^2\right)}\right)^{1/2}\, dR + D(\omega).
		\end{align*}
		Now we do a change of variables $X = R\left(\omega^2-m^2\right)$ in the integral and use that $r(\omega^2-m^2) \gg 1$ to obtain that
		\[\frac{1}{10}r\left(\omega^2-m^2\right)^{1/2} + D(\omega) \leq \xi \leq 10 r\left(\omega^2-m^2\right)^{1/2} + D(\omega).\]
		This is easily seen to complete the proof of the lemma.
	\end{proof}
	
	Next we provide an estimate for the terms $W_{\rm error}$ and $W_{\rm main}$ from Lemma~\ref{neweqninxicoord}.
	\begin{lemma}Let $W_{\rm error}$ be as in Lemma~\ref{neweqninxicoord}. Then, for $s \geq 2S_{\xi}$, we have that
		\begin{equation}\label{fristestforwerror1}
			\left|W_{\rm error}\right| \lesssim {\rm min}\left(r^{-2}\left(\omega^2-m^2\right)^{-2},1\right)r^{-1},
		\end{equation}
		\begin{equation}\label{fristestforwerror2}
			\left|W_{\rm error}\right| \lesssim \xi^{-2},
		\end{equation}
		\begin{equation}\label{fristestforwerror3}
			\left|\partial_{\xi}W_{\rm error}\right| \lesssim {\rm min}\left(r^{-5/2}\left(\omega^2-m^2\right)^{-5/2},1\right)r^{-3/2},
		\end{equation}
		\begin{equation}\label{fristestforwerror4}
			\left|\partial_{\xi}W_{\rm error}\right| \lesssim \xi^{-3}.
		\end{equation}
		
		For $W_{\rm main}$ we have that $s \geq 2S_{\xi}$ implies that
		\begin{equation}\label{lowerboundformainw}
			-\partial_{\xi}W_{\rm main} \gtrsim L^2\xi^{-3}.
		\end{equation}
	\end{lemma}
	\begin{proof}We have
		\[\left|W_{\rm error}\right| \lesssim P^{-3}\left(\partial_sP\right)^2 + P^{-2}\left|\partial_s^2P\right|,\]
		and also the estimates
		\[\left|\partial_sP\right| \lesssim r^{-2},\qquad \left|\partial_s^2P\right| \lesssim r^{-3},\qquad \left|P\right|^{-1} \lesssim {\rm min}\left(\left(\omega^2-m^2\right)^{-1},r\right).\]
		This establishes~\eqref{fristestforwerror1}. The estimate~\eqref{fristestforwerror2} then follows from Lemma~\ref{somestuffaboutthexicoor}.
		
		Next, keeping in mind that $\partial_{\xi} = P^{-1/2}\partial_s$, we note that
		\begin{align*}
			\left|\partial_{\xi}W_{error}\right| &\lesssim P^{-4}\left|\partial_{\xi}P\right|\left|\partial_sP\right|^2 + P^{-3}\left|\partial_{\xi}P\right|\left|\partial_s^2P\right| + P^{-3}\left|\partial_sP\right|\left|\partial^2_{\xi s}P\right| + P^{-2}\left|\partial_{ss\xi}^3P\right|
			\\ \nonumber &\lesssim P^{-9/2}\left|\partial_sP\right|^3 + P^{-7/2}\left|\partial_sP\right|\left|\partial_s^2P\right| + P^{-5/2}\left|\partial_s^3P\right|.
		\end{align*}
		One then easily obtains~\eqref{fristestforwerror3} and~\eqref{fristestforwerror4}.
		
		Finally, arguing in a similar fashion,~\eqref{lowerboundformainw} is straightforward to obtain.
	\end{proof}
	
	\subsection{Basic multipliers}
	In this section we present the building blocks for our main estimate. 
	
	We start with an energy estimate, which will be later used to control various boundary terms.
	\begin{lemma}\label{agoodenergyestimateiguess}Suppose that $v$ is an outgoing solution as in Definition~\ref{outxi}. Then 
		\[\left|v(\infty)\right|^2 + \omega \left|v(-\infty)\right|^2 = \int_{-\infty}^{\infty} \Im\left(\check{H}\overline{v}\right)\, d\xi.\]
	\end{lemma}
	\begin{proof}We note that
		\[\partial_{\xi}\Im\left(\partial_{\xi}v\overline{v}\right) = \Im\left(\check{H}\overline{v}\right).\]
		Then we integrate this identity over $\xi \in (-\infty,\infty)$ and use the outgoing boundary conditions.
	\end{proof}
	
	The following multipliers will be used to produce integrated estimates for the solution $v$.
	\begin{lemma}\label{bunchofmultipliers}Suppose that $v$ is an outgoing solution as in Definition~\ref{outxi}.
		\begin{enumerate}
			\item If $f(\xi)$ is a $C^2$ function which is piecewise $C^3$ and which is constant for sufficiently large $|s|$, then we have that 
			\begin{align}\label{fmulteqn}
				&\int_{-\infty}^{\infty}\left[2\partial_{\xi} f \left|\partial_{\xi}v\right|^2 - \left(f \partial_{\xi}W + \frac{1}{2}\partial_{\xi}^3f\right)\left|v\right|^2 \right]\, d\xi =
				\\ \nonumber &\qquad f(\infty)|v(\infty)|^2 - f(-\infty)\omega^2 \left|v\left(-\infty\right)\right|^2 - \int_{-\infty}^{\infty} \left(2f\Re\left(\check{H}\overline{\partial_{\xi}v}\right) + \partial_{\xi}f \Re\left(\check{H}\overline{v}\right)\right)\, d\xi.
			\end{align}
			\item If $y(\xi)$ is a $C^0$ function which is piecewise $C^1$ and satisfies that the limits $y(\pm\infty) \doteq \lim_{\xi\to\pm\infty}y(s)$ exist, then
			\begin{align}\label{ymulteqn}
				&\int_{-\infty}^{\infty}\left[\partial_{\xi}y \left|\partial_{\xi}v\right|^2 - \partial_{\xi}\left(y W\right) \left|v\right|^2\right]\, d\xi =
				\\ \nonumber &\qquad y(\infty)\left|v(\infty)\right|^2 - y(-\infty)\omega^2\left|v(-\infty)\right|^2 + \int_{-\infty}^{\infty}2y\Re\left(\check{H}\overline{\partial_{\xi}v}\right)\, d\xi.
			\end{align}
			\item If $z(s) $ is a $C^0$ function which is piecewise $C^1$, vanishes for large $r$, and is such that $lim_{r\to r_+}(r-r_+)z$ exists, then 
			\begin{align}\label{redmicro}
				&\int_{-\infty}^{\infty}\partial_{\xi}z \left|\partial_{\xi}v + i\omega v\right|^2 - \partial_{\xi}\left(z(W+\omega^2)\right)\left|v\right|^2\, d\xi = 
				\\ \nonumber &\qquad -2\int_{-\infty}^{\infty}z\Re\left(\check{H}\left(\overline{\partial_{\xi}v + i\omega v}\right)\right)\, d\xi + \left(\lim_{r \to r_+}z\left(W+\omega^2\right)\right)\left|v(-\infty)\right|^2.
			\end{align}

		\end{enumerate}
	\end{lemma}
	\begin{proof}
		The first identity follows from the fundamental theorem of calculus and the identity
		\begin{align*}
			&\partial_{\xi}\left(f\left|\partial_{\xi}v\right|^2 - fW\left|v\right|^2 + \partial_{\xi}f \Re\left(\partial_{\xi}v\overline{v}\right) - \frac{1}{2}\partial_{\xi}^2f \left|v\right|^2\right) = 
			\\ \nonumber &\qquad 2\partial_{\xi} f \left|\partial_{\xi}v\right|^2 - \left(f \partial_{\xi}W + \frac{1}{2}\partial_{\xi}^3f\right)\left|v\right|^2 + 2f\Re\left(\check{H}\overline{\partial_{\xi}v}\right) + \partial_{\xi}f\Re\left(\check{H}\overline{v}\right).
		\end{align*}
		
		Similarly, the second identity follows from 
		\begin{align*}
			&\partial_{\xi}\left(y\left|\partial_{\xi}v\right|^2 - yW\left|v\right|^2\right) = \partial_{\xi}y\left|\partial_{\xi}v\right|^2 - \partial_{\xi}\left(yW\right)\left|v\right|^2 + 2y\Re\left(\check{H}\overline{\partial_{\xi}v}\right).
		\end{align*}
		
		The third identity follows from
		\begin{align*}
			&\partial_{\xi}\left(z\left|\partial_{\xi}v + i\omega v\right|^2 - z\left(W+\omega^2\right)\left|v\right|^2\right) = 
			\\ \nonumber &\qquad \partial_{\xi}z \left|\partial_{\xi}v + i\omega v\right|^2 - \partial_{\xi}\left(z(W+\omega^2)\right)\left|v\right|^2 + 2z\Re\left(\check{H}\overline{\left(\partial_{\xi}v + i\omega v\right)}\right).
		\end{align*}

	\end{proof}
	\subsection{The basic Morawetz estimate for $v$.}
	In this section we will prove a weighted $L^2_{\xi}$ estimate for $v$ which covers the whole range of frequencies~\eqref{omegabiggerthanm} and~\eqref{Llarge}.

	The following lemma will be useful to control certain terms near $r = r_{\rm crit}$ where certain of our estimates degenerate.
	\begin{lemma}\label{interpolatetrap}Let $h(x) : (-1,1) \to \mathbb{C}$ be a smooth function and $A \gg 1$ be a large constant not depending on $h$. Then 
		\begin{equation}\label{interpolatetoprove}
			A\int_{-1}^1\left|h\right|^2\, dx \lesssim \int_{-1}^1\left[\left|\partial_xh\right|^2 + A^2x^2\left|h\right|^2\right]\, dx,
		\end{equation}
		where the implied constant is independent of $A$ and $h$.
	\end{lemma}
	\begin{proof}Since the right hand side of~\eqref{interpolatetoprove} only degenerates at $x = 0$, it suffices to consider the case when $h$ vanishes at $x = -1$ and $x = 1$. In this case one has
		\[A\int_{-1}^1\left|h\right|^2\, dx = A\int_{-1}^1
		\frac{d}{dx}\left(x\right)\left|h\right|^2\, dx = 2A\int_{-1}^1x\Re\left(\partial_xh \overline{h}\right)\, dx \leq \int_{-1}^1\left[\left|\partial_xh\right|^2 + 4A^2x^2\left|h\right|^2\right]\, dx.\]
	\end{proof}

	It is convenient to split the proof of the Morawetz estimate into two different regimes: $\{\omega \gg L\}$ and $\{L \lesssim \omega\}$. We start with $\omega \gg L$.
	\begin{lemma}\label{mora1}Suppose that $v$ is an outgoing solution as in Definition~\ref{outxi} and suppose that in addition to~\eqref{Llarge}, we have that 
		\begin{equation}\label{omegaiswaybiggerthanL}
			\omega \gg L.
		\end{equation}
		Then
		\[\int_{-\infty}^{\infty}\left[ \left(1+|\xi|\right)^{-1-\delta}\left|\partial_{\xi}v\right|^2 + \left(1+|\xi|\right)^{-1-\delta}\left|v\right|^2\right]\, d\xi \lesssim_{\delta} \omega^2\int_{-\infty}^{\infty}\left|\check{H}\right|^2 \left(1+|\xi|\right)^{1+\delta}\, d\xi,\]
		where  $\delta \in (0,1)$ is arbitrary.
	\end{lemma}
	\begin{proof}We first note that one consequence of~\eqref{omegaiswaybiggerthanL} is that $W \lesssim -\omega^2$ for all $s \leq S_{\xi}$ and $W \lesssim -1$ for all $s \geq S_{\xi}$. 
		
		Now define a function $y(\xi)$ by
		\[y\left(\xi\right) \doteq P\left(\xi\right)\exp\left(A\int_{-\infty}^{\xi}\left(1+|\xi'|\right)^{-1-\delta}\, d\xi'\right),\]
		for a sufficiently large constant $A$ which is independent of $\omega$ (but may depend on $S_{\xi}$).
		
		When $s \leq S_{\xi}$, then $P = 1$ and we have that
		\[\partial_{\xi}y =A\left(1+|\xi|\right)^{-1-\delta}y > 0.\]
		When $s \in [S_{\xi},2S_{\xi}]$, we have that 
		\[\partial_{\xi}y = A\left(1+|\xi|\right)^{-1-\delta}y + (\partial_{\xi}\log P) y.\]
		In view of the bound $-\partial_s\log P \lesssim \omega^{-2} r^{-2}$, we have that this expression is positive. When $s \in [2S_{\xi},\infty)$, we have that 
		\[\partial_{\xi}y = A\left(1+|\xi|\right)^{-1-\delta}y + P^{-3/2}\left(\frac{-2Mm^2}{r^2}\right) \frac{dr}{ds} y,\]
		which is also positive. Thus $\partial_{\xi}y (\xi)\gtrsim (1+|\xi|)^{-1-\delta}$ is everywhere.
		
		Next we analyze $-\partial_{\xi}\left(yW\right)$. When $s \leq S_{\xi}$ we have that 
		\[-\partial_{\xi}\left(yW\right) = -A\left(1+|\xi|\right)^{-1-\delta}y W - y \partial_{\xi}W.\]
		Since $W \lesssim -\omega^2$, and $\left|\partial_sW\right| \lesssim L^2(r-r_+)r^{-3} + (r-r_+)r^{-2}$ this term is clearly positive. When $s \in [S_{\xi},2S_{\xi}]$, we  have  
		\[-\partial_{\xi}\left(yW\right) = \left(-A\left(1+|\xi|\right)^{-1-\delta}-\partial_{\xi}\log P\right)y W - y \partial_{\xi}W.\]
		In view of the bound $-\partial_s\log P \lesssim \omega^{-2} r^{-2}$, we have that this expression is positive. Finally, in the region $s \in [2S_{\xi},\infty)$ we have that 
		\[-\partial_{\xi}\left(yW\right) \gtrsim \frac{A}{\xi^{1+\delta}}.\]
		The lemma then follows from an application of the multiplier~\eqref{ymulteqn} and Lemma~\ref{agoodenergyestimateiguess}.
	\end{proof}
	
	Next we consider the case when $\omega \lesssim L$.
	\begin{lemma}\label{mora2}Suppose that $v$ that $v$ is an outgoing solution as in Definition~\ref{outxi} and suppose that in addition to~\eqref{Llarge}, we have that 
		\begin{equation}\label{omegaisnottoobig}
			\omega \lesssim L.
		\end{equation}
		
		\begin{align*}
			&\int_{-\infty}^{\infty}\left[\left(1+|\xi|\right)^{-1-\delta}\left|\partial_{\xi}v\right|^2 + \left(1+|\xi|\right)^{-1-\delta}\left|v\right|^2\right] 
			\lesssim_{\delta} \left(1+\omega^2\right)\int_{-\infty}^{\infty}\left|\check{H}\right|^2\left(1+|\xi|\right)^{1+\delta}\, d\xi,
		\end{align*}
	\end{lemma}
	where  $\delta \in (0,1)$ is arbitrary.
	\begin{proof}

		Choose $s_0 < s_1 < S_{\xi}$ so that $s_0$ is sufficiently negative and $s_1$ is sufficiently positive, to be fixed later. We apply~\eqref{fmulteqn} with a function $f$ which is identically $-1$ when $s < s_0$, is identically $1$ when $s > s_1$, satisfies $\partial_sf \geq 0$ an $\partial_sf > 0$ for $s \in [s_0+1,s_1-1]$, and so that $f$ vanishes at $r = r_{\rm crit}$. We obtain using  Lemma~\ref{critv}, Lemma~\ref{interpolatetrap},~\eqref{lowerboundformainw}, and Lemma~\ref{agoodenergyestimateiguess} the estimate
		\begin{align}\label{thefirstfestimate}
			&\int_{\xi(s_0)}^{\xi(s_1)}\left[(\partial_{\xi}f\right)\left|\partial_{\xi}v\right|^2 + [1+L^2\left|r-r_{\rm crit}\right|^2]\left|r-r_+\right|\left|v\right|^2]\, d\xi 
			\\ \nonumber &\qquad + \int_{\xi \geq \xi(s_1)}L^2\xi^{-3}\left|v\right|^2\, d\xi+\int_{\xi \leq \xi(s_0)}L^2(r-r_+)\left|v\right|^2\, d\xi \lesssim 
			\\ \nonumber &\qquad \int_{-\infty}^{\infty}\left[\left|f\right|\left|\partial_{\xi}v\right| + \left|\partial_{\xi}f\right|\left|v\right| + \omega\left|v\right|\right]\left|\check{H}\right|\, d\xi.
		\end{align}
		
		We next observe that, for any $\alpha > 0$ sufficiently small then for all $\xi \gg 1$, using the fact that $-W(\xi)\gtrsim 1-\frac{L^2}{\xi^2}$:
		\[ -\partial_{\xi}\left(\left(1-\alpha \xi^{-\delta}\right)W \right)\gtrsim \frac{L^2}{\xi^3} + \frac{\alpha}{\xi^{1+\delta}}. \]
		Similarly, if $-\xi \gg 1$, then 
		\[ -\partial_{\xi}\left(\left(-1+\alpha (-\xi)^{-\delta}\right)W \right)\gtrsim L^2 [r-r_+]+ \frac{\alpha}{|\xi|^{1+\delta}}. \]
		Now we let $y(\xi)$ be a function which is equal to $1-\alpha \xi^{-\delta}$ when $s \geq s_1$, is equal to $-1+\alpha(-\xi)^{-\delta}$,  $y\equiv 0$ on $[\xi(r_{crit}-1),\xi(r_{crit}+1)]$ and always satisfies $\partial_{\xi}y \geq 0$. Adding a sufficiently small multiple of~\eqref{ymulteqn} to~\eqref{thefirstfestimate}, applying Lemma~\ref{agoodenergyestimateiguess}, and then applying Cauchy-Schwarz then completes the proof.
		
	\end{proof}
	
	\subsection{Putting everything together}
	We are almost ready to give the proof of  Proposition~\ref{almosthereexcepthorizon} and Proposition~\ref{highfreqest}.
	
	We will need to carry out derivatives with respect to $\omega$ in the $\xi$-coordinate. In order to reduce the potential for confusion we define $\partial_{\tilde{\omega}}$ to be derivative with respect to $\omega$ in the $\left(\xi,\omega\right)$ coordinate system. In the next lemma we estimate $\partial_{\tilde{\omega}}W$.
	\begin{lemma}\label{iguessthisisgoodenough}For $s \geq 2S_{\xi}$ and $|\omega| \lesssim 1$, we have, for any $\delta > 0$ sufficiently small,
		\begin{equation}\label{whenomegaisnottoolargetilde}
			\left|\partial_{\tilde{\omega}}W\right| \lesssim \frac{L^2}{\left(r\left(\omega^2-m^2\right)+1\right)^2} \lesssim L^2 \xi^{-1-10\delta}\left(\omega^2-m^2\right)^{-1/2-5\delta}.
		\end{equation}
		
		When $\omega - m \gtrsim 1$, then we have
		\begin{equation}\label{whenomegaistoolargetilde}
			\left|\partial_{\tilde{\omega}}W\right| \lesssim L^2 \omega^{-3}r^{-2} \lesssim L^2 \omega^{-1}\xi^{-2}.
		\end{equation}
	\end{lemma}
	\begin{proof}When $\left|\omega\right| \lesssim 1$, we have 
		\[\left|\partial_{\tilde{\omega}}r\right| = \left|\left(\partial_{\xi}r\right)\left(\partial_{\omega}\xi\right)\right| \lesssim P^{-1/2}\left(\int_{2S_{\xi}}^sP^{-1/2}\, dS + 1\right) \lesssim r^2.\]
		
		Then we have that 
		\[\left|\partial_{\tilde{\omega}}W\right| \lesssim \left|\partial_{\tilde{\omega}}\left(P^{-1}V\right)\right| + \left|\partial_{\tilde{\omega}}W_{\rm error}\right|.\]
		We start with $W_{\rm error}$. We have
		\begin{align*}
			\left|\partial_{\tilde{\omega}}W_{\rm error}\right| &\lesssim P^{-4}\left|\partial_{\tilde{\omega}}P\right|\left|\partial_sP\right|^2 + P^{-3}\left|\partial_sP\right|\left|\partial^2_{\tilde{\omega}s}P\right| + P^{-3}\left|\partial_{\tilde{\omega}}P\right|\left|\partial_s^2P\right| + P^{-2}\left|\partial^3_{ss\tilde{\omega}}P\right|
			\\ \nonumber &\lesssim \frac{1}{\left(r\left(\omega^2-m^2\right)+1\right)^2}.
		\end{align*}
		Similarly, we obtain
		\[\left|\partial_{\tilde{\omega}}\left(P^{-1}V\right)\right| \lesssim \frac{L^2}{\left(r\left(\omega^2-m^2\right)+1\right)^2}.\]
		This suffices to establish~\eqref{whenomegaisnottoolargetilde}.
		
		For $\omega - m \gtrsim 1$, we have that 
		\[\left|\partial_{\tilde{\omega}}r\right| = \left|\partial_{\xi}r\right|\left|\partial_{\omega}\xi\right| \lesssim r \omega^{-1},  \]
		\begin{align*}
			\left|\partial_{\tilde{\omega}}W_{\rm error}\right| &\lesssim P^{-4}\left|\partial_{\tilde{\omega}}P\right|\left|\partial_sP\right|^2 + P^{-3}\left|\partial_sP\right|\left|\partial^2_{\tilde{\omega}s}P\right| + P^{-3}\left|\partial_{\tilde{\omega}}P\right|\left|\partial_s^2P\right| + P^{-2}\left|\partial^3_{ss\tilde{\omega}}P\right|
			\\ \nonumber &\lesssim r^{-3}\omega^{-5},
		\end{align*}
		\[\left|\partial_{\tilde{\omega}}\left(P^{-1}V\right)\right| \lesssim L^2 r^{-2}\omega^{-3}.\]
		
	\end{proof}

	Now we give the proof of Proposition~\ref{almosthereexcepthorizon}.
	\begin{proof} Let $\chi(\xi)$ be a smooth cutoff function which is $1$ for $\xi \leq -1$ and vanishes for $\xi \geq 0$. Then we set 
		\[\mathring{v} \doteq \left(-\chi \frac{\xi}{\omega}\partial_{\xi} + \partial_{\tilde{\omega}}\right)v.\]
		A computation indicates that $\mathring{v}$ will satisfy the following equation:
		\begin{equation}\label{eqnformathringv}
			\partial_{\xi}^2\mathring{v} - W\mathring{v} = \mathring{H},
		\end{equation}
		where
		\[\mathring{H} \doteq \left(-\chi \frac{\xi}{\omega}\partial_{\xi} + \partial_{\tilde{\omega}}\right)\check{H} + 2\left(-\chi' \frac{\xi}{\omega} - \frac{\chi}{\omega}\right)\left(Wv + \check{H}\right) + \left(\partial_{\tilde{\omega}}W-\frac{\chi \xi}{\omega} \rd_{\xi}W\right)v + \left(-\chi'' \frac{\xi}{\omega} - \frac{\chi'}{\omega}\right)\partial_{\xi}v.\]
		Since $-2\omega^{-1}W + \partial_{\tilde{\omega}}W = O(r-r_+)$ as $r\to r_+$, it is straightforward to see that, in view of Lemma~\ref{mora1}, Lemma~\ref{mora2}, and Lemma~\ref{iguessthisisgoodenough} that, for any $\delta > 0$ sufficiently small,
		\[\int_{-\infty}^{\infty}\left(1+|\xi|\right)^{1+\delta}\left|\mathring{H}\right|^2\, d\xi \lesssim_{\delta}  L^2\left(\omega^4+\left(\omega^2-m^2\right)^{-1/2-5\delta}\right)\int_{-\infty}^{\infty}\left(1+|\xi|\right)^{1+\delta}\left[\left|\check{H}\right|^2 + [1+(-\xi)_+^2]\left|\partial_{\xi}\check{H}\right|^2 + \left|\partial_{\tilde{\omega}}\check{H}\right|^2\right]\, d\xi.\]
		
		Since $\mathring{v}$ will still be an outgoing solution, we may apply Lemmas~\ref{mora1} and~\ref{mora2} to obtain that
		\begin{align*}
			&\int_{-\infty}^{\infty}\left(1+|\xi|\right)^{-1-\delta}\left[\left|\partial_{\xi}\mathring{v}\right|^2+\left|\mathring{v}\right|^2\right]\, d\xi 
			\\ \nonumber &\qquad \lesssim_{\delta}  L^2 (1+\omega^2)\left(\omega^4+\left(\omega^2-m^2\right)^{-1/2-5\delta}\right)\int_{-\infty}^{\infty}\left(1+|\xi|\right)^{1+\delta}\left[\left|\check{H}\right|^2 + [1+(-\xi)_+^2] \left|\partial_{\xi}\check{H}\right|^2 + \left|\partial_{\tilde{\omega}}\check{H}\right|^2\right]\, d\xi
		\end{align*}
		Using directly the equation~\eqref{eqnformathringv}, we can also obtain 
		\begin{align*}
			&\int_{-\infty}^{\infty}\left(1+|\xi|\right)^{-1-\delta}\left[\left|\partial_{\xi}^2\mathring{v}\right|^2+\left|\partial_{\xi}\mathring{v}\right|^2+\left|\mathring{v}\right|^2\right]\, d\xi 
			\\ \nonumber &\qquad \lesssim_{\delta} \underbrace{L^4 (1+\omega^2)\left(\omega^4+\left(\omega^2-m^2\right)^{-1/2-5\delta}\right)\int_{-\infty}^{\infty}\left(1+|\xi|\right)^{1+\delta}\left[\left|\check{H}\right|^2 + [1+(-\xi)_+^2]\left|\partial_{\xi}\check{H}\right|^2 + \left|\partial_{\tilde{\omega}}\check{H}\right|^2\right]\, d\xi}_{\mathcal{E}}.
		\end{align*}
		We can re-express in terms of $\partial_{\tilde{\omega}}v$ to obtain that for any $0 < \xi_0 < \infty$:
		\[\sum_{j=0}^2\int_{-\xi_0}^{\xi_0}\left|\partial^j_{\xi}\partial_{\tilde{\omega}}v\right|^2\, d\xi \lesssim_{\xi_0} \mathcal{E} \Rightarrow \sum_{j=0}^1\sup_{\xi \in [-\xi_0,\xi_0]}\left|\partial_{\xi}^j\partial_{\tilde{\omega}}v\right|^2 \lesssim \mathcal{E}.\]
		It is then immediate that~\eqref{toprovewhenomegaissortofbig} follows.
	\end{proof}
	
	Finally we come to the proof of Proposition~\ref{highfreqest}.
	\begin{proof}We start with a new Morawetz estimate with an improved weight as $r\to r_+$. Observe that $\omega^2+W = O\left(L^2(r-r_+)\right)$ as $r\to r_+$, and, moreover, $\frac{d}{dr}|_{r=r_+}\left(\omega^2+W\right) \gtrsim L^2$. Now let $z(\xi)$ be any smooth function of $\xi$ which vanishes for large $\xi$ and which is equal to $-\left(L^{-2}\left(\omega^2+W\right)\right)^{-1}\left(1+(-\xi)^{-\delta}\right)$ for $\xi$ sufficiently negative. Applying~\eqref{redmicro} and using Lemmas~\ref{mora1} and~\ref{mora2} leads to 
		\begin{equation}\label{redshiftyay}
			\int_{-\infty}^0\left[\left(r-r_+\right)^{-1}\left|\partial_{\xi}v+i\omega v\right|^2 + L^2\left(1+|\xi|\right)^{-1-\delta}|v|^2\right]\, d\xi \lesssim_{\delta,R_0} \left(1+\omega^2\right)\int_{-\infty}^{\infty}\left(r-r_+\right)^{-1} \left|\check{H}\right|^2\, d\xi.
		\end{equation}
		
		Now we set $\tilde{v} \doteq e^{-i\omega \xi}\partial_{\tilde{\omega}}\left(e^{i\omega \xi}v\right)$. A computation yields the following equation for $\tilde{v}$:
		\begin{equation}\label{tildeveqn}
			\partial_{\xi}^2\tilde{v} - W\tilde{v} = \tilde{H},
		\end{equation}
		where 
		\[\tilde{H} = e^{-i\omega \chi}\partial_{\tilde{\omega}}\left(e^{i\omega \xi}\check{H}\right) + 2i\left(\partial_{\xi}v+i\omega v\right).\]
		In particular,
		\[\int_{-\infty}^0\left|\tilde{H}\right|^2(r-r_+)^{-1} d\xi \lesssim \left(1+\omega^2\right)\int_{-\infty}^{\infty}\left[\left|\partial_{\tilde{\omega}}\left(e^{i\omega \xi}\check{H}\right)\right|^2 + \left|\check{H}\right|^2\right](r-r_+)^{-1}\, d\xi.\]
		Thus, applying the estimate~\eqref{redshiftyay} to~\eqref{tildeveqn} with a suitable cut-off yields
		\begin{align}\label{okthisisnice}
			&\int_{-\infty}^0\left[\left(r-r_+\right)^{-1}\left|\partial_{\xi}\tilde{v}+i\omega \tilde{v}\right|^2 + \left(1+|\xi|\right)^{-1-\delta}|\tilde{v}|^2\right]\, d\xi
			\\ \nonumber &\qquad \lesssim \left(1+\omega^2\right)^2\int_{-\infty}^{\infty}\left[\left|\partial_{\tilde{\omega}}\left(e^{i\omega \xi}\check{H}\right)\right|^2 + \left|\check{H}\right|^2\right](r-r_+)^{-1}\, d\xi + \int_{-1}^1\left[\left|\tilde{v}\right|^2 + \left|\partial_{\xi}\tilde{v}\right|^2\right]\, d\xi.
		\end{align}
		
		Now we observe that $\partial_{\xi}\tilde{v} + i\omega \tilde{v} = e^{-i\xi\omega}\partial_{\xi}\left(\partial_{\tilde\omega}\left(e^{i\omega \xi}v\right)\right)$. In particular, we may conclude the proof of the proposition by applying the fundamental theorem of calculus to control $\partial_{\tilde\omega}\left(e^{i\omega \xi}v\right)$ from the estimate~\eqref{okthisisnice} and using Proposition~\ref{almosthereexcepthorizon} to control the terms with $\tilde{v}$ on the right hand side of~\eqref{okthisisnice}. 
	\end{proof}

	\section{Frequencies where $m - \omega \gtrsim L^{-p}$ }\label{regimeB1.section}
	In this section we will study solutions to~\eqref{eq:mainVrhs} satisfying outgoing boundary conditions~\eqref{outgoingboundary2} under the additional assumptions that

	\begin{equation}\label{mbiggerthanomega}
		m - \omega > 0 \text{ and }\omega \geq 0
	\end{equation}
	\begin{equation}\label{Llarge2}
		L \gg 1,
	\end{equation}
	\begin{equation}\label{thebasiccomparableassump}
		L^{-p} \lesssim (m-\omega),
	\end{equation}
	where $p$ is any large positive integer.
	
	Our main result will be the following.
	\begin{prop}\label{themainpropinthebigthanL2} Let $p > 0$ be a large positive integer and $L$ be sufficiently large depending on $p$. Also, let $R > r_+$ be an arbitrary large constant, and let $\check{\delta} > 0$ be an arbitrarily small constant. Unless noted otherwise, all constants that follow in this proposition may depend on $p$, $R$, and $\check{\delta}$.
		
		There exists a positive integers $N(p)$ and ``critical frequencies'' $\left\{\lambda_n\right\}_{n=1}^{N(p)}$ so that
		\[\left|N\left(p\right)\right| \lesssim L^{p/2},\qquad  -\lambda_n \sim \left(n+L\right)^{-2} \qquad i \in [1,N(p)-1] \Rightarrow \left|\lambda_{i+1}-\lambda_i\right| \gtrsim \left(-\lambda_i\right)^{3/2}.\]
		
		We define the following sets:
		\[I_{\rm good} \doteq \left\{ \omega \in [0,m) : m^2-\omega^2 \geq  L^{-p}\text{ and }{\rm inf}_n \left(m^2-\omega^2 + \lambda_n\right) \geq e^{-\frac{1}{2} \LL \log \LL}\right\},\]
		\[I_{\rm bad} \doteq \left\{ \omega \in [0,m) : m^2-\omega^2 \geq  L^{-p}\text{ and }{\rm inf}_n \left(m^2-\omega^2 + \lambda_n\right) < e^{-\frac{1}{2} \LL \log \LL}\right\}.\]
		
		Suppose that $u$ is an outgoing solution to~\eqref{eq:mainVrhs}, and assume that $H$ is compactly supported in $\{r \leq R\}$. 
		\begin{enumerate}
			\item Suppose that $\omega \in I_{\rm good}$:
			\begin{enumerate}
				\item\label{bigL1} $\sup_{r \leq R}\left|u\right|^2 \lesssim \LL^{2}\int_{-\infty}^R\frac{r}{r-r_+}\left|H\right|^2$.
				\item\label{bigL2} $\sup_{r \leq R}\left|\partial_{\omega}\left(e^{i\omega s}u\right)\right|^2 \lesssim \LL^{4}\int_{-\infty}^R\frac{r}{r-r_+}\left[\left|H\right|^2 + \left|\partial_{\omega}\left(e^{i\omega s}H\right)\right|^2\right]\, ds$.
			\end{enumerate}
			\item Suppose that $\omega \in I_{\rm bad}$: There exist a constant $D > 0$ (independent of $\check{\delta}$)
			\begin{enumerate}
				\item\label{bigL3} $\sup_{r \leq R}\left|u\right|^2 \lesssim e^{D \LL}\int_{-\infty}^R\frac{r}{r-r_+}\left|H\right|^2\, ds$.
				\item\label{bigL4} $\sup_{r \leq R}\left|\partial_{\omega}\left(e^{i\omega s}u\right)\right|^2\lesssim e^{4(1+\check{\delta}) \LL \log \LL}\int_{-\infty}^R\frac{r}{r-r_+}\left[\left|H\right|^2 + \left|\partial_{\omega}\left(e^{i\omega s}H\right)\right|^2\right]\, ds$.
				\item\label{bigL5} 
				\begin{align*}
					&\sup_{r \leq R}\left|\partial_{\omega}\left(e^{i\omega s}u\right)\right|^2 \lesssim \left(\left({\rm inf}_n \left(m^2-\omega^2 + \lambda_n\right)\right)^{-2}e^{\check{\delta}\LL \log\LL}+e^{-4\left(1-\check{\delta}\right)\LL \log\LL}\left({\rm inf}_n \left(m^2-\omega^2 + \lambda_n\right)\right)^{-4}\right)
					\\ \nonumber &\qquad \times\int_{-\infty}^R\frac{r}{r-r_+}\left[\left|H\right|^2 + \left|\partial_{\omega}\left(e^{i\omega s}H\right)\right|^2\right]\, ds.
				\end{align*}
			\end{enumerate}
		\end{enumerate}
	\end{prop}
	
	It will be convenient to introduce the convention that, throughout this section, we let $\alpha$ denote any suitably large positive integer depending on $p$ which we allow to increase line to line.
	\subsection{Multiplers}\label{multiplier.section}
	In this section we will list various multiplier identities which will be useful in the following sections. 
	
	We start by stating a useful identity which is the key ingredient for the well-known Agmon estimates for solutions to one dimensional inhomogeneous Schr\"{o}dinger equations. The specific identity we cite will be taken from Section 5 of~\cite{quasimodeads} since it is stated in a form which will be convenient for us.
	\begin{lemma}[Lemma 5.1 of~\cite{quasimodeads}]\label{agmonidentity} Let $-\infty < s_0 < s_1 < \infty$, and suppose that $v \in C^{\infty}\left(s_0,s_1\right)$ is a real valued function with $v(s_0) = v(s_1) = 0$ and which satisfies
		\[\frac{d^2v}{ds^2} - Wv = P,\]
		for real valued $W,P \in C^{\infty}\left(s_0,s_1\right)$. Then, for any Lipschitz function $\phi(s) : [s_0,s_1] \to \mathbb{R}$, we have
		\begin{equation}\label{theactualagmonidentity}
			\int_{s_0}^{s_1}\left[\left(\frac{d}{ds}\left(e^{\phi}v\right)\right)^2 + \left(W - \left(\frac{d\phi}{ds}\right)^2\right)\left(e^{\phi}v\right)^2\right]\, ds = -\int_{s_0}^{s_1}e^{2\phi}Pu\, ds.
		\end{equation}
	\end{lemma}
	\begin{rmk}Lemma~\ref{agmonidentity} is useful when the potential $W$ is large and positive so that the left hand side of~\eqref{theactualagmonidentity} is positive definite even when $\phi$ is large.
	\end{rmk}
	
	It will be useful to slightly broaden the scope of solutions to~\eqref{eq:mainVrhs} which we call ``outgoing.''
	\begin{defn}\label{weakoutgoingdef}Assume that $\omega^2 < m^2$. We then say that $u$ is a ``weakly outgoing'' solution to~\eqref{eq:mainVrhs} if 
		\[u \sim e^{-i\omega s}\text{ as }s\to -\infty,\qquad \exists \epsilon > 0\text{ so that }\lim_{s\to \infty} e^{\epsilon s}\left(\left|u\right|+\left|\partial_su\right|\right) = 0.\]
	\end{defn}
	\begin{rmk}
		The reason we have introduced this definition is that when $u$ is an outgoing solution to~\eqref{eq:mainVrhs} it will generally not be the case   that $\partial_{\omega}u$ has the correct boundary behavior as $s\to \infty$ to be considered outgoing. It will, however, be weakly outgoing as $s\to\infty$.
	\end{rmk}

	Next, we present an identity which is the ODE analogue of the usual $\partial_t$-energy identitiy (cf.~Lemma~\ref{agoodenergyestimateiguess}).
	\begin{lemma}\label{agoodenergyestimateiguess2}Suppose that $u$ is a weakly outgoing solution to~\eqref{eq:mainVrhs}. Then 
		\begin{equation} \label{agoodenergyestimateiguess2.eq}
			\omega \left|u(-\infty)\right|^2 = \int_{-\infty}^{\infty} \Im\left(H\overline{v}\right)\, ds.
		\end{equation}
	\end{lemma}
	\begin{proof}We note that
		\[\partial_s\Im\left(\partial_su\overline{u}\right) = \Im\left(H\overline{u}\right).\]
		Then we integrate this identity over $s \in (-\infty,\infty)$ and use the weakly outgoing boundary conditions.
	\end{proof}
	
	Next, we state a family of multipliers which we will use to generate $L^2_s$ estimates for $u$ (cf.~Lemma~\ref{bunchofmultipliers}).
	\begin{lemma}\label{bunchofmultipliers2} Suppose that $u$ is a weakly outgoing solution to ~\eqref{eq:mainVrhs}.
		\begin{enumerate}
			
			\item If $y(s)$ is a $C^0$ function which is piecewise $C^1$ and satisfies that the limits $y(\pm\infty) \doteq \lim_{s\to\pm\infty}y(s)$ exist, then
			\begin{align}\label{ymulteqn2}
				&\int_{-\infty}^{\infty}\left[\partial_sy \left|\partial_su\right|^2 - \partial_s\left(y V\right) \left|u\right|^2\right]\, ds =
				\\ \nonumber &\qquad  - y(-\infty)\omega^2\left|u(-\infty)\right|^2 + \int_{-\infty}^{\infty}2y\Re\left(H\overline{\partial_su}\right)\, d\xi.
			\end{align}
			\item If $z(s) $ is a $C^0$ function which is piecewise $C^1$, vanishes for large $r$, and is such that $lim_{r\to r_+}(r-r_+)z$ exists, then 
			\begin{align}\label{redmicro2}
				&\int_{-\infty}^{\infty}\partial_sz \left|\partial_su + i\omega u\right|^2 - \partial_s\left(z(V+\omega^2)\right)\left|u\right|^2\, ds = 
				\\ \nonumber &\qquad -2\int_{-\infty}^{\infty}z\Re\left(H\left(\overline{\partial_su + i\omega u}\right)\right)\, ds + \left(\lim_{r \to r_+}z\left(V+\omega^2\right)\right)\left|u(-\infty)\right|^2.
			\end{align}
			\item If $h(s)$ a $C^1$ function which is piecewise $C^2$ so that $\limsup_{s\to\infty}|h(s)|  < \infty$ and $\limsup_{s\to\infty}|h'(s)| = 0$, then 
			\begin{align}\label{hcurent}
				\int_{-\infty}^{\infty}\left[h\left|\partial_su\right|^2 + \left(hV - \frac{1}{2}\partial_s^2h\right)\left|u\right|^2\right]\, ds = -\int_{-\infty}^{\infty}h\Re\left(H\overline{u}\right)\, ds.
			\end{align}
		\end{enumerate}
	\end{lemma}
	\begin{proof}
		The identities~\eqref{ymulteqn2} and~\eqref{redmicro2} are derived just as the identities~\eqref{ymulteqn} and~\eqref{redmicro} were derived in the proof of Lemma~\ref{bunchofmultipliers}.  To prove~\eqref{hcurent} we observe that
		\[\partial_s\left(h\Re\left(\partial_su\overline{u}\right) - \frac{1}{2}\partial_sh\left|u\right|^2\right) = h\left|\partial_su\right|^2  + \left(hV - \frac{1}{2}\partial_s^2h\right)\left|u\right|^2 + h\Re\left(H\overline{u}\right).\]
		The identity~\eqref{hcurent} then follows by integrating this identity and applying the fundamental theorem of calculus.
	\end{proof}
	
	Our use of the identity~\eqref{redmicro2} will be through the following lemma which is an ODE version of the red shift multiplier of Dafermos--Rodnianski~\cite{Red,claylecturenotes}.
	\begin{lemma}\label{redshiftode} Let $u$ be an outgoing solution of~\eqref{eq:mainVrhs} with $\omega^2 < m^2$. Let $s_0 < 0$ be sufficiently negative. Then
		\begin{align}\label{tocomplementhcurr}
			&\int_{-\infty}^{s_0-1}\left[\left(r-r_+\right)^{-1}\left|\partial_su + i\omega u\right|^2 + L^2\left(1+|s|\right)^{-1-\delta}\left|u\right|^2\right]\, ds \lesssim 
			\\ \nonumber &\qquad \qquad \int_{-\infty}^{s_0}\left(r-r_+\right)^{-1}\left|H\right|^2\,ds + \int_{s_0-1}^{s_0}L^2\left|u\right|^2\, ds.
		\end{align}
	\end{lemma}
	\begin{proof}
		Now we let $z(s)$ be a function which is equal to $\left(-L^{-2}\left(\omega^2+V\right)\right)^{-1}\left(1+(-s)^{-\delta}\right)$ for $s \leq s_0-1$ and vanishes for $s \geq s_0$. We then apply the identity~\eqref{redmicro2} and argue as in the proof of the estimate~\eqref{redshiftyay}.
		
	\end{proof}
	
	The following computation will be useful in what follows.
	\begin{lemma}\label{qeqnusefultodiffu}Let $u$ solve~\eqref{eq:mainVrhs}, and define $q$  by
		\[q \doteq e^{-i\omega s \chi(s)}\partial_{\omega}\left(e^{i\omega s \chi(s)}u\right).\]
		for smooth function $\chi(s)$.  Then a computation shows that $q$ will satisfy the following equation:
		\begin{equation}\label{qeqnfordiff}
			\partial_s^2q - Vq = 2i\left(\chi+s\chi'\right) \left(\partial_su+i\omega u\right)+i\left(2\chi' + s\chi''\right)u + e^{-iws \chi}\partial_{\omega}\left(e^{i\omega s \chi}H\right) + 2\omega \left(\chi+s\chi'-1\right)u.
		\end{equation}
	\end{lemma}

	\subsection{The reference elliptic operator at infinity}\label{ellipatinfinity}
	In this section we will analyze an operator which will approximate the left hand side of~\eqref{eq:mainVrhs} well when $s$ is sufficiently large. Recall that for any non-negative integer $k$ and $x_0 \in \mathbb{R}$, the space $H^k_0\left(\left(x_0,\infty)\right)\right)$ denotes the closure of smooth compactly supported functions in $(x_0,\infty)$ under the Sobolev norm-$\left\vert\left\vert\cdot\right\vert\right\vert_{H^k}$.
	\begin{defn}\label{thisistheoperatorq}Let $x_0 \geq 2$ so that $L$ is sufficiently large with respect to $x_0$ and suppose that $\mathscr{E}_1(x),\mathscr{E}_2(x) : (x_0,\infty) \to \mathbb{R}$ satisfy the estimates
		\[\left|\mathscr{E}_1(x)\right| \lesssim \frac{\log(x)}{x^2},\qquad \left|\mathscr{E}_2(x)\right| \leq {\rm min}\left(\frac{1}{10},C\frac{\log(x)}{x}\right),\]
		for some constant $C > 0$.
		
		Then we define a differential operator $Q : H^2_0(\left(x_0,\infty\right)) \to L^2\left(\left(x_0,\infty\right)\right)$ by
		\[Qf \doteq -\frac{d^2f}{dx^2} + \underbrace{\left(-\frac{2Mm^2}{x} + \frac{L(L+1)+\mathbf{e}^2m^2}{x^2}\left(1+\mathscr{E}_2(x)\right) + \mathscr{E}_1(x)\right)}_{\doteq \tilde{V}}f.\]
	\end{defn}
	
	It is straightforward to show that $Q$ is self-adjoint.
	\begin{lemma}The operator $Q$ with domain $H^2_0\left(\left(x_0,\infty\right)\right)$ is self-adjoint on $L^2\left(\left(x_0,\infty\right)\right)$.
	\end{lemma}
	\begin{proof}This follows from standard arguments and thus we omit the proof. (See, for example, Theorem VIII.3 from~\cite{reedsimonI}, and note that any function $f \in H_0^2\left(\left(x_0,\infty\right)\right)$ will vanishes at $x = x_0$.)
	\end{proof}
	
	Next, we consider the rough qualitative properties of the spectrum of $Q$
	\begin{lemma}\label{basicspectQ}We have $\sigma_{\rm ess}\left(Q\right) = [0,\infty)$, and there exist an infinite sequence of eigenvalues $\{\lambda_n\}_{n=1}^{\infty}$ with $\lambda_n <0$ which satisfy $\lambda_n \to 0$ as $n\to\infty$.
	\end{lemma}
	\begin{proof}This is straightforward, so we will be brief.
		
		By explicit diagonalization with the ($x_0$-shifted) Fourier sine transform, one has that, with the same domain as $Q$, $\sigma_{\rm ess}\left(-\frac{d^2}{dx^2}\right) = [0,\infty)$. Since $Q- \left(-\frac{d^2}{dx^2}\right)$ is a relatively compact perturbation of $-\frac{d^2}{dx^2}$, we obtain that $\sigma_{\rm ess}\left(Q\right) = [0,\infty)$. (See Corollary 2 of Chapter XIII of~\cite{reedsimonIV}.)
		
		The statement about the eigenvalues follows from the fact that $Q$ is lower semibounded, Theorem XIII.1 (the min-max principle) of~\cite{reedsimonIV}, a mild adaption of Theorems XIII.6 of~\cite{reedsimonIV}, and the fact that elliptic operators cannot have eigenvalues of infinite multiplicity.

	\end{proof}
	
	We will be interested in the studying the eigenvalues of $Q$. For this it will be useful to compare with the eigenvalues of the Hydrogen atom. We review the relevant well-known facts in the following.
	\begin{theo}Let $\beta > 0$ and $q \gg 1$ both be real numbers. Consider the operator 
		\[Q_{\beta,q} \doteq -\frac{d^2}{dx^2} + \left(-\frac{2\beta}{x} + \frac{q(q+1)}{x^2}\right),\]
		acting on $L^2(0,\infty)$ with domain $C_0^{\infty}(0,\infty)$. Then
		\begin{enumerate}
			\item $Q_{\beta,q}$ is essentially self-adjoint. We then continue to denote the closure of $Q_{\beta,q}$ by $Q_{\beta,q}$.
			\item The discrete spectrum of $Q_{\beta,q}$ is given by the sequence
			\[\left\{\frac{-\beta^2}{(n+q)^2}\right\}_{n=1}^{\infty},\]
			where all eigenvalues have multiplicity one.
		\end{enumerate}
	\end{theo}
	\begin{proof}The essential self-adjointness of $Q_{\beta,q}$ follows because $Q_{\beta,q}$ is of limit-point type near $x = 0$ and $x = \infty$. (See the Appendix to X.1 in~\cite{reedsimonII}.) For the calculation of the discrete spectrum, see, for example, Chapter 32 of~\cite{quantformath}.

	\end{proof}
	
	In this next lemma we will use the Agmon multiplier identity from Lemma~\ref{agmonidentity} to show that the eigenfunctions of both $Q$ and $Q_{\beta,q}$ must be exponentially small relative to $L$ in the region $\{x \leq cL^2 \}$ 
	\begin{lemma}\label{expdecayeigeninclassicalforbid} Assuming that $q \gtrsim L$, there exists constants $c_1,c_2 > 0$ and so that for every eigenfunction $\psi$ of $Q$ or $Q_{\beta,q}$, we have
		\[\left\vert\left\vert \psi\right\vert\right\vert_{H^1\left(\left\{x \leq c_1L^2\right\}\right)} \leq e^{-c_2 L}\left\vert\left\vert \psi\right\vert\right\vert_{L^2}.\]
		In the case of the operator $Q$, the constants may be chosen independently of $x_0$. In the case of the operator $Q_{\beta,q}$, the constants $c_1$ and $c_2$ may be chosen to only depend on $\beta$.
	\end{lemma}
	\begin{proof}We will prove the lemma only for the operator $Q$ as it will then be clear how to repeat the argument \emph{mutatis mutandis} for $Q_{\beta,q}$.
		
		Choose a constant $\tilde{c}_1$ so that 
		\begin{equation}\label{lowerboundedtildeV}
			x \leq \tilde{c}_1L^2 \Rightarrow \tilde{V} \geq \frac{1}{2} \frac{L(L+1)}{x^2}.
		\end{equation}
		
		Now let $\psi$ be an eigenfunction of $Q$ with eigenvalue $\lambda< 0$ normalized so that $\left\vert\left\vert \psi\right\vert\right\vert_{L^2} = 1$. We may, moreover, assume that $\psi$ is real valued without loss of generality. Let $C_0 > 0$ be a large constant and choose a cut-off function $\chi(x)$ so that $\chi(x)$ is identically $1$ for $x \leq C_0L^p$ and is identically $0$ for $x \geq C_0L^p+ 1$. Then, if we set $\tilde{\psi} \doteq \chi\psi$, we will have that
		\begin{equation}\label{eigenfucntioneqnpsitoexponendecay}
			-\frac{d^2\tilde{\psi}}{dx^2} +\left(-\lambda + \tilde{V}\right)\tilde{\psi} = \tilde{H},
		\end{equation}
		where we have that ${\rm supp}\left(\tilde{H}\right) \subset \{x \geq C_0L^p\}$ and, by a straightforward elliptic estimate (consisting of multiplying \eqref{eigenfucntioneqnpsitoexponendecay} by  $\tilde{\psi}$ and integrating by parts) and the $L^2$ normalization of $\psi$ (note that $\left|{\rm min}\left(\tilde{V},0\right)\right|$ is uniformly bounded),
		\begin{equation}\label{goodestimatefortildeHinexpdecay}
			\left\vert\left\vert \tilde{H}\right\vert\right\vert_{L^2} \lesssim 1.
		\end{equation}

		Then we apply the identity~\eqref{theactualagmonidentity} to the equation~\eqref{eigenfucntioneqnpsitoexponendecay} with a function $\phi(x)$ which is identically $0$ for $x \geq \tilde{c}_1L^2$ and so that 
		\[x \leq \tilde{c}_1L^2 \Rightarrow \phi(x) = \frac{1}{2}\int_x^{\tilde{c}_1L^2}\sqrt{-\lambda + \tilde{V}}\, dx.\]
		Let $\tilde{c}_2$ be a positive constant so that $\tilde{c}_2 < \tilde{c}_1$. In view of~\eqref{lowerboundedtildeV}, we have that
		\[x \leq \tilde{c}_2L^2 \rightarrow \phi(x) \gtrsim \log\left(\frac{\tilde{c}_1}{\tilde{c}_2}\right)L.\]
		Using then also~\eqref{goodestimatefortildeHinexpdecay} we may immediately conclude from our application of~\eqref{theactualagmonidentity} (note that $\left(-\lambda + \tilde{V}\right) - \left(\phi'\right)^2$ is positive for $\{x \leq \tilde{c}_1L^2\}$) that there is a constant $\tilde{c}_3 > 0$ so that 
		\[\left\vert\left\vert \psi\right\vert\right\vert_{L^2\left(x \leq \tilde{c}_2L^2\right)} \lesssim L^{\alpha}e^{-L\tilde{c}_3\log\left(\frac{\tilde{c}_1}{\tilde{c}_2}\right)}.\]
		Carrying out an additional elliptic estimate and choosing $c_1$ and $c_2$ suitably then concludes the proof.
		
	\end{proof}
	
	Next, we show that eigenfunctions of $Q$ and $Q_{\beta,q}$ must also be exponentially small if $x$ is sufficiently large. However, now the necessary largeness of $x$ depends on the smallness of the eigenvalue.
	\begin{lemma}\label{expdecayeigeninclassicalforbid2}Let $\psi$ be an eigenfunction of $Q$ or $Q_{\beta,q}$ so that the corresponding eigenvalue $\lambda$ satisfies
		\[-\lambda \geq \frac{a}{L^p},\]
		for some $a \gtrsim 1$. In the case of the operator $Q_{\beta,q}$, assume also that $q \gtrsim L$. Then, there exists constants $c_3,c_4 > 0$ depending only a lower bound for $a$, and, in the case of $Q_{\beta,q}$, depending also on $\beta$ and a lower bound for $q L^{-1}$,  so that 
		\[\left\vert\left\vert \psi\right\vert\right\vert_{H^1\left(\left\{x \geq c_3^{-1}L^p\right\}\right)} \leq e^{-c_4 L}\left\vert\left\vert \psi\right\vert\right\vert_{L^2}.\] 
	\end{lemma}
	\begin{proof} The proof of this lemma is very similar to Lemma~\ref{expdecayeigeninclassicalforbid}, so we will just quickly note the essential differences. As in the proof of Lemma~\ref{expdecayeigeninclassicalforbid} we will consider only the case of the operator $Q$. 
		
		The role of~\eqref{lowerboundedtildeV} is now played by the existence of a constant $\tilde{c}_3$ so that 
		\begin{equation}\label{lowerboundedtildeV2}
			x \geq \tilde{c}_3L^p \Rightarrow -\lambda + \tilde{V} \geq \frac{1}{2}\left(-\lambda\right) \gtrsim_a L^{-p}.
		\end{equation}
		
		Let $A_1$ be an arbitrary large constant and then choose another large constant $A_2$ with $A_2$ sufficiently large, depending on  $A_1$. The cut-off $\chi$ is now defined to be identically $1$ for $x \in [2\tilde{c}_3^{-1}L^p + 1,A_2\tilde{c}_3^{-1}L^p-1]$ and is identically $0$ for $x \leq 2\tilde{c}_3^{-1}L^p$ and $x \geq A_2\tilde{c}_3^{-1}L^p$. The function $\phi(x)$ is defined by the following:
		\begin{enumerate}
			\item $\phi(x) = 0$ for $x \leq 2\tilde{c}_3^{-1}L^p + 1$.
			\item $\phi(x) = \frac{1}{2}\int_{2\tilde{c}_3^{-1}L^p + 1}^x\sqrt{-\lambda + \tilde{V}}\, dx$ for $x \in (2\tilde{c}_3^{-1}L^p + 1,A_1\tilde{c}_3^{-1}]$.
			\item $\phi(x) = \phi\left(A_1\tilde{c}_3^{-1}\right) - \frac{1}{2}\int_{A_1\tilde{c}_3^{-1}}^x\sqrt{-\lambda - \tilde{V}}\, dx$ for $x \in (A_1\tilde{c}_3^{-1},\tilde{A})$ where $\tilde{A}$ is chosen so that $\lim_{x \to \tilde{A}}\phi(x) = 0$.
			\item $\phi(x) = 0$ for $x \geq \tilde{A}$.
		\end{enumerate}
		Note that we can pick $A_2$ sufficiently large so that $\tilde{A} \ll A_2$.
		
		Then, since $\phi(x) \gtrsim L^{p/2}$ for $x \in [3\tilde{c}_3^{-1}L^p + 1,A_1\tilde{c}_3^{-1}L^p]$,  we may repeat the steps exactly as in the proof of Lemma~\ref{expdecayeigeninclassicalforbid} to obtain that 
		\[\left\vert\left\vert \psi\right\vert\right\vert_{L^2\left(x \in [3\tilde{c}_3^{-1}L^p + 1,A_1\tilde{c}_3^{-1}L^p]\right)} \lesssim e^{-\tilde{c}_4L},\]
		for a suitable constant $c_4$. Since the constant $A_1$ may be taken to be arbitrarily large, the proof is concluded.
	\end{proof}
	
	Finally, we are ready for the main eigenvalue estimates for $Q$.
	\begin{lemma}\label{eigenshavespace}Let $\{\lambda_n\}_{n=1}^{\infty}$ denote the eigenvalues of $Q$. 
		
		Let $N(p)$ denote the number of eigenvales $\lambda_n$ which satisfy $-\lambda_n \geq L^{-p}$. We have
		\[ N(p) \sim L^{p/2}.\]
		
		Next define a positive real number $P_L$ by 
		\[P_L \doteq -\frac{1}{2} + \frac{1}{2}\sqrt{1+4\left(L(L+1) + \mathbf{e}^2m^2\right)},\]
		and let $A \gg 1$ be a suitable large positive constant and $c > 0$ be a sufficiently small constant (independent of $L$).
		Then, for all $n \in [1,N(p)]$, we have
		\begin{equation}\label{theeigenvaluebounds}	
			\frac{\left(Mm^2\right)^2}{\left(n+P_L + \frac{A\log(P_L)}{P_L}\right)^2}-e^{-cP_L} \leq -\lambda_n \leq \frac{\left(Mm^2\right)^2}{\left(n+P_L - \frac{A\log(P_L)}{P_L}\right)^2}+e^{-cP_L}.
		\end{equation}
	\end{lemma}
	\begin{proof}We will use the Rayleigh--Ritz technique, Theorem XIII.3 of~\cite{reedsimonIV}.
		
		Let $c_1$, $c_2$, $c_3$, and $c_4$ be the constants from Lemmas~\ref{expdecayeigeninclassicalforbid} and~\eqref{expdecayeigeninclassicalforbid2}. with $\beta = Mm^2$. Then we observe that there exists a large constant $A \gg 1$, independent of $L$, so that if we define
		\[P_L^{\pm} \doteq P_L \pm \frac{A\log L}{L},\]
		then, for all real valued $\varphi \in D\left(Q\right) \cap D\left(Q_{\beta,P_L^-}\right) \cap D\left(Q_{\beta,P_L^+}\right)$
		\begin{equation}\label{upperboundforQ}
			\int_{x \in \left[c_1 L^2,c_3^{-1}L^p\right]} \left(Q\varphi\right)\varphi\, dx \leq \int_{x \in \left[c_1 L^2,c_3^{-1}L^p\right]}\left(Q_{\beta,P_L^+}\varphi\right)\varphi, dx,
		\end{equation}
		\begin{equation}\label{lowerboundforQ}
			\int_{x \in \left[c_1 L^2,c_3^{-1}L^p\right]} \left(Q\varphi\right)\varphi\, dx \geq \int_{x \in \left[c_1 L^2,c_3^{-1}L^p\right]}\left(Q_{\beta,P_L^-}\varphi\right)\varphi, dx.
		\end{equation}
		
		Choose $\tilde{N} \in \mathbb{Z}_{\geq 1}$ be arbitrary subject to the constraint that the first $\tilde{N}$ eigenvalues $\{\tilde{\lambda}_n\}_{n=1}^{\tilde{N}}$ of $Q_{Mm^2,P_L^+}$ all satisfy $-\tilde{\lambda}_N \gtrsim L^{-p}$. Then let $\left\{\psi_n\right\}_{n=1}^{\tilde{N}}$ denote the first $\tilde{N}$ eigenfunctions of $Q_{Mm^2,P_L^+}$ all normalized to have $\left\vert\left\vert \psi_n\right\vert\right\vert_{L^2} = 1$ and taken to be real valued. Let $\chi(x)$ be a cut-off function which is identically $0$ for $x \leq x_0 + 1$ and is identically $1$ for $x \geq x_0 + 2$. Then set $\tilde{\psi}_n \doteq \chi \psi_n$. Set $V^{\tilde{N}} \doteq {\rm span}\left(\tilde{\psi}_1,\cdots,\tilde{\psi}_{\tilde{N}}\right)$. Standard elliptic theory implies that $V^N$ is a subset of the domain of both $Q$ and $Q_{\beta,p}$ for any $\beta$ and $p$.
		
		Let $n \in [1,\tilde{N}]$. By the Rayleigh--Ritz technique (Theorem XIII.3 of~\cite{reedsimonIV}) we have that
		\begin{equation}\label{rayritzforlambdan}
			\lambda_n \leq \sup_{\varphi_1,\cdots,\varphi_{n-1} \in V^{\tilde{N}}}\left(\inf_{\varphi \in V^{\tilde{N}} : \left\vert\left\vert \varphi\right\vert\right\vert_{L^2} = 1,\ \varphi \in \{\varphi_1,\cdots,\varphi_{n-1}\}^{\perp}}\int_{x_0}^{\infty}\left(Q\varphi\right)\varphi\, dx\right)
		\end{equation}
		There exists a constant $c > 0$ so that, in view of~\eqref{upperboundforQ}, Lemma~\ref{expdecayeigeninclassicalforbid}, and Lemma~\ref{expdecayeigeninclassicalforbid2}, any $\varphi \in V^{\tilde{N}}$ with $\left\vert\left\vert \varphi\right\vert\right\vert_{L^2} = 1$ will satisfy
		\begin{equation}\label{uppertopluginrayritz}
			\int_{x_0}^{\infty}\left(Q\varphi\right)\varphi\, dx \leq e^{-c P_L} + \int_{x_0}^{\infty}\left(Q_{Mm^2,P_L^+}\varphi\right)\varphi\, dx.
		\end{equation}
		Plugging~\eqref{uppertopluginrayritz} into~\eqref{rayritzforlambdan}, using that the $\psi_n$ are eigenfunctions of $Q_{Mm^2,P_L^+}$ and using Lemmas~\ref{expdecayeigeninclassicalforbid} and~\ref{expdecayeigeninclassicalforbid2}  again, then yields, (after possibly redefining the constant $c > 0$)
		\begin{equation}\label{wegottheupperboundforlambdan}
			\lambda_n \leq e^{-cP_L} -\frac{\left(Mm^2\right)^2}{\left(n+P_L + \frac{A\log(P_L)}{P_L}\right)^2}.
		\end{equation}
		
		It is now clear that in this argument we can replace $Q$ wit $Q_{Mm^2,P_L^-}$ and $Q_{Mm^2,P_L^+}$ with $Q$ to obtain a lower bound for $\lambda_n$. Putting everything together concludes the proof.
	\end{proof}

	\begin{rmk} \label{approx.poor}When $n \gtrsim e^{cP_L}$, the bound~\eqref{theeigenvaluebounds} becomes essentially useless. However, we will only be interested in applying this estimate in regimes where $n \lesssim L^p \sim P_L^p$.
	\end{rmk}

	\subsection{Exponential multipliers for all frequencies}
	In this section we prove estimate for $u$ and $\partial_{\omega}u$ which holds for all frequencies satisfying~\eqref{mbiggerthanomega},~\eqref{Llarge2}, and~\eqref{thebasiccomparableassump}. The key downside of these estimates is that they will have bad dependence on $L$. The estimates of this section will neither use the eigenvalue estimates of Section~\ref{ellipatinfinity} nor the Agmon multiplier identity of Lemma~\ref{agmonidentity}. 
	
	We first consider the (easier) case when $m^2 -\omega^2 \gg L^{-2}$.

	\begin{lemma}\label{easymwaybig}Assume there exists a sufficiently large constant $A > 0$, independent of $L$, so that $m^2-\omega^2 \geq \frac{A}{L^2}$ and that $L$ is sufficiently large. Let $u$ be a weakly outgoing solution to~\eqref{eq:mainVrhs}. Then we have, for any $\delta > 0$,
		\[\int_{-\infty}^{\infty}\left[\frac{r}{r-r_+}\left|\partial_su +i\omega u\right|^2 + \frac{1}{\left({\rm max}\left(0,-s\right)\right)^{1+\delta}+1}\left(\left|\partial_su\right|^2 + \left|u\right|^2\right)\, \right]ds \lesssim_{\delta} L^4\int_{-\infty}^{\infty}\frac{r}{|r-r_+|}\left|H\right|^2\, ds.\]
	\end{lemma}
	\begin{proof}Let $s_0$ be as in Lemma~\ref{redshiftode}. Then let $h(s)$ be a function so that $h(s)$ vanishes for $s \leq s_0-10$ and $h(s)$ is identically $1$ for $s \geq s_0 -9$. 
		
		We observe that under the assumptions of the lemma, we have that $V|_{s \geq s_0} \gtrsim L^{-2} + L^2\frac{r-r_+}{r^3}$.
		Applying the identity~\eqref{hcurent} and using Cauchy--Schwarz then leads to the estimate
		\begin{equation}\label{thehcurrentgavemethis}
			\int_{s_0-9}^{\infty}\left[\left|\partial_su\right|^2 + \left(L^{-2}+L^2r^{-2}\right)\left|u\right|^2\right]\, ds \lesssim \int_{s_0-10}^{s_0-9}\left|u\right|^2\, ds + L^2\int_{s_0-10}^{\infty}\left|H\right|^2\, ds.
		\end{equation}

		In view of the fact that there is no $L^2$ multiplying the $|u|^2$ on the right hand side of~\eqref{thehcurrentgavemethis}, it is clear that by combining Lemma~\ref{redshiftode} and~\eqref{thehcurrentgavemethis} we may conclude the proof.
	\end{proof}
	
	In the next lemma, we provide an analogue of Lemma~\ref{easymwaybig} for $\partial_{\omega}u$.
	\begin{lemma}\label{easymwaybigderivativeomega}Assume there exists a sufficiently large constant $A > 0$, independent of $L$, so that $m^2-\omega^2 \geq \frac{A}{L^2}$ and that $L$ is sufficiently large. Let $u$ be an outgoing solution to~\eqref{eq:mainVrhs}. Let $\chi(s)$ be a cut-off function which is $1$ for $s \leq -1$ and is identically $0$ for $s \geq 0$. Then define
		\[q \doteq e^{-i\omega s \chi(s)}\partial_{\omega}\left(e^{i\omega s \chi(s)}u\right).\]
		Then we have, for any $\delta > 0$,
		\begin{align*}
			&\int_{-\infty}^{\infty}\left[\frac{r}{r-r_+}\left|\partial_sq +i\omega q\right|^2 + \frac{1}{\left({\rm max}\left(0,-s\right)\right)^{1+\delta}+1}\left(\left|\partial_sq\right|^2 + \left|q\right|^2\right)\, \right]ds \lesssim_{\delta}
			\\ \nonumber &\qquad \qquad \qquad L^8\int_{-\infty}^{\infty}\frac{r}{|r-r_+|}\left[\left|H\right|^2+\left|e^{-iws \chi}\partial_{\omega}\left(e^{i\omega s \chi}H\right)\right|^2\right]\, ds.
		\end{align*}
	\end{lemma}
	\begin{proof}The outgoing condition for $u$ implies that $q$ will be weakly outgoing. The result then follows from repeated applications of Lemma~\ref{easymwaybig} to the equation for $q$ from Lemma~\ref{qeqnusefultodiffu}.
	\end{proof}
	
	In view of Lemmas~\ref{easymwaybig} and~\ref{easymwaybigderivativeomega} we will be able to assume, without loss of generality in the rest of the section that $m^2-\omega^2 \lesssim L^{-2}$. In particular, this allows us to use Lemma~\ref{agoodenergyestimateiguess2} to control $\left|u(-\infty)\right|$. We start with an estimate for $u$.
	\begin{lemma}\label{thisisgeneralforubutbadconstant} Assume that $ L^{-p} \leq m^2-\omega^2 \lesssim L^{-2}$, and let $u$ be a weakly outgoing solution to~\eqref{eq:mainVrhs}. Then we have the following two estimates for $u$:
		\begin{enumerate}
			\item Let $R > r_+$. Then there exists a constant $D_0$, depending only on $R$ so that, for every $\delta > 0$,
			\begin{align}\label{Rboundedonlyexplossforu}
				&\int_{-\infty}^{s\left(R\right)}\left(\frac{r}{r-r_+}\left|\partial_su+i\omega u\right|^2 + \frac{1}{1+|s|^{1+\delta}}\left[\left|\partial_su\right|^2 + \left|u\right|^2\right]\right)\, ds \lesssim_{\delta}
				\\ \nonumber &\qquad e^{D_0 L}\int_{-\infty}^{s(R)}\frac{r}{r-r_+}\left|H\right|^2\, ds + e^{D_0L}\int_{s(R)}^{\infty}\left|H\right|\left|u\right|\, ds.
			\end{align}
			\item If $H$ is compactly supported in $\{r \leq R\}$, then there exists a constant $D_0$, depending only on $R$ so that, for every $\delta > 0$,
			\begin{equation}\label{Rboundedonlyexplossforu2}
				\int_{-\infty}^{s\left(R\right)}\left(\frac{r}{r-r_+}\left|\partial_su+i\omega u\right|^2 + \frac{1}{1+|s|^{1+\delta}}\left[\left|\partial_su\right|^2 + \left|u\right|^2\right]\right)\, ds \lesssim_{\delta} e^{D_0 L}\int_{-\infty}^{s(R)}\frac{r}{r-r_+}\left|H\right|^2\, ds.
			\end{equation}
			\item For every $\delta,\tilde{\delta} > 0$,
			\begin{align}\label{nocompactsupportsolosellogl}
				&\int_{-\infty}^{\infty}\left[\frac{r}{r-r_+}\left|\partial_su +i\omega u\right|^2 + \frac{1}{\left({\rm max}\left(0,-s\right)\right)^{1+\delta}+1}\left(\left|\partial_su\right|^2 + \left|u\right|^2\right)\, \right]ds \lesssim_{\delta} 
				\\ \nonumber &\qquad \qquad e^{(4+\tilde{\delta}) \LL \log \LL}\int_{-\infty}^{\infty}\frac{r}{r-r_+}\left|H\right|^2\, ds.
			\end{align}
		\end{enumerate}
	\end{lemma}
	\begin{proof}Let $\delta_1,\delta_2 > 0$  be arbitrary small constants, let $C \gg 1$ be sufficiently large depending only $p$, and $A \gg 1$ be sufficiently large constant, possibly depending on $C$, $\delta_1$, and $\delta_2$. Then let $\delta_1, \delta_2 > 0$, and define
		\[x(s) \doteq \sqrt{\frac{r-r_+}{r}}\left(2\LL\left(1+\delta_1\right) r^{-1} + A \LL r^{-1-\delta_2} + A r^{-1/2}
		\right),\qquad y(s) \doteq -\exp\left(-\int_{-\infty}^s x(s)\, ds\right).\]
		
		Let $\check{s}$ be sufficiently large and let $\chi(s)$ be a cut-off function which is $1$ for $s \leq \check{s}$, is identically $0$ for $s \geq \check{s}$, and satisfies $\partial_s\chi \leq 0$ everywhere. Then, choose and fix a constant $C_0$ (after picking $\check{s}$ big enough and then fixing $\check{s}$) so that 
		\begin{equation}\label{definetildev}
			\left|V+\chi(s)\omega^2\right| \leq \tilde{V}\left(s\right) \doteq C_0\LL^2 \left(\frac{r-r_+}{r}\right)\chi(s) + \left(1-\chi(s)\right)\left(m^2-\omega^2+\frac{C_0}{r} + \frac{\LL^2\left(1+\delta_1^{100}\right)}{r^2}\right).
		\end{equation}
		
		We note that (keeping in mind that the largeness of $A$ may depend on $C$ and that $m^2-\omega^2 \lesssim \LL^{-2}$)
		\begin{equation}\label{usefulcomparisonforpandV}
			r \in [r_+,C\LL^2] \Rightarrow	\left|\tilde{V}\right| \leq {\rm min}\left(\frac{r^{2+2\delta_2}}{A^{2-\delta_1} 
			},1\right)\frac{1}{4}\left(1+\delta_1\right)^{-2}\left(1+\delta_1^{100}\right)x^2,
		\end{equation}
		\begin{equation}\label{derxgainsanr}
			|\partial_s x| \lesssim r^{-1} x,
		\end{equation}
		\begin{equation}\label{useful22}
			\left|\partial_s\tilde{V}\right| \lesssim  r^{-1}\left|\tilde{V}\right|,	\end{equation}
		and also that $\partial_sy =x[-y]> 0$ everywhere, $y(-\infty) = -1$, and $y(\infty) = 0$.
		
		Now we apply~\eqref{ymulteqn2} and obtain (keeping in mind that $y\partial_s\chi \geq 0$)
		\begin{align}\label{expmultbeginningoftheend}
			&\int_{-\infty}^{\infty}\left[\left(-y\right)x\left|\partial_su\right|^2 + \omega^2 \left(-y\right)x \chi(s)\left|u\right|^2\right]\, ds \leq 
			\\ \nonumber &\qquad \omega^2\left|u(-\infty)\right|^2 s +  \left|\int_{-\infty}^{\infty}\partial_s\left(y\left(\chi(s)\omega^2+V\right)\right)\left|u\right|^2\, ds\right| + 2\int_{-\infty}^{\infty}(-y)\left|H\right|\left|\partial_su\right|\, ds.
		\end{align}
		For the second term on the right hand side of~\eqref{expmultbeginningoftheend} we can integrate by parts to move the $\partial_s$ derivative onto the $\left|u\right|^2$ and use Cauchy-Schwarz to obtain that (keeping~\eqref{definetildev} in mind) 
		\begin{align}\label{expmultbeginningoftheend2}
			&\int_{-\infty}^{\infty}\left[\frac{1}{2}\left(-y\right)x\left|\partial_su\right|^2 + \omega^2 \left(-y\right)x \chi\left|u\right|^2\right]\, ds \leq 
			\\ \nonumber &\qquad \omega^2\left|u(-\infty)\right|^2  + 2\int_{-\infty}^{\infty}(-y)\frac{\tilde{V}^2}{x}\left|u\right|^2\, ds + 2\int_{-\infty}^{\infty}(-y)\left|H\right|\left|\partial_su\right|\, ds.
		\end{align}

		Our plan is now to use a Hardy type inequality to combine the first and second term on the left hand side of~\eqref{expmultbeginningoftheend2} to absorb the second term  on the right hand side of~\eqref{expmultbeginningoftheend2} when restricted to $r \in [r_+,C\LL^2]$.

		More specifically:
		\begin{align}\label{aboundarytermdropped}
			\int_{-\infty}^{s\left(C\LL^2\right)} (-y)\frac{\tilde{V}^2}{x}\left|u\right|^2\, ds &= \int_{-\infty}^{s\left(C\LL^2\right)}\left(\partial_sy\right)\frac{\tilde{V}^2}{x^2}\left|u\right|^2\, ds
			\\ \nonumber &\leq \int_{-\infty}^{s\left(C\LL^2\right)}(-y)\partial_s\left(\frac{\tilde{V}^2}{x^2}\right)\left|u\right|^2\, ds + 2\int_{-\infty}^{s\left(C\LL^2\right)}(-y)\frac{\tilde{V}^2}{x^2}\left|\partial_su\right|\left|u\right|\, ds \Rightarrow
		\end{align}
		\begin{align}\label{ahardytypeineq}
			&2\int_{-\infty}^{s\left(C\LL^2\right)} (-y)\frac{\tilde{V}^2}{x}\left|u\right|^2\, ds \leq
			\\ \nonumber &\qquad4\int_{-\infty}^{s\left(C\LL^2\right)}(-y)\partial_s\left(\frac{\tilde{V}^2}{x^2}\right)\left|u\right|^2\, ds + 8\int_{-\infty}^{s\left(C\LL^2\right)}(-y)\frac{\tilde{V}^2}{x^3}\left|\partial_su\right|^2\, ds.
		\end{align}
		Note that the boundary term we dropped in the transition from the first to the second line of~\eqref{aboundarytermdropped} has a favorable sign. We now turn to a further analysis of the terms on the right hand side of~\eqref{ahardytypeineq}.

		We start with the first term on the right hand side of~\eqref{ahardytypeineq}. Choose $\tilde{s}\left(\LL\right) < 0$ be the most negative value of $s$ so that $\left(r\left(\tilde{s}\right)-r_+\right)^{1/2}\LL \geq 1$. In view of~\eqref{derxgainsanr} and~\eqref{useful22} we have that
		\begin{align}\label{ithinkthistermdoesntmatter}
			\left|\partial_s\left(\frac{\tilde{V}^2}{x^2}\right)\right| &\lesssim r^{-1} \frac{\tilde{V}^2}{x^2}
			\\ \nonumber &= r^{-1} \frac{\tilde{V}^2}{x^2} 1|_{\left\{s \leq \tilde{s}\left(\LL\right)\right\}}+r^{-1} \frac{\tilde{V}^2}{x^2} 1|_{\left\{\tilde{s}\left(\LL\right)\leq  s \leq 0\right\}} + r^{-1} \frac{\tilde{V}^2}{x^2} 1|_{ \left\{s\geq 0\right\}}
			\\ \nonumber &\lesssim \frac{\LL^2}{A^{2-\delta_1}}\left(r-r_+\right) 1|_{\left\{s \leq \tilde{s}\right\}} +A^{-1} \frac{\tilde{V}^2}{x} 1|_{\left\{\tilde{s}\left(\LL\right) \leq s \leq 0\right\}}+ r^{-1/2}A^{-1} \frac{\tilde{V}^2}{x} 1|_{ \left\{s\geq 0\right\}}\\ \nonumber &\lesssim A^{-2+\delta_1}x 1|_{\left\{s \leq \tilde{s}\right\}} +A^{-1} \frac{\tilde{V}^2}{x} 1|_{\left\{\tilde{s}\left(\LL\right) \leq s \leq 0\right\}}+ r^{-1/2}A^{-1} \frac{\tilde{V}^2}{x} 1|_{ \left\{s\geq 0\right\}}.
		\end{align}

		Then, using again~\eqref{usefulcomparisonforpandV}, we finally obtain 
		\begin{align}\label{ahardytypeineq2}
			&2\int_{-\infty}^{s\left(C\LL^2\right)} (-y)\frac{\tilde{V}^2}{x}\left|u\right|^2\, ds \leq
			\\ \nonumber &\qquad \frac{1}{2}\left(1+\delta_1\right)^{-1/2}\int_{\infty}^{s\left(C\LL^2\right)}(-y)x\left|\partial_su\right|^2\, ds +  A^{-2+2\delta_1}\int_{-\infty}^{\tilde{s}}(-y)x\left|u\right|^2\, ds.
		\end{align}

		Combining~\eqref{expmultbeginningoftheend2} and~\eqref{ahardytypeineq2} thus leads to
		\begin{align}\label{expmultbeginningoftheend3}
			&\int_{-\infty}^{\infty}\left[\left(-y\right)x\left|\partial_su\right|^2 + \omega^2 \left(-y\right)x\left|u\right|^2\right]\, ds \lesssim_{\delta_1} 
			\\ \nonumber &\qquad \qquad \qquad \left|u(-\infty)\right|^2+\int_{s\left(C\LL^2\right)}^{\infty}(-y)\frac{\tilde{V}^2}{x}\left|u\right|^2\, ds+ \int_{-\infty}^{\infty}(-y)\left|H\right|\left|\partial_su\right|\, ds.
		\end{align}
		Now we add in a suitable multiple of~\eqref{agoodenergyestimateiguess2.eq} to obtain
		\begin{align}\label{expmultbeginningoftheend4}
			&\int_{-\infty}^{\infty}\left[\left(-y\right)x\left|\partial_su\right|^2 + \omega^2 \left(-y\right)x\left|u\right|^2\right]\, ds \lesssim_{\delta_1} 
			\\ \nonumber &\qquad \qquad \qquad 	\int_{s\left(C\LL^2\right)}^{\infty}(-y)\frac{\tilde{V}^2}{x}\left|u\right|^2\, ds+ \int_{-\infty}^{\infty}\left[(-y)\left|H\right|\left|\partial_su\right|+ \left|H\right|\left|u\right|\right]\, ds.
		\end{align}
		
		Next we observe that in view of the largeness of $C$ we have that  $r \geq \frac{C}{100}\LL^2$ implies that
		\[\partial_sV \gtrsim r^{-2}.\]
		Now let $\tilde{y}(s)$ be a function which is $0$ for $r \leq \frac{C}{100}\LL^2-1$ and is equal to $-1$ for $r \geq \frac{C}{100}\LL^2 $. Applying~\eqref{ymulteqn2} with $y = \tilde{y}$ leads to the estimate
		\begin{align}\label{gettinguptowherevispositive}
			\int_{\frac{C}{100}\LL^2}^{\infty}r^{-2}\left|u\right|^2\, ds \lesssim \int_{\frac{C}{100}\LL^2-1}^{\frac{C}{100}\LL^2}\left[\left|\partial_su\right|^2 + \left|u\right|^2\right]\, ds + \int_{\frac{C}{100}\LL^2-1}^{\infty}\left|H\right|\left|\partial_su\right|\, ds.
		\end{align}

		Next, we observe that, after possibly increasing the constant $C$, depending only on $p$, we will have that
		\[r \geq C\LL^p \Rightarrow V \gtrsim  \LL^{-p}.\]
		Furthermore,
		\[r \in [\frac{C}{100}\LL^2,C\LL^p] \Rightarrow V \gtrsim \LL^{-p} - r^{-1} \gtrsim \LL^{-p} - \LL^p r^{-2}.\] 
		Let $h(s)$ be any smooth function so that $h$ is identically $1$ for $r \geq \frac{C}{100}\LL^2$ and vanishes identically for $r \leq \frac{C}{100}\LL^2-1$ and so that $\partial^2_sh$ is uniformly bounded. From~\eqref{hcurent} we then get
		\begin{equation}\label{fixthelargerpartwithabetterweight}
			\int_{\frac{C}{100}L^2}^{\infty}\left[\left|\partial_su\right|^2 + \LL^{-p}\left|u\right|^2\right]\, ds \lesssim \int_{\frac{C}{100}\LL^2-1}^{\frac{C}{100}\LL^2}\left|u\right|^2\, ds + \LL^p \int_{\frac{C}{100}\LL^2}^{C\LL^p} r^{-2}\left|u\right|^2\, ds + \int_{\frac{C}{100}\LL^2-1}^{\infty}\left|H\right|\left|u\right|\, ds.
		\end{equation}
		Next we observe that,
		\[e^{-4 \left(1+\delta_1\right)\LL \log\LL + O\left(\LL\right)} \leq y|_{r = \frac{C}{100}\LL^2} \leq e^{-4\left(1+\delta_1\right) \LL \log\LL + O\left(\LL\right)},\]
		\[  \left|y\right||_{r \geq C\LL^2} \leq \exp\left(-\frac{1}{4}A \left(\sqrt{r} -\sqrt{\frac{C}{100}\LL^2}\right)\right)y|_{r = \frac{C}{100}\LL^2} \lesssim \exp\left(-A^{1/2}\LL\right) y|_{r = \frac{C}{100}\LL^2} r^{-100}.\]
		Thus, we can now multiply~\eqref{fixthelargerpartwithabetterweight} with a suitable constant and combine with~\eqref{expmultbeginningoftheend4} and~\eqref{gettinguptowherevispositive} to obtain
		\begin{align}\label{wowwemadeittotheendoftheproof}
			&\int_{-\infty}^{\frac{C}{100}L^2}\left[\left(-y\right)x\left|\partial_su\right|^2 + \omega^2 \left(-y\right)x\left|u\right|^2\right]\, ds + e^{-4\left(1+\delta_1\right)\LL \log \LL + O\left(\LL\right)}\int_{\frac{C}{100}L^2}^{\infty}\left[\left|\partial_su\right|^2 + \left|u\right|^2\right]\, ds \lesssim
			\\ \nonumber &\qquad \qquad \LL^{\alpha}\int_{-\infty}^{\infty}\left[(-y)\left|H\right|\left|\partial_su\right|+ \left|H\right|\left|u\right|\right]\, ds + \LL^{\alpha}e^{-4\left(1+\delta_1\right)\LL \log \LL + O\left(\LL\right)}\int_{\frac{C}{100}L^2-1}^{\infty}\left|H\right|\left|u\right|\, ds.
		\end{align}
		The proof is then concluded by applying Cauchy--Schwarz, adding a suitable multiple of the estimate from Lemma~\ref{redshiftode}, and noting that for any fixed $R$, we will have that there exists a constant $C_R$, depending only on $R$ with 
		\[y|_{r \leq R} \geq e^{-C_R \LL}.\]
	\end{proof}
	
	We now establish an analogue of Lemma~\ref{thisisgeneralforubutbadconstant} to $\partial_{\omega}u$.
	\begin{lemma}\label{thisisgeneralforderofubutbadconstant}Let $u$ be an outgoing solution to~\eqref{eq:mainVrhs}, let $\chi(s)$ be a cut-off function which is $1$ for $s \leq -1$ and is identically $0$ for $s \geq 0$, and then define
		\[q \doteq e^{-i\omega s \chi(s)}\partial_{\omega}\left(e^{i\omega s \chi(s)}u\right).\]
		
		Then, for every $\delta, \tilde{\delta} > 0$,
		\begin{align}\label{nocompactsupportsolosellogl2}
			&\int_{-\infty}^{\infty}\left[\frac{r}{r-r_+}\left|\partial_sq +i\omega q\right|^2 + \frac{1}{\left({\rm max}\left(0,-s\right)\right)^{1+\delta}+1}\left(\left|\partial_sq\right|^2 + \left|q\right|^2\right)\, \right]ds \lesssim_{\delta} 
			\\ \nonumber &\qquad \qquad e^{\left(8+\tilde{\delta}\right) \LL \log \LL}\int_{-\infty}^{\infty}\frac{r}{r-r_+}\left[\left|H\right|^2+\left|e^{-i\omega s \chi}\partial_{\omega}\left(e^{i\omega s\chi(s)}H\right)\right|^2\right]\, ds.
		\end{align}
	\end{lemma}
	\begin{proof}We will have that $q$ is weakly outgoing. Then we may repeatedly apply Lemma~\ref{thisisgeneralforubutbadconstant} to the equation that $q$ satisfies which is computed in Lemma~\ref{qeqnusefultodiffu}. 
	\end{proof}
	\begin{rmk}We emphasize one crucial difference between Lemma~\ref{thisisgeneralforderofubutbadconstant} and Lemma~\ref{thisisgeneralforubutbadconstant}. In Lemma~\ref{thisisgeneralforderofubutbadconstant}, there is no way to exploit an assumed compact support of $H$ to replace the $e^{(4+\tilde{\delta}) \LL \log \LL}$ on the right hand side with a $e^{DL}$. 
	\end{rmk}

	\subsection{Improved estimates away from the critical frequencies}
	In this section we will establish estimates for $u$ and $\partial_{\omega}u$ which, while not useful for all frequencies $\omega$, will, for many frequencies, offer a dramatic improvement over the constants on the right hand side of Lemmas~\ref{thisisgeneralforubutbadconstant} and~\ref{thisisgeneralforderofubutbadconstant}.
	
	The following lemma is proven by combining an Agmon estimate with the red shift effect at the event horizon. The resulting estimate allows us to control the solution $u$ in a compact region if we have relatively weak a priori control of $u$ in $L^2$ globally. 
	\begin{lemma}\label{agmonplusredest}Let $u$ be a weakly outgoing solution to~\eqref{eq:mainVrhs}, and $\delta_1$ and $\tilde{\delta}$ be arbitrarily small positive constants and $S_0$ be a large positive constant. Then, as long as $\LL$ is sufficiently large, we have that 
		\begin{align}
			&\int_{-\infty}^{S_0}\left[\frac{r}{r-r_+}\left|\partial_su+i\omega u\right|^2 + \frac{1}{\left({\rm max}\left(0,-s\right)\right)^{1+\delta_1}+1}\left(\left|\partial_su\right|^2 + \left|u\right|^2\right)\right]\, ds \lesssim_{\delta_1,\tilde{\delta}}
			\\ \nonumber &\qquad e^{-\left(4-\tilde{\delta}\right) \LL \log \LL}\int_{0}^{\infty}\left|u\right|^2\, ds + \LL^2\int_{-\infty}^{\infty}\frac{r}{r-r_+}\left|H\right|^2\, ds.
		\end{align}
	\end{lemma}
	\begin{proof}Let $\delta_2,\delta_3 > 0$ be  a arbitrary small constants. We can find a sufficiently small constant $\tilde{c} > 0$ and a large constant $C > 0$, both independent of $L$, so that 
		\begin{equation}\label{viscomptosthisinabigrange}
			r \in [s_0-2,\tilde{c}L^2] \Rightarrow C^{-1}\frac{L^2(r-r_+)}{r^3} \leq V \leq C \frac{L^2 (r-r_+)}{r^3},
		\end{equation}
		where $s_0<0$ is as in Lemma~\ref{redshiftode}.
		
		In particular,
		\begin{equation}\label{thephigoesupandgoesdown}
			\int_{2S_0}^{\tilde{c} \LL^2}\sqrt{V}\, ds \geq 2\left(1- \delta_2\right)\LL\log \LL,
		\end{equation}
		
		We define a Lipschitz function $\phi(s)$ by
		\begin{enumerate}
			
			\item $\phi(s) = -\left(1-\delta_3\right)\int_s^{\frac{1}{2}S_0}\sqrt{V}\, ds+\left(1-\delta_3\right)\int_{2S_0}^{\tilde{c}\LL^2}\sqrt{V}\, ds$ for $s \leq \frac{1}{2}S_0$.
			\item $\phi(s) = \left(1-\delta_3\right)\int_{2S_0}^{\tilde{c}\LL}\sqrt{V}\, ds$ for $s \in \left[\frac{1}{2}S_0,2S_0\right]$.
			\item $\phi(s) = \left(1-\delta_3\right)\int_s^{\tilde{c}\LL^2}\sqrt{V}\, ds$ for $s \in \left[2S_0,\tilde{c}\LL^2\right]$.
			\item $\phi(s) = 0$ for $s \geq \tilde{c}\LL^2$.
		\end{enumerate}
		
		Now we let $\chi_0(s)$ be any cut-off function which is $1$ for $s \in [s_0-1,\tilde{c}\LL^2-1]$ and vanishes for $s \leq s_0-2$ and $s \geq \tilde{c}\LL^2$. Then set $\tilde{u} \doteq \chi_0 u$. We will have that
		\begin{equation}\label{uhasbenncutbychi0}
			\partial_s^2\tilde{u}-V\tilde{u}  = \chi_0 H+ 2\chi_0' \partial_su + \chi_0'' u.
		\end{equation}
		We now apply Lemma~\ref{agmonidentity} to~\eqref{uhasbenncutbychi0} with the function $\phi$ defined above. We obtain, for some constant $c > 0$,  (note that $\phi(s)=O(1)$ for $s \in [\tilde{c}\LL^2-1,\tilde{c}\LL^2]$, while $\phi(s)- \underbrace{(1-\delta_3) \int_{2S_0}^{\tilde{c}\LL^2} \sqrt{V} ds}_{2(1-\delta_3)\LL \log(\LL)+O(\LL)}=O(\LL)$ for $s \in [s_0-2,s_0-1]$) the estimate
		\begin{align}\label{afirstconseqofagmonforlarger}
			&\int_{\frac{1}{2}S_0}^{2S_0}\left[\left|u\right|^2+ \left|\rd_s u\right|^2\right]\, ds \\ \nonumber &\qquad \lesssim L^2\int_{s_0-2}^{\tilde{c}\LL^2}\left|H\right|^2 + e^{-4\left(1-\delta_3\right)\left(1-\delta_2\right)L \log L}\int_{\tilde{c}\LL^2-1}^{\tilde{c}\LL^2}\left(\left|\partial_su\right|^2 + \left|u\right|^2\right)\, ds + e^{-c\LL}\int_{s_0-2}^{s_0-1}\left(\left|\partial_su\right|^2 + \left|u\right|^2\right)\, ds
			\\ \nonumber &\qquad \lesssim L^2\int_{s_0-2}^{\tilde{c}\LL^2}\left|H\right|^2 + e^{-4\left(1-\delta_3\right)\left(1-\delta_2\right)L \log L}\int_{\tilde{c}\LL^2-1}^{\tilde{c}\LL^2} \left|u\right|^2\, ds + e^{-c\LL}\int_{s_0-2}^{s_0-1} \left|u\right|^2\, ds,
		\end{align}
		where we used a standard elliptic estimate to move from the first line to the second line.
		
		Now let $h(s)$ be a cut-off function which is identically $1$ for $s \in [s_0-1,\frac{3}{4}S_0]$ and is identically $0$ for $s \leq s_0-2$ and $s \geq \frac{5}{6}S_0$. We then apply~\eqref{hcurent} to obtain
		\[\int_{s_0-1}^{\frac{3}{4}S_0}\left[\frac{r-r_+}{r^3}L^2\left|u\right|^2+ |\rd_s u|^2\right]\, ds \lesssim \int_{s_0-2}^{s_0-1}\left|u\right|^2 + \int_{\frac{3}{4}S_0}^{\frac{5}{6}S_0}\left|u\right|^2\, ds + \int_{s_0-2}^{\frac{5}{6}S_0}\left|H\right|^2\, ds.\]
		Combining with Lemma~\ref{redshiftode} then yields
		\begin{align*}
			&\int_{-\infty}^{\frac{3}{4}S_0}\left[\frac{r}{r-r_+}\left|\partial_su+i\omega u\right|^2 + \frac{1}{\left({\rm max}\left(0,-s\right)\right)^{1+\delta_1}+1}\left(\left|\partial_su\right|^2 + \left|u\right|^2\right)\right]\, ds \lesssim_{\delta_1} 
			\\ \nonumber &\qquad \qquad  \int_{\frac{3}{4} S_0}^{\frac{5}{6} S_0}\left|u\right|^2\, ds + \int_{-\infty}^{\frac{5}{6} S_0}\frac{r}{r-r_+}\left|H\right|^2\, ds.
		\end{align*}
		Combining with~\eqref{afirstconseqofagmonforlarger} finishes the proof.
	\end{proof}
	
	Before proceeding with the main estimates of the section, we note that Lemma~\ref{agmonplusredest} allows to improve the $\LL$-dependence of the estimate from Lemma~\ref{thisisgeneralforderofubutbadconstant} is we restrict the range of integration on the left hand side to a compact range of $r$.
	\begin{cor}
		\label{thisisgeneralforderofubutbadconstant2}Let $u$ be an outgoing solution to~\eqref{eq:mainVrhs}, let $\chi(s)$ be a cut-off function which is $1$ for $s \leq -1$ and is identically $0$ for $s \geq 0$, and then define
		\[q \doteq e^{-i\omega s \chi(s)}\partial_{\omega}\left(e^{i\omega s \chi(s)}u\right).\]
		
		Then, for every large constant $S_0$ and small constants $\delta, \tilde{\delta} > 0$,
		\begin{align}\label{nocompactsupportsolosellogl22}
			&\int_{-\infty}^{2S_0}\left[\frac{r}{r-r_+}\left|\partial_sq +i\omega q\right|^2 + \frac{1}{\left({\rm max}\left(0,-s\right)\right)^{1+\delta}+1}\left(\left|\partial_sq\right|^2 + \left|q\right|^2\right)\, \right]ds \lesssim_{S_0,\delta} 
			\\ \nonumber &\qquad \qquad e^{\left(4+\tilde{\delta}\right) \LL \log \LL}\int_{-\infty}^{\infty}\frac{r}{r-r_+}\left[\left|H\right|^2+\left|\partial_{\omega}\left(e^{i\omega s}H\right)\right|^2\right]\, ds.
		\end{align}
	\end{cor}
	\begin{proof}This is an immediate consequence of combining an application of Lemma~\ref{agmonplusredest} to the equation from Lemma~\ref{qeqnusefultodiffu} with~\eqref{nocompactsupportsolosellogl} 
		and Lemma~\ref{thisisgeneralforderofubutbadconstant}. 
	\end{proof}
	
	We will now be interested in determining various situations where the constant on the right hand sides of the estimates in Lemma~\ref{easymwaybig} may be improved. A crucial role will be played by the eigenvalues of the operator $Q$. This next lemma provides estimates for $u$ when $\omega^2-m^2$ is sufficiently far away from an eigenvalue of $Q$.
	\begin{lemma}\label{estimatesforunottoofarfromspec}Let $u$ be a weakly outgoing solution to~\eqref{eq:mainVrhs} so that $L^{-p} \leq m^2-\omega^2 \lesssim L^{-2}$. Let $S_0$ be any large positive constant and $\{\lambda_n\}_{n=1}^{N(p)}$ denote the collection of the eigenvalues which are greater than $ L^{-p}$ of the operator $Q$  from Definition~\ref{thisistheoperatorq} with $x_0 = S_0$. Let $\check{\delta} > 0$ be any suitably small positive constant, and assume 
		\begin{equation}\label{protectedfromspectrum}
			{\rm inf_{n \in [1,N(p)]}}\left(m^2-\omega^2 + \lambda_n\right) \geq e^{-\frac{1}{2} \LL \log \LL}.
		\end{equation}
		
		Then, for any small constant $\delta> 0$,  we have the following estimates:
		\begin{enumerate}
			\item
			\begin{align}\label{spectralgoodcompactest}
				&\int_{-\infty}^{2S_0 }\left[\frac{r}{r-r_+}\left|\partial_su+i\omega u\right|^2 + \frac{1}{\left({\rm max}\left(0,-s\right)\right)^{1+\delta}+1}\left(\left|\partial_su\right|^2 + \left|u\right|^2\right)\right]\, ds \lesssim 
				\\ \nonumber &\qquad \LL^{2}\int_{-\infty}^{\infty}\frac{r}{r-r_+}\left|H\right|^2\, ds.
			\end{align}
			
			\item 
			\begin{align}\label{spectralgoodestL2}
				&\int_{-\infty}^{\infty}\left[\frac{r}{r-r_+}\left|\partial_su+i\omega u\right|^2 + \frac{1}{1 + \left({\rm max}\left(-s,0\right)\right)^{1+\delta}}\left(\left|\partial_su\right|^2 + \left|u\right|^2\right)\right]\, ds \lesssim 
				\\ \nonumber &\qquad \qquad e^{\LL \log \LL}\LL^{2}\int_{-\infty}^{\infty}\frac{r}{r-r_+}\left|H\right|^2\, ds.
			\end{align}
			nd{align}

		\end{enumerate}
		
	\end{lemma}
	\begin{proof}In view of~\eqref{protectedfromspectrum}, $\omega^2-m^2$ is at least distance $e^{-\frac{1}{2} \LL \log \LL}$ away from the spectrum of $Q$ (see Lemma~\ref{basicspectQ}). In particular, after multiplying $u$'s equation with a suitable cut-off $\chi(s)$ which is vanishes for $s \leq S_0$ and is identically $1$ for $s \geq S_0 + 1$, using this spectral information, and elliptic estimates, we immediately obtain the following bound:
		\begin{equation}\label{fromthegoodspectrum}
			\int_{S_0+1}^{\infty}\left[\left|\partial_su\right|^2 + \left|u\right|^2\right]\, ds \lesssim e^{ \LL \log \LL}\left[\int_{S_0}^{S_0+1}\left|u\right|^2\, ds + \int_{S_0}^{\infty}\left|H\right|^2\, ds\right].
		\end{equation}
		We may choose and fix $\delta > 0$ small and apply Lemma~\ref{agmonplusredest} replacing $S_0$ by $2S_0$ to obtain that 
		\begin{align}\label{theagmonthingfrombefore}
			&\int_{-\infty}^{2S_0}\left[\frac{r}{r-r_+}\left|\partial_su+i\omega u\right|^2 + \frac{1}{\left({\rm max}\left(0,-s\right)\right)^{1+\delta}+1}\left(\left|\partial_su\right|^2 + \left|u\right|^2\right)\right]\, ds \lesssim
			\\ \nonumber &\qquad e^{-\left(4-\tilde{\delta}\right) \LL \log \LL}\int_{0}^{+\infty}\left|u\right|^2\, ds + \LL^{2}\int_{-\infty}^{\infty}\frac{r}{r-r_+}\left|H\right|^2\, ds.
		\end{align}
		We may thus establish~\eqref{spectralgoodcompactest}  by combining~\eqref{theagmonthingfrombefore} and~\eqref{fromthegoodspectrum}.
		
		The estimates~\eqref{spectralgoodestL2} then follow from using these bounds to control the term $\int_{S_0}^{S_0+1}\left|u\right|^2\, ds$ on the right hand side of~\eqref{fromthegoodspectrum}.
	\end{proof}
	
	In the next lemma, we prove an analogue of Lemma~\ref{estimatesforunottoofarfromspec} for $\partial_{\omega}u$.
	\begin{lemma}\label{estimatesforderuomegaunottoofarfromspec}Let $u$ be an outgoing solution to~\eqref{eq:mainVrhs} so that $L^{-p} \leq m^2-\omega^2 \lesssim L^{-2}$. Let $S_0$ be any large positive constant and $\{\lambda_n\}_{n=1}^{N(p)}$ denote the collection of the eigenvalues which are greater than $ L^{-p}$ of the operator $Q$  from Definition~\ref{thisistheoperatorq} with $x_0 = S_0$. 
		
		Let $\delta > 0$ be a small constant.  Assume that  \begin{equation}\label{protectedfromspectrum2}
			{\rm inf_{n \in [1,N(p)]}}\left(m^2-\omega^2 + \lambda_n\right) \geq e^{-\frac{1}{2} \LL \log \LL}.
		\end{equation}

		Furthermore, assume that $H$ is compactly supported within $\{s \leq S_0\}$ for a large constant $S_0$. Define
		\[q \doteq e^{-i\omega s \chi(s)}\partial_{\omega}\left(e^{i\omega s \chi(s)}u\right),\]
		where $\chi(s)$ is a cut-off function which is $1$ for $s \leq 0$ and $0$ for $s \geq 1$.
		
		Then we have 
		\begin{align}\label{derivativeofuestimatewhenthespectralisgood}
			&\int_{-\infty}^{2S_0}\left[\frac{r}{r-r_+}\left|\partial_sq+i\omega q\right|^2 + \frac{1}{\left({\rm max}\left(0,-s\right)\right)^{1+\delta}+1}\left(\left|\partial_sq\right|^2 + \left|q\right|^2\right)\right]\, ds \lesssim  
			\\ \nonumber &\qquad \qquad \LL^4\int_{-\infty}^{S_0}\frac{r}{r-r_+}\left[\left|H\right|^2 + \left|\partial_{\omega}\left(e^{i\omega s \chi}H\right)\right|^2\right]\, ds.
		\end{align}
	\end{lemma}
	\begin{proof} We apply the estimates~\eqref{spectralgoodcompactest} and~\eqref{spectralgoodestL2} to the equation for $q$ computed in Lemma~\ref{qeqnusefultodiffu} and use the estimates~\eqref{spectralgoodcompactest} and~\eqref{spectralgoodestL2} to control the terms involving $u$.
	\end{proof}
	
	In the next lemma we will consider the case when $\omega^2-m^2$ is not exactly equal to an eigenvalue of $Q$ but is too close to an eigenvalue to apply Lemmas~\ref{estimatesforunottoofarfromspec} and~\ref{estimatesforderuomegaunottoofarfromspec}.
	\begin{lemma}\label{nottooclosebutkindofclose}Let $u$ be a weakly outgoing solution to~\eqref{eq:mainVrhs} so that $ L^{-p} \leq m^2-\omega^2 \lesssim L^{-2}$. Let $S_0$ be any large positive constant and $\{\lambda_n\}_{n=1}^{N(p)}$ denote the collection of the eigenvalues which are greater than $L^{-p}$ of the operator $Q$  from Definition~\ref{thisistheoperatorq} with $x_0 = S_0$. Further assume that $H$ is compactly supported in the region $\left\{r \leq S_0\right\}$.
		
		Let $D_0$ be as in Lemma~\ref{thisisgeneralforubutbadconstant} and assume that there exists an eigenvalue $\lambda_n$ so that 
		\begin{equation}\label{almosteig}
			0 < \underbrace{\left|m^2-\omega^2 + \lambda_n\right|}_{\doteq \mathscr{E}} \leq e^{-\frac{1}{2} \LL \log \LL}.
		\end{equation}
		Then there exists a constant $D > 0$, depending on $S$, so that the following estimates hold:
		\begin{enumerate}
			
			\item We have
			\begin{align}\label{sortofclosetoeigenvalue2}
				&\int_{-\infty}^{\infty}\left[\frac{r}{r-r_+}\left|\partial_su+i\omega u\right|^2 + \frac{1}{1 + \left({\rm max}\left(-s,0\right)\right)^{1+\delta}}\left(\left|\partial_su\right|^2 + \left|u\right|^2\right)\right]\, ds \lesssim 
				\\ \nonumber &\qquad \qquad \mathscr{E}^{-2}e^{DL}\int_{-\infty}^{S_0}\frac{r}{r-r_+}\left|H\right|^2\, ds.
			\end{align}
			
			\item Let $\tilde{\delta} > 0$ be a small constant. Then we have that 
			\begin{align}\label{sortofclosetoeigenvalue4}
				&\int_{-\infty}^{2S_0}\left[\frac{r}{r-r_+}\left|\partial_su+i\omega u\right|^2 + \frac{1}{1 + \left({\rm max}\left(-s,0\right)\right)^{1+\delta}}\left(\left|\partial_su\right|^2 + \left|u\right|^2\right)\right]\, ds \lesssim 
				\\ \nonumber & \left(\LL^{2} + e^{-\left(4-\tilde{\delta}\right) \LL \log \LL}\mathscr{E}^{-2}\right)\int_{-\infty}^{S_0}\frac{r}{r-r_+}\left|H\right|^2\, ds.
			\end{align}
			
		\end{enumerate}
		
	\end{lemma}
	\begin{proof}In view of~\eqref{almosteig}, it follows from Lemmas~\ref{basicspectQ} and~\ref{eigenshavespace} that
		\[{\rm dist}\left(\omega^2-m^2,{\rm spec}\left(Q\right)\right) \gtrsim \left|m^2-\omega^2 + \lambda_n\right|.\]
		In particular, we may argue as in the derivation of~\eqref{fromthegoodspectrum} to obtain that
		\begin{align}\label{fromthesortofokspectrum}
			&\int_{S_0+1}^{\infty}\left[\left|\partial_su\right|^2 + \left|u\right|^2\right]\, ds \lesssim 
			\mathscr{E}^{-2}\int_{S_0}^{S_0+1}\left|u\right|^2\, ds.
		\end{align}
		We then combine with~\eqref{Rboundedonlyexplossforu2} to obtain~\eqref{sortofclosetoeigenvalue2}.
		
		To obtain~\eqref{sortofclosetoeigenvalue4}, we apply Lemma~\ref{agmonplusredest} and use~\eqref{sortofclosetoeigenvalue2} to control the first term on the right hand side.

	\end{proof}
	
	In the next lemma, we establish an analogue of Lemma~\ref{nottooclosebutkindofclose} for $\partial_{\omega}u$.
	\begin{lemma}\label{nottooclosebutkindofcloseforderu}Let $u$ be an outgoing solution to~\eqref{eq:mainVrhs} so that $L^{-p} \leq m^2-\omega^2 \lesssim L^{-2}$.. Let $S_0$ be any large positive constant and $\{\lambda_n\}_{n=1}^{N(p)}$ denote the collection of the eigenvalues which are greater than $ L^{-p}$ of the operator $Q$  from Definition~\ref{thisistheoperatorq} with $x_0 = S_0$.
		
		Assume that there exists an eigenvalue $\lambda_n$ so that 
		\begin{equation}\label{almosteig2}
			0 < \underbrace{\left|m^2-\omega^2 + \lambda_n\right|}_{\doteq \mathscr{E}} \leq e^{-\frac{1}{4} L \log L}.
		\end{equation}
		
		Furthermore, assume that $H$ is compactly supported within $\{s \leq S_0\}$. Define
		\[q \doteq e^{-i\omega s \chi(s)}\partial_{\omega}\left(e^{i\omega s \chi(s)}u\right),\]
		where $\chi(s)$ is a cut-off function which is $1$ for $s \leq 0$ and $0$ for $s \geq 1$.
		
		Let $\tilde{\delta}, \delta > 0$ be small constants, then  we have 
		\begin{align}\label{nottoclosetoaneigforderu}
			&\int_{-\infty}^{2S_0}\left[\frac{r}{r-r_+}\left|\partial_sq+i\omega q\right|^2 + \frac{1}{1 + \left({\rm max}\left(-s,0\right)\right)^{1+\delta}}\left(\left|\partial_sq\right|^2 + \left|q\right|^2\right)\right]\, ds \lesssim_{\tilde{\delta},\delta, S_0} 
			\\ \nonumber &\left(\mathscr{E}^{-2}e^{\tilde{\delta}\LL \log\LL} + e^{-\left(4-\tilde{\delta}\right) \LL \log \LL}\mathscr{E}^{-4}\right)\int_{-\infty}^{S_0}\frac{r}{r-r_+}\left[\left|H\right|^2 + \left|\partial_{\omega}\left(e^{i\omega s \chi(s)}H\right)\right|^2\right]\, ds. 
		\end{align}
		
	\end{lemma}
	\begin{proof}We start with the equation for $q$ computed in Lemma~\ref{qeqnusefultodiffu} and then proceed as in the proof of Lemma~\ref{nottooclosebutkindofclose}. The analogue of~\eqref{fromthesortofokspectrum} is 
		\begin{align*}
			\int_{S_0+1}^{\infty}\left[\left|\partial_sq\right|^2 + \left|q\right|^2\right]\, ds \lesssim \mathscr{E}^{-2}\int_{S_0}^{S_0+1}\left|q\right|^2\, ds + \mathscr{E}^{-2}\int_{S_0}^{\infty}\left|u\right|^2\, ds.
		\end{align*}
		
		Then we apply~\eqref{nocompactsupportsolosellogl22} and~\eqref{sortofclosetoeigenvalue2} to estimate the right hand side  
		and obtain, for any $\hat{\delta} > 0$,
		\begin{align*}
			&\int_{S_0+1}^{\infty}\left[\left|\partial_sq\right|^2 + \left|q\right|^2\right]\, ds \lesssim_{\hat{\delta}} \mathscr{E}^{-4}e^{DL}\int_{-\infty}^{S_0}\frac{r}{r-r_+}\left|H\right|^2\, ds
			\\ \nonumber &\qquad + \mathscr{E}^{-2}e^{\left(4+\hat{\delta}\right) \LL \log \LL}\int_{-\infty}^{S_0}\frac{r}{r-r_+}\left[\left|H\right|^2+\left|\partial_{\omega}\left(e^{i\omega s\chi }H\right)\right|^2\right]\, ds.
		\end{align*}
		
		Finally, we combine this estimate with an application of Lemma~\ref{agmonplusredest} to the equation for $q$ computed in Lemma~\ref{qeqnusefultodiffu} and obtain~\eqref{nottoclosetoaneigforderu}.  
		
	\end{proof}
	\subsection{Putting everything together}
	Finally, we are ready to prove Proposition~\ref{themainpropinthebigthanL2}:
	\begin{proof} It is useful to keep in the mind the fact that for any function $f(s)$ and $1 \ll S < \infty$, we have
		\begin{equation}\label{linftyfroml2}
			\sup_{s \leq S}\left|f(s)\right|^2 \lesssim_S \int_{-\infty}^S(r-r_+)^{-1}\left|\partial_sf + i\omega f\right|^2\, ds + \int_0^1\left|f\right|^2\, ds.
		\end{equation}
		
		The estimate~\ref{bigL1} follows from~\eqref{spectralgoodcompactest} and~\eqref{linftyfroml2}. The estimate~\ref{bigL2} follows from~\eqref{derivativeofuestimatewhenthespectralisgood} and~\eqref{linftyfroml2}. The estimate~\ref{bigL3} follows from~\eqref{Rboundedonlyexplossforu2} and~\eqref{linftyfroml2}. The estimate~\ref{bigL4} follows from~\eqref{nocompactsupportsolosellogl22} and~\eqref{linftyfroml2}. Finally,~\ref{bigL5} follows from~\eqref{nottoclosetoaneigforderu} and~\eqref{linftyfroml2}.    
	\end{proof}
	
	\section{The turning point regimes in the low frequencies}\label{regimeB2.section}
	We now address the regime of low frequencies $0<\omega<m$ where $\omega$ is very close to $m$. These set of frequencies are the most singular contributors to the tail and require a precise control of the Fourier transform of the solution $u(\omega,s)$ capturing discrete sets of more singular frequencies. To do so, we  carry out a	 WKB-type analysis of the ODE \eqref{eq:mainV}.	Throughout the proof, we will assume $m-L^{-p}<\omega<m$ for some large $p> 0$  (note in particular $\omega>0$, chosen for convenience). Define		
	\begin{equation}\label{k.def}
		k^{-2} = \frac{m^2-\omega^2}{(Mm^2)^2}.
	\end{equation} and \begin{equation}\label{L.def}
		\LL = \sqrt{L(L+1)}.
	\end{equation}

	We now give further details before stating the main result of this section. Under our notations we have \begin{equation}\label{V.schematic}
		\VV(s) = (Mm^2)^2 k^{-2} - \frac{2Mm^2}{r}+ \frac{\LL^2}{r^2} +  \underbrace{\Big(-\frac{2M}{r}+\frac{\DD^2}{r^2}\Big)\frac{\LL^2}{r^2}}_{=O(\LL^2 r^{-3})}+ \underbrace{\frac{\DD^2 m^2}{r^2}+\Big(1-\frac{2M}{r}+\frac{\DD^2}{r^2}\Big)\Big(\frac{2rM -2\DD^2}{r^4}\Big)}_{=O(r^{-2})}
	\end{equation}
	We also introduce
	
	\begin{equation}
		\alpha= \frac{\LL^2}{k^2}.
	\end{equation}

	Note that the solution to $(Mm^2)^2 k^{-2} - \frac{2Mm^2}{s}+ \frac{\LL^2}{s^2}=0$ are $s=	s^{\pm}(\omega,\LL)$ with \begin{equation}\label{spm.def}
		s^{\pm}(\omega,\LL) = \frac{k^2}{Mm^2}  \left( 1 \pm \sqrt{1-\alpha}\right)=  \frac{\LL^2}{Mm^2}  \frac{ 1 \pm \sqrt{1-\alpha}}{\alpha} .
	\end{equation}
	
	First, let's define $\tilde{s}_I < s_{II} <s_{III}$ to be the three distinct values of the turning points, i.e.\ \begin{align*}
		&	V(\tilde{s}_I)=0\ \&\	V(s) < 0 \text{ for } s< \tilde{s}_I,\\ &	V(s_{II})=0\ \&\	V(s) > 0 \text{ for } s_{I}<s< s_{II},\\ &	V(s_{III})=0\ \&\	V(s) < 0 \text{ for } s_{II}<s<s_{III},\ V(s)>0  \text{ for } s>s_{III}.
	\end{align*}
	Now, as it turns out, what truly matters for the first turning point, is the zero of  the modified potential $V-\frac{\K^2}{4}$, which is explained by the exponential decay of the potential $V$ as $s\rightarrow -\infty$ (see \cite{Schlag_exp} for more details). We thus define $s_I> \tilde{s}_I$ as the smallest $s<0$ for which \begin{equation*}
		V(s_I) - \frac{\K^2}{4}=0.
	\end{equation*}
	
	Defining $\tilde{\omega}:= \sqrt{\omega^2+  \frac{\K^2}{4} } > \omega>0 $, we have, in particular \begin{align}
		& \label{sI}	s_I(\LL,\omega) = \frac{1}{\K(M,e)} \log( \frac{r_+(M,e) \tilde{\omega}}{\LL})  + O (\LL^{-2})<0,\\ & 	s_{II}(\LL,\omega)= s_-(\LL,\omega)+O(1)=\frac{\LL^2}{Mm^2}\  \underbrace{\frac{1-\sqrt{1-\alpha}}{\alpha}}_{\in  (\frac{1}{2},1 )}+ O(1),\label{sII}\\&   	s_{III}(\LL,\omega)= s_+(\LL,\omega)+O(1)=\frac{k^2}{Mm^2} \left(1+\sqrt{1-\alpha}\right)+O(1).  \label{sIII}
	\end{align}

	We will introduce two Jost solutions $u_H$ and $u_I$ of \eqref{eq:mainV}. The first Jost solution 		$u_H$  is ``regular at the event horizon $\{s=-\infty\}$'', and obeys the weakly outgoing condition \eqref{weakoutgoingdef}. More precisely $u_H$ is defined through the following asymptotics as $s\rightarrow -\infty$:
	\begin{equation}\label{uH.def} u_H(\omega,s) = e^{-i \omega s} [1+o(1)]. \end{equation}  The second Jost solution $u_I$ is ``regular at infinity $\{s=+\infty\}$ and obeys decaying-exponential asymptotics.   More precisely $u_I$ is defined through the following asymptotics as $s\rightarrow +\infty$: \begin{equation}\label{uI.asymp}
		u_I(\omega,s) \sim \frac{k^{-2/3}}{2\pi^{1/2} (Mm^2)^{1/2}}\ |V|^{-1/4}(s) \exp( -\int_{s_{III}}^{s} |V|^{1/2}(s') ds').
	\end{equation} (See \cite{KGSchw1}, Section 4 for a justification of the fact that $u_I$ is well-defined).
	
	Our strategy will be to use  Green's formula (see \cite{KGSchw1}) to connect  $u$, the actual solution to \eqref{eq:mainVrhs} to the Jost solutions $u_H$ and $u_I$: \begin{equation}
		u(\omega,s) =\frac{ u_H(\omega,s) \int_s^{+\infty} u_I(s') H(\omega,s') ds'}{W(u_I,u_H)}+ \frac{ u_I(\omega,s) \int_{-\infty}^{s} u_H(s') H(\omega,s') ds'}{W(u_I,u_H)}.
	\end{equation}
	
	We will now state the main result in this section, which  involves $w_{1,\pm}$, solutions of \eqref{eq:mainV} which are well characterized in a compact $s$-region. $w_{1,\pm}$ should be intuitively though of as generalizations of the elementary solutions $\{r^{L+1},\ r^{-L}\}$ of the wave equation on Minkowski space $\frac{d^2 u }{dr^2}= \frac{\LL^2}{r^2} u $. We remind the reader that we will assume compact support of $H$ in $r$, but we allow $H$ to be supported near the event horizon. \begin{thm}\label{TPsection.mainprop}  We assume that $H$ from \eqref{eq:mainVrhs} is supported in $\{s \leq s_H\}$ for some $s_H>1$, and recall also that we have assumed that $ m-\LL^{-p} <\omega< m$, for some large, but fixed $p>2$.
		\begin{enumerate}
			
			\item (Construction of the solutions $w_{1,\pm}$ in the compact $s$ region) 
			Let $A\gg 1$ and $0<\ep\ll1$.
			There exists $w_{1,\pm}(\omega,s)$, solutions of \eqref{eq:mainV} obeys the following estimates for all $s_I(\LL) +A \leq s \leq  \LL^{2-\ep}$:
			
			\begin{equation}\begin{split}
					&	|w_{1,\pm}|(\omega,s),\	\LL^{-3}	|\rd_{\omega} w_{1,\pm}|(\omega,s) \ls \LL^{-1/2} \left( e^{-\frac{\kappa_+ s}{2} }  1_{s\leq 1}+ s^{1/2} 1_{s\geq 1}\right) \exp(\pm \int_{s_I}^{s}|V|^{1/2}(s') ds'),\\ &  \int_{s_I}^{s}|V|^{1/2}(s') ds' \lesssim \LL \cdot [   D_H \cdot e^{\kappa_+ s}  1_{s\leq 1}+ \log(1+s) 1_{s\geq 1}], \end{split}
			\end{equation} for some constant $D_H>0$ independent of $\LL$ and $\omega$.
			
			\item (Construction of $\tilde{K}_{\pm}$ functions near the event horizon and extension of $w_{1,\pm}$) There exists  $\tilde{K}_{\pm}$  solutions of \eqref{eq:mainV} for all $s \leq 0$ and obeying the following estimates in the restricted range $s \leq s_I+A$:
			
			\begin{align*}
				|\tilde{K}_{\pm}|(\omega,s)\ls 1,\ |\rd_{\omega}(e^{-i\omega s}[\tilde{K}_+ + i \tilde{K}_-])|(\omega,s),\ |\rd_{\omega}(e^{i\omega s}[\tilde{K}_+ - i \tilde{K}_-])|(\omega,s)  \ls \log(\LL).
			\end{align*}
			Then, $w_{1,\pm}$  can be connected to $\tilde{K}_{\pm}$,  as such: for any $s_I + A \leq s \leq 0$:
			
			\begin{align}\label{w1pm.connectionH}
				&	w_{1,\pm}(\omega,s)= \beta^{H}_{+,\pm}(\omega) \tilde{K}_+(\omega,s)+ \beta^{H}_{-,\pm}(\omega) \tilde{K}_-(\omega,s),\\ & \beta_{\pm \pm}^H(\omega),\ \rd_{\omega}\beta_{\pm \pm}^H(\omega) = O(1),\ \beta_{-+}^H(\omega),\ \rd_{\omega}\beta_{-+}^H(\omega)  = O(1),\\ & \beta^H_{+ -} (\omega)= O (\exp(-2D_H \LL)),\ \rd_{\omega}\beta^H_{+ -} (\omega)=O(1).
			\end{align}	 	In what follows, we adopt the convention to still denote the RHS of \eqref{w1pm.connectionH} as $w_{1,\pm}$, even in the region $\{s \leq s_I+A\}$ where $w_{1,\pm}(\omega,s)$ did not originally make sense.

			\item (Description of $u_H$ in terms of $w_{1,\pm}$) $u_H$ is connected to $w_{1,\pm}$ as such: for all $s \leq  \LL^{2-\ep}$: \begin{align}\label{uH.mainprop}
				&u_H(\omega,s)=  \alpha_H^{+}(\omega) w_{1,+}(\omega,s)+\alpha_H^{-}(\omega) w_{1,-}(\omega,s),\\ &| \alpha_H^{\pm}|(\omega,s) \approx 1,\ |\rd_{\omega}  \alpha_H^{\pm}|(\omega,s) \lesssim \log(\LL) \label{uH.mainprop2}. \end{align}

			\item (Green's formula) We have a Green's formula  connecting $u$ to $u_H$ and $w_{1,\pm}$ (and combining with \eqref{uH.mainprop}, connecting indirectly $u$ to  $w_{1,\pm}$ as well): we first introduce the key phase quantity \begin{equation}\label{tildek.def}\begin{split} &  \kkq(\omega)  \doteq\pi^{-1}
					\left( \int_{s_{II}}^{s_{III}} |V|^{1/2}(s) ds-\phi_3(\omega)\right),\\ & |\phi_3|(\omega)\ls k^{-1},\  |\rd_{\omega}^n\phi_3|(\omega)\ls k^{2n-1}\log(k), \text{ for any } n\geq 1. \end{split}
			\end{equation} 
			Then, we have the  Green's formula:
			
			\begin{equation}\begin{split}\label{secondGreenFormula}
					&	u(\omega,s)= e^{i\kkq}\left(\frac{ \sum_{i \in \{1,2\},\ j \in \{\pm\}}u_{i,j}^{reg}(\omega,s)+\sum_{i \in \{1,2\},\ j \in \{\pm\}}\ep_{i,j}^{cos}(\omega,s) \cos(\kkq)+ \ep_{i,j}^{sin}(\omega,s) \sin(\kkq)}{1+\Gamma(m)\cdot e^{2i\kkq}}\right),			\end{split}
			\end{equation} where 
			\begin{align*}
				& u_{1,+}^{reg}(\omega,s)\\=&\exp(-4 \LL\log(\LL)+ D_I \LL) \frac{u_H(m,s) }{C_W(m,\LL)}\   \int_s^{+\infty} w_{1,+}(m,s') H(m,s') ds'  \left(C_{cos,+}(m) \cos(\kkq) +C_{sin,+}(m) \sin(\kkq)\right) , \\ &u_{1,-}^{reg}(\omega,s)= \frac{u_H(m,s) }{ C_W(m,\LL)}    \int_s^{+\infty} w_{1,-}(m,s') H(m,s') ds'   \left(C_{cos,-}(m) \cos(\kkq) +C_{sin,-}(m) \sin(\kkq)\right),\end{align*} and \begin{align*}&u_{2,+}^{reg}(\omega,s)\\ =& \exp(-4 \LL\log(\LL)+ D_I \LL) \frac{w_{1,+}(m,s) }{ C_W(m,\LL)}\     \int^{s}_{-\infty}u_H(m,s')  H(m,s') ds' \left(C_{cos,+}(m) \cos(\kkq) +C_{sin,+}(m) \sin(\kkq)\right), \\ &u_{2,-}^{reg}(\omega,s)= \frac{ w_{1,-}(m,s)}{ C_W(m,\LL)}    \int^s_{-\infty} u_H(m,s') H(m,s') ds'   \left(C_{cos,-}(m) \cos(\kkq) +C_{sin,-}(m) \sin(\kkq)\right).
			\end{align*}  There exists  $D_I(\LL)=O(1)$ independent of $\omega$, and a large $A>0$  depending  on $(M,\DD)$ such that \begin{align*}
				&	|C_W|(m,\LL)\approx 1,\ 	|C_{cos,+}(m,\LL)-\frac{1}{2}|,\  |C_{sin,-}(m,\LL)-1|\ls e^{-\kappa_+ A},\\ & |C_{sin,+}(m,\LL)| \ls \exp(O(\LL)),\  	|C_{cos,-}(m,\LL) \ls_A \exp(-4\LL \log\LL+ D_I \LL) ,
			\end{align*} 
			Moreover, we have the following estimate on 
			$\Gamma$: \begin{align} &    |\Gamma|= 1-   \tilde{C}\  \exp( -4\LL \log(\LL)+D_I \LL),\qquad \tilde{C} \sim 1,
				\\ &\bigl|\Gamma+1 \big| \ls  \exp( -4\LL\log\LL + D_I\LL).
			\end{align}
			
			Moreover, $\ep_{i,j}^{\pm}(\omega,s)$ obeys the following estimates: there exists $D_H=O(1)$ such that  for all $s\leq s_{max}$ (with implicit constants depending on $s_{max}$ and $s_H$, recalling that $H$ is supported on $\{s\leq s_H\}$):\begin{equation}\begin{split}
					& k^{1/2}|\ep_{i,+}^{cos}|(\omega,s),\ k^{1/2}|\ep_{i,+}^{sin}|(\omega,s) \ls\exp(-4 \LL\log(\LL) + D_H \LL) \cdot N[H]\\ & k^{-2}\log^{-1}(k)|\rd_{\omega}\ep_{i,+}^{cos}|(\omega,s),\ k^{-2}\log^{-1}(k)|\rd_{\omega}\ep_{i,+}^{sin}|(\omega,s) \ls\exp(-4 \LL\log(\LL) + D_H \LL) \cdot N[H], \\ & k^{1/2}|\ep_{i,-}^{cos}|(\omega,s),\  k^{-2}\log^{-1}(k)|\rd_{\omega}\ep_{i,-}^{cos}|(\omega,s) \ls\exp(-4 \LL\log(\LL) + D_H \LL) \cdot N[H],\\	& k^{1/2}|\ep_{i,-}^{sin}|(\omega,s),\ k^{-2}\log^{-1}(k)|\rd_{\omega}\ep_{i,-}^{sin}|(\omega,s) \lesssim    N[H],\\ & N[H]:=    \int_{-\infty}^{s_H} \left(\sup_{\omega \in (\frac{m}{2},m)}|H|(\omega,s')+ \sup_{\omega \in (\frac{m}{2},m)}|\rd_{\omega}H|(\omega,s')  \right) ds'.\end{split}
			\end{equation}
		\end{enumerate}
	\end{thm}

	The  Green's formula \eqref{secondGreenFormula} can be rephrased in terms of $\cos( \kkq)$,  $\sin( \kkq)$  with a more familiar formula:
	\begin{align}\label{Green2}& u(\omega,s) = \sum_{i=1..2, j=\pm}   v_{i,j}(\omega,s)\frac{ C_{cos,j}(\omega)\cos(\kkq)+ C_{sin,j}(\omega)\sin(\kkq)}{\cos(\kkq)\exp(-4\LL\log(\LL)+D_I\LL) \delta_{cos}(\omega)+ \sin(\kk)},\\ & v_{1,\pm}(\omega,s)=  \exp(-\frac{(1\pm 1)}{2}[4\LL \log \LL+ D_I \LL])\frac{  u_H(\omega,s) }{D_W(\omega)}\  \int_s^{+\infty} w_{1,\pm}(m,s')  H(m,s') ds',\\ &v_{2,\pm}(\omega,s)= \exp(-\frac{(1\pm 1)}{2}[4\LL \log \LL+ D_I \LL])  \frac{ w_{1,\pm}(m,s) }{D_W(\omega,\LL)}\  \int^s_{-\infty} u_H(\omega,s') H(m,s') ds'.				\end{align}
	
	In fact, we will first prove a formula like \eqref{Green2} (see Proposition~\ref{together.prop}) to then deduce \eqref{secondGreenFormula}.
	
	We finally note  Theorem~\ref{quasimode.thm} easily follows from Theorem~\ref{TPsection.mainprop}, see the  discussion at the end  of Section~\ref{resolvent.section}.
	\subsection{Strategy of the proof and organization of the section}\label{strategy.section}

	The proof  of Theorem~\ref{TPsection.mainprop} 	relies on a strategy similar to our work \cite{KGSchw1} where we addressed the fixed angular mode $L$ case: we take advantage of the following representation  formula	proven in \cite{KGSchw1}, Section 4.3:
	\begin{equation}\label{resolvent1}
		u(\omega,s) = \underbrace{\frac{ u_H(\omega,s) \int_s^{+\infty} u_I(s') H(\omega,s') ds'}{W(u_I,u_H)}}_{:=u_1(\omega,s)}+ \underbrace{\frac{ u_I(\omega,s) \int_{-\infty}^{s} u_H(s') H(\omega,s') ds'}{W(u_I,u_H)}}_{:=u_2(\omega,s)}.
	\end{equation}
	
	In the analysis of \cite{KGSchw1}, the main term in the potential was of  the form $(Mm^2)^2 k^{-2} - \frac{2Mm^2}{r}$, with a unique turning point when $r\approx k^2$.	The main difference is that, because now the potential $V$ has the schematic form  \eqref{V.schematic}, there are \emph{three} turning points to account for (see Figure~\ref{Potential}), with an important $\LL$-dependency specifically for the second turning point (note the presence of a classically forbidden region in which $\{-\log(\LL) \ls s\ls \LL^2\}$ between the fist and second turning points). We will adopt the following steps:
	
	\begin{figure}\label{Potential2}
		
		\begin{center}
			
			\includegraphics[width=100 mm, height=50 mm]{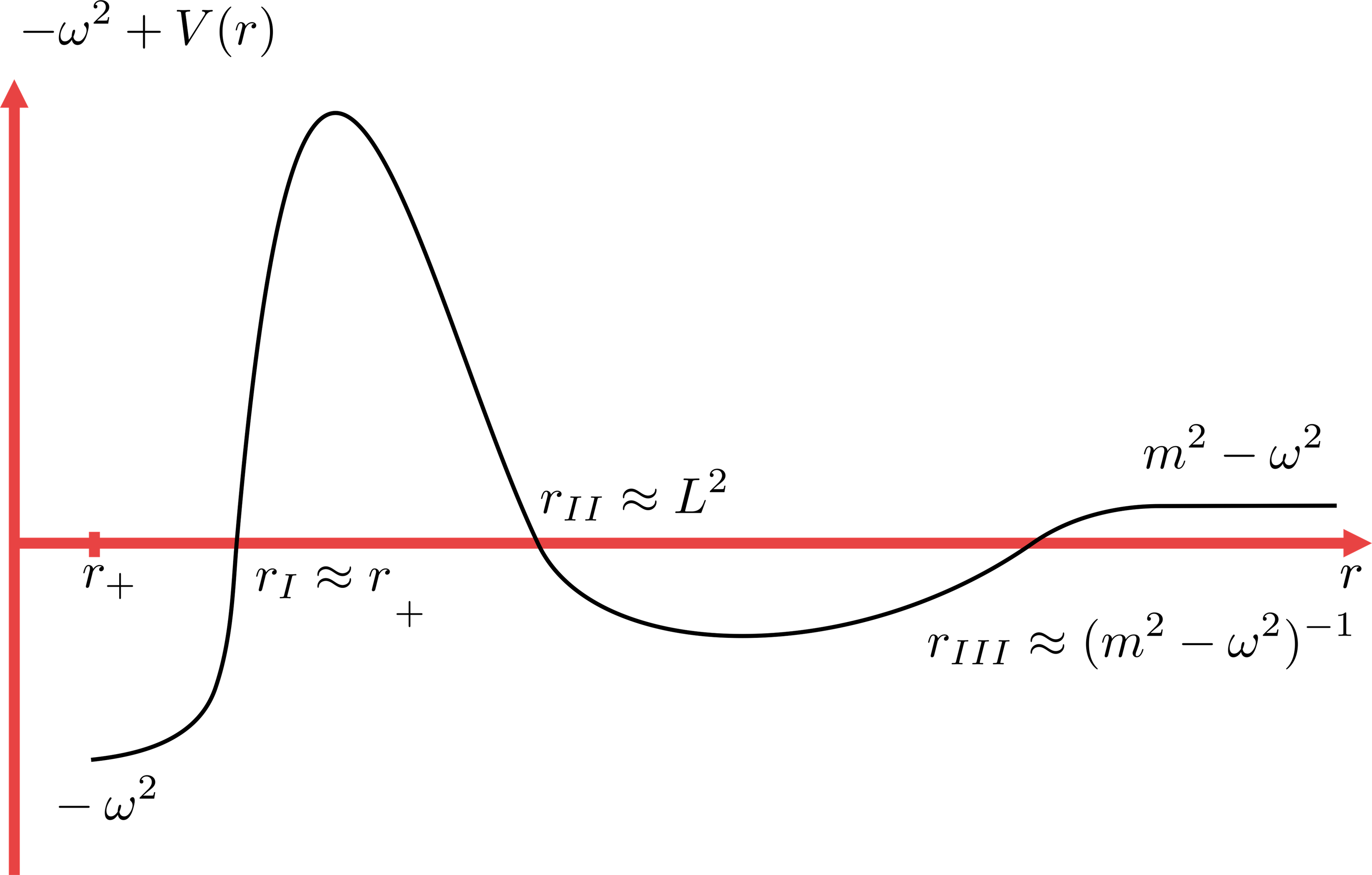}
			
		\end{center}
		\caption{The  effective potential $-\omega^2 + V(s)$, when $\omega \in(-m,m)$, where $V(r)$ is as in \eqref{V.schematic}.}

	\end{figure}

	\paragraph{Step 1: $u_H$ and the first turning point (Section~\ref{uH.section})}	 We  need to control $u_H$, the solution regular at the event horizon, more precisely in terms of $\LL$. In particular, we will need to analyze the first turning point at $s=s_I(\LL,\omega) \approx -\log(\LL^{-1})<0$, which is located near the event horizon $\mathcal{H}^+=\{s=-\infty\}$. While the analysis near the event horizon using multipliers has somewhat become standard (see Section~\ref{multiplier.section} for a discussion in the range $|m-\omega|\gg \LL^{-p}$ and \cite{Red,claylecturenotes}), it cannot be used in the frequency range $(m-L^{-p},m)$ because our analysis relies on the fine structure of the Fourier transform which we capture using WKB-type estimates. Instead we rely on the works   \cite{Schlag2,Schlag_exp}, which exploit the exponentially-decaying character of the potential, and identify the adequate modified-Bessel functions $\tilde{K}_{\pm}$ resolving the ODE near the first turning point $s=s_I$.  The key is a  change of variable $s\rightarrow w(s)$ that reduces to the WKB approximation away from the turning point. Recall that in \cite{KGSchw1}, neglecting the influence of $\LL$, we softly used the smoothness of $u_H(\omega,s)$ in $\omega$ uniformly in any compact region of $s$. As it turns out, our new estimates in Section~\ref{uH.section} show that it is possible to control uniformly $u_H$ in $\LL$ too, up to $\exp(O(\LL))$ factors. Anticipating the discussion of Section~\ref{connection12.section} (Step 4), we remark that the main terms in \eqref{resolvent.intro} involves $\exp( O(\LL \log(\LL)))$, which dominate the  $\exp(O(\LL))$ errors. This observation in the large $\LL$ case is key to considering $u_H$ as ``a regular solution'' and represents a significant improvement over the  argument used in \cite{KGSchw1}.
	
	\paragraph{Step 2:  WKB  estimates above the first turning point and under the second turning point (Section~\ref{WKB.section1})}	 We introduce $w_{1,\pm} \approx \exp(\pm O(\LL))$, two solutions of the ODE constructing using the standard WKB approximation above the first turning point ($s>s_I + A$ for $A>0$) but under   the second turning point ($s< \ep\ s_{II}$). These functions $w_{1,\pm}$ are particularly important in any compact $s$ region $s_0^{-} \leq s \leq s_0^{+}$ where $s_0^{-}< 0 < s_0^{+}$ are independent of $\LL$ and $\omega$. At this stage, we are able to express $u_H$ in terms of a linear combination of  $w_{1,\pm}$, and obtain schematically $u_H \approx w_{1,+}$. Our ultimate objective will be to also express $u_I$ in terms of a linear combination of  $w_{1,\pm}$, which will require further steps below.
	
	\paragraph{Step 3: the second turning point (Section~\ref{TP2.section})} The presence of a second turning point $s_{II} \approx \LL^2$ matters in the large $\LL$ case (unlike the fixed $\LL$ case considered in \cite{KGSchw1} where this turning point was ignored). The resolution of the second turning point is addressed with standard airy functions denoted $\uAii$ and $\uBii$. We will use this Airy approximation in a range $ s_I  \ll s \leq 10 s_{II}$ (note that the range is $\LL$-dependent but $k$-independent). It is important to note that the second turning point delimits the classically forbidden region (for $s<s_{II}$) where the solution displays exponential-type behavior, and the classically allowed region (for $s_{II}<s<s_{III}$) where the solution  displays oscillatory-type behavior (see Figure~\ref{Potential2}).
	
	\paragraph{Step 4: connection between the first and second turning point (Section~\ref{connection12.section})}  Next, we establish connection formula between the $w_{1,\pm}$ functions obtained by WKB-approximations and the $\uAii$, $\uBii$ Airy-type functions near the second turning point: the connection formula will take the following form \begin{align*}
		\uAii&=\underbrace{\alpha_{A_2}(\omega,\LL)}_{\ls \exp(-2\LL \log \LL+O(\LL))} w_{1,-} +  \underbrace{\beta_{A_2}(\omega,\LL)}_{\approx  \exp(-2\LL \log \LL+O(\LL))} w_{1,+},  
		\\ \nonumber  \uBii(\omega,s) &= \underbrace{\alpha_{B_2}(\omega,\LL)}_{\approx  \exp(2\LL \log \LL+O(\LL))} w_{1,-} +  \underbrace{\beta_{B_2}(\omega,\LL)}_{\ls \exp(-2\LL \log \LL+O(\LL))} w_{1,+}.
	\end{align*}
	
	Note that the coefficients $\alpha_{A_2}$, $\alpha_{B_2}$, $\beta_{A_2}$, $\beta_{B_2}$ depend only on $\omega$ and $\LL$ but not on $s$. 
	The above connection formula is a refinement of the standard wave-penetration through a barrier formula, see \cite{olver}, Chapter 11: essentially, $\alpha_{A_2}$, $\alpha_{B_2}$, $\beta_{A_2}$, $\beta_{B_2}$ display exponential-type behavior in $\int_{s_I}^{s_{II}} |V|^{1/2}(s) ds \approx\LL\log(\LL)$.	
	The multiplicative errors generated in these coefficients (compared, say, to the square  potential barrier) decay in $\LL$, but not in $k$. Similarly, the $\partial_{\omega}$ derivative of these errors only  blows up (polynomially) in $\LL$, but not in $k$.

	\paragraph{Step 5: control of the solution slightly above the second turning point (Section~\ref{bessel.used.section})} The above Airy-approximation breaks\footnote{More precisely, the errors remains  small in amplitude, but their $\rd_{\omega}$ derivatives blow-up like $O(k^2)$ as $s \approx k^2$, which is not allowed in our stationary phase arguments. Indeed, we  only handle singular terms whose amplitude is a very small $O(k^{-1/2})$.} down as $s$ gets closer to the third turning point $s_{III}$. To address this issue, we adopt a similar strategy to our previous work \cite{KGSchw1} and set $\omega=m$ in \eqref{eq:mainV}.  In \cite{KGSchw1} (i.e.\ for fixed $\LL$ and neglecting $\LL$-dependence), we obtained Bessel functions. In the present work, we cannot solve \eqref{eq:mainV} as explicitly due to $\LL$-dependence, but we still perform a WKB-approximation to obtain analogous functions $B^{*}_{\pm}(s)$.   As in \cite{KGSchw1}, we then perturb then to construct $B_{\pm}(\omega,s)$ which are actual solutions of \eqref{eq:mainV} (for general $\omega\in [m-L^{-p},m)$). We note that such an approximation is bound to break down as $s$ gets too close to $s_{III}$, but it turns out to be valid for $s\ls k^{4/3}$ as in \cite{KGSchw1}. The solutions $B_{\pm}(\omega,s)$ can be connected with $\uAii$ and $\uBii$ in the region $s\approx \LL^2$, a process  which generates errors obeying $k$-independent bounds: $$ B_{\pm}(\omega,s) = b_{\pm,+}(\omega) \left(\uAii(\omega,s)-i\uBii(\omega,s)\right)+ b_{\pm,-}(\omega) \left(\uAii(\omega,s)+i\uBii(\omega,s)\right). $$
	
	\paragraph{Step 6: WKB estimate above the second turning point and under the third turning point (Section~\ref{WKB.section2})} We obtain two solutions $w_{2,\pm}$ by another WKB-type estimate in the region $s_{II} \ll s \ll s_{III}$ between the second and third turning point. The key point is that, because this is a classically allowed region where the solution has an oscillatory type behavior, we are allowed to set the errors in the WKB-approximation to be $0$ at some value $s\approx k^2$ (as we did in \cite{KGSchw1} in the analogous region).  The advantage is that the error created by the WKB-approximation are $O(k^{-1})$ small when $s\approx k^2$ (and their $\rd_{\omega}$ derivative blows up as $O(k^2)$), but obey $k$-independent bounds when  $s\approx \LL^2$ (which would not be the case if we naively tried to extend the Airy approximation of $\uAii$ and $\uBii$ to the region $s \approx k^2$). The solutions $w_{2,\pm}$ can be  connected to $B_{\pm}$ when $s\approx k$ (recall that the $B_{\pm}$ approximations are not valid for $s \approx k^2$) as such\begin{equation}\begin{split}
			&w_{2,\pm}(\omega,s) =  a_{\pm, +}(\omega) B_+(\omega,s)+  a_{\pm, -}(\omega) B_-(\omega,s),\\ &| a_{\pm \pm}(\omega)-1|,\ | a_{\pm \mp}|(\omega) \ls k^{-1/2},\ | \rd_{\omega}a_{\pm \pm}|(\omega),\ | \rd_{\omega}a_{\pm \mp}|(\omega)  \ls k^2 \log(k) . \end{split}
	\end{equation}
	
	\paragraph{Step 7: Energy identity (Section~\ref{energy.section})} One important component of the Green's formula \eqref{resolvent1} is $W(u_H,u_I)$ the Wronskian of $u_H$ and $u_I$, and specifically its $\LL$ and $k$ dependence. Expressing both $u_H$ and $u_I$ in terms of $w_{1,\pm}$ (as we will do at the conclusion of all the steps), it is possible to obtain upper bounds on $|W(u_H,u_I)|$. However, it is more important to obtain  upper bounds on $|W(u_H,u_I)|^{-1}$. 
	For this, we use an energy identity, as in \cite{KGSchw1}, taking advantage of the fact that $\Im(u'\bar{u})$ is constant for all solutions of \eqref{eq:mainV} with a real-valued potential $V(s)$. In \cite{KGSchw1}, we showed that $$ \frac{\overline{W(u_H,u_I)}}{W(u_H,u_I)}= \frac{1+\bar{\Gamma}(\omega) e^{-2i\pi k}}{1+\Gamma(\omega) e^{2i\pi k}} \text{ with } |\Gamma|(\omega)<1,$$  and we used this condition to expand in geometric series. \emph{Let us emphasize that the  condition  $ |\Gamma|(\omega)<1$ is insufficient  in the present setting where $\LL$ is unbounded}: we will  quantify how close $|\Gamma|$ is to $1$ as such \begin{equation}\label{energy.explanation}
		|\Gamma|(\omega)= 1- \exp(-4\LL \log\LL+ O(\LL))= 1- \exp(-2 \int_{s_I}^{s_{II}}|V|^{1/2}(s) ds +O(1)).
	\end{equation} To prove \eqref{energy.explanation}, the key is to express $u_H$ in terms of the \emph{real-valued} solutions $w_{2,even}$, $w_{2,odd}$ which are simply linear combinations of $w_{2,\pm}$ as such \begin{align}
		& \label{uHstep}u_H(\omega,s) = \gamma_A(\omega) w_{2,even}(\omega,s)+ \gamma_B(\omega) w_{2,odd}(\omega,s) \\  & w_{2,even} = \Re (e^{i\pi/4} w_{2,+}),\  w_{2,odd} = -\Im (e^{i\pi/4} w_{2,+}).
	\end{align} \eqref{energy.explanation} is then obtained as an application of the energy identity $\Im(u'\bar{u})=cst$ to \eqref{uHstep} combined with $$ \Gamma(\omega,s) = \frac{\gamma_B(\omega,s)-i\gamma_A(\omega)}{\gamma_B(\omega)+i\gamma_A(\omega)} + \text{ errors},$$ which connects $W(u_H,u_I)$ to the coefficients $\gamma_A(\omega)$ and $\gamma_B(\omega)$.

	\paragraph{Step 8: $u_I$ and the third turning point (Section~\ref{uI.section})} $u_I$ is defined as the regular solution of \eqref{eq:mainV} as $s\rightarrow+\infty$. Since the region $s>s_{III}$ near infinity is classically forbidden, it means that $u_I$ adopts an exponentially-decaying type behavior modeled by the Airy function (as in \cite{KGSchw1}). Using the Airy approximation near the third turning point, we obtain  estimates on $u_I$ in the region $\ep\ s_{III} < s$, noting that $u_I$ adopts an oscillatory behavior in the sub-region $\ep\ s_{III} < s< s_{III}$.
	
	\paragraph{Step 9: connection between the second and third turning point (Section~\ref{connection23.section})} We then connect $u_I$ to the solutions $w_{2,\pm}$ obtained by WKB-approximation in Section~\ref{WKB.section2} as such $$ u_I(\omega,s) = a_{3,+}(\omega) w_{2,+}(\omega,s)+  a_{3,-}(\omega) w_{2,-}(\omega,s).$$ The value of $a_{3,\pm}(\omega)$ is obtained evaluating at $s\approx \ep \cdot k^2 \ll s_{III}$ and generates $O(k^{-1})$ errors.
	
	\paragraph{Step 10:  Green's formula (Section~\ref{together.section} and Section~\ref{resolvent.section})}
	
	Combining all previous estimates, we connect $u_I$ to $w_{1,\pm}$ (the relevant solutions in a compact region of $s$) as such \begin{equation}\label{step10}
		u_I(\omega,s) = \alpha_{3,-}(\omega) w_{1,-}(\omega,s)+ \alpha_{3,+}(\omega) w_{1,+}(\omega,s).
	\end{equation} 
	
	By	\eqref{resolvent1}, and recalling the principle that $u_H\approx w_{1,+}$ is a reasonably regular solution, it appears that the key quantity is $u_I(\omega,s)/ W(u_H,u_I)$ which by \eqref{step10} involves the following quantities, which we write schematically in two different ways (up to $O(k^{-1/2})$ errors whose $\rd_{\omega}$ derivative blows up like $O(k^2\log(k))$): \begin{align}
		& \label{res+} \frac{\alpha_{3,+}(\omega)}{ W(u_H,u_I)(\omega)}\approx \ep_\LL^2\ \frac{	e^{O(1)}\cdot\cos(\pi k )+  e^{O(\LL)}\cdot \sin(\pi k)}{ 	\ep_{\LL}^2 \cdot \cos(\pi k )+ \sin(\pi k)}\approx \ep_\LL^2\  \frac{	e^{O(1)} + 	e^{O(1)}\cdot e^{2i\pi k}}{ 1+ \Gamma(\omega,s) e^{2i\pi k}},\\ & \label{res-}  \frac{\alpha_{3,-}(\omega)}{ W(u_H,u_I)(\omega)}\approx  \frac{	\ep_{\LL}^2\ e^{O(1)}\cdot\cos(\pi k )+   \sin(\pi k)}{ 	\ep_{\LL}^2 \cdot \cos(\pi k )+ \sin(\pi k)} \approx   \frac{ e^{O(1)}[ e^{2i\pi k}- 1	]+ \ep_{\LL^2}\ e^{O(1)}[1+ 	 e^{2i\pi k}]}{ 1+ \Gamma(\omega,s) e^{2i\pi k}}.
	\end{align}
	where $\Gamma(\omega,\LL)$ is already mentioned in Step 7 (Section~\ref{energy.section})  and \begin{equation}
		\ep_{\LL} \approx  \exp(-2\LL \log\LL+ O(\LL)) \approx \exp( -\int_{s_{I}}^{s_{II}} |V|^{1/2}(s) ds),\ \Gamma(\omega,s) \approx 1-\ep(\LL).
	\end{equation} 
	Note that combining \eqref{resolvent1} with \eqref{res+}, \eqref{res-} provides rigorously a version of the schematic estimate \eqref{resolvent.intro}: note that the $\exp( O(L))$ term are generated by the $u_H$ and $w_{1,+}$ terms in \eqref{resolvent1} and \eqref{step10}. For the future purpose of using a stationary phase argument, we will  also expand \eqref{res+}, \eqref{res-} in  geometric series as: \begin{equation}\label{taylor}
		\frac{1}{1+\Gamma(\omega) e^{2i\pi k}} = \sum_{n=0}^{+\infty}  [-\Gamma(\omega)]^n e^{2i\pi n k}.
	\end{equation}		
	Note that \eqref{taylor} also turns out to be a Fourier-series expansion of the given function.
	
	While $\frac{\alpha_{3,+}(\omega)}{ W(u_H,u_I)(\omega)}$ has an extra $\ep_{\LL}^2$ factor  compared to  $\frac{\alpha_{3,-}(\omega)}{ W(u_H,u_I)(\omega)}$, but the only term without the multiplicative $\ep_{\LL}^2$ factor in  $\frac{\alpha_{3,-}(\omega)}{ W(u_H,u_I)(\omega)}$ involves $(e^{2i\pi k}-1)$. Taking advantage of this fact, we  use  summation by parts  to obtain a telescoping sum and show that in fact, the contribution of  $\frac{\alpha_{3,-}(\omega)}{ W(u_H,u_I)(\omega)}$ is comparable to that of  $\frac{\alpha_{3,+}(\omega)}{ W(u_H,u_I)(\omega)}$.

	\subsection{The Jost solution  $u_H$ and the first turning point}\label{uH.section}
	We start by a basic analysis of the first turning point, only valid asymptotically as $s \rightarrow -\infty$. Our objective will be to connect $u_H$ (defined in \eqref{uH.def})
	to suitable Bessel functions adapted to the first turning point.
	
	First, we introduce the variable $y= (\K)^{-1} e^{\K s}$ and recast  \eqref{eq:mainV} in the following form: \begin{equation*}
		y^2	\frac{d^2 u}{dy^2} +  y \frac{du}{dy}= \left[ -\frac{\omega^2}{\K^2} +  \tilde{C} \LL^2  y^2 + G(y) \right]u,
	\end{equation*} where $\tilde{C}$ is a postive constant and $G$ is a function satisfying\begin{align*}
		\tilde{C} \sim 1,\qquad G(y)= O(\LL^2 y^4),\qquad  G'(y)= O(\LL^2 y^3).
	\end{align*}
	
	Then, we introduce the rescaled variable $z= \sqrt{\tilde{C}}\ \LL y$ and $\nu= \frac{\omega}{\K}$  so that the equation becomes \begin{equation}\label{eq:main.Z}
		z^2	\frac{d^2 u}{dz^2} +  z \frac{du}{dz}- \left[ z^2-\nu^2 + G(z) \right]u=0.
	\end{equation}
	
	We denote $K_{\pm}(\nu,z)$, the modified Bessel functions, the  solutions of \eqref{eq:main.Z} with $G=0$, see Appendix~\ref{bessel.section}. 
	
	\begin{prop} \cite{Schlag2,Schlag_exp}\label{uH.prop} There exist two independent solutions $\uAi$ and $\uBi$ of \eqref{eq:mainV} with the following properties: define $\tilde{\omega}=  \sqrt{\omega^2+\frac{\K^2}{4}}$ (note that $\tilde{\omega}= \K \sqrt{\nu^2+\frac{1}{4}}$). Then for all $s\leq 0$: \begin{align}
			&\tilde{K}_{\pm}(\omega,s)= (\frac{w(s)}{w'(s)})^{1/2}  K_{\pm}(\frac{\omega}{\kappa_+}\color{black},\frac{\tilde{\omega}\ w(s)}{\K})\left( 1+ \LL^{-1} c_{\pm}(\omega,s)\right),
			\\ & \label{w.der}\frac{dw}{ds} = \frac{\K}{\tilde{\omega}}\ \left(\frac{ -V(s)+\frac{\K^2}{4}}{w^{-2}(s)-1} \right)^{1/2},\\ \label{w.1}& 
			\sqrt{w^2(s)-1} - \arctan(\sqrt{w^2(s)-1})=  \frac{\K}{\tilde{\omega}}\int_{s_I}^s \sqrt{\bigl| V(s')-\frac{\K^2}{4} \bigr|}	 ds'	\text{ for } s> s_{I}(\tilde{\omega}),\\ & \text{ with } w(s) = \frac{\K}{\tilde{\omega}}\int_{s_I}^s \sqrt{\bigl| V(s')-\frac{\K^2}{4} \bigr|}	 ds'+\frac{\pi}{2}+O(
			\frac{1}{\int_{s_I}^s \sqrt{\bigl| V(s')-\frac{\K^2}{4} \bigr|}	 ds'}) \text{ as }  \int_{s_I}^s \sqrt{\bigl| V(s')-\frac{\K^2}{4} \bigr|}	 ds' \rightarrow +\infty. \label{w.asymp}\\ \label{w.3} & \sqrt{1-w^2(s)} +\frac{1}{2} \ln(\frac{1-\sqrt{1-w^2(s)}}{1+\sqrt{1-w^2(s)}}) = -\frac{\K}{\tilde{\omega}}\int^{s_I}_s \sqrt{\bigl| V(s')-\frac{\K^2}{4} \bigr|}	 ds'	\text{ for } s< s_{I}(\tilde{\omega}),\\ \label{w.4}& \text{ with } w(s) =e^{\K (s-s_I)}\left[ e^{\ep_{-\infty}(\omega,\LL)}+O( e^{2\K (s-s_I)} )\right] \text{ as } s\rightarrow -\infty,\\ &  \text{ where for any  }  n \in \mathbb{N}, \rd_{\omega}^n\ep_{-\infty}(\omega,\LL)= O(\LL^{-2}),	\\ & |\rd_{\omega}^{q} \rd_{s}^{j}c_{\pm}|(\omega,s) \lesssim   \left(\frac{ -V(s)+\frac{\K^2}{4}}{w^{-2}(s)-1} \right)^{j/2},\ q,j \in \{0,1\}, \text{ and } \lim_{s\to -\infty}c_{\pm}(\omega,s) =0.
		\end{align}
		
	\end{prop}
	
	\begin{proof} 	
		
		This follows from Lemma 2.4 and Corollary 2.5 in \cite{Schlag2}, recast using our $(\omega,s)$-coordinates. Note, in particular, that \eqref{w.1}, \eqref{w.asymp}, \eqref{w.3} follow from integrating \eqref{w.der}. \eqref{w.4} follows from \eqref{w.3} and the asymptotics of  $\frac{1}{\tilde{\omega}}\int^{s_I}_s \sqrt{\bigl| V(s')-\frac{\K^2}{4} \bigr|}	 ds'$ as $s\rightarrow -\infty$: note indeed  that $w(s) \rightarrow 0$ as  $s\rightarrow -\infty$,  and by \eqref{w.3}: \begin{equation}\begin{split}
				&	\frac{\ln( w(s))+	1-\ln(2) +O(w^2(s))}{\K}= \int_{s_I}^s\frac{ \bigl| V(s')-\frac{\K^2}{4} \bigr|^{\frac{1}{2}}}{\tilde{\omega}}	 ds'=  \int_{s_I}^s |1-  e^{2\K(s'-s_I)}f(s')|^{\frac{1}{2}} ds'\\ &= \frac{1}{2\K}\int^{0}_{2\K (s-s_I)} (1-  e^{u}f(\frac{u}{2\K}+s_I))^{\frac{1}{2}} du	= (s-s_I) + \frac{1}{2\K}\int^{0}_{2\K (s-s_I)}\left[ 1-(1-  e^{u}f(\frac{u}{2\K}+s_I))^{\frac{1}{2}}\right] du\\&	= (s-s_I) + \frac{1}{2\K}\underbrace{\int^{0}_{-\infty}\left[ 1-(1-  e^{u}f(\frac{u}{2\K}+s_I))^{\frac{1}{2}}\right] du}_{=2(1-\ln(2))+O(\LL^{-2})}+O( e^{2\K(s-s_I)}) ,\end{split}		\end{equation} where we have used the notation $\frac{\omega^2+V(s)}{\tilde{\omega}^2} = e^{2\K (s-s_I)}  f(s)$, where $f(s) = 1+O(\LL^{-2})$ for any $s< s_I$. Then, \eqref{w.4} indeed follows.
		
	\end{proof}
	
	\begin{cor}\label{uH.cor}
		$u_H$ can be expressed as \begin{equation}\label{u_H.formula}
			u_H(\omega,s) =  C_1(\omega) \cdot \exp\left(i \frac{\omega}{\K} \log\LL \right)\cdot \left(\uAi + i \uBi\right),
		\end{equation}
		and where $|C_1|(\omega,M,\DD)\approx 1$ and  for any $n\in \mathbb{N}$:
		\begin{equation}
			|\rd^{n}_{\omega}C_1|(\omega) \ls_n 1.
		\end{equation}			
		
	\end{cor}\begin{proof}
		The corollary follows immediately from Proposition~\ref{uH.prop} and the asymptotics of $K_{\pm}$ from Section~\ref{bessel.section}. 
	\end{proof}

	\subsection{WKB estimates  below the second turning point}\label{WKB.section1}

	\begin{prop}\label{WKB.TP1.prop} Let $A>0$ (independent of $\omega$ and $\LL$) and $\epsilon > 0$ be sufficiently small. For any $ s_I + A \leq s \leq \ep \LL^2$, the ODE \eqref{eq:mainV} admits the following two solutions\begin{align}& \tilde{w}_{1,\pm}(\omega,s) = |V|^{-1/4}(s) \exp( \pm\int_{s_I}^{s} |V|^{1/2}(s') ds' ) \left(1+ \tilde{\ep}_{1,\pm}^{WKB}(\omega,s)\right),\\ & |\tilde{\ep}_{1,\pm}^{WKB}|(\omega,s),\ |V|^{-1/2}(s)  |\rd_{s}\tilde{\ep}_{1,\pm}^{WKB}|(\omega,s)\lesssim  e^{-\K A}, \label{WKB.error1}\\ &  \\ & |\rd_{\omega}\tilde{\ep}_{1,\pm}^{WKB}|(\omega,s),\ |V|^{-1/2}(s)  |\rd^2_{\omega s}\tilde{\ep}_{1,\pm}^{WKB}|(\omega,s)\lesssim  \LL^3 e^{-\kappa_+ A},
			\label{WKB.error2} 
			\\ &\tilde{\ep}_{1,+}^{WKB}(s=s_I + A)=\rd_{s}\tilde{\ep}_{1,+}^{WKB}(s=s_I + A)=0,\ \tilde{\ep}_{1,-}^{WKB}(s= \ep \LL^2)=\rd_{s}\tilde{\ep}_{1,-}^{WKB}(s= \ep \LL^2)=0 .  \end{align} 
	\end{prop}\begin{proof}

		By Section 2.4 of Chapter 6 in~\cite{olver}, we find that $\epsilon^{WKB}_{1,\pm}$ solves the following Volterra equation
		\begin{equation}\label{tosolvethewkbwkb}
			\tilde{\ep}^{WKB}_{1,\pm}\left(\omega,s\right) = \frac{1}{2i}\int^{\xi_{\pm}(s)}_{0}\left(1-e^{2i\left(v-\xi_{\pm}(s)\right)}\right)\psi(v)\left(1+\tilde{\ep}^{WKB}_{1,\pm}\right)\, dv,
		\end{equation}
		where $dv$ refers to an integration with respect to the $\xi_{\pm} = \int^{s}_{S_{\pm}} |V|^{1/2}(s) ds$ variable, where $S_{-}= s_I + A$, $S_{+}= \ep \LL^2$, and where $\psi$ is defined by
		$$\psi(v) \doteq - V^{-3/4}\frac{\partial^2}{\partial s^2}\left(V^{-1/4}\right). $$

		We remark that $\int_{0}^{\xi_{\pm}(s)} |\psi|(v) dv = \int_{S_{\pm}}^{s} |\psi|(\xi(s)) |V|^{1/2}(s) ds$ and $\rd_{\omega} V = -2\omega$. 
		
		In the region, $s_I +A \leq s \leq \ep \LL^2$, observe the following: 
		\begin{align*}
			&	V(s) \approx \LL^2 \left[ [1+s]^{-2} 1_{s\geq 0} + e^{2\K s} 1_{s\leq 0}\right],\\ &  	|\psi| \ls \LL^{-2}\left[ 1+ e^{-2\K s} 1_{s\leq 0}\right],\\ & 	|\rd_{\omega}\psi| \ls \LL^{-4}\left[ [1+s^2]  1_{s\geq 0} + e^{-4\K s} 1_{s\leq 0}\right],\\ &  |\xi_{\pm}|(s) \ls \LL \log \LL ,\ |\rd_{\omega}\xi_{\pm}|(s) \ls \LL^3.
		\end{align*}

		Thus, we obtain \eqref{WKB.error1} after applying Theorem~\ref{thm:volterra} with $K(\xi,v)= (2i)^{-1} (1- e^{2i(v-\xi)})$, $P_0$, $P_1$ and $Q$ to be constants (independent of $\omega$ and $\LL$) and  $\Phi(s)=  e^{-\K A}$. Similarly, after taking a $\rd_{\omega}$ derivative, 
		we obtain  \eqref{WKB.error2} as an application of Theorem~\ref{thm:volterraparam}.

	\end{proof}
	
	We also note that, evaluating the expression of Proposition~\ref{WKB.TP1.prop} at $s=0$, we obtain:  \begin{equation}\begin{split}\label{WKB.wronskian.prelim}
			&	\bigl| W(\tilde{w}_{1,+}, \tilde{w}_{1,-})-2\bigr| \ls e^{-\K A}, \\ &  |\rd_{\omega}W(\tilde{w}_{1,+}, \tilde{w}_{1,-})|\ls  \LL^3 e^{-\kappa_+ A}.
		\end{split}
	\end{equation} Keeping \eqref{WKB.wronskian.prelim} in mind, we define the following solutions of \eqref{eq:mainV} as $w_{1,\pm}(\omega,s)= \frac{\sqrt{2}\ \tilde{w}_{1,\pm}(\omega,s)}{\sqrt {W(\tilde{w}_{1,+}, \tilde{w}_{1,-})}}$ and \begin{equation}
		w_{1,\pm}(\omega,s)=  |V|^{-1/4}(s) \exp( \pm\int_{s_I}^{s} |V|^{1/2}(s') ds' ) \left(1+ \ep_{1,\pm}^{WKB}(\omega,s)\right)
	\end{equation} so that \begin{equation}\label{WKB.wronskian}
		W(w_{1,+}, w_{1,-})=2.
	\end{equation} We will give up on $\tilde{w}_{1,+}(\omega,s)$  and only work with the solutions $w_{1,\pm}(\omega,s)$ defined as such in what follows. Note that $\ep_{1,\pm}^{WKB}(\omega,s)$ also obeys the  same estimates as  $\tilde{\ep}_{1,\pm}^{WKB}(\omega,s)$, i.e.\ \eqref{WKB.error1} and \eqref{WKB.error2}.
	\begin{cor}\label{cor.TP1}
		We write\begin{align}&\uAi(\omega,s) = \alpha_{1,+}(\omega) w_{1,+}(\omega,s)+  \alpha_{1,-}(\omega) w_{1,-}(\omega,s),\\&\uBi(\omega,s) = \beta_{1,+}(\omega) w_{1,+}(\omega,s)+  \beta_{1,-}(\omega) w_{1,-}(\omega,s), \\ & u_H(\omega,s)= C_1(\omega) e^{i \frac{\omega}{\kappa_+} \log( \LL)} \left( [\alpha_{1,+}(\omega)+i\beta_{1,+}(\omega)] w_{1,+}(\omega,s)+[\alpha_{1,-}(\omega)+i\beta_{1,-}(\omega)] w_{1,-}(\omega,s)\right)\end{align} and the following estimates hold, defining $C_I(\omega,\LL):= \LL^{-1}\int_{s_I}^{1} |V|^{1/2}(s) ds$:  \begin{align} &  \alpha_{1,+}= 			\LL
			\exp(-2 C_I \LL) \underbrace{\tilde{\alpha}_{1,+}}_{=O_{A(1)}}, \\ &  \alpha_{1,-}  = E_- \cdot (1+O(e^{-\K A})),\\ &  \beta_{1,+}=  E_+ \cdot (1+O(			e^{-\K A}
			)), \\ &\beta_{1,-}  = O_A(1),\\& E_{\pm}(\omega,\LL) =\exp\left(\mp\frac{\K^2}{4}\int_{s_I}^{0} \frac{ds}{|V-\frac{\K^2}{4}|^{1/2}(s)+|V|^{1/2}(s)}\right) \approx 1,\\ & C_I(\omega,\LL) \approx 1,\ |\rd_{\omega} C_I|(\omega,\LL) \ls \LL^{-1}.  \end{align} Moreover, we have \begin{equation} \label{domega.connection1WKB}
			|\rd_{\omega} {\alpha}_{1,\pm}|(\omega,s),\ |\rd_{\omega} \beta_{1,\pm}|(\omega,s) \lesssim_A \LL^3 .
	\end{equation}\end{cor}
	
	\begin{proof} For this computation, we  denote (see Appendix~\ref{bessel.section}) \begin{equation}
			m_{\pm}(x)= \sqrt{x} \exp(\mp x) K_{\pm}(x)= m^{\infty}_{\pm}\cdot (1\pm\frac{\nu^2+\frac{1}{4}}{2 x} + O(x^{-2})),
		\end{equation} with   $m^{\infty}_{-}=\sqrt{\frac{\pi}{2}}$ and $m^{\infty}_{+}=\frac{1}{\sinh(\pi \nu)}\sqrt{\frac{\pi}{2}}$. We assume $ s_I+A\leq  s \leq 0$, for $A>0$ large enough (so that $\ep_{WKB}(s)< \frac{1}{2}$ for any $s_I+A \leq s\leq 0$). 
		Note also that $w(s),\ w'(s)  \approx \LL e^{\K s} \gtrsim e^{\K A} $.\begin{align*}
			&W(	\tilde{K}_{\pm},  w_{1,\pm} )\\ & =   W\left([ (\frac{ \K}{ \tilde{\omega}\ w'(s)})^{1/2}  \exp(\pm\frac{\tilde{\omega}\ w(s)}{\K}) m_{\pm}(\frac{\tilde{\omega}\ w(s)}{\K}) \left( 1+ \LL^{-1} c_{\pm}(\omega,s)\right), |V|^{-1/4}(s) \exp( \pm \int_{s_I}^{s} |V|^{1/2}(s') ds' )[1+\ep^{WKB}_{1,\pm}]\right)\\ & =   W( \frac{\bigl| V- \frac{\K^2}{4} \bigr|^{-1/4}(s)}{|1-w^{-2}(s)|^{-1/4}}  \exp(\pm \frac{\tilde{\omega}\ w(s)}{\K}) m_{\pm}(\frac{\tilde{\omega}\ w(s)}{\K}) \left( 1+ \LL^{-1} c_{\pm}(\omega,s)\right), |V|^{-1/4}(s) \exp(\pm \int_{s_I}^{s} |V|^{1/2}(s') ds' )[1+\ep^{WKB}_{1,\pm}])\\ & =  (1+\LL^{-1} c_{\pm}(\omega,s)) \frac{\bigl| V- \frac{\K^2}{4} \bigr|^{-1/4}(s) }{|1-w^{-2}(s)|^{-1/4}} \exp(\pm \frac{\tilde{\omega}\ w(s)}{\K}\pm \int_{s_I}^{s} |V|^{1/2}(s') ds') m_{\pm}(\frac{\tilde{\omega}\ w(s)}{\K}) |V|^{-1/4}(s)\cdot [1+\ep^{WKB}_{1,\pm}(s)]\\ & \cdot (I+II+III+IV_A+IV_B+IV_{C})(s), \end{align*} with \begin{align*}& I= \frac{\LL^{-1} \rd_{s}c_{\pm}(\omega,s)}{1+\LL^{-1} c_{\pm}(\omega,s)}= O( e^{\K s}),\\ &II =  \frac{\tilde{\omega}}{\K} \frac{m_{\pm}'(\frac{\tilde{\omega} w(s)}{\K})\cdot w'(s)}{m_{\pm}(\frac{\tilde{\omega} w(s)}{\K})} = \mp \frac{\frac{\tilde{\omega}}{\K}w'(s)}{2w^2(s)}+ O (w'(s) \cdot w^{-3}(s))		=O(\LL^{-1} e^{-\K s}),\\ & III=- \frac{w'(s)}{2w^3(s)}[1-w^{-2}(s)]^{-1} =O(\LL^{-2} e^{-2\K s}),\\ & IV_A=\pm \left[ \frac{|\tilde{\omega}| w'(s)}{\K}- |V|^{1/2}(s)\right] =\pm  |V|^{1/2}(s) \left[ (1-\frac{\K^2}{4|V|(s)})^{1/2} [1-w^{-2}(s)]^{-1/2}-1 \right]=O(\LL^{-1}e^{-\K s})
			,\\ & IV_B = \frac{1}{4} \left[ \frac{|V|'(s)}{|V|(s)} -\frac{|V|'(s)}{|V-\frac{\K^2}{4}|(s)}\right]=-\frac{\K^2\ |V|'(s)}{16|V|(s) \cdot|V-\frac{\K^2}{4}|(s)}= O(\LL^{-2} e^{-2\K s}),\\ & IV_{C}= \frac{\rd_s \ep^{WKB}_{1,\pm}(s)}{1+ \ep^{WKB}_{1,\pm}(s)}=O(\LL e^{\K (s-A)}) .
		\end{align*} From the above, we deduce \begin{equation*}
			W(	\tilde{K}_{\pm},  w_{1,\pm} )=  |V|^{-1/2}(s)  \exp(\pm \frac{|\tilde{\omega}|\ w(s)}{\K}\pm \int_{s_I}^{s} |V|^{1/2}(s') ds') O \left( 1+ \LL e^{\K (s-A)} + \LL^{-1}e^{-\K s}\right).
		\end{equation*} Thus, we have \begin{equation}
			W(	\tilde{K}_{+},  w_{1,+} )=   W(	\tilde{K}_{+},  w_{1,+} )(s_I+A)=O_A(1),
		\end{equation}
		\begin{equation}
			W(	\tilde{K}_{-},  w_{1,-} )=   W(	\tilde{K}_{-},  w_{1,-} )(s=1)=O(\LL \exp(-2 \int_{s_I}^1 |V|^{1/2}(s) ds)).
		\end{equation}
		Now, we turn to the other sets of Wronskians:
		\begin{align*}
			&W(	\tilde{K}_{\pm},  w_{1,\mp} )\\ & =   W\left([ (\frac{ \K}{ |\tilde{\omega}|\ w'(s)})^{1/2}  \exp(\pm\frac{|\tilde{\omega}|\ w(s)}{\K}) m_{\pm}(\frac{|\tilde{\omega}|\ w(s)}{\K}) \left( 1+ \LL^{-1} c_{\pm}(\omega,s)\right), |V|^{-1/4}(s) \exp( \mp \int_{s_I}^{s} |V|^{1/2}(s') ds' )[1+\ep^{WKB}_{1,\mp}]\right)\\ & =   W( \frac{\bigl| V- \frac{\K^2}{4} \bigr|^{-1/4}(s)}{|1-w^{-2}(s)|^{-1/4}}  \exp(\pm \frac{|\tilde{\omega}|\ w(s)}{\K}) m_{\pm}(\frac{|\tilde{\omega}|\ w(s)}{\K}) \left( 1+ \LL^{-1} c_{\pm}(\omega,s)\right), |V|^{-1/4}(s) \exp(\mp \int_{s_I}^{s} |V|^{1/2}(s') ds' )[1+\ep^{WKB}_{1,\mp}])\\ & =  (1+\LL^{-1} c_{\pm}(\omega,s)) \frac{\bigl| V- \frac{\K^2}{4} \bigr|^{-1/4}(s) }{|1-w^{-2}(s)|^{-1/4}} \exp(\pm \frac{|\tilde{\omega}|\ w(s)}{\K}\mp \int_{s_I}^{s} |V|^{1/2}(s') ds') m_{\pm}(\frac{|\tilde{\omega}|\ w(s)}{\K}) |V|^{-1/4}(s)\cdot [1+\ep^{WKB}_{1,\mp}]\\ & \cdot (I+II+III+IV_A+IV_B+IV_{C})(s),\end{align*} with \begin{align*}& I= \frac{\LL^{-1} \rd_{s}c_{\pm}(\omega,s)}{1+\LL^{-1} c_{\pm}(\omega,s)}= O( e^{\K s}),\\ &II =  \frac{\tilde{\omega}}{\K} \frac{m_{\pm}'(\frac{|\tilde{\omega}| w(s)}{\K})\cdot w'(s)}{m_{\pm}(\frac{|\tilde{\omega}| w(s)}{\K})} = \mp \frac{\frac{\tilde{\omega}}{\K}w'(s)}{2w^2(s)}+ O (w'(s) \cdot w^{-3}(s))
			=O(\LL^{-1} e^{-\K s}),\\ & III=- \frac{w'(s)}{2w^3(s)}[1-w^{-2}(s)]^{-1} =O(\LL^{-2} e^{-2\K s}),\\ & IV_A=\pm \left[ \frac{|\tilde{\omega}| w'(s)}{\K}+ |V|^{1/2}(s)\right] =\pm  2 |V|^{1/2}(s) +O(\LL^{-1}e^{-\K s})
			,\\ & IV_B = \frac{1}{4} \left[ \frac{|V|'(s)}{|V|(s)} -\frac{|V|'(s)}{|V-\frac{\K^2}{4}|(s)}\right]=-\frac{\K^2\ |V|'(s)}{16|V|(s) \cdot|V-\frac{\K^2}{4}|(s)}= O(\LL^{-2} e^{-2\K s}),\\ & IV_{C}= \frac{\rd_s \ep^{WKB}_{1,\mp}(s)}{1+ \ep^{WKB}_{1,\mp}(s)}=O(\LL e^{\K (s-A) }) .
		\end{align*}
		Note also that \begin{equation}\begin{split}
				\frac{\bigl| V- \frac{\K^2}{4} \bigr|^{-1/4}(s) }{|1-w^{-2}(s)|^{-1/4}} \exp(\pm \frac{|\tilde{\omega}|\ w(s)}{\K}\mp \int_{s_I}^{s} |V|^{1/2}(s') ds') m_{\pm}(\frac{|\tilde{\omega}|\ w(s)}{\K}) |V|^{-1/4}(s)\cdot [1+\ep^{WKB}_{1,\mp}(s)]\\= m_{\pm}^{\infty}|V|^{-1/2}(s)  \exp\left(\pm\int_{s_I}^{0} [|V-\frac{\K^2}{4}|^{1/2}(s)-|V|^{1/2}(s)]ds\right)\left( 1+ O(e^{-\K A})\right)\\ = m_{\pm}^{\infty}|V|^{-1/2}(s)  \exp\left(\mp\frac{\K^2}{4}\int_{s_I}^{0} \frac{ds}{|V-\frac{\K^2}{4}|^{1/2}(s)+|V|^{1/2}(s)}\right)\left( 1+ O(e^{-\K A})\right)
			\end{split}
		\end{equation}

		From the above, we deduce \begin{equation*}
			W(	\tilde{K}_{\pm},  w_{1,\mp} ) =\pm 2  \underbrace{ m_{\pm}^{\infty}  \exp\left(\mp\frac{\K^2}{4}\int_{s_I}^{0} \frac{ds}{|V-\frac{\K^2}{4}|^{1/2}(s)+|V|^{1/2}(s)}\right)}_{:=  E_{\pm}}  \left( 1+ O(e^{-\K A})\right).
		\end{equation*} Thus, evaluating the above at $s=0$, we have \begin{equation}
			W(	\tilde{K}_{\pm},  w_{1,\mp} )=   \pm 2E_{\pm} \cdot \left(1+O(e^{-\K A})\right).
		\end{equation}

		Now, we turn to the $\rd_{\omega}$ estimates. We will get content with sub-optimal bounds, which will turn out to suffice. Note that in the view of Proposition~\ref{WKB.TP1.prop}, Proposition~\ref{uH.prop} and the identity $\rd_{\omega} V = -2\omega$, we have, for $0 \geq s\geq s_{I}+A$ with $A\gg 1$ (recall that in this regime $w(s) \gg 1$): \begin{align*} &  \frac{| w_{1,\pm}|(\omega,s)}{|V|^{-1/4}(s)  \exp(\pm \int_{s_I}^{s} |V|^{1/2}(s')ds')} \ls 1\\
			& \frac{|\rd_{\omega} w_{1,\pm}|(\omega,s)}{|V|^{-1/4}(s)  \exp(\pm \int_{s_I}^{s} |V|^{1/2}(s')ds')} \ls \left(|V|^{-1}(s)+\int_{s_I}^{s} |V|^{-1/2}(s')ds' + |\rd_{\omega} \ep^{WKB}_{1,\pm}|(\omega,s)  \right)\\ 	& \frac{|\rd_{\omega s}^2 w_{1,\pm}|(\omega,s)}{|V|^{-1/4}(s)  \exp(\pm \int_{s_I}^{s} |V|^{1/2}(s')ds')} \ls\\ & \left( |\rd^2_{\omega s} \ep^{WKB}_{1,\pm}|+ [|V|^{1/2}+ |V|^{-1}][|\rd_{\omega} \ep^{WKB}_{1,\pm}|+\int_{s_I}^{s} |V|^{-1/2}(s')ds']+[|V|^{-1}(s)+\int_{s_I}^{s} |V|^{-1/2}(s')ds']|\rd_{s}  \ep^{WKB}_{1,\pm}|  \right),\end{align*} and \begin{align*}  & \frac{| \tilde{K}_{\pm}|(\omega,s)}{|V|^{-1/4}(s)  \exp(\pm \frac{\tilde{\omega}\cdot w(s)}{\kappa_+})} \ls 1,\ \frac{| \rd_s\tilde{K}_{\pm}|(\omega,s)}{|V|^{-1/4}(s)  \exp(\pm \frac{\tilde{\omega}\cdot w(s)}{\kappa_+})} \ls \frac{|V'|(s)}{|V|(s)} + |V(s)-\frac{\kappa_+^2}{4}|^{1/2},\\ & |\rd_{\omega} w|(s)\lesssim \int_{s_I}^{s} |V-\frac{\kappa_+^2}{4}|^{1/2}(s') ds'+\int_{s_I}^{s} |V-\frac{\kappa_+^2}{4}|^{-1/2}(s') ds',\\ & |\rd_{\omega s}^2 w|(s)\lesssim  |V-\frac{\kappa_+^2}{4}|^{-1/2}(s) +  w^{-3}(s)|V-\frac{\kappa_+^2}{4}|^{1/2}(s)\ |\rd_{\omega} w|,\\ &  \frac{| \rd_{\omega}\tilde{K}_{\pm}|(\omega,s)}{|V|^{-1/4}(s)  \exp(\pm \frac{\tilde{\omega}\cdot w(s)}{\kappa_+})} \ls |V|^{-1}(s)+ |\rd_{\omega} w|(s)+ |w|(s),\\  &  \frac{| \rd_{\omega s}^2\tilde{K}_{\pm}|(\omega,s)}{|V|^{-1/4}(s)  \exp(\pm \frac{\tilde{\omega}\cdot w(s)}{\kappa_+})} \ls  \frac{ |V'|(s)}{|V|^2(s)}+ |\rd_{\omega s} w|(s)+ |V-\frac{\kappa_+^2}{4}|^{1/2} \left( 1+ |V|^{-1}(s)+ |\rd_{\omega} w|(s)+ |w|(s) \right).
		\end{align*}
		
		We will evaluate all the above expression at $s=s_I+A$, where $A=O(1)\gg 1$ is independent of $\LL$ and $\omega$, thus $|V|(s) \approx |V-\frac{\kappa_+^2}{4}|(s) \approx 1$, $|w|(s) \approx 1$, and \begin{equation*}
			\begin{split}
				&|\rd_{\omega} w|,\ |\rd^2_{\omega s} w|,\ | w_{1,\pm}|,\ |\tilde{K}_{\pm}|, \sum_{\alpha\in \{\omega,s\}}\left[| \rd_{\alpha} w_{1,\pm}|+| \rd_{\alpha}\tilde{K}_{\pm}|+ |\rd_{\alpha} \ep^{WKB}_{1,\pm}|\right] \lesssim_A \LL^3,\\ & | \rd_{\omega s}^2 w_{1,\pm}|,\ | \rd_{\omega s}^2\tilde{K}_{\pm}|,\  |\rd_{\omega } \ep^{WKB}_{1,\pm}|,\ |\rd_{\omega s}^2 \ep^{WKB}_{1,\pm}| \lesssim_A \LL^3.
			\end{split}
		\end{equation*} 
		Note also the formula $$ \rd_{\omega}[ W(f,g)]= \rd^2_{\omega s } f \cdot g+  \rd_{s } f \cdot  \rd_{\omega}g -  \rd_{\omega}f \cdot  \rd_{s} g-  \rd_{\omega s }^2 g \cdot f,$$ therefore we do obtain, using the above inequalities, $$ |\rd_{\omega}[ W(\tilde{K}_{\pm},w_{1,\pm})]|,\ |\rd_{\omega}[ W(\tilde{K}_{\pm},w_{1,\mp})]| \lesssim_A \LL^3 ,$$ from which \eqref{domega.connection1WKB} follows, using also \eqref{WKB.wronskian}.
	\end{proof}

	\subsection{The second turning point}\label{TP2.section}
	
	First, we recall the definition of $s_{\pm}(\omega,\LL)$ from \eqref{spm.def} and introduce the coordinate: \begin{equation}\label{x.def}
		\xx= (Mm^2)\ \frac{s}{  \LL^2}.
	\end{equation}
	We also define correspondingly $\xx_{\pm}(\LL,\omega):= \xx(s_\pm(\LL,\omega))$,  $\xx_{II}(\LL,\omega):= \xx(s_{II}(\LL,\omega))$, and finally $\xx_{III}(\LL,\omega):= \xx(s_{III}(\LL,\omega))$.	Note that by  \eqref{sII}, \eqref{sIII}, we have \begin{align*}
		& 	X_{II}(\LL,\omega)= X_-(\LL,\omega)+O(\LL^{-2})= \underbrace{\frac{1-\sqrt{1-\alpha}}{\alpha}}_{\in  (\frac{1}{2},1 )}+ O(\LL^{-2}),\\&  	X_{III}(\LL,\omega)= X_+(\LL,\omega)+O(\LL^{-2})= \frac{1+\sqrt{1-\alpha}}{\alpha}+O(\LL^{-2}) .\\
	\end{align*}
	In particular, note that $X_{-}< 1 < \alpha^{-1}< X_{+}$ and $X_{II} \in  \left((1-\ep) X_-,  1\right)$ and $X_{III} \in  \left( \alpha^{-1},  (1+\ep) X_+\right)$. Note also that for large $\LL$, $X_{II},\ X_{-} =O(1)$ and $X_{III},\ X_{+} =O(\alpha^{-1})$.

	In what follows, we will construct Airy-type solutions adapted to the second turning point in the region $X\in [(Mm^2) \LL^{-2},10]$ (i.e.\ up to  $s=  \frac{10 \LL^2}{Mm^2}$, which is slightly above the second turning point). 
	
	We also recall that in the present section, we have assumed that $0 < \alpha \ls \LL^{2-p}\ll 1$ for some large $p>2$.

	Note that in the following proposition, $\{X< X_{II}\}$ (below the second turning point) corresponds to a so-called classically forbidden regime, in which the Airy functions adopt an exponentially growing/decaying bheavior, while the region  $\{X> X_{II}\}$ is contained in the so-called classically allowed region where (in contrast to the forbidden region) the Airy functions adopt an oscillatory behavior.
	\begin{prop}\label{TP2.prop}
		There exists		$\tilde{\uAii}$ and $\tilde{\uBii}$ solutions of \eqref{eq:mainV} which  obey the following estimates.
		\begin{align}
			&\tilde{\uAii}= \hat{f_2}^{-1/4} \left( Ai( \LL^{2/3} \zeta_2) + \tilde{\ep}_{Ai_2}\right),\\&\tilde{\uBii}=  \hat{f_2}^{-1/4} \left( Bi( \LL^{2/3} \zeta_2) + \tilde{\ep_{Bi_2}}\right), \\ & f_2(\xx) = \frac{\LL^2}{(Mm^2)^2} V( r(\xx)),\\  
			& \frac{2}{3}\left[sign(X_{II}-X)\zeta_2\right]^{3/2}(\xx)=\int_{\xx}^{\xx_{II}} |f_2|^{1/2}(y) dy= \LL^{-1} \int^{s_{II}}_{s(X)} |V|^{1/2}(s')ds', \\ & \hat{f_2}(\xx) = \frac{f_2(\xx)}{ \zeta_2(\xx)}.
		\end{align}  with $\tilde{\ep_{Ai_2}}(s=s_I+A)=\rd_s\tilde{\ep_{Ai_2}}(s=s_I+A)=\tilde{\ep_{Bi_2}}(s=\frac{10 \LL^2}{Mm^2})=\rd_s\tilde{\ep_{Bi_2}}(s=\frac{10 \LL^2}{Mm^2})=0$, and satisfying  the following estimates for all $ s_I+A \leq s \leq  \frac{10 \LL^2}{Mm^2}$, recalling that $s_+=\max\{0,s\}$ and $s_-=\max\{0,-s\}$: \begin{align} 
			&\frac{|\tilde{\ep_{Ai_2}}|(\omega,s)}{M(\LL^{2/3} \zeta_2)},\ \frac{ \LL^2 |\rd_{s} \tilde{\ep_{Ai_2}}|(\omega,s)}{ \LL^{2/3} \hat{f}_2^{1/2}(s)\ N(\LL^{2/3} \zeta_2)} \lesssim   E^{-1}(\LL^{2/3} \zeta_2)\cdot e^{-\kappa_+ A}
			, \label{TP2.1}\\ &\frac{|\tilde{\ep_{Bi_2}}|(\omega,s)}{M(\LL^{2/3} \zeta_2)},\ \frac{\LL^2|\rd_{s} \tilde{\ep_{Bi_2}}|(\omega,s)}{ \LL^{2/3} \hat{f}_2^{1/2}(s)\ N(\LL^{2/3} \zeta_2)} \lesssim   E(\LL^{2/3} \zeta_2) \cdot  e^{-\kappa_+ A}.
			\label{TP2.2} \end{align}
		We will also rewrite the following estimates in a simpler way, 	in a large region excluding the second turning point $s=s_{II}$: for $ s\in[s_I+A, \frac{ \ep \LL^2}{Mm^2}] \cup [\frac{2\LL^2}{Mm^2},\frac{10\LL^2}{Mm^2}]$, defining the following quantities:
		\begin{align*}
			&\tilde{\mathcal{E}}_{Ai_2}(\omega,s)= \LL^{1/6} |\zeta_2|^{1/4} \tilde{\ep}_{Ai_2}(s) \cdot \exp(\frac{2}{3}\LL |\zeta_2|^{3/2})
			,\\ &  \tilde{\mathcal{E}}_{Bi_2}(\omega,s)=  \LL^{1/6}|\zeta_2|^{1/4} \tilde{\ep}_{Bi_2}(s) \cdot \exp(-\frac{2}{3}\LL |\zeta_2|^{3/2}),
		\end{align*} 	
		\begin{align} 
			&|\tilde{\mathcal{E}}_{Ai_2}|(\omega,s),\ \frac{|\rd_s\tilde{\mathcal{E}}_{Ai_2}|(\omega,s)}{|V|^{1/2}(s)} \lesssim e^{-\kappa_+ A},
			\label{TP2.1'}\\ &|\tilde{\mathcal{E}}_{Bi_2}|(\omega,s),\ \frac{|\rd_s\tilde{\mathcal{E}}_{Bi_2}|(\omega,s)}{|V|^{1/2}(s)} \lesssim  e^{-\kappa_+ A}.\label{TP2.2'} \end{align}
		
		Finally, note the Wronskian estimate \begin{equation}\label{W.TP2.prelim}
			\big|	 \LL^{4/3}W(\tilde{\uAii},\tilde{\uBii})-    \frac{M m^2}{\pi}\bigr| \ls e^{-\K A}.
		\end{equation}
	\end{prop}
	
	\begin{proof}

		We express \eqref{eq:mainV} using coordinate $X$ from \eqref{x.def} and it becomes \begin{equation}
			\frac{d^2 u}{d\xx^2} = \LL^2\  f_2(\xx) u .
		\end{equation}

		We will look at two regimes of interest, where we recall that  $\ep\in(0,1)$ is a fixed small constant  (independent of $\LL$ and $\omega$). Recall that we assume $ \frac{M m^2}{\LL^2}\leq X\leq 10$, and we will take the total variation of the error control  function $H$ on $[\frac{M m^2}{\LL^2},10]$, using the notations from Theorem 3.1, Chapter 11 of \cite{olver} (with $g=0$), where $H$ is defined as $$ H =  -\int_{0}^{\zeta_2} |v|^{-1/2}	\hat{f}_2^{-3/4} \frac{d^2}{d\xx^2}(\hat{f}_2^{-1/4}) dv= \int_{X_{II}}^X \left[|f_2|^{-1/4} \frac{d^2}{d\xx^2}(|f_2|^{-1/4}) -\frac{5|f_2|^{1/2}}{16|\zeta_2|^3}\right] d\xx.$$  We therefore split into several regions according to the value of $X$. \begin{enumerate}

			\item \label{TP2.step3} $ \ep \xx_{-} \leq \xx \leq 10$.
			Recall that the second turning point is contained in this region, i.e.\ $\xx_{II}\in  (\ep \xx_{-}, 10)$. Note that near the turning point $\xx = \xx_{II}$, we have: \begin{align*}
				&	f_2(\xx) \approx (X_{II}-X),\ |\zeta_2|(\xx) \approx  |X-X_{II}|,\\ &  \hat{f}_2(\xx)  \approx 1,\ |\hat{f}_2'|(\xx)  \lesssim 1,\ | \hat{f}_2''|(\xx)  \lesssim  1,
			\end{align*}  hence \begin{equation}
				|	\hat{f}_2^{-1/4} \frac{d^2}{d\xx^2}(\hat{f}_2^{-1/4})|(\xx) \lesssim  1.
			\end{equation}
			\begin{equation}
				\int_{\ep \xx_{-}}^{10}|\zeta_2|^{-1/4}	\hat{f}_2^{-1/4} \frac{d^2}{d\xx^2}(\hat{f}_2^{-1/4})|(\xx) d\xx \lesssim   1.	
			\end{equation}

			The conclusion is that 	\begin{equation}\label{H.TP2.3}
				TV_{\ep \xx_{-},10 }[H] \lesssim 1.
			\end{equation}

			\item  $\frac{M m^2}{\LL^2} \leq \xx \leq \epsilon  \xx_{-}$.
			
			In this region, we have the following estimates \begin{align*}
				&f_2(\xx) \approx \xx^{-2},\\ & |\zeta_2|^{3/2}(\xx) \approx \log(\xx^{-1}), \\ & |f_2|^{-1/4}|\frac{d^2}{d\xx^2}|f_2|^{-1/4}|(\xx) \lesssim \xx^{-1}
			\end{align*} Therefore, we have, for all $\frac{M m^2}{\LL^2}  \leq \xx \leq \epsilon  \xx_{-}$:  
			\begin{equation}\label{H.TP2.1}
				TV_{\xx, \epsilon  \xx_{-}}[H] \lesssim \log(  \xx^{-1}).
			\end{equation}
			\item  
			
			$s_I +A \leq s \leq 1$ (in particular $X \leq \frac{Mm^2}{\LL^2}$). In this region, we write the error differently: introducing $f(s) = \LL^{-2} V(s)$, we will compute (still following \cite{olver}) as such $$ TV_{s \in [s_I +A,1]}[H] =  Mm^2 \int_{s_I +A}^{Mm^2} \left| |f|^{-1/4} \frac{d^2}{ds^2}(|f|^{-1/4})
			-\frac{5|f|^{1/2}}{16|\zeta_2|^3}\right| ds.$$

			In this region, we have the following estimates \begin{align*}
				&f(s) \approx e^{2\kappa_+ s},\\ & |\zeta_2|^{3/2}(s) \approx \log(\LL), \\ & |f|^{-1/4}|\frac{d^2}{ds^2}|f|^{-1/4}|(\xx) \lesssim e^{-\kappa_+ s}.
			\end{align*} Therefore, we have
			\begin{equation}\label{H.TP2.2}
				TV_{[s_I+A,Mm^2]}[H]\ls  \LL\ e^{-\kappa_+ A}.
			\end{equation}

		\end{enumerate}

		Combining \eqref{H.TP2.1},
		\eqref{H.TP2.3}, the conclusion is that 
		\begin{equation}
			TV_{ [s_I+A,\frac{10\LL^2}{Mm^2}]}(H)  \lesssim \LL\  e^{-\kappa_+ A} .
		\end{equation}
		
		Thus we obtain, from Theorem 3.1, Chapter 11 of \cite{olver} with $u=\LL$: \begin{align*} 
			& \frac{|\tilde{\ep}_{Ai_2}|(\omega,\xx)}{M(\LL^{2/3} \zeta_2)},\ \frac{|\rd_{\xx} \tilde{\ep}_{Ai_2}|(\omega,\xx)}{ \LL^{2/3} [\hat{f}_2]^{1/2}(s)\ N(\LL^{2/3} \zeta_2)} \lesssim   E^{-1}(\LL^{2/3} \zeta_2) \cdot e^{-\kappa_+ A}
			, \\ &\frac{|\tilde{\ep}_{Bi_2}|(\omega,\xx)}{M(\LL^{2/3} \zeta_2)},\ \frac{|\rd_{\xx} \tilde{\ep}_{Bi_2}|(\omega,\xx)}{ \LL^{2/3} [\hat{f}_2]^{1/2}(s)\ N(\LL^{2/3} \zeta_2)} \lesssim   E(\LL^{2/3} \zeta_2) \cdot e^{-\kappa_+ A}
			. \end{align*}
		
		Then, \eqref{TP2.1}, \eqref{TP2.2} follow from translating the above estimates in the $s$-coordinates.
	\end{proof}

	In view of \eqref{W.TP2.prelim}, we will define $\uAii$ and $\uBii$ (and $\epsilon_{Ai_2}$, $\epsilon_{Bi_2}$, $\mathcal{E}_{Ai_2}$, and $\mathcal{E}_{Bi_2}$) by multiplying $\tilde{\uAii}$ and $\tilde{\uBii}$ with a suitable factor (in the same fashion as in \eqref{WKB.wronskian}) so that \begin{equation}\label{W.TP2}
		W(\uAii,\uBii) = \LL^{-4/3} \cdot \frac{Mm^2}{\pi}.
	\end{equation} The estimates we established for the tilded quantities then all continue to hold for the quantities without tildes.

	We will now take care of the $\rd_{\omega}$ derivatives of the errors in Proposition~\ref{TP2.prop} in a separate corollary. In view of the fact that we will evaluate these $s$-independent errors at $s=s_I+A$, $s=\frac{\ep \LL^2}{Mm^2}$ or $s=\frac{10\LL^2}{Mm^2}$, we restrict the estimates to an interval  $ s\in[s_I+A, \frac{ \ep \LL^2}{Mm^2}] \cup [\frac{2\LL^2}{Mm^2},\frac{10\LL^2}{Mm^2}]$ excluding the second turning point for simplicity. \begin{cor}\label{TP2.cor}
		We will consider the error estimates in a large region excluding the second turning point $s=s_{II}$: for $ s\in[s_I+A, \frac{ \ep \LL^2}{Mm^2}] \cup [\frac{2\LL^2}{Mm^2},\frac{10\LL^2}{Mm^2}]$:
		\begin{align}\label{dTP2.low}
			\LL^{1/6} |\zeta_2|^{1/4}\frac{|\rd_{\omega}\ep_{Ai_2}|(\omega,s)}{E^{-1}((2\LL)^{2/3}\zeta_2)},\ \frac{\LL^{1/6} |\zeta_2|^{1/4}}{|V|^{1/2}(s)}\frac{|\rd_{\omega s}^2\ep_{Ai_2}|(\omega,s)}{E^{-1}((2\LL)^{2/3}\zeta_2)} \lesssim \LL^3 \log^{1/3}(\LL),\\ 	\LL^{1/6} |\zeta_2|^{1/4}\frac{|\rd_{\omega}\ep_{Bi_2}|(\omega,s)}{E((2\LL)^{2/3}\zeta_2)},\ \frac{\LL^{1/6} |\zeta_2|^{1/4}}{|V|^{1/2}(s)}\frac{|\rd_{\omega s}^2\ep_{Ai_2}|(\omega,s)}{E((2\LL)^{2/3}\zeta_2)} \lesssim \LL^3 \log^{1/3}(\LL),\nonumber
		\end{align} or, equivalently, in terms of $\mathcal{E}_{Ai_2}$ and  $\mathcal{E}_{Bi_2}$ defined in Proposition~\ref{TP2.prop}: \begin{align}\label{dTP2.low2}
			|\rd_{\omega}\mathcal{E}_{Ai_2}|(\omega,s),\ \frac{|\rd_{\omega s}^2\mathcal{E}_{Ai_2}|(\omega,s)}{|V|^{1/2}(s)}\lesssim  \LL^3 \log^{1/3}(\LL),\\ |\rd_{\omega}\mathcal{E}_{Bi_2}|(\omega,s),\ \frac{|\rd_{\omega s}^2\mathcal{E}_{Bi_2}|(\omega,s)}{|V|^{1/2}(s)}\lesssim \LL^3 \log^{1/3}(\LL).\nonumber
		\end{align}
		
	\end{cor}\begin{proof}
		$\ep_{Ai_2}$ and $\ep_{Bi_2}$ satisfy respectively the following Volterra equations \begin{align*}
			&\ep_{Ai_2}(\zeta_2,\omega) = \int_{\zeta_2}^{\zeta_{2,A}(\omega)}  K(\zeta_2,v,\omega) \psi_2(v,\omega) \left(Ai((2\LL)^{2/3} v)+\ep_{Ai_2}(v,\omega)\right)dv, \\  & \ep_{Bi_2}(\zeta_2,\omega) = \int^{\zeta_2}_{\zeta_{2,B}(\omega)} K(\zeta_2,v,\omega) \psi_2(v,\omega) \left(Bi((2\LL)^{2/3} v)+\ep_{Bi_2}(v,\omega)\right) dv,
		\end{align*} where we have denoted $\zeta_{2,A}(\omega)= \zeta_2(s=s_I+A) \approx \log^{2/3}(\LL)$, $\zeta_{2,B}(\omega)= \zeta_2(s=\frac{10\LL^2}{Mm^2}) \approx -1$ (this estimate indeed follows from the proof of Proposition~\ref{TP2.prop}) corresponding to $\ep_{Ai_2}(\zeta_{2,A})=\ep_{Bi_2}(\zeta_{2,B})=0$, by assumption. To avoid any confusion, we note that $\zeta_{2,B}< \zeta_{2,A}$, consistently with the fact that $\zeta_2(X)>0$ for $X<X_{II}$ and $\zeta_2(X)<0$ for $X>X_{II}$. Moreover: \begin{align*}
			& K(\zeta_2,v,\omega)= (2\LL)^{-2/3} \left(Bi((2\LL)^{2/3} \zeta_2)Ai((2\LL)^{2/3} v)-Ai((2\LL)^{2/3} \zeta_2)Bi((2\LL)^{2/3} v) \right),\\ &  \psi_2(\zeta_2)= -(\hat{f}_2)^{-3/4} \frac{d^2}{d\xx^2} (\hat{f}_2)^{-1/4}.
		\end{align*}

		For clarity (noting that $\zeta_2$ contains an implicit $\omega$-dependence), we introduce the coordinate system $(\zeta_2,\omegac)$ with $\omegac=\omega$, denoting its partial derivatives ($\rd_{\zeta_2}$,$\rd_{\omegac}$) (to be distinguished from the coordinate system $(s,\omega)$ or $(X,\omega)$ and their partial derivatives ($\rd_{s}$,$\rd_{\omega}$) or respectively ($\rd_{\xx}$,$\rd_{\omega}$) ). Note in particular that \begin{align*}
			&|K|(\zeta_2,v,\omega)\ls (2\LL)^{-2/3} E^{-1}((2\LL)^{2/3}\zeta_2) M((2\LL)^{2/3}\zeta_2) E((2\LL)^{2/3}v) M((2\LL)^{2/3}v) \text{ for } v \geq \zeta_2,\\ & |K|(\zeta_2,v,\omega)\ls (2\LL)^{-2/3} E^{-1}((2\LL)^{2/3}v) M((2\LL)^{2/3}\zeta_2) E((2\LL)^{2/3}\zeta_2) M((2\LL)^{2/3}v) \text{ for } v \leq \zeta_2, \\ &  \rd_{\omegac} K \equiv 0,\\ &|\rd_{\omegac}(\zeta_{2,A}(\omega))|,\ |\rd_{\omegac}(\zeta_{2,B}(\omega))| \ls \LL^2.
		\end{align*} 
		We also extend the proof of Proposition~\ref{TP2.prop} and obtain the following estimate: \begin{align*}
			& |\psi_2|(\xx) \lesssim \log^{2/3}(\LL) e^{-2\kappa_+ s_-(\xx)}.
		\end{align*}
		Now to estimate $\rd_{\omegac}\psi_2$, note that  in $(\omega,X)$ coordinates, we have \begin{equation}
			|	\rd_{\omega} f_2(\omega,X)| \ls \LL^2, 
		\end{equation} and moreover for any function $F(\omega,X)$ \begin{equation}
			|	\rd_{\omegac} F|\ls  	 |	\rd_{\omega} F|+ \LL^2  |f_2|^{-1/2} \left(\int^{X_{II}}_{X} |f_2|^{-1/2}(Z) dZ\right)  |	\rd_{X} F| \ls 	 |	\rd_{\omega} F|+ \LL^2  |f_2|^{-1/2} 	|	\rd_{X} F|.
		\end{equation}
		
		Eventually, and in view again of the estimates in the proof of Proposition~\ref{TP2.prop},  we obtain  
		\begin{align*}
			&  |\rd_{\omegac}\psi_2|(\omega,\xx)  \ls \LL^2  \text{ for }  \ep X_- \leq X \leq 10\\ &   |\rd_{\omegac}\psi_2|(\omega,\xx) \ls |\zeta_2|(X) \left( \LL^2 X^2 + \LL^2  \right) \ls  \log^{2/3}(X^{-1})  \LL^2  \text{ for } \frac{Mm^2}{\LL^2} \leq X \leq \ep X_-,\\ &   | \rd_{\omegac}\psi_2|(\omega,\xx) \ls |\zeta_2|(X)  \left( \LL^{-2} e^{-4\K s(X)}+e^{-3\K s(X)}\right) \ls_A \log^{2/3}(\LL) \LL   e^{-2\kappa_+ s(\xx)} \text{ for } \frac{Mm^2(s_I+A)}{\LL^2} \leq X \leq\frac{ Mm^2}{\LL^2},
		\end{align*}
		
		and, therefore:
		
		$$		|\rd_{\omegac}\psi_2|(\omega,\xx)\ls  
		\LL^2  \log^{2/3}(\LL)e^{-2\kappa_+ s_-(\xx)}.$$

		Now, note that \begin{align*}
			& \rd_{\omegac}\left(\int_{\zeta_2}^{\zeta_{2,A}(\omega)} K(\zeta_2,v,\omega) \psi_2(v,\omega) Ai((2\LL)^{2/3} v)dv\right) \\ &=[\rd_{\omegac} \zeta_{2,A}(\omega)] K(\zeta_2,\zeta_{2,A}(\omega),\omega) \psi_2(\zeta_{2,A}(\omega),\omega) Ai((2\LL)^{2/3} \zeta_{2,A}(\omega))+  \int_{\zeta_2}^{\zeta_{2,A}(\omega)} K(\zeta_2,v,\omega) \rd_{\omegac}\psi_2(v,\omega) Ai((2\LL)^{2/3} v)dv, \\ & \rd_{\omegac}\left(\int^{\zeta_2}_{\zt} K(\zeta_2,v,\omega)  \psi_2(v,\omega) Bi((2\LL)^{2/3} v)dv\right)\\ & =[\rd_{\omegac} \zt] K(\zeta_2,\zt,\omega) \psi_2(\zt,\omega) Bi((2\LL)^{2/3} \zt)+  \int^{\zeta_2}_{\zt} K(\zeta_2,v,\omega) \rd_{\omegac}\psi_2(v,\omega) Bi((2\LL)^{2/3} v)dv, 
		\end{align*} thus, using the above estimates, we deduce \begin{align*}
			&\frac{\bigl|\rd_{\omegac}\left(\int_{\zeta_2}^{\zeta_{2,A}(\omega)}   K(\zeta_2,v,\omega)\psi_2(v,\omega) Ai((2\LL)^{2/3} v)dv\right)\bigr|}{ E^{-1}((2\LL)^{2/3}\zeta_2) M((2\LL)^{2/3}\zeta_2)},\ \frac{\bigl|\rd_{\omegac \zeta_2}^2\left(\int_{\zeta_2}^{\zeta_{2,A}(\omega)}   K(\zeta_2,v,\omega)\psi_2(v,\omega) Ai((2\LL)^{2/3} v)dv\right)\bigr|}{ E^{-1}((2\LL)^{2/3}\zeta_2) M((2\LL)^{2/3}\zeta_2)}\\ & \ls \LL^3 \log^{1/3}(\LL),\\ & \frac {\bigl|\rd_{\omegac}\left(\int^{\zeta_2}_{\zt}  K(\zeta_2,v,\omega) \psi_2(v,\omega) Bi((2\LL)^{2/3} v)dv\right)\bigr|}{ E((2\LL)^{2/3}\zeta_2) M((2\LL)^{2/3}\zeta_2)},\ \frac {\bigl|\rd_{\omegac \zeta_2}^2\left(\int^{\zeta_2}_{\zt}  K(\zeta_2,v,\omega) \psi_2(v,\omega) Bi((2\LL)^{2/3} v)dv\right)\bigr|}{ E((2\LL)^{2/3}\zeta_2) M((2\LL)^{2/3}\zeta_2)} \\ &\ls \LL^3 \log^{1/3}(\LL).
		\end{align*} In particular, we obtain  \eqref{dTP2.low} and its equivalents \eqref{dTP2.low2}. 
		To obtain \eqref{dTP2.low2}, we have used the fact that $$\big| \int_{s_I+A}^{s_{II}} \rd_{\omega}\left[|V|^{1/2}(s)\right]  ds \bigr| \ls \LL^3. $$
		
	\end{proof}

	\begin{cor}\label{cor.TP2}
		In the   region $\{  \frac{2\LL^2}{Mm^2} \leq s \leq  \frac{10\LL^2}{Mm^2} \}$, one has
		
		\begin{align}
			&u_{Ai_2}= \frac{\pi^{-1/2} }{2}\LL^{-1/6}  {f}_2^{-1/4} \left( e^{\frac{2i\LL}{3} |\zeta_2|^{3/2}(s) -\frac{i \pi}{4} }+ e^{-\frac{2i\LL}{3} |\zeta_2|^{3/2}(s) +\frac{i \pi}{4} } + \eta_{Ai_2}(s,\omega)\right),\\ 	&u_{Bi_2}= i\  \frac{\pi^{-1/2} }{2}\LL^{-1/6}  {f}_2^{-1/4} \left( e^{\frac{2i\LL}{3} |\zeta_2|^{3/2}(s) -\frac{i \pi}{4} }- e^{-\frac{2i\LL}{3} |\zeta_2|^{3/2}(s) +\frac{i \pi}{4} } + \eta_{Bi_2}(s,\omega)\right),
		\end{align} with the errors $\eta_{Ai_2}(s,\omega)$, $\eta_{Bi_2}(s,\omega)$ satisfying \begin{align}
			& |\eta_{Ai_2}|(s,\omega),\ |\eta_{Bi_2}|(s,\omega),\     \LL|\rd_{s} \eta_{Ai_2}|(\omega,s),\    \LL|\rd_{s} \eta_{Bi_2}|(\omega,s) \lesssim  e^{-\kappa_+ A},\\  	&  |\rd_{\omega}\eta_{Ai_2}|(s,\omega),\ |\rd_{\omega}\eta_{Bi_2}|(s,\omega),\  \LL |\rd_{\omega s}^2 \eta_{Ai_2}|(\omega,s),\  \LL |\rd_{\omega s}^2 \eta_{Bi_2}|(\omega,s) \ls \LL^3 \log^{1/3}(\LL).
		\end{align}
		
		We can also rewrite this in a differently convenient form 
		\begin{align}\label{Airy2+-}
			u_{Ai_2}(\omega,s) \mp i  u_{Bi_2}(s,\omega)=(Mm^2)^{1/2} \pi^{-1/2}\LL^{-2/3}  |V|^{-1/4}(s) \left(e^{\pm i \int_{s_{II}}^{s} |V|^{1/2}(s')ds' \mp \frac{i \pi}{4} }+ \eta_{2}^{\pm}(\omega,s)\right),
		\end{align} where \begin{equation*}
			\eta_{2}^{\pm}(s,\omega):=\eta_{Ai_2}(\omega,s)\pm i  \eta_{Bi_2}(\omega,s).
		\end{equation*}
	\end{cor}
	
	\begin{proof}
		The above estimates follow directly from Proposition~\ref{TP2.prop} \& Corollary~\ref{TP2.cor}, together with the asymptotics of the $Ai$ and $Bi$ functions (see for instance in \cite{olver}, Chapter 11, Section 1 or \cite{KGSchw1}, Appendix A).
	\end{proof}

	\subsection{Connection between the first and second turning point}\label{connection12.section}
	
	We now connect $\uAii$ and $\uBii$ to $w_{1,\pm}$: this connection will take place in the forbidden region: in fact, we restrict to $s\in [1,\ep \LL^2]$ which is the common domain of validity of  $\uAii$, $\uBii$ and $w_{1,\pm}$. We note that the situation (and estimates we obtain) is very analogous to the case of a wave penetrating through a barrier (see \cite{olver}, Chapter 13).

	\begin{prop} \label{connection12.prop} For all  $s\in [s_I+A,\ep \LL^2]$, we have, recalling that $A>0$ is a large constant: \begin{align*}
			&	\uAii = \alpha_{A_2}(\omega,\LL) w_{1,-} +  \beta_{A_2}(\omega,\LL) w_{1,+},\  \uBii = \alpha_{B_2}(\omega,\LL) w_{1,-} +  \beta_{B_2}(\omega,\LL) w_{1,+},\\ & \alpha_{A_2}(\omega,\LL)= \ \LL^{-2/3} \cdot O_A\left(\exp(-  \int^{s_{II}}_{s_I} |V|^{1/2}(s) ds ) \right),\\ & \alpha_{B_2}(\omega,\LL)= \pi^{-1/2}  (Mm^2)^{1/2}\LL^{-2/3} \exp\left(
			\int_{s_I}^{s_{II}} |V|^{1/2}(s') ds'\right) \left[1+ O( e^{-\kappa_+ A})\right],\\ &  \beta_{A_2}(\omega,\LL)=\frac{\pi^{-1/2}  (Mm^2)^{1/2}}{2}\LL^{-2/3} \exp\left(-
			\int_{s_I}^{s_{II}} |V|^{1/2}(s') ds'\right) \left[1+ O( e^{-\kappa_+ A})\right],\\ & \beta_{B_2}(\omega,\LL)=O\left(\exp(-  \int^{s_{II}}_{s_I} |V|^{1/2}(s) ds+O(\LL)) \right).
		\end{align*}

	\end{prop}\begin{proof}

		We define the number  $u= \frac{\LL}{M m^2}$, and note that $f_2= u^2 V$ and $\frac{2}{3} \LL |\zeta_2|^{3/2}(s)=  \int_{s}^{s_{II}} |V|^{1/2}(s') ds'$. We write  \begin{align*}
			&	\uAii(s) = |f_2|^{-1/4}(s) \left( |\zeta_2|^{1/4} Ai(\LL^{2/3}\zeta_2) + |\zeta_2|^{1/4} \epsilon_{Ai_2}(s) \right),\\
			&	\uBii(s) = |f_2|^{-1/4}(s) \left( |\zeta_2|^{1/4} Bi(\LL^{2/3}\zeta_2) + |\zeta_2|^{1/4} \epsilon_{Bi_2}(s) \right).
		\end{align*}  
		
		We will work with the following definitions \begin{align*}
			&m_{Bi}(x) = |x|^{1/4} Bi(x)\cdot \exp(-\frac{2}{3} x^{3/2}) ,\\ & m_{Ai}(x) = |x|^{1/4} Ai(x)\cdot \exp(\frac{2}{3} x^{3/2}),
		\end{align*} noting that as $x\rightarrow + \infty$ (see \cite{olver}, chapter 11): \begin{align*}
			&  m_{Bi}(x) = \pi^{-1/2} \left( 1+ O(|x|^{-3/2})\right),\ m_{Bi}'(x) =- \frac{5 \pi^{-1/2}}{72} x^{-5/2} + O(x^{-4}),\ m_{Bi}''(x) =O(x^{-7/2}), \\&   m_{Ai}(x) =  \frac{\pi^{-1/2}}{2} \left( 1+ O(|x|^{-3/2})\right),\  m_{Ai}'(x) = \frac{5 \pi^{-1/2}}{144} x^{-5/2}+O(x^{-4}),\ \ m_{Ai}''(x) =O(x^{-7/2}),\\ & W(m_{Ai},m_{Bi}) = \frac{5\pi^{-1}}{72} x^{-5/2} +O(x^{-4}).
		\end{align*}
		
		We will also recall \begin{align*}
			& \mathcal{E}_{Ai_2}(\omega,s)= \LL^{1/6} |\zeta_2|^{1/4} \ep_{Ai_2}(s) \cdot \exp(\frac{2}{3}\LL |\zeta_2|^{3/2})
			,\\ &  \mathcal{E}_{Bi_2}(\omega,s)=  \LL^{1/6}|\zeta_2|^{1/4} \ep_{Bi_2}(s) \cdot \exp(-\frac{2}{3}\LL |\zeta_2|^{3/2})
			,
		\end{align*} so that \begin{align*}
			&u_{Ai_2}(s)= |f_2|^{-1/4}(s) \exp(-\frac{2}{3} \LL |\zeta_2|^{3/2})  \left(  m_{Ai}( \LL^{2/3}\zeta_2)  +  \mathcal{E}_{Ai_2}(s)\right) \LL^{-1/6},\\ &  u_{Bi_2}(s)= |f_2|^{-1/4}(s)   \exp(\frac{2}{3} \LL |\zeta_2|^{3/2})\left(  m_{Bi}(\LL^{2/3}\zeta_2) +  \mathcal{E}_{Bi_2}(s)\right)  \LL^{-1/6}.
		\end{align*}

		Note the following formulae:
		\begin{equation}	\begin{split}&	W\left( |f_2|^{-1/4} \exp(\pm\frac{Mm^2}{\LL}\int_{s_I}^{s} |f_2|^{1/2}(s') ds') , |f_2|^{-1/4} \exp(\pm\frac{Mm^2}{\LL}\int^{s_{II}}_{s} |f|^{1/2}(s') ds')\right)\\ &= \pm \frac{2Mm^2}{\LL}\ \exp\left(	\pm \frac{Mm^2}{\LL} \int_{s_I}^{s_{II}} |f_2|^{1/2}(s') ds' \right) =\pm \frac{2Mm^2}{\LL} \exp\left(	\pm  \int_{s_I}^{s_{II}} |V|^{1/2}(s') ds' \right)\end{split}		\end{equation}

		\begin{equation}		W\left( |f_2|^{-1/4} \exp(\pm\frac{Mm^2}{\LL}\int_{s_I}^{s} |f_2|^{1/2}(s') ds') , |f_2|^{-1/4} \exp(\mp\frac{Mm^2}{\LL}\int^{s_{II}}_{s} |f_2|^{1/2}(s') ds')\right)=0.		\end{equation}

		Conclusion: using Proposition~\ref{WKB.TP1.prop}, we get
		
		\begin{equation}\begin{split}
				&	W(w_{1,-},\uAii) \\ =\   & \frac{\LL^{1/2} \LL^{-1/6}}{(Mm^2)^{1/2}}\ |f_2|^{-1/2}(s) \exp\left(
				- \int_{s_I}^{s_{II}} |V|^{1/2}(s') ds'\right) \cdot [1+\ep^{WKB}_{1,-}(s)]\cdot \left( I_{AA}+ II_{AA}+III_{AA}+IV_{AA}  \right), \\ & I_{AA}=  -\frac{2 Mm^2}{\LL} |f_2|^{1/2}(s)\ \left(  m_{Ai}( \LL^{2/3}\zeta_2)  +  \mathcal{E}_{Ai_2}(s)\right) ,\\ & II_{AA}=-  \frac{(Mm^2)\LL^{-4/3} |f_2|^{1/2}(s)}{|\zeta_2|^{1/2}(s)}m_{Ai}'( \LL^{2/3}\zeta_2)=  -  \frac{ (Mm^2) \LL^{-3} |f_2|^{1/2}(s)}{|\zeta_2|^{3}(s)}\left[  \LL^{5/3}|\zeta_2|^{5/2}\cdot m_{Ai}'( \LL^{2/3}\zeta_2) \right]  ,\\ & III_{AA}= - \rd_s \mathcal{E}_{Ai_2}(s), \\ & IV_{AA}=  \frac{\rd_s \ep^{WKB}_{1,-}(s)}{1+  \ep^{WKB}_{1,-}(s)}.
			\end{split}
		\end{equation}  
		\begin{equation}\begin{split}
				&	W(w_{1,+},\uAii)\\=\ &  \frac{\LL^{1/2} \LL^{-1/6}}{(Mm^2)^{1/2}}\ |f_2|^{-1/2}(s) \exp\left(	 \int^{s}_{s_{I}}  |V|^{1/2}(s') ds'
				- \int_{s}^{s_{II}} |V|^{1/2}(s') ds'\right) \cdot [1+\ep^{WKB}_{1,+}(s)]\cdot \left( II_{BA}+III_{BA}+IV_{BA}  \right), \\ &  II_{BA}=-  \frac{(Mm^2)\LL^{-4/3} |f_2|^{1/2}(s)}{|\zeta_2|^{1/2}(s)}m_{Ai}'( \LL^{2/3}\zeta_2)=  -  \frac{(Mm^2) \LL^{-3} |f_2|^{1/2}(s)}{|\zeta_2|^{3}(s)}\left[  \LL^{5/3}|\zeta_2|^{5/2}\cdot m_{Ai}'( \LL^{2/3}\zeta_2) \right] \\ & III_{BA}= - \rd_s \mathcal{E}_{Ai_2}(s),\\ & IV_{BA}=  \frac{\rd_s \ep^{WKB}_{1,+}(s)}{1+  \ep^{WKB}_{1,+}(s)}.
			\end{split}
		\end{equation}
		
		\begin{equation}\begin{split}
				&	W(w_{1,+},\uBii)=  \frac{\LL^{1/2} \LL^{-1/6}}{(Mm^2)^{1/2}}\ |f_2|^{-1/2}(s) \exp\left(
				\int_{s_I}^{s_{II}} |V|^{1/2}(s') ds'\right) \cdot [1+\ep^{WKB}_{1,+}(s)]\cdot \left( I_{BB}+ II_{BB}+III_{BB}+IV_{BB}  \right), \\ & I_{BB}=  \frac{2 Mm^2}{\LL} |f_2|^{1/2}(s)\ \left(  m_{Bi}( \LL^{2/3}\zeta_2)  +  \mathcal{E}_{Bi_2}(s)\right) ,\\ & II_{BB}=-  \frac{(Mm^2)\LL^{-4/3} |f_2|^{1/2}(s)}{|\zeta_2|^{1/2}(s)}m_{Bi}'( \LL^{2/3}\zeta_2)=  -  \frac{ \LL^{-3} |f_2|^{1/2}(s)}{|\zeta_2|^{3}(s)}\left[  \LL^{5/3}|\zeta_2|^{5/2}\cdot m_{Bi}'( \LL^{2/3}\zeta_2) \right]  ,\\ & III_{BB}= - \rd_s \mathcal{E}_{Bi_2}(s),\\ & IV_{BB}=  \frac{\rd_s \ep^{WKB}_{1,+}(s)}{1+  \ep^{WKB}_{1,+}(s)}.
			\end{split}
		\end{equation}  
		
		\begin{equation}\begin{split}
				&	W(w_{1,-},\uBii)\\ =\ & \frac{\LL^{1/2} \LL^{-1/6}}{(Mm^2)^{1/2}}\ |f_2|^{-1/2}(s) \exp\left(	- \int^{s}_{s_{I}}  |V|^{1/2}(s') ds'
				+\int_{s}^{s_{II}} |V|^{1/2}(s') ds'\right)  [1+\ep^{WKB}_{1,-}(s)] \left( II_{AB}+III_{AB}+IV_{AB}  \right), \\ &  II_{AB}=-  \frac{(Mm^2)\LL^{-4/3}|f_2|^{1/2}(s)}{|\zeta_2|^{1/2}(s)}m_{Bi}'( \LL^{2/3}\zeta_2)=  -  \frac{ \LL^{-3} |f_2|^{1/2}(s)}{|\zeta_2|^{3}(s)}\left[  \LL^{5/3}|\zeta_2|^{5/2}\cdot m_{Bi}'( \LL^{2/3}\zeta_2) \right] \\ & III_{AB}= - \rd_s \mathcal{E}_{Bi_2}(s),\\ & IV_{AB}=  \frac{\rd_s \ep^{WKB}_{1,-}(s)}{1+  \ep^{WKB}_{1,-}(s)}.
			\end{split}
		\end{equation}

		We get immediately,  applying  Proposition~\ref{TP2.prop} and Proposition~\ref{WKB.TP1.prop} to $1\leq s \leq\ep \LL^2$: \begin{align}
			& \label{W1}W(w_{1,-},\uAii) \nonumber \\ &= -\pi^{-1/2}  (Mm^2)^{1/2}\LL^{-2/3} \left[ 1+  O(e^{-\kappa_+ A})]\right]  \exp\left(-
			\int_{s_I}^{s_{II}} |V|^{1/2}(s') ds'\right),  \\  	& \label{W2} W(w_{1,-},\uBii)= \left( O(\LL^{-2/3}e^{-\kappa_+ A})\right) \exp\left(-	 \int^{s}_{s_{I}}  |V|^{1/2}(s') ds'
			+\int_{s}^{s_{II}} |V|^{1/2}(s') ds'\right), \\ &  \label{W3} W(w_{1,+},\uBii) = 2\pi^{-1/2}  (Mm^2)^{1/2}\LL^{-2/3} \left[1+ O(e^{-\kappa_+ A})\right]  \exp\left(
			\int_{s_I}^{s_{II}} |V|^{1/2}(s') ds'\right),
		\end{align} and   applying  Proposition~\ref{TP2.prop} and Proposition~\ref{WKB.TP1.prop} to $s_I+A\leq s \leq1$ \begin{equation}
			\label{W4} W(w_{1,+},\uAii)= O(\LL^{-2/3} e^{-\K A}) \cdot \exp\left(	\int^{s}_{s_{I}} |V|^{1/2}(s') ds'
			-\int_{s}^{s_{II}} |V|^{1/2}(s') ds'\right), 
		\end{equation}
		
		We evaluate \eqref{W1}, \eqref{W3} at $s=\ep \LL^2$ and find \begin{equation}
			W(w_{1,-},\uAii)= -\pi^{-1/2}  (Mm^2)^{1/2}\LL^{-2/3} \exp\left(-
			\int_{s_I}^{s_{II}} |V|^{1/2}(s') ds'\right) \left[1+ O( e^{-\kappa_+ A})\right],
		\end{equation}
		\begin{equation}
			W(w_{1,+},\uBii)= 2\pi^{-1/2}  (Mm^2)^{1/2}\LL^{-2/3} \exp\left(
			\int_{s_I}^{s_{II}} |V|^{1/2}(s') ds'\right) \left[1+ O( e^{-\kappa_+ A})\right].
		\end{equation}
		
		Now, we want to evaluate the RHS of \eqref{W2} at $s=\ep \LL^2$. Note that $$ -\int_{s_I}^{\ep\LL^2} |V|^{1/2}(s) ds+ \int^{s_{II}}_{\ep\LL^2} |V|^{1/2}(s) ds= -  \int^{s_{II}}_{s_I} |V|^{1/2}(s) ds+ O(\LL),$$ thus \begin{equation*}
			W(w_{1,-},\uBii)=O\left(\exp(-  \int^{s_{II}}_{s_I} |V|^{1/2}(s) ds+O(\LL)) \right).
		\end{equation*}
		
		Finally, we want to evaluate the RHS of \eqref{W4} at $s=s_I+A$. Thus,  since   $\int^{s_{I}+A}_{s_I} |V|^{1/2}(s) ds=O_A(1)$:
		
		\begin{equation*}
			W(w_{1,+},\uAii)	=  O_A\left(\LL^{-2/3}  \cdot\exp(-   \int^{s_{II}}_{s_I} |V|^{1/2}(s) ds) \right).
		\end{equation*} The estimates claimed in the proposition follow, using \eqref{WKB.wronskian}.

	\end{proof} In the next corollary, we estimate the coefficients $\alpha_{A_2}(\omega,\LL)$, $\alpha_{B_2}(\omega,\LL)$, $\beta_{A_2}(\omega,\LL)$, $\beta_{B_2}(\omega,\LL)$ more precisely, together with their $\rd_{\omega}$ derivatives.

	\begin{cor}\label{TP2.connection.cor} Recalling $C_I(\omega,\LL):= \LL^{-1}\int_{s_I}^{1} |V|^{1/2}(s) ds$ and defining $C_{II}(\omega,\LL):= \LL^{-1}\log^{-1}(\LL)\int_{s_I}^{s_{II}} |V|^{1/2}(s) ds$, we have  \begin{equation}\label{C.eq}
			1 \lesssim 	 C_I(\omega,\LL) \lesssim 1,\  |C_{II}(\omega,\LL)-2|\ls \log^{-1}(\LL),\ |\rd_{\omega} C_{II}|(\omega,\LL) \ls \LL^2 \log^{-1}(\LL).
		\end{equation} Thus, the coefficients $\alpha_{A_2}$, $\beta_{A_2}$, $\alpha_{B_2}$ and $\beta_{B_2}$ in   Proposition~\ref{TP2.prop} can be expressed as
		\begin{align*}
			& \alpha_{A_2}(\omega,\LL)=\frac{\pi^{-1/2}  (Mm^2)^{1/2}}{2}\LL^{-2/3}  {\exp(-   C_{II}\LL\log(\LL))} \cdot \delta_{\alpha,A_2}(\omega,\LL),\\ & \alpha_{B_2}(\omega,\LL)= \pi^{-1/2}  (Mm^2)^{1/2}\LL^{-2/3} \exp\left(
			C_{II}\LL\log(\LL)\right) \left[1+\delta_{\alpha,B_2}(\omega,\LL) \right],\\ &  \beta_{A_2}(\omega,\LL)=\frac{\pi^{-1/2}  (Mm^2)^{1/2}}{2}\LL^{-2/3} \exp\left(-
			C_{II}\LL\log(\LL)\right) \left[1+ \delta_{\beta,A_2}(\omega,\LL)\right],\\ & \beta_{B_2}(\omega,\LL)= \pi^{-1/2}  (Mm^2)^{1/2}\LL^{-2/3} \cdot\exp(-  C_{II}\LL\log(\LL)) \cdot \tilde{\beta}_{B_2}(\omega,\LL),\\ & |\delta_{\alpha,A_2}(\omega,\LL)|\ls_A 1,\  |\delta_{\alpha,B_2}|(\omega,\LL),\ |\delta_{\beta,A_2}|(\omega,\LL) \ls e^{-\K A}
			, \\  & |\rd_{\omega}\delta_{\alpha,A_2}(\omega,\LL)|,\ |\rd_{\omega}\delta_{\alpha,B_2}|(\omega,\LL),\ |\rd_{\omega}\delta_{\beta,A_2}|(\omega,\LL)\ls_A  \LL^{3} \log(\LL),\\ & |\tilde{\beta}_{B_2}|(\omega,\LL),\  \LL^{-3}|\rd_{\omega}\tilde{\beta}_{B_2}|(\omega,\LL) \ls_A \exp( O(\LL)) .
		\end{align*}
	\end{cor}
	
	\begin{proof}
		Recovering the computations in the proof of Proposition~\ref{connection12.prop}, we obtain, using the estimate $|\rd_{\omega} V|\ls 1$:
		
		\begin{equation}\begin{split}
				&	|\rd_{\omega}W(w_{1,-},\uAii)| \ls \left(|V|^{-1}+\int_{s_I}^{s_{II}} |V|^{-1/2}(s') ds'+ |\rd_{\omega}\ep^{WKB}_{1,-}| \right)|W(w_{1,-},\uAii)|\\ & + \LL^{1/3} |f_2|^{-1/2}(s) \exp\left(
				- \int_{s_I}^{s_{II}} |V|^{1/2}(s') ds'\right) \cdot  \left(| \rd_{\omega}I_{AA}|+ |\rd_{\omega}II_{AA}|+|\rd_{\omega}III_{AA}|+|\rd_{\omega}IV_{AA}|  \right), \\ & |\rd_{\omega}I_{AA}| \lesssim   \LL^{-1} |f_2|^{1/2}(s)\ \left( |V|^{-1}+ \LL^{2/3}|\rd_{\omega}\zeta_2| |m_{Ai}'|( \LL^{2/3}\zeta_2)  +  |\rd_{\omega}\mathcal{E}_{Ai_2}|\right) ,\\ & |\rd_{\omega}II_{AA|} \ls   \frac{\LL^{-4/3} |f_2|^{1/2}(s)}{|\zeta_2|^{1/2}(s)}\left( [|V|^{-1}+ \frac{|\rd_{\omega}\zeta_2|}{|\zeta_2|}]|m_{Ai}'|( \LL^{2/3}\zeta_2)+ \LL^{2/3} |\rd_{\omega}\zeta_2||m_{Ai}''|( \LL^{2/3}\zeta_2) \right)  ,\\ & |\rd_{\omega}III_{AA}| \ls |\rd^2_{\omega s} \mathcal{E}_{Ai_2}|,	\\ & |\rd_{\omega}IV_{AA}| \ls  |\rd_{\omega s}^2 \ep^{WKB}_{1,-}| +|\rd_s \ep^{WKB}_{1,-}|\ |\rd_{\omega} \ep^{WKB}_{1,-}|.
			\end{split}
		\end{equation}

		Consistently with	the proof of Proposition~\ref{connection12.prop}, we will estimate $\rd_{\omega}W(w_{1,-},\uAii)$ at $s=\ep \LL^2$. Note that \begin{equation}\label{eq1.lemma}\begin{split}
				& |V|(s=\ep \LL^2) \approx \LL^{-2},\ |\zeta_2|(s=\ep \LL^2) \approx 1,\  |\rd_{\omega}\zeta_2|(s=\ep \LL^2) \approx \LL^2, \\ & |\mathcal{E}_{Ai_2}|(s=\ep \LL^2),\ |\mathcal{E}_{Bi_2}|(s=\ep \LL^2) \ls e^{-\kappa_+ A},  \\ & |\rd_{s}\mathcal{E}_{Ai_2}|(s=\ep \LL^2),\ |\rd_{s}\mathcal{E}_{Bi_2}|(s=\ep \LL^2) \ls \LL^{-1} e^{-\kappa_+ A},\\   & |\rd_{\omega}\mathcal{E}_{Ai_2}|(s=\ep \LL^2),\ |\rd_{\omega}\mathcal{E}_{Bi_2}|(s=\ep \LL^2) \ls \LL^3  \log^{1/3}(\LL), \\   & |\rd_{\omega s}^2\mathcal{E}_{Ai_2}|(s=\ep \LL^2),\ |\rd_{\omega s}^2\mathcal{E}_{Bi_2}|(s=\ep \LL^2) \ls  \LL^2 \log^{1/3}(\LL),\\ &   |\ep^{WKB}_{1,\pm}|(s=\ep\LL^2) \ls  e^{-\K A},\ |\rd_{s}\ep^{WKB}_{1,\pm}|(s=\ep\LL^2) \ls \LL^{-1}e^{-\K A},\\  &   |\rd_{\omega}\ep^{WKB}_{1,\pm}|(s=\ep\LL^2) \ls   \LL^3\log\LL,\ |\rd_{\omega s}^2\ep^{WKB}_{1,\pm}|(s=\ep\LL^2) \ls    \LL^2\log\LL.\end{split}
		\end{equation}  Thus we obtain \begin{equation}\label{W--.domega}\begin{split}
				&	|\rd_{\omega}W(w_{1,-},\uAii)| \ls  \LL^3|W(w_{1,-},\uAii)|+ \LL^{\frac{7}{3}}\log(\LL) \exp\left(
				- \int_{s_I}^{s_{II}} |V|^{1/2}(s') ds'\right)\\   &\lesssim  \LL^{\frac{7}{3}}\log(\LL)  \exp\left(
				- \int_{s_I}^{s_{II}} |V|^{1/2}(s') ds'\right).\end{split}
		\end{equation} The other computations (except for $	W(w_{1,+},\uAii)$) are done similarly: making use \eqref{eq1.lemma}, we obtain \begin{align}
			&  \label{W++.domega} |\rd_{\omega} W(w_{1,+},\uBii)|\ls \LL^{\frac{7}{3}}\log(\LL)  \exp\left(
			\int_{s_I}^{s_{II}} |V|^{1/2}(s') ds'\right),   \\ & \label{W-+.domega} |\rd_{\omega} W(w_{1,-},\uBii)|\ls  \exp\left(
			- \int_{s_I}^{s_{II}} |V|^{1/2}(s') ds'+ O(\LL)\right).
		\end{align}
		
		For $	W(w_{1,+},\uAii)$, recall from the proof of Proposition~\ref{TP2.prop} that we evaluate at $s=s_I+A$. Thus, it will be useful to remark that
		
		\begin{equation}\label{eq2.lemma}\begin{split}
				& |V|(s=s_I+A) \approx_A {1},\ |\zeta_2|({s=s_I+A}) \approx \log^{2/3}(\LL),\  |\rd_{\omega}\zeta_2|({s=s_I+A})\approx \LL^{2} \log^{-1/3}(\LL), \\ & |\mathcal{E}_{Ai_2}|({s=s_I+A}),\ |\mathcal{E}_{Bi_2}|({s=s_I+A}),\ |\rd_{s}\mathcal{E}_{Ai_2}|({s=s_I+A}),\ |\rd_{s}\mathcal{E}_{Bi_2}|({s=s_I+A})  \ls_A {1},  \\   & |\rd_{\omega}\mathcal{E}_{Ai_2}|({s=s_I+A}),\ |\rd_{\omega}\mathcal{E}_{Bi_2}|({s=s_I+A}),\  |\rd_{\omega s}^2\mathcal{E}_{Ai_2}|({s=s_I+A}),\ |\rd_{\omega s}^2\mathcal{E}_{Bi_2}|({s=s_I+A})  \ls_A \LL^3 \log^{1/3} \LL, \\ &   |\ep^{WKB}_{1,\pm}|({s=s_I+A}),\ |\rd_{s}\ep^{WKB}_{1,\pm}|({s=s_I+A}),\  |\rd_{\omega}\ep^{WKB}_{1,\pm}|({s=s_I+A}),\ |\rd_{\omega s}^2\ep^{WKB}_{1,\pm}|({s=s_I+A}) \ls_A \LL^3 \log \LL.
			\end{split}
		\end{equation} Thus, we eventually get  \begin{equation}   \label{W+-.domega}
			|\rd_{\omega} W(w_{1,+},\uAii)| \ls_A \LL^{\frac{7}{3}}\log(\LL)  \exp\left(-\int_{s_I}^{s_{II}}|V|^{1/2}(s) ds \right).
		\end{equation}
		
		The estimates claimed in the corollary thus follow.
	\end{proof}

	\subsection{Control of the solution slightly above the second turning point}\label{bessel.used.section}
	This section is  inspired by Section 5.4 in our previous work \cite{KGSchw1}.  Denote $\Vs(s) = V(\omega=m,s)$ so that $V(\omega,s)= (Mm^2)^2 k^{-2} +  \Vs(s)$.

	\begin{lemma}		There exists  $B_\pm^{*}$ solutions of \eqref{eq:mainV}	with $\omega=m$ of the following form 	
		
		\begin{equation*}		B^{*}_\pm(s)= |V|^{-1/4}(s) \exp( \pm i \int_{s_{II}}^{s} |V|^{1/2}(s') ds') \left[1+\ep_{*,\pm}^{WKB}(s)\right],		\end{equation*} with	$B^{*}_-(s)=	\overline{B^{*}_+}(s)$ and satisfying the following estimates for $s \geq \frac{2\LL^2}{Mm^2}$: \begin{equation}|\ep_{*,\pm}^{WKB}|(s),\  s^{1/2}|\frac{d\ep_{*,\pm}^{WKB}}{ds}|(s) \ls s^{-1/2}. 		\end{equation}							\end{lemma}						\begin{proof}					Note that since $\omega=m$, there is only one turning point in the $\{s\geq 0\}$ region, located at $s\approx \LL^2$.	The proof then follows straightforwardly from the WKB-approximation (away from the turning point) and is the same as in  the first part of the proof of Lemma 5.7 in \cite{KGSchw1}.						\end{proof}
	
	\begin{prop}\label{bessel.prop}			For any $\omega \in (\frac{m}{2},m)$, 	there exists  $B_\pm$ solutions of \eqref{eq:mainV} of the following form 			\begin{equation*}			B_\pm(\omega,s)= |V|^{-1/4}(s) \exp( \pm i \int_{s_{II}}^{s} |V|^{1/2}(s') ds')  + \eta_{\pm}^{B}(\omega,s),			\end{equation*}  with $B_-(s)=	\overline{B_+}(s)$ with the following estimates for  any $\frac{2\LL^2}{Mm^2} \leq s \leq k^{4/3}$:   \begin{equation}\begin{split}		&	|\eta_{\pm}^{B}|(\omega,s),\  s^{1/2}|\rd_s\eta_{\pm}^{B}|(s) \ls k^{-2} s^{7/4}+  s^{-1/4},\\ & |\rd_\omega\eta_{\pm}^{B}|(\omega,s),\  s^{1/2}|\rd^2_{\omega s}\eta_{\pm}^{B}|(\omega,s) \ls s^{7/4}.  \end{split}			\end{equation}			Moreover, we have the Wronskian identity $$ W(B_{+},B_-)= 2i. $$		\end{prop}
	
	\begin{proof}				The proof is the same as the second part of the proof of Lemma 5.7 in \cite{KGSchw1}, working in $(\omega,s)$ coordinate instead of introducing the auxiliary variable $x=\frac{s}{k^2}$: the main idea to  introduce $$ B_{\pm}(\omega,s)=B_{\pm}^{*}(s)+\ep_{\pm}(\omega,s)  $$ and write \eqref{eq:mainV} as an integral Volterra equation for $\ep_{\pm}$, treating the $k^{-2}$ term in $V$ as an inhomogeneity (see \cite{KGSchw1} for more details).			\end{proof}
	Now we will connect the Bessel-like functions $B_{\pm}(\omega,s)$ to the Airy functions $\uAii$ and $\uBii$ adapted to the second turning point. The connection will take place in the classically forbidden region, close enough (but above) the second turning point $\{ \frac{2\LL^2}{Mm^2} \leq s \leq   \frac{10\LL^2}{Mm^2} \}$.
	\begin{cor}\label{cor.bessel.airy}We have the following estimates for all $\frac{2\LL^2}{Mm^2} \leq s \leq  \frac{10\LL^2}{Mm^2}$: \begin{equation}\label{bessel.airy}\begin{split}	& B_{\pm}(\omega,s) = b_{\pm,+}(\omega) \left(\uAii(\omega,s)-i\uBii(\omega,s)\right)+ b_{\pm,-}(\omega) \left(\uAii(\omega,s)+i\uBii(\omega,s)\right),\\ & |\frac{b_{\pm \pm}(\omega)}{{\LL^{2/3}}}- (Mm^2)^{-1/2} \pi^{1/2}  e^{\mp i \frac{\pi}{4}}|\ls{e^{-\kappa_+ A}},\ \LL^{-3}|\rd_{\omega}b_{\pm \pm}|(\omega) \ls {\log^{1/3}\LL}, \\ & \frac{|b_{\pm \mp}|(\omega)}{{\LL^{2/3}}}\ls{e^{-\kappa_+ A}},\ \LL^{-3}|\rd_{\omega}b_{\pm \mp}|(\omega) \ls {\log^{1/3}\LL},	\end{split}	 \end{equation}
		where  $b_{\pm \pm}=\bar{b}_{\mp \mp}$, and $b_{\pm \mp}=\bar{b}_{\mp \pm}$.  
	\end{cor}
	
	\begin{proof}
		
		By Corollary~\ref{cor.TP2} and Proposition~\ref{bessel.prop}, we have:	\begin{equation}	\uAii(\omega,s)\mp i\uBii(\omega,s) =   (Mm^2)^{1/2}\pi^{-1/2} \LL^{-2/3} \left[ e^{\mp i \frac{\pi}{4}}[ B_{\pm}(\omega,s)- \eta_{\pm}^{B}(\omega,s)]+ |V|^{-1/4} \eta_{2,\pm}(\omega,s)\right].	\end{equation}
		
		Now note from Proposition~\ref{bessel.prop} and Corollary~\ref{cor.TP2} that for all $\frac{2\LL^2}{Mm^2} \leq s \leq  \frac{10\LL^2}{Mm^2}$:  \begin{align}	&|W(B_{\pm}, \eta_{\pm}^{B})|,\ |W(B_{\mp}, \eta_{\pm}^{B})|\ls k^{-2} s^{3/2}+ s^{-1/2} \ls 	\LL^{-1},\\ & |W(B_{\pm},  |V|^{-1/4}\eta_{2,\pm})|,\ |W(B_{\mp}, |V|^{-1/4}\eta_{2,\pm})|\ls 
			e^{-\kappa_+ A},\\  	&|\rd_{\omega}W(B_{\pm}, \eta_{\pm}^{B})|,\ |\rd_{\omega}W(B_{\mp}, \eta_{\pm}^{B})|\ls k^{-2} s^{3}+ s^{3/2}\ls \LL^3\log^{1/3}\LL ,\\ & |\rd_{\omega}W(B_{\pm},  |V|^{-1/4}\eta_{2,\pm})|,\ |\rd_{\omega}W(B_{\mp}, |V|^{-1/4}\eta_{2,\pm})|\ls \LL^3\log^{1/3}\LL.	\end{align} Note that we have used the fact that $k^{-2} \ls \LL^{-p}$ for some large $p>2$ to ignore the terms proportional to $k^{-2}$. \eqref{bessel.airy} thus follows from an elementary computation.\end{proof}

	\subsection{WKB estimates slightly below the third turning point}\label{WKB.section2}
	As in \cite{KGSchw1}, the strategy consisting in setting $\omega=m$ in \eqref{eq:mainV} and treating the rest as an error does not work all the way until the third turning point $s=s_{III} \approx k^2$ (in fact, from Proposition~\ref{bessel.prop}, we see that the estimates stop at $s\ls k^{4/3}$). That is why, as in \cite{KGSchw1}, we perform an additional WKB-approximation, up to a region  $\frac{k^{\delta}}{Mm^2} \leq s \leq   \frac{k^2}{Mm^2} $ slightly below the third turning point.
	
	\begin{prop}\label{WKB.TP2.prop}  Let $\delta>0$ (and assume $p$ is sufficiently large depending on $\delta$). For all  $\frac{k^{\delta}}{Mm^2} \leq s \leq   \frac{k^2}{2Mm^2} $,  the ODE \eqref{eq:mainV} admits the following two  oscillating solutions\begin{align}& w_{2,\pm}(\omega,s) = |V|^{-1/4}(s) \exp( \pm i \int_{s_{II}}^{s} |V|^{1/2}(s') ds' ) \left(1+ \ep_{2,\pm}^{WKB}(\omega,s)\right),\\ & |\ep_{2,\pm}^{WKB}|(\omega,s),\ s^{1/2} |\rd_{s}\ep_{2,\pm}^{WKB}|(\omega,s)\lesssim  s^{-1/2},\label{WKB2.error1}\\  &   |\rd_{\omega}\ep_{2,\pm}^{WKB}|(\omega,s),\ s^{1/2} |\rd_{\omega } \rd_s\ep_{2,\pm}^{WKB}|(\omega,s)\lesssim k^{2} \log(k). \label{WKB2.error2}\\ &\ep_{2,\pm}^{WKB}(s=\frac{k^2}{Mm^2})=\rd_{s}\ep_{2,\pm}^{WKB}(s=\frac{k^2}{Mm^2})=0.  \end{align}  
		Moreover, for $ \frac{k^2}{4Mm^2}\leq s< \frac{k^2}{2Mm^2}$, we have the following higher derivative estimates: for any $n\geq 1$ \begin{align}
			|\rd_{\omega}^n\ep_{2,\pm}^{WKB}|(\omega,s),\ s^{1/2} |\rd^n_{\omega } \rd_s\ep_{2,\pm}^{WKB}|(\omega,s)\lesssim_{n} k^{2n-1}.
			\label{WKB2.error3}
		\end{align}

		Finally, note that 	$w_{2,-}=\overline{w_{2,+}}$, and we  have the Wronskian identity $$ W(w_{2,+},w_{2,-})= 2i.$$
	\end{prop}
	\begin{proof}
		
		\eqref{WKB2.error1} and\eqref{WKB2.error2} follow immediately from an easy adaption of Lemma 5.5 in \cite{KGSchw1}. 	We extend  and improve the proof to higher derivatives (estimate \eqref{WKB2.error3} in what follows. Introducing the variable $Y=(Mm^2) \frac{s}{k^2}$, and rephrasing \eqref{eq:mainV} in terms of $Y$ (note that $1/4<Y<1/2$), we write \begin{equation}\label{Y.coordinate}\begin{split}		
				&		\frac{d^2 u}{dY^2 }  =
				k^2 f(Y) u
				, \\ &  f(Y)= 1-2Y^{-1}+ \frac{\LL^2}{k^2 Y^2}+ O( k^{-2} Y^{-2} \log(k^2 Y))+ O( k^{-4} \LL^2 Y^{-3} \log(k^2 Y)). 
			\end{split}
		\end{equation}
		
		Then, we invoke Theorem 3.1 in \cite{olver}, Chapter 10 to obtain an asymptotic expansion (in terms of $u=k$) of  	$\ep_{2,\pm}^{WKB}$ as such, for any $N\in \mathbb{N}$, and introducing the variable $\xi(Y)=\int^{1/2}_Y |f|^{1/2}(Y') dY'$: \begin{equation}\begin{split}
				&\ep_{2,\pm}^{WKB}(\omega,s) = \sum_{q=1}^{N-1} \frac{A_q(\xi(s))}{k^q}+ e^{\mp i \int_{s_{II}}^{s} |V|^{1/2}(s') ds' } \mathcal{E}_{N,\pm}(\omega,s),\\ &  |\mathcal{E}_{N,\pm}|(\omega,s),\  \frac{|\rd_s\mathcal{E}_{N,\pm}|(\omega,s)}{|V|^{1/2}(s)} \leq 2 k^{-N} \exp\left( k^{-1} TV_{[\xi(s), \xi(\frac{k^2}{2Mm^2})] }[A_1(\xi)] \right) TV_{[\xi(s), \xi(\frac{k^2}{2Mm^2})] }[A_N(\xi)], 
			\end{split}
		\end{equation} where the functions $A_q(\xi)$ are defined recursively with $A_0(\xi)=1$ and   (in what follows the integration variable $v$ corresponds to $\xi$): \begin{align}\label{A.def}
			&A_{q+1}(\omega,\xi) = -\frac{1}{2} \rd_{\xi}A_q(\omega,\xi) + \frac{1}{2} \int \psi(\omega,v) A_q(\omega,v) dv,  \\ & \psi(\omega,v)=  f^{-3/4} \frac{d^2}{dY^2} (f^{-1/4}).
		\end{align} 
		
		Considering the change of coordinates between $(\tilde{\omega},\xi)$ (with $\tilde{\omega}=\omega$) and $({\omega},Y)$, we have $|\rd_{\tilde{\omega}} h|\ls|\rd_{{\omega}} h|+(\LL^2+\log(k)) |\rd_{Y} h|$ for any function $h$: we will use this fact in what follows. In the coordinate system  $(\tilde{\omega},\xi)$,  we have

		\begin{align} &  |  \partial^m_{\xi}\rd_{\tilde{\omega}}^n f(\omega,\xi)|,\ | \rd_{\tilde{\omega}}^n \psi(\omega,\xi)|\ls_{n,q,m}  [1+k^{2(n-1)} 1_{n\geq 1}+(\LL^2+\log(k))1_{n= 1}] 
			,\\   
			& |\xi|(\omega,Y)  \approx 1
			,\   |\frac{d\xi}{dY}|(\omega,Y)\ls1,
		\end{align} in the region $ 1/4< Y < 1/2$. Using the above equations, we get that for any $n\in \mathbb{N}$, $m\in \mathbb{N}$, $q\in \mathbb{N}$
		
		\begin{equation}	|\partial^m_{\xi}\rd_{\tilde{\omega}}^n	A_q( \xi(Y))|,\ 	|\partial^m_{\xi}\rd_{\tilde{\omega}}^n	\rd_{\xi}	A_q( \xi(Y))|,\  TV_{[\xi(Y), \xi(Y=1/2)]}	[\partial^m_{\xi}\rd_{\tilde{\omega}}^nA_q]\ls_{n,q,p}   [1+k^{2(n-1)} 1_{n\geq 1}+(\LL^2+\log(k))1_{n= 1}] ,		\end{equation}
		in particular $k^{-1} TV_{[\xi(s), \xi(\frac{k^2}{2Mm^2})] }[A_1(\xi)]
		\ls  1$. Now, note (see the discussion in \cite{olver} following  Theorem 3.1 in  Chapter 10) that $\mathcal{E}_{N,\pm}(\omega,s)$ satisfies the following Volterra equation:
		
		\begin{align}
			&\mathcal{E}_{N,\pm}(\omega,\xi) = \int_0^{\xi}  K(\xi,v)\left( 2k^{-N} e^{i k v} \rd_{\xi}A_N(\omega,v)+ k^{-1} \psi(v) 	\mathcal{E}_{N,\pm}(\omega,v) \right) dv,  \\& K(\xi,v) = \frac{1}{2} \{ e^{ik(\xi-v)}- e^{-ik(\xi-v)}\}.
		\end{align}

		Therefore,  we deduce as an application of Theorem~\ref{thm:volterraparam} that for all $n\in \mathbb{N}$, $m\in \mathbb{N}$   (recalling that $k^{-1}\lesssim \LL^{-p}$ for a large $p>0$): \begin{equation}
			|\partial^m_{\xi}\rd_{\tilde{\omega}}^n\mathcal{E}_{N,\pm}(\omega,\xi)|,\ |\partial^m_{\xi}\rd_{\tilde{\omega}}^n \rd_{\xi}\mathcal{E}_{N,\pm}(\omega,\xi)| \lesssim k^{-N+3n}.
		\end{equation} 
		
		Now taking $n=N-1$  results in the following bound, as claimed (after changing the $(\omega,\xi)$ coordinates to $(\omega,s)$ coordinates):
		\begin{equation}
			|\rd_{\omega}^n\ep_{2,\pm}(\omega,s)|,\ \frac{|\rd_{\omega}^n \rd_{s}\ep_{2,\pm}(\omega,s)|}{|V|^{1/2}(s)} \lesssim k^{2n-1}.
		\end{equation}

	\end{proof}

	\begin{cor}\label{cor.bessel.connection}
		\begin{equation}\label{bessel.connection}\begin{split}
				&w_{2,\pm}(\omega,s) =  a_{\pm, +}(\omega) B_+(\omega,s)+  a_{\pm, -}(\omega) B_-(\omega,s),\\ &| a_{\pm \pm}(\omega)-1|,\ | a_{\pm \mp}|(\omega) \ls k^{-1/2}, \\& | \rd_{\omega}a_{\pm \pm}|(\omega),\ | \rd_{\omega}a_{\pm \mp}|(\omega)  \ls k^2 \log(k) . \end{split}
		\end{equation}
	\end{cor}
	\begin{proof}

		We have \begin{equation}
			B_{\pm}(\omega,s)= w_{2,\pm}(\omega,s) + \eta_{\pm}^B(\omega,s)-  |V|^{-1/4}(s) \exp(\pm i \int_{s_{II}}^{s} |V|^{1/2}(s')ds')\ep_{2,\pm}^{WKB}(\omega,s).
		\end{equation}
		
		Now note that from Proposition~\ref{bessel.prop} and Proposition~\ref{WKB.TP2.prop}   \begin{align}
			&|W(\eta_{\pm}^B, w_{2,\pm})|,\ |W(\eta_{\pm}^B, w_{2,\mp})| \ls[ k^{-2} s^{3/2} + s^{-1/2}],\\ & |\rd_{\omega}W(\eta_{\pm}^B, w_{2,\pm})|,\ |\rd_{\omega}W(\eta_{\pm}^B, w_{2,\mp})| \ls (k^2 + s^2)\log(k),\\ & |W( w_{2,\pm},|V|^{-1/4}(s) \exp(\pm i \int_{s_{II}}^{s} |V|^{1/2}(s')ds')\ep_{2,\pm}^{WKB}(\omega,s))|\ls s^{1/2}|\rd_s \ep_{2,\pm}^{WKB}|(\omega,s)\ls s^{-1/2},\\  & |W( w_{2,\pm},|V|^{-1/4}(s) \exp(\mp i \int_{s_{II}}^{s} |V|^{1/2}(s')ds')\ep_{2,\pm}^{WKB}(\omega,s))|\ls | \ep_{2,\pm}^{WKB}|(\omega,s)+ s^{1/2}|\rd_s \ep_{2,\pm}^{WKB}|(\omega,s)\ls s^{-1/2},\\  & |\rd_{\omega}W( w_{2,\pm},|V|^{-1/4}(s) \exp(\pm i \int_{s_{II}}^{s} |V|^{1/2}(s')ds')\ep_{2,\pm}^{WKB}(\omega,s))|\ls s^{1/2}|\rd_{s \omega}^2 \ep_{2,\pm}^{WKB}|(\omega,s)\ls k^2\log(k),\\  & |\rd_{\omega}W( w_{2,\pm},|V|^{-1/4}(s) \exp(\mp i \int_{s_{II}}^{s} |V|^{1/2}(s')ds')\ep_{2,\pm}^{WKB}(\omega,s))|\ls | \rd_{\omega}\ep_{2,\pm}^{WKB}|(\omega,s)+ s^{1/2}|\rd_{s \omega}^2 \ep_{2,\pm}^{WKB}|(\omega,s)\\ &\ls k^2\log(k).
		\end{align}
		
		Evaluating the above at $s=k$ gives \eqref{bessel.connection}.

	\end{proof}
	
	While $w_{2,\pm}(\omega,s)$ initially only make sense for $s \geq \frac{k^{\delta}}{Mm^2}$, we use Corollary~\ref{cor.bessel.connection} to also extend the definition of $w_{2,\pm}(\omega,s)$ up to $s\geq  \frac{2\LL^2}{Mm^2}$, using the same notation.	We continue with the connection between the Airy functions adapted to the \emph{second} turning point $\uAii$ and $\uBii$ and the WKB-obtained functions $w_{2,\pm}$ slightly under the third turning point.
	\begin{lemma}\label{connection2.3.lemma2} We have the following estimates, where $w_{2,-}=\bar{w}_{2,+}$:
		\begin{equation}\label{estprem}\begin{split}
				&	w_{2,\pm}(\omega,s)=  (Mm^2)^{-1/2} \pi^{1/2} \LL^{2/3}e^{\mp \frac{i \pi}{4}} \left[\left(1+ \eta_{\alpha}^{reg}(\omega)+ \eta_{\alpha}^{sing}(\omega)\right) \uAii\mp i \left(1+ \eta_{\beta}^{reg}(\omega)+ \eta_{\beta}^{sing}(\omega)\right) \uBii\right],\\
								&|\eta_{\alpha}^{reg}|(\omega),\ |\eta_{\beta}^{reg}|(\omega)\ls e^{-\kappa_+ A},\  |\rd_{\omega}\eta_{\alpha}^{reg}|(\omega),\ |\rd_{\omega}\eta_{\beta}^{reg}|(\omega)\ls \LL^3 \log^{1/3}\LL,\\ &   |\eta_{\alpha}^{sing}|(\omega),\ |\eta_{\beta}^{sing}|(\omega)\ls k^{-1/2},\  |\rd_{\omega}\eta_{\alpha}^{sing}|(\omega),\ |\rd_{\omega}\eta_{\beta}^{sing}|(\omega)\ls k^2\log(k).
			\end{split}
		\end{equation}
		
		We can also rewrite the above equation in the more convenient form only involving the  linear combination  $w_{2,even}$ and $w_{2,odd}$ (defined below), which are real-valued.
		\begin{align}
			&	w_{2,even}(\omega,s):=	\Re( e^{ \frac{i \pi}{4}}	w_{2,+})(\omega,s)=  (Mm^2)^{-1/2} \pi^{1/2} \LL^{2/3}\left(1+ \eta_{\alpha}^{reg}(\omega)+ \eta_{\alpha}^{sing}(\omega)\right) \uAii,\\ & 	 w_{2,odd}(\omega,s):=-\Im( e^{ \frac{i \pi}{4}}	w_{2,+})(\omega,s)=  (Mm^2)^{-1/2} \pi^{1/2} \LL^{2/3}\left(1+ \eta_{\beta}^{reg}(\omega)+ \eta_{\beta}^{sing}(\omega)\right) \uBii.\\
		\end{align}
		
		Note, finally, that \begin{equation}\label{W.even.odd}
			W( w_{2,even},w_{2,odd})= 1,
		\end{equation} from which we also deduce that \begin{equation}
			1+ \eta_{\beta}^{reg}(\omega)+ \eta_{\beta}^{sing}(\omega)=\left(1+ \eta_{\alpha}^{reg}(\omega)+ \eta_{\alpha}^{sing}(\omega)\right)^{-1}. 
		\end{equation}
	\end{lemma}
	\begin{proof}
		From the definitions of Corollary~\ref{cor.bessel.airy} and Corollary~\ref{cor.bessel.connection}, we have \begin{equation}\begin{split}
				&w_{2,\pm}(\omega,s)= \left[ a_{\pm + }(\omega) b_{+ +}(\omega)+ a_{\pm + }(\omega) b_{+ -}(\omega)+ a_{\pm - }(\omega) b_{- +}(\omega)+ a_{\pm - }(\omega) b_{- -}(\omega)\right]\uAii(\omega,s)\\ & +i\left[-a_{\pm + }(\omega) b_{+ +}(\omega)+ a_{\pm + }(\omega) b_{+ -}(\omega)- a_{\pm - }(\omega) b_{- +}(\omega)+ a_{\pm - }(\omega) b_{- -}(\omega)\right]\uBii(\omega,s).\end{split}
		\end{equation}\eqref{estprem} then follows immediately from the estimates in  Corollary~\ref{cor.bessel.airy} and Corollary~\ref{cor.bessel.connection}. Note also that  $\eta_{\alpha}^{reg}(\omega)$, $\eta_{\alpha}^{sing}(\omega)$, $\eta_{\beta}^{reg}(\omega)$, $\eta_{\beta}^{sing}(\omega)$ are all real-valued, because $\uAii$ and $\uBii$ are real-valued.
	\end{proof}
	
	\subsection{Energy identity}\label{energy.section}
	
	Recalling that $ w_{2,even(\omega,s)}$, $ w_{2,odd(\omega,s)}$ are real-valued, but $u_H$ is not, we define the (complex-valued) coefficients $\gamma_A$  and $\gamma_B$ as such \begin{equation}\label{u_H.Airy2}
		u_H = \gamma_A(\omega,\LL)   w_{2,even}(\omega,s)+\gamma_B(\omega,\LL)  w_{2,odd}(\omega,s).
	\end{equation}
	
	Note that in view of~\eqref{WKB.wronskian},~\eqref{W.TP2}, and Proposition~\ref{connection12.prop}, we have
	\begin{equation}\label{usefulidentityforscatteringconstatnst}
		2\left(\alpha_{A_2}\beta_{B_2}-\alpha_{B_2}\beta_{A_2}\right) = -\LL^{-4/3}\frac{Mm^2}{\pi}.
	\end{equation}

	Then, using~\eqref{usefulidentityforscatteringconstatnst}, Corollary~\ref{uH.cor}, Proposition~\ref{connection12.prop}, Corollary~\ref{cor.bessel.connection} and equation \eqref{W.TP2} we have
	\begin{align}& \gamma_A(\omega,\LL) = \frac{2\pi^{1/2}\LL^{2/3} C_1(\omega) \exp(i\frac{\omega}{\kappa_+} \log(\LL))}{ (Mm^2)^{1/2}(1+\eta_{\alpha}^{reg}(\omega)+ \eta_{\alpha}^{sing}(\omega))} \left([\alpha_{1,+}(\omega) + i \beta_{1,+}(\omega)] \alpha_{B_2}(\omega)-[\alpha_{1,-}(\omega) + i \beta_{1,-}(\omega)] \beta_{B_2}(\omega)\right),\\ &\gamma_B(\omega,\LL) \\= & -\frac{2 \LL^{2/3} \pi^{1/2}C_1(\omega) \exp(i\frac{\omega}{\kappa_+} \log(\LL))}{(Mm^2)^{1/2}(1+\eta_{\beta}^{reg}(\omega)+ \eta_{\beta}^{sing}(\omega))}  \left([\alpha_{1,+}(\omega,\LL) + i \beta_{1,+}(\omega,\LL)] \alpha_{A_2}(\omega,\LL)-[\alpha_{1,-}(\omega,\LL) + i \beta_{1,-}(\omega,\LL)] \beta_{A_2}(\omega,\LL)\right).\end{align}

	We will actually rewrite the above definitions using the quantities defined in Corollary~\ref{TP2.connection.cor} as: 
	
	\begin{align}& \frac{\gamma_A}{\exp(C_{II}\LL \log(\LL))} \\ &=\frac{2 [\alpha_{1,+} + i \beta_{1,+}]}{(1+\eta_{\alpha}^{reg}(\omega)+ \eta_{\alpha}^{sing}(\omega))} C_1(\omega) \exp(i\frac{\omega}{\kappa_+} \log( \LL)) \left(-\exp(-2C_{II}\LL \log(\LL))\frac{\alpha_{1,-} + i \beta_{1,-}}{\alpha_{1,+} + i \beta_{1,+}} \tilde{\beta}_{B_2}+ [1+\delta_{\alpha,B_2}]\right),\\ &\frac{\gamma_B(\omega,\LL)}{\exp(-C_{II}\LL \log(\LL))} =  -\frac{C_1(\omega) \exp(i\frac{\omega}{\kappa_+} \log( \LL))}{ (1+\eta_{\beta}^{reg}(\omega)+ \eta_{\beta}^{sing}(\omega))}\left([\alpha_{1,+} + i \beta_{1,+}]\delta_{\alpha,A_2} -[\alpha_{1,-} + i \beta_{1,-}] [1+\delta_{\beta,A_2}]\right).\end{align}

	\begin{lemma}\label{energy.lemma}

		We also have,
		\begin{equation}\label{gamma.AB.est}\begin{split}
				&\tilde{\gamma}_A(\omega,\LL):=\frac{\gamma_A(\omega,\LL)}{ \exp( C_{II} \LL \log(\LL))}= 2 C_1(\omega) e^{i\frac{\omega}{\kappa_+} \log( \LL)}  [iE_+] (1+O(e^{-\kappa_+ A})),\\ & \tilde{\gamma}_B(\omega,\LL):= \frac{\gamma_B(\omega,\LL)}{ \exp( -C_{II} \LL \log(\LL))} =O(1),\\ & \tilde{\gamma}_A(\omega,\LL)= \tilde{\gamma}_A^{reg}(\omega)+\tilde{\gamma}_A^{sing}(\omega,\LL),\  \tilde{\gamma}_B(\omega,\LL)= \tilde{\gamma}_B^{reg}(\omega,\LL)+\tilde{\gamma}_B^{sing}(\omega,\LL),  \\&| \rd_{\omega}\tilde{\gamma}_A^{reg}|(\omega,\LL),\  | \rd_{\omega}\tilde{\gamma}_B^{reg}|(\omega,\LL) \ls \LL^{3}\log(\LL),\\ &  |\tilde{\gamma}_A^{sing}|(\omega,\LL),\  | \tilde{\gamma}_B^{sing}|(\omega,\LL) \ls k^{-1/2},\ |\rd_{\omega}\tilde{\gamma}_A^{sing}|(\omega,\LL),\  | \rd_{\omega}\tilde{\gamma}_B^{sing}|(\omega,\LL) \ls k^2 \log(k),\\
			\end{split}
		\end{equation}
		We also have $\Gamma(\omega,\LL)$, $\Gamma^{\rm reg}(\omega,\LL)$, and $\Gamma^{\rm sing}(\omega,\LL)$ which satisfy
		\begin{equation}\begin{split}\label{Gammastuff}
				&  \Gamma(\omega,\LL):=\frac{\gamma_B(\omega,\LL) -i\gamma_A(\omega,\LL)}{\gamma_B(\omega,\LL) + i\gamma_A(\omega,\LL)}= \Gamma^{reg}(\omega,\LL) + \Gamma^{sing}(\omega,\LL) ,\\
				& 	\bigl|\Gamma^{reg}(\omega,\LL)+1 \big|,\ 
				\LL^{-3}\log^{-1}(\LL)	\bigl|\rd_{\omega}\Gamma^{reg} \big|(\omega,\LL) 
				\exp( -2C_{II} \LL \log(\LL)),\\ & k^{1/2}\bigl| \Gamma^{sing} \big|(\omega,\LL),\ k^{-2}\log^{-1}(k)   \bigl|\rd_{\omega} \Gamma^{sing} \big|(\omega,\LL) \ls  \exp( -2C_{II} \LL \log(\LL)).
		\end{split}\end{equation}
		
		Thus, we also have the identity/estimate (recalling that we will use it in the range $\omega \in (m-\LL^{-p},m)$, in particular $\omega>m/2>0$): \begin{equation}\label{energy.identity}
			|\Gamma|^2(\omega,\LL)= 1 -\frac{4 \omega}{|\gamma_B(\omega,\LL)+ i\gamma_A(\omega,\LL)|^2}= 1-   \frac{ \omega (1+O(e^{-\kappa_+ A}))}{|C_1|^2 |E_+|^2}\  \exp( -2C_{II} \LL \log(\LL)), 
		\end{equation} \begin{equation} \label{energy.identity2}
			|\tilde{\gamma}_B|(\omega,s) \geq \frac{\omega}{2|C_1| |E_+|} \cdot (1+O(e^{-\kappa_+ A})).
	\end{equation}\end{lemma}
	\begin{proof}
		\eqref{gamma.AB.est} is merely a rewriting of Corollary~\ref{cor.TP1} and Corollary~\ref{TP2.connection.cor}.		As for~\eqref{Gammastuff}, \eqref{energy.identity},  and 
		\eqref{energy.identity2}, we note that for every solution $u$ of \eqref{eq:mainV},  $\Im(\bar{u} \rd_s u)$ is constant. Hence, applying  $u \rightarrow \Im(\bar{u} \rd_s u)$  to the identity \eqref{u_H.Airy2} (recall that $\uAii\in \R,\ \uBii\in \R$ and the Wronskian identity \eqref{W.TP2}) gives the following energy identity\begin{equation}\label{energy.lemma.eq}
			- \omega  = \Im(\gamma_A \bar{\gamma}_B) = \frac{1}{4} \left[|i\gamma_A-\gamma_B|^2- |i\gamma_A+\gamma_B|^2 \right]
		\end{equation}
		\eqref{Gammastuff}, \eqref{energy.identity}, 	and \eqref{energy.identity2} then follow combining  \eqref{gamma.AB.est} and \eqref{energy.lemma.eq}. 
	\end{proof}

	\subsection{The Jost solution  $u_I$  and the third turning point}\label{uI.section}
	We recall $u_I(\omega,s)$ is the unique solution of \eqref{eq:mainV} satisfying \eqref{uI.asymp}.  Note that, up to a multiplicative constant, $u_I$ is the only bounded solution of \eqref{eq:mainV}.
	
	As in Section~\ref{TP2.section}, we start recalling the definition of $s_{III}$ \eqref{sIII}	 and introduce the new coordinate \begin{equation} \yy=  (M m^2) \frac{s}{  k^2}.
	\end{equation}
	
	Note also that \begin{equation*}
		Y_{III}(\LL,\omega):= Y(s=s_{III})= 1+\sqrt{1-\alpha} + O\left(k^{-2}\right) \in (1,2+O(k^{-2})\color{black}),
	\end{equation*}
	
	First, note that $u_I$ is given as such (the argument, already present in \cite{KGSchw1}, Lemma 5.2, will be briefly repeated).  \begin{lemma}\label{uI.lemma}
		We have $u_I= u_{Ai_3}$, where $u_{Ai_3}$ is defined as \begin{equation}\begin{split}
				&\label{uAiii.def}	u_{Ai_3}= \hat{f}_3^{-1/4} \left( Ai( k^{2/3} \zeta_3) + \ep_{Ai_3}\right),\\ &  f_3(\yy) = \frac{k^2}{(Mm^2)^2} V( r(\yy)),\\ 
				& \frac{2}{3}|\zeta_3|^{3/2}(\yy)=\int_{\yy}^{\yy_{III}} |f_3|^{1/2}(y) dy=k^{-1} \int^{s_{III}}_{s(\yy)} |V|^{1/2}(s')ds' \\ & \hat{f_3}(\yy) = \frac{f_3(\yy)}{ \zeta_3(\yy)}.\end{split}
		\end{equation} with $$\lim_{s \rightarrow +\infty}\ep_{Ai_3}(s)=\lim_{s \rightarrow +\infty}\rd_s \ep_{Ai_3}(s)=0.$$
	\end{lemma} \begin{proof}
		This is immediate from the fact that $TV_{[Y_{III},+\infty)}[H]<+\infty$, where $$ H(\yy):= \int_{\yy_{III}}^{\yy}  \left(|f_3|^{-1/4}(\yy') \frac{d^2}{d\yy^2}(|f_3|^{-1/4})(\yy') - \frac{5 |f_3|^{1/2}(\yy')}{16|\zeta_3|^3(\yy')} \right) d\yy' $$ in the notations of \cite{olver}, Theorem 3.1, Chapter 11. For more details, see also \cite{KGSchw1}, Lemma 5.2.
	\end{proof}
	Now, we turn to the error estimates in the Airy approximation. For simplicity, we restrict our estimates to a region   $   \frac{k^2}{4Mm^2} \leq  s \leq \frac{k^2}{2Mm^2} $ strictly under the third turning point  $s=s_{III}$.
	\begin{prop}\label{usefulformofuI}
		In the   region $\{  \frac{k^2}{4Mm^2} \leq  s \leq \frac{k^2}{2Mm^2}\}$, $u_I=u_{Ai_3} $ defined in \eqref{uAiii.def}  adopts the following oscillating behavior:
		
		\begin{equation}\begin{split}
				&u_{I}= \frac{\pi^{-1/2} }{2}k^{-1/6}  {f}_3^{-1/4} \left( e^{i\int_{s}^{s_{III}}|V|^{1/2}(s') ds' -\frac{i \pi}{4} } [1+\eta_{Ai_3,+}(s,\omega)]+ e^{-i\int_{s}^{s_{III}}|V|^{1/2}(s') ds' +\frac{i \pi}{4} } [1+\eta_{Ai_3,-}(s,\omega)] \right)\\=&(Mm^2)^{1/2} \frac{\pi^{-1/2} }{2}k^{-2/3}  |V|^{-1/4}(s)\\ & \left( e^{i\int_{s}^{s_{III}}|V|^{1/2}(s') ds' -\frac{i \pi}{4} } [1+\eta_{Ai_3,+}(s,\omega)]+ e^{-i\int_{s}^{s_{III}}|V|^{1/2}(s') ds' +\frac{i \pi}{4} }[1+\eta_{Ai_3,-}(s,\omega)]\right),\end{split}
		\end{equation} with the error $ \eta_{Ai_3,\pm}(s,\omega)$ satisfying, for all $ \frac{k^2}{4Mm^2} \leq  s \leq \frac{k^2}{2\color{black}Mm^2} $: \begin{align}
			&	|\eta_{Ai_3,\pm}|(s,\omega),\ k  |\rd_{s} \eta_{Ai_3,\pm}|(\omega,s) 
			\lesssim    k^{-1},\\ &\label{many.deriv}	|\rd_{\omega}^n\eta_{Ai_3,\pm}|(s,\omega),\ k  |\rd^n_{\omega} \rd_s \eta_{Ai_3,\pm}|(\omega,s) 
			\lesssim_{n}    k^{2n-1}\log(k), 
		\end{align}  for all $n\geq 1$.

	\end{prop}
	
	\begin{proof}  The proof for $n=1$ is essentially identical to that of Lemma 5.2 and Lemma 5.3 in \cite{KGSchw1} (note indeed that the $\frac{\LL^2}{r^2}$ term in \eqref{eq:mainV} is sub-leading in the range $ \frac{k^2}{4Mm^2} \leq  s \leq \frac{k^2}{2Mm^2} $, and thus does not impact the  error estimates). 	For the case $n\geq 1$, we come back to the equation written as \eqref{Y.coordinate}  	and we invoke Theorem 7.1 in \cite{olver}, Chapter 11 to obtain an asymptotic expansion (in terms of $u=k$) of  	$\ep_{Ai_3}$ as such, for any $N\in \mathbb{N}$, and we will set $\zeta=\zeta_3$ for a lighter notation in the proof. We introduce the variable $\zeta$ by  $\frac{2}{3}|\zeta|^{3/2}(Y)={\rm sign}\left(Y-Y_{III}\right)\int^{Y_{III}}_{Y} |f|^{1/2}(Y') dY'$:
		\begin{equation}\begin{split}
				&\ep_{Ai_3}(\omega,s) = Ai(k^{2/3}\zeta)\sum_{q=1}^{N} \frac{\AC_q(\zeta(s))}{k^{2q}}+ k^{-4/3}Ai'(k^{2/3}\zeta)\sum_{q=0}^{N-1} \frac{B_q(\zeta(s))}{k^{2q}}+  \mathcal{E}_{2N+1,Ai}(\omega,s),\\ &  
				\frac{|\mathcal{E}_{2N+1,Ai}|(\omega,s)}{M(k^{2/3} \zeta)},\  \frac{|\rd_s\mathcal{E}_{2N+1,Ai}|(\omega,s)}{k^{2/3}|V|^{1/2}(s) N(k^{2/3} \zeta)} \leq 2 k^{-2N-1} \exp\left( 4 k^{-1} TV_{[\zeta(s),+\infty] }[|\zeta|^{1/2}B_0(\zeta)] \right) TV_{[\zeta(s), +\infty] }[|\zeta|^{1/2}B_N(\zeta)], 
			\end{split}
		\end{equation}   (see \cite{KGSchw1}, Appendix A, or \cite{olver}, chapter 11 for the definition of the $M$, $N$ functions), where the functions $\AC_q(\xi)$, $B_q(\xi)$ are defined recursively with $\AC_0(\xi)=1$ and   (the integration variable $v$ corresponds to $\zeta$): \begin{align}\label{A.def2}
			&	\AC_{q+1}(\zeta)= -\frac{1}{2}\rd_{\zeta} B_q(\zeta) + \frac{1}{2} \int_0^{\zeta} \psi(v) B_q(v) dv,\\ & B_q(\zeta)= \frac{{\rm sign}\left(\zeta\right)}{2|\zeta|^{1/2}} \int_0^{\zeta} \{ \psi(v) \AC_q(v)- \rd_{\zeta}^2\AC_q(v)  \}\frac{dv}{|\color{black}v|\color{black}^{1/2}},\\ & \psi(v)=- \hat{f}^{-3/4} \frac{d^2}{dY^2} (\hat{f}^{-1/4}).
		\end{align}

		As we have done before, we use $\partial_{\tilde{\omega}}$ to denote differentiation with respect to $\omega$ in the $\left(\omega,\zeta\right)$ coordinates and use $\partial_{\omega}$ for the $\left(\omega,Y\right)$ coordinates. We note that in the region $\frac{1}{4} \leq Y \leq3\color{black}$, we have the following estimates for any $N,q,n\in \mathbb{N}$: \begin{align*}
			&		|\hat{f}(\omega,Y)|  \approx 1,\ 	|\rd_{\omega}^n\hat{f}(\omega,Y)|\ls  [1+k^{2(n-1)} 1_{n\geq 1}+(\LL^2+\log(k))1_{n= 1}]  , \\ &  |\zeta|(Y) \approx |Y-Y_{III}|,\  |\rd_Y\zeta(Y)| \approx 1,\ |\rd_{\omega}^n\zeta(Y) |,\ |\rd_Y\rd_{\omega}^n\zeta(Y) |\lesssim  [1+k^{2(n-1)} 1_{n\geq 1}+(\LL^2+\log(k))1_{n= 1}] , \\   & |\rd_Y^q \rd_{\omega}^n\psi|,\  | \rd_{\tilde{\omega}}^n\psi| \ls  [1+k^{2(n-1)} 1_{n\geq 1}+(\LL^2+\log(k))1_{n= 1}]  , \\ &  |\rd_{\zeta}^{q}\rd_{\tilde{\omega}}^n \AC_{N}(\omega,\zeta)|\ls [1+k^{2(n-1)} 1_{n\geq 1}+(\LL^2+\log(k))1_{n= 1}] , \\& 
			|\rd_{\zeta}^{q}\rd_{\tilde{\omega}}^n B_N(\omega,\zeta)|\ls  [1+k^{2(n-1)} 1_{n\geq 1}+(\LL^2+\log(k))1_{n= 1}]  .
		\end{align*}
		
		While for $Y\geq 3$, we have: 
		\begin{align*}
			&	 \zeta(Y) \approx Y^{2/3},\  \rd_Y\zeta(Y) \approx Y^{-1/3},\ |\rd_{\omega}\zeta(Y) |\lesssim[ \LL^2+\log(k)] Y^{-1/3},\	\hat{f}(Y) \approx  Y^{-2/3}, \\   & |\rd_Y^q \rd_{\omega}^n\psi| \ls  Y^{-4/3-q}{ [1+k^{2(n-1)} 1_{n\geq 1}+(\LL^2+\log(k))1_{n= 1}] },\ | \rd_{\tilde{\omega}}^n\psi| \ls  Y^{-4/3}{ [1+k^{2(n-1)} 1_{n\geq 1}+(\LL^2+\log(k))1_{n= 1}] }, \\ &  |\rd_{\tilde{\omega}}^n \AC_{N+1}(\omega,\zeta)|\ls { [1+k^{2(n-1)} 1_{n\geq 1}+(\LL^2+\log(k))1_{n= 1}] } ,
			\\ & |\rd_{\zeta}^{1+q}\rd_{\tilde{\omega}}^n \AC_{N+1}(\omega,\zeta)|\ls { [1+k^{2(n-1)} 1_{n\geq 1}+(\LL^2+\log(k))1_{n= 1}] }|\zeta|^{-3/2-q} , 
			\\& 
			|\rd_{\tilde{\omega}}^n \left(|\zeta|^{1/2}B_N(\omega,\zeta)\right)|\ls { [1+k^{2(n-1)} 1_{n\geq 1}+(\LL^2+\log(k))1_{n= 1}] } ,
			\\ & |\rd_{\zeta}^{1+q}\rd_{\tilde{\omega}}^n \left(|\zeta|^{1/2}B_N(\omega,\zeta)\right)|\ls { [1+k^{2(n-1)} 1_{n\geq 1}+(\LL^2+\log(k))1_{n= 1}] } |\zeta|^{-5/2-q},
		\end{align*} where to obtain this, we  have  also taken advantage of the estimate 
		$  |\rd_{\tilde{\omega}} F|\ls  |\rd_{\omega} F|+\frac{|\rd_{\omega}\zeta|}{|\rd_Y\zeta|} |\rd_{Y} F|\ls|\rd_{\omega} F|+[\LL^2+\log(k) ]|\rd_{Y} F|$ 
		for any function $F$, and its iterated applications. In particular, we have \begin{align*}
			&   k^{-1} TV_{[\zeta(s),+\infty] }[|\zeta|^{1/2}B_0(\zeta)]\ls 1,\\ & TV_{[\zeta(s), +\infty] }[|\zeta|^{1/2}B_N(\zeta)] \ls 1.
		\end{align*}

		Now, note (see the discussion in \cite{olver} following  Theorem 7.1 in  Chapter 11) that $\mathcal{E}_{2N+1,Ai}(\omega,\zeta)$ satisfies the following Volterra equation:
		
		\begin{align}\label{Volterra2}
			&\mathcal{E}_{2N+1,Ai}(\omega,\zeta) = \pi k^{-2/3} \int^{+\infty}_{\zeta}  K(\zeta,v)\left(k^{-2N} Ai( k^{2/3} v) [2v\rd_{\zeta}B_N(\omega,v)+B_N(\omega,v)]+ \psi(v) 	\mathcal{E}_{2N+1,Ai}(\omega,v) \right) dv,  \\& 	K(\zeta,v,\omega)= (2k)^{-2/3}\left(Bi( [2k]^{2/3} \zeta)Ai( [2k]^{2/3} v)-Ai( [2k]^{2/3} \zeta)Bi( [2k]^{2/3} v)\right).
		\end{align}

		Applying $\rd_{\tilde{\omega}}^n$ derivatives to \eqref{Volterra2}, we see that the  contribution of each $\rd_{\omega}$ derivative is of order $O(k^3)$: so from Theorem~\ref{thm:volterraparam}, we get 
		$$ 	\frac{| \rd_{\tilde{\omega}}^n\mathcal{E}_{2N+1,Ai}|(\omega,s)}{M(k^{2/3} \zeta)},\  \frac{| \rd_{\tilde{\omega}}^n\rd_\zeta\mathcal{E}_{2N+1,Ai}|(\omega,s)}{k^{2/3} N(k^{2/3} \zeta)} \ls k^{-2N-1+3n}
		.$$
		
		Now, using the asymptotic expansion of $Ai(x)$ and $Bi(x)$ (see \cite{olver}, Chapter 11) when $x<0$: for any $M\geq 1$, there exist universal constants $( a_{q,\pm}, a_{q,+}')$ [independent of $\omega$] such that for all $n \geq 0$:  \begin{align*}
			&\rd_{\tilde{\omega}}^n\left[ Ai( k^{\frac{2}{3}} \zeta)- k^{-\frac{1}{6}} |\zeta|^{-\frac{1}{4}}  ( e^{\frac{2ik |\zeta|^{\frac{3}{2}}}{3}-i\frac{\pi}{4}} \sum_{q=0}^{M-1} a_{q,+} k^{-q} |\zeta|^{-\frac{3q}{2}}+    e^{\frac{-2ik |\zeta|^{\frac{3}{2}}}{3}+i\frac{\pi}{4}} \sum_{q=0}^{M-1} a_{q,-} k^{-q} |\zeta|^{-\frac{3q}{2}})\right]\\= & O\left(k^{-M+3n} |\zeta|^{-\frac{3M}{2}-\frac{1}{4}}\right),\\ & \rd_{\tilde{\omega}}^n\left[Ai'( k^{2/3} \zeta)- k^{1/6} |\zeta|^{\frac{1}{4}}   e^{\frac{2ik |\zeta|^{\frac{3}{2}}}{3}-i\frac{\pi}{4}} \sum_{q=0}^{M-1} a_{q,+}' k^{-q} |\zeta|^{-3q/2}+    e^{\frac{-2ik |\zeta|^{\frac{3}{2}}}{3}+i\frac{\pi}{4}} \sum_{q=0}^{M-1} a_{q,-}' k^{-q} |\zeta|^{-\frac{3q}{2}}\right]\\ =& O\left(k^{-M+3n} |\zeta|^{-\frac{3M}{2}-\frac{1}{4}}\right).
		\end{align*}
		\eqref{many.deriv} then follows, taking $M$ large enough, e.g.\ $M=n+1$ and $N$ large enough, e.g.\ $2N \geq n$.

	\end{proof}

	\subsection{Connection between the second and third turning point}\label{connection23.section}
	We start with the connection between the WKB estimates slightly under the third turning point, and $u_I$.
	\begin{lemma}\label{connection2.3.lemma} $u_I$ satisfies the following connection formula to $\uAii$ and $\uBii$, defining $\kk:=\int_{s_{II}}^{s_{III}} |V|^{1/2}(s) ds$:
		\begin{equation}\label{connection2.3.eq}
			\begin{split}
				&u_I(\omega,s)= a_{3,+}(\omega) w_{2,+}(\omega,s)+ a_{3,-}(\omega) w_{2,-}(\omega,s),\\  & a_{3,\pm}(\omega) = \frac{1}{2}  (Mm^2)^{1/2} \pi^{-1/2} k^{-2/3}  e^{\mp i\kk \pm \frac{i\pi}{4}}[ 1+ \delta_{3,\pm}(\omega)],\\ &  |\delta_{3,\pm}|(\omega)\ls k^{-1},\  |\rd_{\omega}^{n}\delta_{3,\pm}|(\omega)\ls k^{2n-1}\log(k),
			\end{split}
		\end{equation}for any $n\geq 1$, where $\delta_{3,-}(\omega)= \overline{\delta_{3,+}}(\omega)$. In particular, note that defining $\phi_3(\omega) \in \R$  as $e^{i\phi_3(\omega)} := \frac{1+\delta_{3,+}(\omega)}{|1+\delta_{3,+}(\omega)|}$, we have for any $n\geq 1$ \begin{equation}
			|\phi_3|(\omega)\ls k^{-1},\ |\rd_{\omega}^n\phi_3|(\omega)\ls k^{2n-1}\log(k).
		\end{equation} We can also rewrite the above relation in terms of $w_{2,even}$ and $w_{2,odd}$ as \begin{equation}\label{useful}\begin{split}&u_I(\omega,s)= (Mm^2)^{1/2} \pi^{-1/2} k^{-2/3} (1+\delta_I(\omega))\left[  \cos(\kk-\phi_3(\omega)) w_{2,even}(\omega,s)- \sin(\kk-\phi_3(\omega)) w_{2,odd}(\omega,s)\right],\end{split} \end{equation}
		where $\delta_I(\omega):= |1+\delta_{3,+}(\omega)|-1$ satisfies also
		\[|\delta_I|(\omega)\ls k^{-1},\  |\rd_{\omega}^{n}\delta_I|(\omega)\ls k^{2n-1}\log(k).\]
	\end{lemma}

	\begin{proof}
		We write, using  Proposition~\ref{WKB.TP2.prop}, Proposition~\ref{usefulformofuI}, and  Lemma~\ref{uI.lemma}\begin{align}
			u_I &= (Mm^2)^{1/2} \frac{\pi^{-1/2}}{2} k^{-2/3} ( e^{i\int_{s_{II}}^{s_{III}} |V|^{1/2}(s) ds-\frac{i\pi}{4}} \frac{w_{2,-}(\omega,s)}{1+\ep_{2,-}^{WKB}(\omega,s)}+ e^{-i\int_{s_{II}}^{s_{III}} |V|^{1/2}(s) ds+\frac{i\pi}{4}} \frac{w_{2,+}(\omega,s)}{1+\ep_{2,+}^{WKB}(\omega,s)}
			\\ \nonumber &+  e^{i\int_{s_{II}}^{s_{III}} |V|^{1/2}(s) ds-\frac{i\pi}{4}} \frac{w_{2,-}(\omega,s)}{1+\ep_{2,-}^{WKB}(\omega,s)}\eta_{Ai_3,+}(\omega,s)+e^{-i\int_{s_{II}}^{s_{III}} |V|^{1/2}(s) ds+\frac{i\pi}{4}} \frac{w_{2,+}(\omega,s)}{1+\ep_{2,+}^{WKB}(\omega,s)}\eta_{Ai_3,-}(\omega,s)).
		\end{align} Thus \eqref{connection2.3.eq} follows from the estimates of Proposition~\ref{usefulformofuI} and Lemma~\ref{uI.lemma} evaluated at $s=\frac{k^2}{4Mm^2}$. Finally, \eqref{useful} follows from the identities: \begin{equation*}\begin{split}&u_I(\omega,s)= (Mm^2)^{1/2} \pi^{-1/2} k^{-2/3}  \left[  (\cos(\kk) + \Re(\delta_{3,+} e^{-i\kk}))w_{2,even}(\omega,s)+(-\sin(\kk) + \Im(\delta_{3,+} e^{-i\kk})) w_{2,odd}(\omega,s)\right] \\& =  (Mm^2)^{1/2} \pi^{-1/2} k^{-2/3}  \left[   \Re([1+\delta_{3,+}] e^{-i\kk})w_{2,even}(\omega,s)+ \Im([1+\delta_{3,+}] e^{-i\kk}) w_{2,odd}(\omega,s)\right].\end{split} \end{equation*}
	\end{proof}
	
	In the next proposition, we combine Lemma~\ref{connection2.3.lemma}	and Lemma~\ref{connection2.3.lemma2}.	
	\begin{prop}\label{connection23.prop}
		$u_I$ obeys the following estimates, recalling  $\eta_{\alpha}(\omega)=\eta_{\alpha}^{reg}(\omega)+\eta_{\alpha}^{sing}(\omega)$ and $\eta_{\beta}(\omega)=\eta_{\beta}^{reg}(\omega)+\eta_{\beta}^{sing}(\omega)$ defined in Lemma~\ref{connection2.3.lemma2}:
		\begin{align}
			&u_I = \alpha_{A_3}(\omega) \uAii+ \beta_{A_3}(\omega) \uBii,\label{uI1}\\
			& \alpha_{A_3}(\omega)= k^{-2/3}\LL^{2/3}\left(1+\delta_I(\omega)\right) (1+\eta_{\alpha}(\omega)) \cos(\kk-\phi_3(\omega)),\label{uI2} \\  & \beta_{A_3}(\omega)= -k^{-2/3}\LL^{2/3}\left(1+\delta_I(\omega)\right)(1+\eta_{\beta}(\omega)) \sin(\kk-\phi_3(\omega)),\label{uI3}.
		\end{align}
	\end{prop}
	
	\begin{proof} This is an immediate computation combining Lemma~\ref{connection2.3.lemma}	and Lemma~\ref{connection2.3.lemma2}.

	\end{proof}
	
	\subsection{Putting everything together: $u_H$ and $u_I$ formulae}\label{together.section}
	
	Writing $u_H$ (see Section~\ref{energy.section}) and $u_I$ (see Section~\ref{connection23.section}) in terms of $w_{2,\pm}$, we obtain the following formula, recalling that $e^{i\phi_3(\omega)} := \frac{1+\delta_{3,+}(\omega)}{|1+\delta_{3,+}(\omega)|}$: \begin{equation*}\begin{split}
			&W(u_I,u_H)=  (Mm^2)^{1/2} \pi^{-1/2} k^{-2/3} \left( \gamma_B(\omega) (\cos(\kk)+ \Re(\delta_{3,+}e^{-i\kk})) + \gamma_A(\omega) (\sin(\kk) -\Im(\delta_{3,+}e^{-i\kk})) \right) \\ & =  (Mm^2)^{1/2} \pi^{-1/2} k^{-2/3}  \tilde{\gamma}_A(\omega) \exp( C_{II} \LL\log\LL)\left( \frac{\tilde{\gamma}_B(\omega)}{\tilde{\gamma}_A(\omega)}  \frac{\cos(\kk)+ \Re(\delta_{3,+}e^{-i\kk})}{\exp(2 C_{II} \LL\log\LL) }+  (\sin(\kk) -\Im(\delta_{3,+}e^{-i\kk})) \right) \\ & =   (Mm^2)^{1/2} \pi^{-1/2} k^{-2/3}  \tilde{\gamma}_A(\omega) \exp( C_{II} \LL\log\LL)\left( \frac{\tilde{\gamma}_B(\omega)}{\tilde{\gamma}_A(\omega)}  \frac{  \Re([1+\delta_{3,+}]e^{-i\kk})}{\exp(2 C_{II} \LL\log\LL) }- \Im([1+{\delta}_{3,+}]e^{-i\kk}) \right)\\ &  =     (Mm^2)^{1/2} \pi^{-1/2} k^{-2/3}  \tilde{\gamma}_A(\omega) \exp( C_{II} \LL\log\LL) |1+\delta_{3,+}(\omega)|\left( \frac{\tilde{\gamma}_B(\omega)}{\tilde{\gamma}_A(\omega)}  \frac{  \cos(\kk - \phi_3(\omega))}{\exp(2 C_{II} \LL\log\LL) }+  \sin(\kk - \phi_3(\omega)))\right).\end{split} \end{equation*} Moreover, we can also write \begin{align*} & W(u_I,u_H) = \frac{1}{2}(Mm^2)^{1/2} \pi^{-1/2} k^{-2/3} \left( [\gamma_B(\omega)-i\gamma_A(\omega)][1+
		\overline{\delta_{3,+}}]e^{i\kk} +[\gamma_B(\omega)+i\gamma_A(\omega)][1+
		{\delta_{3,+}}]e^{-i\kk} \right) \\ & =  \frac{1}{2}(Mm^2)^{1/2} \pi^{-1/2} k^{-2/3} e^{-i(\kk-\phi_3(\omega))}  [\gamma_B(\omega)+i\gamma_A(\omega)]|1+\delta_{3,+}(\omega)|\left( 1 +\underbrace{\frac{\gamma_B(\omega)-i\gamma_A(\omega)}{\gamma_B(\omega)+i\gamma_A(\omega) }}_{:=\Gamma(\omega,s)}\ e^{2i(\kk-\phi_3(\omega))} \right).
	\end{align*}

	Now we turn to splitting $u_I$ in terms of $w_{1,\pm}$:

	$$ u_I =  \left[\underbrace{  \alpha_{A_3}\alpha_{A_2}+\beta_{A_3}\alpha_{B_2}}_{:=\alpha_{3,-}}\right] w_{1,-}+\left[ \underbrace{ \alpha_{A_3}\beta_{A_2}+\beta_{A_3}\beta_{B_2}}_{:= \alpha_{3,+}}\right] w_{1,+},$$
	
	and the connection coefficients $\alpha_{3,\pm}$ are as follows:

	\begin{align}	 \label{alpha3-1}&\alpha_{3,-}=  - \frac{k^{-2/3} (Mm^2)^{1/2}}{ \pi^{1/2}} \exp\left(   C_{II}\LL\log(\LL) \right) [  \left(1+\eta_{\beta}(\omega)\right) \left(1+\delta_{\alpha,B_2}(\omega)\right) (1+\delta_I)\sin(\kkp)  \\ &-\exp( -2  C_{II}\LL\log(\LL) )  (1+\eta_{\alpha}(\omega)) (1+\delta_I(\omega))\cos(\kkp)\frac{ \delta_{\alpha,A_2}(\omega)}{2}],   \end{align} and  \begin{align}
		&\alpha_{3,+} = \frac{k^{-2/3} (Mm^2)^{1/2}}{ \pi^{1/2}} \exp\left(  - C_{II}\LL\log(\LL) \right)(1+\delta_I(\omega))\\ & \left[- (1+\eta_{\beta}(\omega))  \sin(\kkp) \tilde{\beta}_{B_2}(\omega) +\frac{1}{2}(1+\eta_{\alpha}(\omega)) \cos(\kkp)[1+\delta_{\beta,A_2}(\omega)]\right].\label{alpha3+1}
	\end{align} 
	
	Using the above computations, we obtain estimates in the following proposition.

	\begin{prop}\label{together.prop} Defining the new phase $\kkq:= \kkp$, we introduce constants $C_{cos,\pm}(\omega)$, $C_{sin,\pm}(\omega)$, $\delta_{cos}(\omega)$, $\delta_{sin}(\omega)$ and write:
		\begin{align}
			&  u_I = \alpha_{3,-}(\omega,\LL)  w_{1,-}+  \alpha_{3,+}(\omega,\LL)  w_{1,+},\\	 
			& \label{alpha3+-2} \alpha_{3,\pm}(\omega,\LL)= k^{-2/3}  \pi^{-1/2}  (Mm^2)^{1/2} \exp\left(  \mp C_{II}\LL\log(\LL) \right) \left( C_{cos,\pm}(\omega) \cos(\kkq)+C_{sin,\pm}(\omega) \sin(\kkq)\right) ,\end{align}  
		
		Moreover, the following estimates hold, defining $C_{cos,\pm}=C_{cos,\pm}^{reg}+C_{cos,\pm}^{sing} $, $C_{sin,\pm}=C_{sin,\pm}^{reg}+C_{sin,\pm}^{sing} $: \begin{align}
			&	|C_{sin,+}^{reg}|(\omega),\ |\rd_{\omega}C_{sin,+}^{reg}|(\omega) \ls_{A} \exp( O(\LL)),\  k^{1/2}	|C_{sin,+}^{sing}|(\omega),\ k^{-2} \log^{-1}(k) |\rd_{\omega}C_{sin,+}^{sing}|(\omega) \ls_{A} \exp( O(\LL)),\\  & 	|C_{sin,-}^{reg}(\omega)+1|\ls e^{-\kappa_+ A},\ |\rd_{\omega}C_{sin -}^{reg}|(\omega) \ls_A  \LL^{3}\log(\LL),\ 	|C_{sin -}^{sing}|(\omega)\ls  k^{-1/2},\  |\rd_{\omega}C_{sin -}^{sing}|(\omega) \ls k^{2} \log(k),\\ &|C_{cos,+}^{reg}(\omega)-\frac{1}{2}|\ls e^{-\kappa_+ A},\ |\rd_{\omega}C_{cos,+}^{reg}|(\omega)\ls_A \LL^{3}\log(\LL),\  	|C_{cos,+}^{sing}|(\omega)\ls k^{-1/2},\  |\rd_{\omega}C_{cos,+}^{sing}|(\omega) \ls k^{2} \log(k),\\  & 	|C_{cos,-}^{reg}(\omega)|,\ \LL^{-3} {\log^{-1}(\LL)} |\rd_{\omega}C_{cos,-}^{reg}|(\omega)\ls_{A}  \exp( -2  C_{II}\LL\log(\LL))  ,\\ & 	k^{1/2}|C_{cos,-}^{sing}|(\omega),\  k^{-2} \log^{-1}(k)  |\rd_{\omega}C_{cos, -}^{sing}|(\omega)  \ls_{A} \exp( -2  C_{II}\LL\log(\LL)) . \end{align}
		Now, for the Wronskian estimates, introducing the non-zero constants $D_W(\omega), C_W(\omega)$ and $\delta_{cos}(\omega)$:  \begin{align}		& W(u_I,u_H)=2k^{-2/3}  \pi^{-1/2}  (Mm^2)^{1/2} D_W(\omega) \exp( C_{II} \LL\log\LL)\left(  \exp(-2C_{II} \LL \log\LL) \cos(\kkq) \delta_{cos}(\omega) +\sin(\kkq)) \right) \\ &=k^{-2/3}  \pi^{-1/2}  (Mm^2)^{1/2}C_W(\omega) \exp( C_{II} \LL \log(\LL))  e^{-i\kkq} \left(    1+ 		\Gamma(\omega)\cdot e^{2i\kkq}\right),
		\end{align}			w
		with $\Gamma(\omega)$  defined in Proposition~\ref{energy.lemma}. 	The following estimates hold, with  $\delta_{cos}=\delta_{cos}^{reg}+\delta_{cos}^{sing} $,\\ $D_W= D_W^{reg}+ D_W^{sing}$,  $C_W= C_W^{reg}+ C_W^{sing}$:	 \begin{align}\\ &  |C_W^{reg}|(\omega)\approx 1,\ |D_W^{reg}|(\omega) \approx 1,\  |\rd_{\omega}C_W^{reg}|(\omega),\ |\rd_{\omega}D_W^{reg}|(\omega)\ls \LL^3,\\ & |C_W^{sing}|(\omega),\  |D_W^{sing}|(\omega)\ls k^{-1/2},\  |\rd_{\omega}C_W^{sing}|(\omega),\  |\rd_{\omega}D_W^{sing}|(\omega)\ls k^2\log(k),\\ &  1 \lesssim |\delta_{cos}^{reg}|(\omega)\ls 1,\  |\rd_{\omega}\delta_{cos}^{reg}|(\omega)\ls    \LL^{3} \log\LL,\\ & |\delta_{cos}^{sing}|(\omega)\ls k^{-1/2} ,\ |\delta_{cos}^{sing}|(\omega) \lesssim  k^{2}\log(k).   \end{align} and recalling 
		
	\end{prop}
	
	\begin{proof}
		These estimates  follow from combining all the previous estimates, notably from Corollary~\ref{TP2.connection.cor}, Proposition~\ref{connection23.prop}, Lemma~\ref{energy.lemma}. Note that we have used the trivial inequality $\log^{1/3}(\LL) \ls \log(
		\LL)$.
	\end{proof}

	\subsection{Green's formula and conclusion of the proof of Theorem~\ref{TPsection.mainprop}}\label{resolvent.section}
	
	The computations in the above sections now allow us to derive a precise representation formula that we will be able to exploit in the next stationary phase lemmata arguments. We now conclude with the proof of Theorem~\ref{TPsection.mainprop}.

	\begin{proof}

		By \cite{KGSchw1}, Section 4.3, we have the following representation formula:
		$$u(\omega,s) = \underbrace{\frac{ u_H(\omega,s) \int_s^{+\infty} u_I(s') H(\omega,s') ds'}{W(u_I,u_H)}}_{:=u_1(\omega,s)}+ \underbrace{\frac{ u_I(\omega,s) \int_{-\infty}^{s} u_H(s') H(\omega,s') ds'}{W(u_I,u_H)}}_{:=u_2(\omega,s)}$$
		
		Then, we use the splitting of $u_I$ in terms of $w_{1,\pm}$ invoked in Section~\ref{together.section}
		\begin{align*}
			&	u_1(\omega,s)=  \underbrace{\frac{u_H(\omega,s) \alpha_{3,+}(\omega)}{W(u_H,u_I)} \int_s^{+\infty} w_{1,+}(\omega,s') H(\omega,s') ds' }_{:=u_{1,+}(\omega,s)}+\underbrace{\frac{u_H(\omega,s) \alpha_{3,-}(\omega)}{W(u_H,u_I)} \int_s^{+\infty} w_{1,-}(\omega,s') H(\omega,s') ds'}_{:=u_{1,-}(\omega,s)}\\ & u_2(\omega,s)=  \underbrace{\frac{ w_{1,+}(\omega,s) \alpha_{3,+}(\omega)}{W(u_H,u_I)} \int^s_{-\infty}u_H(\omega,s') H(\omega,s') ds'}_{:=u_{2,+}(\omega,s)}+\underbrace{\frac{ w_{1,-}(\omega,s) \alpha_{3,-}(\omega)}{W(u_H,u_I)} \int^s_{-\infty}u_H(\omega,s') H(\omega,s') ds'}_{:=u_{2,-}(\omega,s)}
		\end{align*}

		Now, for the  Green's formula, noting that from Proposition~\ref{together.prop}, one has the following formula: \begin{equation*}\begin{split}
				&\frac{\alpha_{3,\pm}(\omega)}{W(u_I,u_H)(\omega)}= e^{i\kkq} \left(  \frac{ \exp\left(-[1\pm 1] C_{II} \LL\log(\LL)\right)}{C_W(\omega)}  \frac{ C_{cos,\pm}(\omega) \cos(\kkq) +C_{sin,\pm}(\omega) \sin(\kkq) }{1+ \Gamma(\omega) e^{2i\kkq}}\right),
			\end{split}
		\end{equation*}  from which we obtain \begin{align*}
			& u_{1,+}(\omega,s)\\=  & e^{i\kkq} \left(\exp(-2 C_{II} \LL\log(\LL)) \frac{u_H(\omega,s) }{C_W(\omega,\LL)}\int_s^{+\infty} w_{1,+}(\omega,s') H(s') ds' \frac{ C_{cos,+}(\omega) \cos(\kkq) +C_{sin,+}(\omega) \sin(\kkq) }{1+ \Gamma(\omega) e^{2i\kkq}}\right),
		\end{align*}  \begin{align*} &u_{1,-}(\omega,s)=  e^{i\kkq} \left(  \frac{u_H(\omega,s) }{C_W(\omega,\LL)} \int_s^{+\infty} w_{1,-}(\omega,s') H(s') ds' \frac{ C_{cos,-}(\omega) \cos(\kkq) +C_{sin,-}(\omega) \sin(\kkq) }{1+ \Gamma(\omega) e^{2i\kkq}}\right),\end{align*} 
		\begin{align*}
			&u_{2,+}(\omega,s)\\= & e^{i\kkq} \left(\exp(-2 C_{II} \LL\log(\LL) \frac{w_{1,+}(\omega,s) }{C_W(\omega,\LL)}\int^s_{-\infty} u_H(\omega,s') H(s') ds' \frac{ C_{cos,+}(\omega) \cos(\kkq) +C_{sin,+}(\omega) \sin(\kkq) }{1+ \Gamma(\omega) e^{2i\kkq}}\right),
		\end{align*}
		\begin{align*}
			u_{2,-}(\omega,s)=  e^{i\kkq} \left(\frac{ w_{1,-}(\omega,s)}{ C_W(\omega,\LL)}    \int^s_{-\infty} u_H(\omega,s') H(s') ds' \frac{ C_{cos,-}(\omega) \cos(\kkq) +C_{sin,-}(\omega) \sin(\kkq) }{1+ \Gamma(\omega) e^{2i\kkq}}\right).
		\end{align*}

		Finally, we define $u_{i,j}^{reg}$ evaluating all the regular coefficients (i.e.\ all coefficients except $e^{\pm i\kkq}$) at $\omega=m$ and $\ep_{i,j}(\omega,s):= (1+\Gamma(m) e^{2i\kkq}) (u_{i,j}-u_{i,j}^{reg})$. Then, we write  $\ep_{i,j}(\omega,s)= \ep_{i,j}^{cos}(\omega,s) \cos(\kkq)+\ep_{i,j}^{sin}(\omega,s)\sin(\kkq)$.
		The estimates on $\ep_{i,j}^{cos}, \ep_{i,j}^{sin}$ follow immediately from Proposition~\ref{together.prop} (noting that, for any function $f$, $|f(\omega)- f(m)| \ls k^{-2} \sup| \rd_{\omega} f|$). Note  that we have also used \eqref{C.eq}, introducing $D_I$  so that $-2C_{II}(\omega=m)\LL \log \LL= -4 \LL\log \LL+ D_I \LL$.

	\end{proof}
	
	\paragraph{Proof of Theorem~\ref{quasimode.thm}} The proof of Theorem~\ref{quasimode.thm}  follows from Theorem~\ref{TPsection.mainprop}. For Statement~\ref{blowup}, we start to obtain \eqref{fourier.blow-up.intro} for the equation \eqref{KG.inh.intro} with a source $F$ and zero initial data. For the solution $\phi$ of  \eqref{KG.inh.intro}, and fixed $R>r_+$, it is straightforward to construct a source $F$ such that its corresponding Fourier transform $H_L(\omega=m,r)$ (projected on a fixed spherical harmonic $Y_L$) is supported on $[r_+,R]$ and positive  so that the $v_{1,+}$, $v_{2,+}$ terms in  Proposition~\ref{together.prop} dominate the other terms and give $e^{C \LL}$ point-wise lower bounds, namely \eqref{fourier.blow-up.intro} holds. Now, returning to \eqref{KG.intro} with smooth, compactly supported initial data: a straightforward contradiction argument, based on the Duhamel formula and the fact that \eqref{fourier.blow-up.intro} holds for \eqref{KG.inh.intro}, allows one to also establish \eqref{fourier.blow-up.intro} for \eqref{KG.intro}.

	Similarly, Statement~\ref{failure} can be obtained constructing a source $F$ concentrating around the bad frequencies where $\kkq= N \in \mathbb{N}$ for large $N$; we will omit the details. For Statement~\ref{quasimode.statement}, we mimic the standard quasimode constructions: we first take $\omega_L= m -  \LL^p$ for some large $p>2$. Then, we construct the quasimode sequence $\Phi_L$ to be  a cut-off of $u_I(\omega_L,s)$ in the region $\eta \LL^2 \leq s \leq \eta^{-1}\LL^p$ for a sufficiently small $L$-independent constant $\eta$, namely:  $$\Phi_L(t^{*},s,\theta,\varphi)= e^{-i \omega_L t^{*}}u_I(\omega_L,s) Y_L(\theta,\varphi) [1-\chi(\frac{s}{\eta \LL^2})] \chi(\frac{s}{\eta^{-1} \LL^p})$$  where $\chi$ is a smooth, compactly supported function with $\chi(x)=0$ for $x>2$ and $\chi(x)=1$ for $x<1/2$. It is straightforward to see that $\Phi_L$ solves \eqref{KG.intro} up to $e^{-O(\LL)}$ errors. Subsequently, the logarithmic lower bounds \eqref{log.lower} follow from a standard argument, see  Section 7 of~\cite{quasimodeads}.
	
	\section{Inverse Fourier transform and stationary phase}\label{fourier.section}
	\subsection{Analysis of $\tilde{k}$ and the key Phase function $\Phi$}
	
	\begin{lemma}\label{estimatesfortildek}Suppose that $m - \omega > 0$ and $|m-\omega| \leq \LL^{-p}$ for some large integer $p$. Recall that $\tilde{k}$ is then defined so that
		\[\pi \tilde{k} = \int_{s_{II}}^{s_{III}}\left|V\right|^{1/2}\, ds,\]
		where $s_{II}$ and $s_{III}$ are the two largest zeros of $V$.
		We have the following asymptotics for $\tilde{k}$ as $\omega \to m$:
		\begin{enumerate}
			\item There exists a constant $D\left(\LL\right)$ which satisfies $\left|D(\LL)\right| \lesssim \LL^3$ so that
			\begin{equation}\label{firstexpansionfortildek}
				\tilde{k} = k + D\left(\LL\right) + O\left(\LL^2 k^{-1}\right).
			\end{equation}
			\item For any integer $n \geq 1$, we have we have
			\[\partial^{(n)}_{\omega}\tilde{k} = \partial_{\omega}^{(n)}k + O_n\left(k^{2n-1+\delta(p)}\right),\]
			where $\delta \to 0$ as $p \to \infty$.
		\end{enumerate}
		
	\end{lemma}
	\begin{proof}It is useful throughout the proof to keep in mind the estimates~\eqref{sII} and~\eqref{sIII}. In fact by applying the implicit function theorem in the variable $x = r \LL^{-2}$, one also obtains that $s_{II}$ is a smooth function of $\omega$ with $\left|\partial_{\omega}^{(n)}s_{II}\right| \lesssim_n \LL^2$. Similarly, after an application of the implicit function theorem in the variable $x = r\left(m^2-\omega^2\right)$, we obtain that $\left(m^2-\omega^2\right)s_{III}$ is smooth and satisfies $\left|\partial_{\omega}^{(n)}\left(\left(m^2-\omega^2\right)s_{III}\right)\right| \lesssim \LL^2.$
		
		Let $A > 0$ be a suitably large constant and $\delta > 0$ be an arbitrarily small constant. We start by decomposing $\pi\tilde{k}$ into three different integrals:
		\[\pi \tilde{k} = \int_{s_{II}}^{As_{II}}\left|V\right|^{1/2}\, ds + \int_{s_{III}-k^{2/3}}^{s_{III}}\left|V\right|^{1/2}\, ds + \int_{As_{II} }^{s_{III}-k^{2/3}}\left|V\right|^{1/2} \doteq P_0 + P_1 + P_2.\]
		
		We start with the estimate of $P_0$. Trivially, we have that
		\[\int_{s_{II}}^{s_{II} + k^{-1}}\left|V\right|^{1/2} \lesssim k^{-1}.\]
		On the other hand, a Taylor expansion of $V|_{\omega =m}$ near $s = s_{II}$ and the monotonicty of $V$ in $[s_{II},As_{II}]$ imply that when $ s \in [s_{II}+k^{-1},As_{II}]$, we have
		\begin{equation}\label{approxwithsqaureroot1}
			\frac{\left|V-V_{\omega=m}\right|}{\left|V_{\omega=m}\right|} \ll 1,\qquad \frac{\left|V - V|_{\omega = m}\right|}{\left|V|_{\omega = m}\right|^{1/2}} \lesssim \LL^2 k^{-3/2} \Rightarrow 
		\end{equation}
		\begin{equation}\label{sothenpoisgoodiguess}
			\int_{s_{II}+k^{-1}}^{As_{II}}\left|V\right|^{1/2}\, ds = \int_{s_{II}+k^{-1}}^{As_{II}}\left|V|_{\omega = m}\right|^{1/2}\, ds + \int_{s_{II}+k^{-1}}^{As_{II}}O\left(\LL^2k^{-3/2}\right)\, ds.
		\end{equation}
		In particular, we then may easily obtain the existence of a constant $D_0\left(\LL\right)$ satisfying $\left|D_0\right| \lesssim_A \LL^2$ so that
		\[P_0\left(\omega,\LL\right) = D_0\left(\LL\right) + O\left(k^{-1}\right).\]

		Now we turn to $P_1$. We have $V\left(s_{III}\right) = 0$. Furthermore, $s \in [s_{III}-k^{2/3},s_{III}]$ implies that $|V'(s_{III})| \sim k^{-4} +O\left(k^{-5}\right)$ and $\left|V''(s)\right| \lesssim k^{-6}$. Thus, $s \in [s_{III}-k^{2/3},s_{III}]$ implies that $|V(s)| \sim k^{-4}\left(s_{III}-s\right)$ and we may easily check that
		\[\left|P_1\left(\omega,\LL\right)\right| \lesssim k^{-1}.\]
		
		In order to investigate $P_2$, we split the potential into $V = V_0 + V_1$ where
		\[V_0 \doteq \frac{\left(Mm^2\right)^2}{k^2} - \frac{2Mm^2}{r},\qquad V_1 \doteq V - V_0 = O\left(\frac{\LL^2}{r^2}\right).\]
		For $s \in [A s_{II},s_{III}-k^{2/3}]$ one may check that
		\[\left|\frac{V_1}{V_0}\right| \lesssim \frac{\LL^2}{r(2Mm^2) - r^2 \frac{(Mm^2)^2}{k^2}} \lesssim A^{-1}.\]
		We now fix a choice of $A$ so that $s \in [A s_{II},s_{III}-k^{2/3}]$ implies
		\begin{equation*}
			\left|\frac{V_1}{V_0}\right| \ll 1.
		\end{equation*}
		In particular, we then have that $s \in [A s_{II},s_{III}-k^{2/3}]$ implies
		\begin{equation}\label{goodesimtatefordifference}
			\left|V\right|^{1/2} = |V_0|^{1/2} + O\left(\frac{\left|V_1\right|}{|V_0|^{1/2}}\right).
		\end{equation}
		We then further split
		\begin{align*}
			P_2 &= \int_{As_{II}}^{s_{III}-k^{2/3}}\left|V_0\right|^{1/2}\, ds +\int_{As_{II}}^{\left(1-\delta\right)s_{III}}\left(\left|V\right|^{1/2}-\left|V_0\right|^{1/2}\right)\, ds +  \int_{\left(1-\delta\right)s_{III}}^{s_{III}-k^{2/3}}\left(\left|V\right|^{1/2}-\left|V_0\right|^{1/2}\right)\, ds 
			\\ \nonumber &\doteq P_3 + P_4+P_5.
		\end{align*}
		
		For $P_3$, we have
		\[P_3 = \int_{r\left(As_{II}\right)}^{r\left(s_{III}-k^{2/3}\right)}\sqrt{\frac{2Mm^2}{r} - \frac{(Mm^2)^2}{k^2}}\left(1-\frac{2M}{r} + \frac{\mathbf{e}^2}{r^2}\right)\, dr.\]
		We then introduce a new variable $x = \frac{Mm^2 r}{2k^2}$ and obtain 
		\begin{equation}\label{forP3changevar}
			P_3 = 2k \int_{r(As_{II})\left(\frac{Mm^2}{2k^2}\right)}^{\left(r(s_{III}-k^{2/3})\right)\left(\frac{Mm^2}{2k^2}\right)}\sqrt{\frac{1}{x} - 1}\left(1 - \frac{M^2m^2}{xk^2} + \frac{\mathbf{e}^2M^2m^4}{4x^2k^4}\right)\, dx.
		\end{equation}
		In view of the fact that
		\[\int_0^1\sqrt{\frac{1}{x}-1}\, dx = \left(x\sqrt{\frac{1}{x}-1} - \arctan\left(\sqrt{\frac{1}{x}-1}\right)\right)\Bigg|_0^1 = \pi/2,\]
		we obtain the existence of a constant $D_3\left(\LL\right)$ with $|D_3| \lesssim 1$ so that
		\[P_3 = \pi k + D_3 + O\left(k^{-1}\right).\]
		For $P_4$ we note in view of~\eqref{goodesimtatefordifference} we have that $s \in [As_{II},\left(1-\delta\right)s_{III}]$ implies
		\[\left|\left|V_0\right|^{1/2} - \left|V\right|^{1/2}\right| \lesssim \LL^2 r^{-3/2},\]
		from which we easily may obtain the existence of a constant $D_4(\LL)$  with $|D_4| \lesssim \LL$ so that 
		\[P_4 = D_4 + O\left(k^{-1}\LL^2\right).\]
		Finally, for $P_5$ we note that $s \in [\left(1-\delta\right)s_{III},s_{III} - k^{2/3}]$ implies that
		\[\left|V_1\right| \lesssim \LL^2 k^{-4},\qquad \left|V_0\right| \sim k^{-4}\left(s_{III}-s\right),\]
		from which we obtain
		\[P_5 = O\left(\LL^2k^{-1}\right).\]
		Adding up all of the estimates for $P_0$ through $P_5$ then establishes~\eqref{firstexpansionfortildek}. 
		
		The derivative estimates are obtained by applying $\partial^{(n)}_{\omega}$ and using a similar decomposition of the integral. We omit the details.

	\end{proof}

	\begin{lemma}\label{foroscest}Let $p \gg 1$ be suitably large. For any $A > 0$ and $t > 0$ define a function $\Phi\left(t,A,\omega\right):(0,\infty) \times (0,\infty) \times (m-\LL^{-p},m) \to \mathbb{R}$ by
		\[\Phi\left(t,A,\omega\right) \doteq \omega - \frac{A}{t} \kkq.\]
		
		Then, either there is no critical point of $\Phi$, or there exists a unique critical $\omega_c\left(t,A\right) \in (m-\LL^{-p},m)$ with
		\[\partial_{\omega}\Phi\left(t,A,\omega_c\left(t,A\right)\right) = 0,\]
		For any $\tilde{\delta} > 0$, we can take $p$ sufficiently large so that we have the following estimate for $\omega_c\left(t,A\right):$
		\[\omega_c(t,A) = m - \left(\frac{t}{A}\right)^{-2/3}\frac{(\pi M)^{2/3}m}{2} + O_{\tilde{\delta}}\left(\left(\frac{t}{A}\right)^{-\frac{4}{3}+\tilde{\delta}/3}\right).\]

		Let $0 < d < \LL^{-p}$. There is no critical point of $\Phi$ in $(m-r,m)$ if and only if
		\[\partial_{\omega}\Phi\left(t,A,m-d\right) \leq 0.\]

		Let $\delta > 0$ be arbitrarily small, and $R$ be an arbitrarily large positive integer. For suitable $D \gg 1$ (depending only on $m$, $M$, $\delta$, and $R$), set \[\mathscr{D} \doteq \left\{ (t,A) \in (1,\infty) \times (1,\infty) : \Phi(t,A,\omega) \text{ has a critical point in }\omega \in (m-\LL^{-p},m)\text{ and }\frac{t}{A} \geq D\right\}.\] Define a function $G\left(t,A\right) : \mathscr{D} \to \mathbb{R}$ by $G\left(t,A\right) \doteq t\Phi\left(t,A,\omega_c\left(t,A\right)\right)$. Then there exists a constant $c \sim -1$ so that for every $0 \leq r \leq R$,
		\begin{align*}
			&\left|\left(\frac{\partial}{\partial A}\right)^{r+1} G - c(-1)^r(1/3)(1/3+1)\cdots(1/3+r)t^{1/3}A^{-1/3-r}\right| \leq
			\\ \nonumber &\qquad \delta (1/3)(1/3+1)\cdots(1/3+r)t^{1/3}A^{-1/3-r}. 
		\end{align*}
		For a suitable constant $\tilde{c} \sim -1$, we also have that
		\[t\Phi\left(t,A,\omega_c\left(t,A\right)\right) = tm + \tilde{c} t^{1/3}A^{2/3} + O_{\LL}(1).\]
		
	\end{lemma}
	\begin{proof}This is a straightforward consequence of Lemma~\ref{estimatesfortildek} and~\eqref{tildek.def}.
	\end{proof}
	\subsection{Stationary Phase Estimates}
	The following estimate is a standard ``brute-force'' oscillatory integral estimate.
	\begin{lemma}\label{bruteforcestatphase}Let $\beta > 0$ and $F(\omega) : (m-\beta,m) \to \mathbb{C}$ be a $C^1$ function. Assume that there exists non-negative constants $B$, $p_0$, and $p_1$ such that for each $\alpha \in (0,\beta)$:
		\[ \left\vert\left\vert F \right\vert\right\vert_{L^{\infty}\left(m-\alpha,m\right)} \leq B \left({\rm min}\left(\alpha,1\right)\right)^{p_0},\qquad \left\vert\left\vert \partial_{\omega}F\right\vert\right\vert_{L^1\left(m-\beta,m-\alpha\right)} \leq B \alpha^{-p_1}.\]
		Then we have the following two estimates:
		\begin{enumerate}
			\item For all $t$ sufficiently large depending on $m$, $p_0$, and $p_1$: 
			\[\left|\int_{\beta}^m e^{-it\omega}F(\omega)\, d\omega\right| \lesssim B t^{-\frac{p_0+1}{p_0+p_1+1}}.\]
			\item For every $d > 0$ and $t$ sufficiently large depending on $m$, $p_0$, and $p_1$:
			\[\left|\int_{\beta}^{m-d}e^{-it\omega}F(\omega)\, d\omega\right| \lesssim B t^{-1}d^{-p_1}.\]
		\end{enumerate}
	\end{lemma}
	\begin{proof}
		For any $\alpha > 0$ we have the following straightforward estimates:
		\[\left|\int_{m-\alpha}^me^{-it\omega} F(\omega)\, dw\right| \leq B \alpha^{p_0+1},\qquad \left|\int_{m-\beta}^{m-\alpha}e^{-it\omega} F(\omega)\, d\omega\right| \leq \frac{2B}{t} + B t^{-1}\alpha^{-p_1}.\]
		One then writes
		\[\left|\int_{m-\beta}^me^{-it\omega}F(\omega)\, d\omega\right| \leq\left|\int_{m-\alpha}^me^{-it\omega}F(\omega)\, d\omega\right|+\left|\int_{m-\beta}^{m-\alpha}e^{-it\omega}F(\omega)\, d\omega\right|,\]
		for $\alpha = t^{-\frac{1}{p_1+p_0+1}}$.
	\end{proof}

	Next we recall some consequences of the Van de Corput lemma.
	\begin{lemma}\label{vancorput}Let $\phi(x) : (a,b) \to \mathbb{R}$ be smooth and suppose that there exists $k \in \mathbb{Z}_{\geq 1}$ and $\lambda > 0$ so that
		\begin{enumerate} 
			\item If $k = 1$, then $\phi'(x)$ is monotonic and ${\rm inf}_{(a,b)}\left|\phi'(x)\right| \geq \lambda$.
			\item If $k > 1$, then ${\rm inf}_{(a,b)}\left|\phi^{(k)}(x)\right| \geq \lambda$.
		\end{enumerate}
		
		Then, for every $C^1$ function $\psi(x) : (a,b) \to \mathbb{C}$, we have
		\[\left|\int_a^be^{i\phi(x)}\psi(x)\, dx\right| \lesssim \lambda^{-1/k}\left(\sup_{(a,b)}\left|\psi(x)\right| + \int_a^b\left|\psi'(x)\right|\, dx\right).\]
	\end{lemma}
	\begin{proof}See Chapter VIII of~\cite{BigStein}.
	\end{proof}

	This lemma is a variant of Lemma 6.6 from~\cite{KGSchw1}
	\begin{lemma}\label{precisestatphase}Let $A \gtrsim 1$, $t \gtrsim 1$, and $0 < d \ll 1$. Then choose any small constant $\delta > 0$ and any large constant $D > 0$. 
		
		We define the phase function
		\begin{equation}\label{thisisaphasefunction}		
			\Phi(t,A,\omega) \doteq \omega - \frac{A}{t}\kkq.
		\end{equation}
		
		Let $p \gg 1$ be suitably large. Suppose that $\frac{t}{A}$ are such that there exists a critical point $\omega_c(t,A) \in (m-\LL^{-p},m)$ where
		\[\partial_{\omega}\Phi(t,A,\omega_c(t,A)) = 0.\]
		We have the following 
		\[\omega_c(t,A) = m - \left(\frac{t}{A}\right)^{-2/3}\frac{(\pi M)^{2/3}m}{2} + O_{\delta}\left(\left(\frac{t}{A}\right)^{-\frac{4}{3}+\delta/3}\right),\]
		\[\partial_{\omega\omega}^2\Phi\left(t,A,\omega_c(t,A)\right) = -\left(\frac{t}{A}\right)^{2/3}3\cdot 2^{3/2}\frac{m^{3/2}}{\pi^{5/3}M^{2/3}} + O_{\delta}\left(\left(\frac{t}{A}\right)^{\delta/3}\right).\]
		
		Define
		$I \doteq \left[\omega_c - \left(\frac{t}{A}\right)^{-2/3 -\delta},\omega_c + \left(\frac{t}{A}\right)^{-2/3-\delta}\right]$ and let $0 < d \leq \LL^{-p}$. Then we have the following:
		\begin{enumerate}
			\item Suppose that  $\frac{t}{A}$ is sufficiently large, depending only on $m$ and $M$, and $I \subset (m-d,m)$. Then there exists a constant $c$ depending only on $m$ and $M$ so that  
			\begin{align*}			
				&  \int_{m-d}^me^{-it\omega} e^{iA \kkq}\, d\omega  
				\\ \nonumber &\qquad = c e^{-it\Phi\left(t,A,\omega_c\right)}t^{-1/2}\left(\partial_{\omega\omega}^2\Phi\left(t,A,\omega_c(t,A)\right)\right)^{-1/2}+O_{\delta,D}\left(t^{-1}\left(\frac{t}{A}\right)^{\delta}\right) + O_{D,\delta}\left(\left(\frac{t}{A}\right)^{-D}\right).
			\end{align*}
			
			\item Suppose that  $\partial_{\omega}\Phi\left(t,A,m-d\right) < 0$ (with no assumption made on the largeness of $\frac{t}{A}$). Then
			\begin{equation}\label{Iisnotin2}
				\left|\int_{m-d}^me^{-it\omega} e^{iA \kkq}\, d\omega\right| \lesssim \left|t\partial_{\omega}\Phi\left(t,A,m-d\right)\right|^{-1}. 
			\end{equation}
			
		\end{enumerate}
		
	\end{lemma}
	\begin{proof}We follow closely the proof of Lemma 6.6 from~\cite{KGSchw1}.
		
		We define the phase function
		\[\Phi \doteq \omega - \frac{A}{t}\kkq,\]
		so that we have
		\[\int_{m-d}^me^{-it\omega} e^{iA \kkq}\, d\omega = \int_{m-d}^me^{-it\Phi} \, d\omega.\]
		
		We have that
		\[\partial_{\omega}\Phi = 1 - \frac{A}{t}\left(\frac{Mm^2\omega}{\left(m^2-\omega^2\right)^{3/2}} + O\left(\left(m-\omega\right)^{-1/2-\delta/2}\right)\right),\]
		and also that
		\[\partial_{\omega\omega}^2\Phi = -\frac{A}{t}\left(\frac{3Mm^2\omega^2}{(m^2-\omega^2)^{5/2}} +O\left(\left(m-\omega\right)^{-3/2-\delta/2}\right)\right).\]
		Since $\partial_{\omega\omega}^2\Phi < 0$, we have that $\partial_{\omega}\Phi$ is monotonically decreasing. Thus~\eqref{Iisnotin2} follows from Lemma~\ref{vancorput}. From now on we assume that $\frac{t}{A}$ is sufficiently large, and consider the situation  when $I \subset (m-d,m)$. 
		
		We have that
		\[\partial_{\omega\omega}^2\Phi\left(\omega_c\right) = -\left(\frac{t}{A}\right)^{2/3}\underbrace{3 \cdot 2^{3/2} \frac{m^{3/2}}{\pi^{5/3}M^{2/3}}}_{\doteq c_0} + O\left(\left(\frac{t}{A}\right)^{\delta/3}\right).\]
		We also have that
		\[\sup_{\omega \in I}\left|\left(\frac{\partial}{\partial \omega}\right)^n\Phi\right| \lesssim \left(\frac{t}{A}\right)^{\frac{2(n-1)}{3}}.\]

		We note also that 
		\[\omega \in (m-d,m) \Rightarrow \partial_{\omega\omega}^2\Phi < 0.\]

		Then we have that
		\[\omega \in \left((m-d,m)\setminus I\right) \Rightarrow \left|\partial_{\omega}\Phi\right| \gtrsim \left(\frac{t}{A}\right)^{-\delta}.\]
		Thus Lemma~\ref{vancorput} implies that 
		\[\left|\int_{[m-d,m] \cap I^c}e^{-it\omega} e^{iA \kkq}\, d\omega \right| \lesssim t^{-1}\left(\frac{t}{A}\right)^{-\delta}.\]
		
		We now focus on the integral over $I$. Following the strategy from  Lemma 6.6 from~\cite{KGSchw1}, we choose $N$ so that $N\delta \gg M$ and observe that a Taylor expansion yields
		\[\Phi\left(t,A,\omega\right) = \Phi\left(t,A,\omega_c\right) + \frac{1}{2}\partial_{\omega\omega}^2\Phi\left(t,A,\omega_c\right)\left(\omega-\omega_c\right)^2 + \sum_{n=3}^N a_n(t)\left(\omega-\omega_c\right)^n + e_N\left(\omega,t\right),\]
		where
		\[\left|a_n\right| \lesssim \left(\frac{t}{A}\right)^{\frac{2(n-1)}{3}},\qquad \left|e_N\right| \lesssim \left(\frac{t}{A}\right)^{\frac{2N}{3}}\left(\omega-\omega_c\right)^{N+1}.\]
		We then introduce a new variable $x$ by
		\[x \doteq \sqrt{\frac{\Phi(t,A,\omega) - \Phi\left(t,A,\omega_c\right)}{\partial^2_{\omega\omega}\Phi\left(t,A,\omega_c\right)}}.\]
		Letting $x_{\pm} \doteq x\left(\omega_c \pm \left(\frac{t}{A}\right)^{-2/3-\delta}\right)$, we obtain that
		\begin{align*}
			&\int_I e^{-it\Phi}\, d\omega = 2e^{-it \Phi\left(t,A,\omega_c\right)}\int_{x_-}^{x_+}e^{-itx^2\partial_{\omega\omega}^2\Phi\left(t,A,\omega_c\right)}\underbrace{\frac{\left|\Phi\left(t,A,\omega\right) - \Phi\left(t,A,\omega_c\right)\right|^{1/2}\partial_{\omega\omega}^2\Phi\left(t,A,\omega_c\right)}{\partial_{\omega}\Phi\left(t,A,\omega\right)\left|\partial_{\omega\omega}^2\Phi\left(t,A,\omega_c\right)\right|^{1/2}}}_{\doteq \tilde{F}(t,x)}\, dx
			\\ \nonumber &=2\tilde{F}\left(t,0\right)e^{-it\Phi\left(t,A,\omega_c\right)}\int_{x_-}^{x_+}e^{-itx^2\partial_{\omega\omega}^2\Phi\left(t,A,\omega_c\right)}\, dx 
			\\ \nonumber &\qquad \qquad + 2e^{-it\Phi\left(t,A,\omega_c\right)}\int_{x_-}^{x_+}e^{-itx^2\partial_{\omega\omega}^2\Phi\left(t,A,\omega_c\right)}\left(\tilde{F}\left(t,x\right) - \tilde{F}\left(t,0\right)\right)\, dx
			\\ \nonumber &\doteq J_1 + J_2.
		\end{align*}
		
		We start with $J_1$. Noting that $x_{\pm} \sim \pm \left(\frac{t}{A}\right)^{-2/3-\delta}$, we introduce a variable $y = x \left(\frac{t}{A}\right)^{2/3+\delta}$, set \\$y_{\pm} \doteq \left(\frac{t}{A}\right)^{2/3+\delta}x_{\pm}$, argue as in Chapter VIII of~\cite{BigStein}, and obtain that
		\begin{align*}
			J_1 &= 2\tilde{F}\left(t,0\right) e^{-it \Phi\left(t,A,\omega_c\right)}\left(\frac{t}{A}\right)^{-2/3-\delta}\int_{y_-}^{y_+}e^{-it \left(\frac{t}{A}\right)^{-4/3-2\delta}y^2\partial_{\omega\omega}^2\Phi\left(t,A,\omega_c\right)}\, dy
			\\ \nonumber &= c e^{-it\Phi\left(t,A,\omega_c\right)}t^{-1/2}\left(\partial_{\omega\omega}^2\Phi\left(t,A,\omega_c\right)\right)^{-1/2} + O\left(t^{-1}\left(\frac{t}{A}\right)^{\delta}\right),
		\end{align*}
		for a constant $c$ depending only on $m$ and $M$.
		
		For $J_2$ (following still the proof of Lemma 6.6 from~\cite{KGSchw1}) we first note that
		\[\tilde{F}\left(t,x\right) - \tilde{F}\left(t,0\right) = \sum_{n=1}^{N-2}\tilde{a}_n(t)\left(\omega(x)-\omega_c\right)^n + \tilde{e}_{N-2}\left(t,\omega(x)\right),\]
		where
		\[\left|\tilde{a}_n\right| \lesssim_n \left(\frac{t}{A}\right)^{\frac{2n}{3}},\qquad \left|\tilde{e}_{N-2}\right| \lesssim_N \left(\frac{t}{A}\right)^{\frac{2(N-1)}{3}}\left(\omega(x)-\omega_c\right)^{N-1}.\]
		Thus,
		\[J_2 \lesssim_N \left(\sum_{n=1}^{N-1} Q_n\right) + K,\]
		where
		\[Q_n \doteq \left(\frac{t}{A}\right)^{\frac{2n}{3}}\left|\int_{x_-}^{x_+}e^{-itx^2 \partial_{\omega\omega}^2\Phi\left(t,A,\omega_c\right)}\left(\omega(x)-\omega_c\right)^n\, dx\right|,\qquad K \doteq \left|\int_{x_-}^{x_+}e^{-it x^2\partial^2_{\omega\omega}\Phi\left(t,A,\omega_c\right)}\tilde{e}_{N-2}\, dx\right|.\]
		Just as in the proof of Lemma 6.6 from~\cite{KGSchw1} we then introduce a new variable $u = tx^2\partial_{\omega\omega}^2\Phi\left(t,\omega_c\right)$, observe that $u(x_{\pm}) \sim \pm t\left(\frac{t}{A}\right)^{-2/3-2\delta}$, and then obtain
		\begin{align*}
			&\left|Q_n\right| \lesssim 
			\\ \nonumber &t^{-1/2 - n/2}\left(\frac{t}{A}\right)^{-1/3 + n/3}\left[\left|\int_{u(x_-)}^0e^{-it u}\frac{(\omega(x)-\omega_c)^n}{x^n}u^{\frac{n-1}{2}}\, du\right| + \left|\int_0^{u(x_+)}e^{-it u}\frac{(\omega(x)-\omega_c)^n}{x^n}u^{\frac{n-1}{2}}\, du\right|\right] \\ \nonumber &\lesssim \left(t^{-1/2 - n/2}\left(\frac{t}{A}\right)^{-1/3 + n/3}\right)\left(t^{\frac{n-1}{2}}\left(\frac{t}{A}\right)^{\left(-1/3-\delta\right)(n-1)}\right) \leq t^{-1}.
		\end{align*}

		Finally, for $K$ we have, using again the change of variables $u = tx^2\partial_{\omega\omega}^2\Phi$ and that $u(x_{\pm}) \sim \pm t\left(\frac{t}{A}\right)^{-2/3-2\delta}$:
		\begin{align*}
			\left|K\right| &\lesssim t^{-1/2-(N-1)/2}\left(\frac{t}{A}\right)^{-1/3+(N-1)/3}\int_{u(x_-)}^{u(x_+)}\left|u\right|^{\frac{N-2}{2}}\, du 
			\\ \nonumber &\lesssim  t^{-1/2-(N-1)/2}\left(\frac{t}{A}\right)^{-1/3+(N-1)/3}\left(t^{N/2}\left(\frac{t}{A}\right)^{(-1/3-\delta)N}\right)
			\\ \nonumber &\lesssim \left(\frac{t}{A}\right)^{-2/3 -\delta N}.
		\end{align*}
		The proof of the lemma follows taking $N\gg 1$ large enough relatively to $\delta$.
	\end{proof}
	
	\subsection{Estimates for the Fourier transform of $u$}
	The goal of this section is to prove the following proposition.
	\begin{prop}\label{estfourtransofumainprop}Let $S_0$ be a large positive constant and $p \geq 1$ be a sufficiently large constant. Suppose that $H\left(\omega,L,s\right) : \left(\mathbb{R}_{>0}\setminus\{m\}\right)\times \mathbb{Z}_{\geq 0} \times (-\infty,\infty) \to \mathbb{C}$ vanishes for $s \geq S_0$ and satisfies 
		\[\mathscr{D} \doteq \sup_{(\omega,L) \in \left(\mathbb{R}_{>0}\setminus \{m\} \right)\times \mathbb{Z}_{\geq 0}} \left(1+\omega^{10}\right)\left(1+L^{10+p}\right)\int_{-\infty}^{S_0}\left[\left|\partial_{\omega}\left(e^{i\omega s}H\right)\right|^2 +\left|\partial_sH\right|^2+ \left|H\right|^2\right]\left(r-r_+\right)^{-1}\, ds < \infty.\]
		Then suppose that $u\left(\omega,L,s\right) :  \left(\mathbb{R}_{>0}\setminus\{m\}\right) \times \mathbb{Z}_{\geq 0} \times (-\infty,\infty) \to \mathbb{C}$ solves~\eqref{eq:mainVrhs}, for $\omega > m$ satisfies the outgoing boundary conditions~\eqref{outgoingboundary}, and for $\omega < m$ satisfies the outgoing boundary conditions~\eqref{outgoingboundary2}. Then, for any $t \geq 1$, $R > r_+$, and $\check{\delta} > 0$, we have that 
		\begin{align*}		
			&\sup_{r \in [r_+,R]}\left|\int_0^{\infty}e^{-it^*\omega} \left(e^{i\omega p(s)}u\right)\, d\omega\right| \lesssim_{R,\check{\delta}} 
			\\ \nonumber &\qquad \left(1+L\right)^{-4}\mathscr{D} (t^*)^{-5/6+\check{\delta}}+1|_{\left\{t \geq e^{\check{\delta}\LL \log\LL}\right\}}\left(1+\LL\right)^{-4}\mathscr{D} \exp\left(-4\LL\log\LL + D_I\LL\right)(t^*)^{-5/6}\times 
			\\ \nonumber &\qquad \qquad \qquad \qquad\qquad \left|\sum_{q=1}^{\lceil (t^*)^{1/2}\rceil}e^{-it\Phi\left(t^*,2q,\omega_c\left(t^*,2q\right)\right)}\left((t^*)^{-2/3}\partial^2_{\omega\omega}\Phi\left(t^*,2q,\omega_c\left(t^*,2q\right)\right)\right)^{-1/2}\left(-\Gamma(m)\right)^q\right|,
		\end{align*}
		where $\Phi$ is the phase function from Lemma~\ref{precisestatphase}, $\Gamma(m)$ is as in Theorem~\ref{TPsection.mainprop}, and $1|_{\left\{t \geq e^{\check{\delta}\LL \log\LL}\right\}}$ is the indicator function for the set $\left\{t \geq e^{\check{\delta}\LL \log\LL}\right\}$. We also recall that $p(s)=t^{*}(s)-t$.
	\end{prop}
	
	We start with the (easier) region $\omega > m$.
	\begin{lemma}Let $u$ satisfy the hypothesis of Proposition~\ref{estfourtransofumainprop}. Then we have
		\[\sup_{r \in [r_+,R]}\left|\int_m^{\infty}e^{-it^*\omega}\left(e^{i\omega p(s)}u\right)\, d\omega\right| \lesssim_R \left(1+L\right)^{-4}t^{-1}\mathscr{D}.\]
	\end{lemma}
	\begin{proof}In view of Propositions~\ref{almosthereexcepthorizon} and~\ref{highfreqest} we have that
		\[\sup_{r \in [r_+,R]}\left(\int_m^{\infty}\left|\partial_{\omega}\left(e^{i\omega p(s)}u\right)\right|\, d\omega + \sup_{\omega \in (m,\infty)}\left(1+\omega^2\right)\left|u\right|\right) \lesssim \left(1+L\right)^{-4}\mathscr{D}.\]
		Thus the proof is concluded by a straightforward integration by parts.
	\end{proof}
	
	Now we turn to the region $\omega < m$ and $m-\omega \geq L^{-p}$ for $p \gg 1$.
	\begin{lemma}Let $u$ satisfy the hypothesis of Proposition~\ref{estfourtransofumainprop}. Then, for every small constant $\delta > 0$ and large constant $p \gg 1$, we have that
		\[\sup_{r \in [r_+,R]}\left|\int_0^{m-L^{-p}}e^{-it^*\omega}\left(e^{i\omega p(s)}u\right)\, d\omega\right| \lesssim_{\delta,p,R} \left(1+L\right)^{-4}\mathscr{D} (t^*)^{-1+\delta}.\]
	\end{lemma}
	\begin{proof}This proof will rely heavily on Proposition~\ref{themainpropinthebigthanL2}.
		
		Let $\tilde{\delta} > 0$ be an arbitrarily small constant. We consider separately the cases $t \leq e^{\tilde{\delta}\LL\log\LL}$ and $t \geq e^{\tilde{\delta}\LL\log \LL}$.
		
		Suppose that $t \leq e^{\tilde{\delta}\LL \log\LL}$. Let $d_0$ denote the number of intervals in the set $I_{\rm good}$ and $d_1$ denote the number of intervals in the set $I_{\rm bad}$ from Proposition~\ref{themainpropinthebigthanL2}. It follows from Proposition~\ref{themainpropinthebigthanL2} that
		\begin{equation}\label{boundintervalnumber}
			\left|d_0\right| + \left|d_1\right| \lesssim L^{p/2}.
		\end{equation}
		In view of the estimates~\eqref{bigL1},~\eqref{bigL2}, and~\eqref{bigL3} from Proposition~\ref{themainpropinthebigthanL2}, and~\eqref{boundintervalnumber}, we have that there exists a constant $D$ depending only on $m$ and $M$ so that 
		\begin{align}\label{tissmallandthisisnothard}
			\left|\int_0^{m-L^{-p}}e^{-it^*\omega}\left(e^{i\omega p(s)}u\right)\, d\omega\right| &\leq \left|\int_{[0,m-L^{-p}]\cap I_{\rm good}}e^{-it^*\omega}\left(e^{i\omega p(s)}u\right)\, d\omega\right|
			\\ \nonumber &\qquad + \left|\int_{[0,m-L^{-p}]\cap I_{\rm bad}}e^{-it^*\omega}\left(e^{i\omega p(s)}u\right)\, d\omega\right|
			\\ \nonumber &\lesssim \left(1+L\right)^{-4}\mathscr{D}\left( (t^*)^{-1} +e^{D\LL}e^{-\frac{1}{2}\LL\log\LL}\right)
			\\ \nonumber &\lesssim \left(1+L\right)^{-4}\mathscr{D} (t^*)^{-1}
		\end{align}
		where in the passage from the first to the second line we integrated by parts on $I_{\rm good}$ and simply used an $L^{\infty}$ bound for $u$ on $I_{\rm bad}$ (as well as the definition of $I_{\rm bad}$).
		
		Next we consider the case when $t \geq e^{\tilde{\delta} \LL \log \LL}$. In this case it does not suffice to estimate the integral over $I_{\rm bad}$ using only the $L^{\infty}$ bound~\eqref{bigL3} from Proposition~\ref{themainpropinthebigthanL2}. Let $\tilde{I}$ denote one of the intervals in $I_{\rm bad}$. In view of~\eqref{bigL4} and~\eqref{bigL5} from Proposition~\ref{themainpropinthebigthanL2} (and still keeping~\eqref{boundintervalnumber} in mind) we have that, for any sufficiently small constant $\check{\delta} > 0$, 
		\begin{align*}
			&d_1\left|\int_{\tilde{I}}e^{-it^*\omega}\left(e^{i\omega p(s)}u\right)\, d\omega\right| 
			\\ \nonumber &\qquad \lesssim d_1\left(t^*\right)^{-1}\int_{\tilde{I}}\left|\partial_{\omega}\left(e^{i\omega s}u\right)\right|\, d\omega + \left(1+L\right)^{-4}\mathscr{D} \left(t^*\right)^{-1}
			\\ \nonumber &\qquad\lesssim 
			d_1 \left(t^*\right)^{-1}\left(\int_{\tilde{I}\cap\left\{ \left|m^2-\omega^2+\lambda_n\right| \geq e^{-2\LL \log\LL}\right\}}\left|\partial_{\omega}\left(e^{i\omega s}u\right)\right|\, d\omega
			+\int_{\tilde{I}\cap \left\{\left|m^2-\omega^2+\lambda_n\right| \leq e^{-2\LL \log\LL}\right\}}\left|\partial_{\omega}\left(e^{i\omega s}u\right)\right|\, d\omega\right)
			\\ \nonumber & \qquad \qquad + \left(1+L\right)^{-4}\mathscr{D} \left(t^*\right)^{-1} 
			\\ \nonumber & \qquad  \lesssim_{\check{\delta}} \left(1+L\right)^{-4}\mathscr{D}\left(t^*\right)^{-1}\Bigg(1
			+ e^{-2(1-\check{\delta})\LL \log\LL}\int_{\tilde{I}\cap\left\{ \left|m^2-\omega^2+\lambda_n\right| \geq e^{2\LL \log\LL}\right\}}\left|m^2-\omega^2+\lambda_n\right|^{-2}\, d\omega 
			\\  \nonumber &  \qquad \qquad \qquad +e^{2\left(1+\check{\delta}\right)\LL\log\LL}\int_{\tilde{I}\cap \left\{\left|m^2-\omega^2+\lambda_n\right| \leq e^{-2\LL \log\LL}\right\}}\, d\omega\Bigg)
			\\ \nonumber &\qquad \lesssim \left(1+L\right)^{-4} \mathscr{D} \left(t^*\right)^{-1} e^{-2\check{\delta}\LL \log\LL}.
		\end{align*}
		Thus, as long as $\check{\delta}$ is sufficiently small compared to $\tilde{\delta}$, we have that 
		\[d_1\left|\int_{\tilde{I}}e^{-it^*\omega}\left(e^{i\omega p(s)}u\right)\, d\omega\right| \lesssim \left(1+L\right)^{-4} \mathscr{D} \left(t^*\right)^{-1+\tilde{\delta}},\]
		which in turn implies that
		\[\left|\int_{[0,m-L^{-p}]\cap I_{\rm bad}}e^{-it^*\omega}\left(e^{i\omega p(s)}u\right)\, d\omega\right| \lesssim \left(1+L\right)^{-4} \mathscr{D} \left(t^*\right)^{-1+\tilde{\delta}}.\]
		Since $\left|\int_{[0,m-L^{-p}]\cap I_{\rm good}}e^{-it^*\omega}\left(e^{i\omega p(s})u\right)\, d\omega\right|$ may be estimated just as we did in~\eqref{tissmallandthisisnothard}, the proof is concluded. 
	\end{proof}
	
	Finally, we come to the region $\omega \in [m-L^{-p},m]$. We will consider separately the cases when $t \leq e^{\tilde{\delta}\LL \log\LL}$ for a small constant $\tilde{\delta} > 0$ and when $t \geq e^{\tilde{\delta}\LL \log\LL}$. An important role will be played by the Green's formula~\eqref{secondGreenFormula} from Theorem~\ref{TPsection.mainprop}. Setting
	\begin{equation}\label{thisisumain}
		u_{\rm main} \doteq e^{i\omega p(s)}\frac{ e^{i\kkq}\sum_{i \in \{1,2\},\ j \in \{\pm\}}u_{i,j}^{reg}(\omega,s)}{1+\Gamma(m) e^{2i \kkq}},
	\end{equation}
	\begin{equation}\label{thisisuerror1}
		u_{\rm error1} \doteq e^{i\omega p(s)}\frac{ e^{i\kkq}\sum_{i \in \{1,2\},\ j \in \{\pm\}}\ep_{i,j}^{cos}(\omega,s) \cos(\kkq)}{1+\Gamma(m)\cdot e^{2i\kkq}}, 
	\end{equation}
	\begin{equation}\label{thisisuerror2}
		u_{\rm error2} \doteq e^{i\omega p(s)}\frac{ e^{i\kkq} \sum_{i \in \{1,2\},\ j \in \{\pm\}}\ep_{i,j}^{sin}(\omega,s) \sin(\kkq)}{1+\Gamma(m)\cdot e^{2i\kkq}}, 
	\end{equation}
	we have that
	\[ e^{i\omega p(s)}u = u_{\rm main} + u_{\rm error1}+u_{\rm error2}.\]
	(Recall that $p(s)$ is defined by~\eqref{thisisps}.)
	
	In the following lemma we gather some estimates for the expressions $\Gamma(m)$ and  $1+\Gamma(m)e^{2i \kkq}$ which show up in the denominators of~\eqref{thisisumain},~\eqref{thisisuerror1}, and~\eqref{thisisuerror2}.
	\begin{lemma}\label{someestimatesforgamma}Let $\Gamma(m)$ be as in Theorem~\ref{TPsection.mainprop}. Then we may write 
		\[\Gamma(m) = r_{\Gamma}e^{i\theta_{\Gamma}},\]
		for $r_{\Gamma} \in [0,\infty)$ and $\theta_{\Gamma} \in [0,2\pi)$ with
		\[r_{\Gamma} = 1 - c\left(\LL\right)\exp\left(-4\LL \log\LL + D_I\LL\right),\]
		\[\left|\theta_{\Gamma}-\pi\right| = \exp\left(-4\LL \log\LL + D_I\LL\right),\]
		where $c\left(\LL\right) \sim 1$, $\hat{D}\left(\LL\right) = O(1)$ and $D_I$ is as in Theorem~\ref{TPsection.mainprop}.
		
		Next assume that for some integer $n$, we have that $\left|\check{k} - n\right| \ll 1$, where the implied constant depends only on $m$ and $M$. Then, we have
		\[\Re\left(1+\Gamma e^{2i\kkq}\right) \geq c\exp\left(-4 \LL\log\LL + D_I\LL\right),\]
		\[\left|\Im\left(1+\Gamma e^{2i \kkq}\right)\right| \geq  \left| \check{k}-n\right| +  O\left(\exp\left(-4\LL \log\LL + D_I\LL\right)\right).\]
	\end{lemma}
	\begin{proof}These are immediate consequences of Theorem~\ref{TPsection.mainprop}.
	\end{proof}
	
	Now we consider $u_{\rm error1}$ in the case when $t \leq e^{\tilde{\delta} \LL \log\LL}$.
	\begin{lemma}\label{uerror1smalltime}Let $u$ satisfy the hypothesis of Proposition~\ref{estfourtransofumainprop} define $u_{\rm error1}$ by~\eqref{thisisuerror1}, and let $\tilde{\delta} > 0$ be any small constant. Then, $t^* \leq e^{\tilde{\delta} \LL \log\LL}$ implies that
		\[\sup_{r \in [r_+,R]}\left|\int_{m-L^{-p}}^me^{-it^*\omega}u_{\rm error1}\, d\omega\right| \lesssim_R (t^*)^{-1}\left(1+\LL\right)^{-4}\mathscr{D}.\]
	\end{lemma}
	\begin{proof} We let all constants in this section depend on $R$. Define a set $X_{\rm bad}$ by
		\[X_{\rm bad} \doteq \left\{\omega \in (m-\LL^{-p},m) : {\rm inf}_{n \in \mathbb{Z}}\left|\check{k} - n\right| \leq \exp\left(-\frac{1}{2}\LL \log\LL\right) \right\}.\]
		Then we set $X_{\rm good} \doteq (m-\LL^{-p},m)\setminus X_{\rm bad}$. In view of Theorem~\ref{TPsection.mainprop}, we have that
		\[\omega \in X_{\rm good} \Rightarrow \left|u_{\rm error1}\right| \lesssim \left(1+\LL\right)^{-4}\mathscr{D} \exp\left(-\LL\log\LL\right).\]
		For $X_{\rm bad}$ we have 
		\[\omega \in X_{\rm bad} \Rightarrow \left|u_{\rm error1}\right| \lesssim \left(1+\LL\right)^{-4}\mathscr{D} e^{D\LL},\]
		for an $\LL$ independent constant $D$. Since the total $\omega$-volume of the set $X_{\rm bad}$ may be estimated as follows:
		\[\int_{X_{\rm bad}}\, d\omega \lesssim \exp\left(-\frac{1}{2}\LL\log\LL\right),\]
		we easily obtain that
		\[\left|\int_{m-\LL^{-p}}^me^{-it^*\omega}u_{\rm error1}\, d\omega\right| \lesssim \mathscr{D} e^{-\frac{1}{4}\LL \log\LL} \lesssim_{\tilde{\delta}} \left(1+\LL\right)^{-4}\mathscr{D} (t^*)^{-1}.\]
	\end{proof}

	Now we consider $u_{\rm error2}$. Here we can treat all ranges of $t^*$ at once. 
	\begin{lemma}\label{uerror2isokwithme}Let $u$ satisfy the hypothesis of Proposition~\ref{estfourtransofumainprop}, and define $u_{\rm error2}$ by~\eqref{thisisuerror2}. Then
		\[\sup_{r \in [r_+,R]}\left|\int_{m-L^{-p}}^me^{-it^*\omega}u_{\rm error2}\, d\omega\right| \lesssim_R \left(1+\LL\right)^{-4}(t^*)^{-5/6}\mathscr{D}.\]
	\end{lemma}
	\begin{proof}We let all constants in this section depend on $R$.

		Now, let $\hat{\delta} > 0$ be a sufficiently small constant (independent of $\LL$) and consider the sets $X_{\rm good}$ and $X_{\rm bad}$ defined by
		\[X_{\rm bad} \doteq \left\{\omega \in (m-\LL^{-p},m) : {\rm inf}_{n \in \mathbb{Z}}\left|\check{k} - n\right| \leq \delta \right\},\qquad  X_{\rm good} \doteq (m-\LL^{-p},m)\setminus X_{\rm bad}.\]

		We start by noting that in view of Theorem~\ref{TPsection.mainprop} and Lemma~\ref{someestimatesforgamma} we have that
		\begin{equation}\label{linftyuerror2bound}		
			\left|u_{\rm error2}\right| \lesssim \left(1+\LL\right)^{-4}\left(m-\omega\right)^{1/4}\mathscr{D} .
		\end{equation}
		For the derivative of $u_{\rm error2}$ we have 
		\[\left|\partial_{\omega}u_{\rm error2}\right| \lesssim \left(m-\omega\right)^{-5/4}\left(1+\frac{1}{\left|1+\Gamma(m) e^{2i \kkq}\right|}+\frac{|\sin(\kkq)|}{\left|1+\Gamma(m) e^{2i \kkq}\right|^2} \right)\left(1+\LL\right)^{-6}\mathscr{D}.\]

		We have that
		\[\omega \in X_{\rm good} \Rightarrow \left|\partial_{\omega}u_{\rm error2}\right| \lesssim \left(m-\omega\right)^{-5/4}\left(1+\LL\right)^{-4}\mathscr{D}.\]
		In particular, for every $d \in (0,\LL^{-p})$ we have that
		\begin{equation}\label{goodestimateuerror2der}		
			\int_{(m-\LL^{-p},m-d) \cap X_{\rm good}}\left|\partial_{\omega}u_{\rm error2}\right|\, d\omega \lesssim \left(1+\LL\right)^{-4}d^{-1/4}\mathscr{D}.
		\end{equation}
		
		For $\omega \in X_{\rm bad}$ we may change variables from $\omega $ to $\check{k}$ and obtain, for suitable constants $c_1$ and $c_2$ which may be uniformly bounded above and below depending only $M$ and $m$:
		\begin{align}\label{changethevariableforl1}
			&\int_{(m-\LL^{-p},m-d)\cap X_{\rm bad}}\left|\partial_{\omega}u_{\rm error2}\right|\, d\omega \lesssim 
			\\ \nonumber &\qquad \left(1+\LL\right)^{-6}\mathscr{D} \sum_{n=c_1 \LL^{p/2}}^{c d^{-1/2}}n^{-1/2}\int_{n-\delta}^{n+\delta}\left(1+\frac{1}{\left|1+\Gamma(m) e^{2i \kkq}\right|}+\frac{|\sin(\kkq)|}{\left|1+\Gamma(m) e^{2i \kkq}\right|^2} \right)\, d\check{k}.
		\end{align}
		
		We have, keeping Lemma~\ref{someestimatesforgamma} in mind,
		\begin{align}
			&\int_{n-\delta}^{n+\delta}\frac{d\check{k}}{\left|1+\Gamma(m) e^{2\pi i \check{k}}\right|} =
			\\ \nonumber &\left(\int_{n-\LL^2e^{-4\LL \log\LL + D_I\LL}}^{n+\LL^2e^{-4\LL \log\LL + D_I\LL}}+\int_{n-\delta}^{n-\LL^2e^{-4\LL \log\LL + D_I\LL}} + \int_{n+\LL^2e^{-4\LL \log\LL + D_I\LL}}^{n+\delta}\right)\frac{d\check{k}}{\left|1+\Gamma(m) e^{2\pi i \check{k}}\right|} \lesssim 
			\\ \nonumber &c\exp\left(4\LL\log\LL + D_I\LL\right)\int_{n-\LL^2e^{-4\LL \log\LL + D_I\LL}}^{n+\LL^2e^{-4\LL \log\LL + D_I\LL}}\, d\check{k}
			\\ \nonumber &\qquad + \left(\int_{n-\delta}^{n-\color{red}\LL^2\color{black}e^{-4\LL \log\LL + D_I\LL}} + \int_{n+\color{red}\LL^2\color{black}e^{-4\LL \log\LL + D_I\LL}}^{n+\delta}\right)\frac{d\check{k}}{\left|\check{k}-n\right|} \lesssim \LL^2.
		\end{align}
		Similarly,
		\[\int_{n-\delta}^{n+\delta}\frac{|\sin(\kkq)| d\check{k}}{\left|1+\Gamma(m)e^{2i\kkq}\right|^2} \lesssim \LL^2.\]
		Plugging in these bounds into~\eqref{changethevariableforl1} leads to
		\begin{align}\label{changethevariableforl1more}
			&\int_{(m-\LL^{-p},m-d)\cap X_{\rm bad}}\left|\partial_{\omega}u_{\rm error2}\right|\, d\omega \lesssim \left(1+\LL\right)^{-4}\mathscr{D} \sum_{n = c_1\LL^{p/2}}^{cd^{-1/2}}n^{-1/2} \lesssim \left(1+\LL\right)^{-4}d^{-1/4}\mathscr{D}.
		\end{align}
		
		Putting together~\eqref{goodestimateuerror2der} and~\eqref{changethevariableforl1more} yields
		\begin{align}\label{changethevariableforl1more223}
			&\int_{(m-\LL^{-p},m-d)}\left|\partial_{\omega}u_{\rm error2}\right|\, d\omega \lesssim  \left(1+\LL\right)^{-4}d^{-1/4}\mathscr{D}.
		\end{align}
		
		In view of~\eqref{linftyuerror2bound},~\eqref{changethevariableforl1more223}, and applying Lemma~\ref{bruteforcestatphase} with $p_0=p_1=\frac{1}{4}$, the proof is concluded.

	\end{proof}
	
	Now we consider $u_{\rm error1}$ in the case when $t^* \geq e^{\tilde{\delta} \LL \log\LL}$.
	\begin{lemma}Let $u$ satisfy the hypothesis of Proposition~\ref{estfourtransofumainprop} define $u_{\rm error1}$ by~\eqref{thisisuerror1}, and let $\tilde{\delta} > 0$ be any small constant. Suppose that $t^* \geq e^{\tilde{\delta} \LL \log\LL}$ and $\delta > 0$ is an arbitrary small constant. Assuming $\LL$ is large enough, depending on $\delta$, then
		\[\sup_{r \in [r_+,R]}\left|\int_{m-L^{-p}}^me^{-it^*\omega}u_{\rm error1}\, d\omega\right| \lesssim_R \left(1+\LL\right)^{-4}(t^*)^{-5/6+\delta}\mathscr{D}.\]
	\end{lemma}
	\begin{proof}We let all constants in this section depend on $R$.
		
		In view of Theorem~\ref{TPsection.mainprop} and Lemma~\ref{someestimatesforgamma} we have that
		\begin{equation}\label{uerror2hasanupperbound}
			\left|u_{\rm error1}\right| \lesssim \left(m-\omega\right)^{1/4}\mathscr{D} e^{D\LL},
		\end{equation}
		for some constant $D = D(\LL) = O(1)$.
		
		For $\partial_{\omega}u_{\rm error1}$ we have 
		\[\left|\partial_{\omega}u_{\rm error1}\right| \lesssim \left(m-\omega\right)^{-5/4}\frac{e^{D \LL-4\LL \log\LL + D_H \LL}}{\left|1+\Gamma(m) e^{2i \kkq}\right|^2}\mathscr{D}.\]
		
		In particular, we may argue just as we did to establish~\eqref{changethevariableforl1more223} in the proof of Lemma~\ref{uerror2isokwithme} to establish that  for some constant $D'=D'(\LL)=O(1)>D$:
		\begin{align}\label{changethevariableforl1more223uerror2}
			&\int_{(m-\LL^{-p},m-d)}\left|\partial_{\omega}u_{\rm error1}\right|\, d\omega \lesssim  d^{-1/4}e^{D'\LL}\mathscr{D}.
		\end{align}
		Then we may combine~\eqref{uerror2hasanupperbound}, ~\eqref{changethevariableforl1more223uerror2}, and Lemma~\ref{bruteforcestatphase} to establish that 
		\[\left|\int_{m-L^{-p}}^me^{-it\omega}u_{\rm error1}\, d\omega\right| \lesssim (t^*)^{-5/6}\mathscr{D} e^{D'\LL} \lesssim \left(1+\LL\right)^{-4}(t^*)^{-5/6+\delta}\mathscr{D}.\]
	\end{proof}
	
	Next we consider $u_{\rm main}$ in the case when $t^* \leq e^{\tilde{\delta}\LL\log\LL}$.
	\begin{lemma}\label{mainsmalltfourier}Let $u$ satisfy the hypothesis of Proposition~\ref{estfourtransofumainprop} define $u_{\rm main}$ by~\eqref{thisisuerror1}, and let $\tilde{\delta} > 0$ be any small constant. Then, $t^* \leq e^{\tilde{\delta} \LL \log\LL}$ implies that
		\[\sup_{r \in [r_+,R]}\left|\int_{m-\LL^{-p}}^me^{-it^*\omega}u_{\rm main}\, d\omega\right| \lesssim_{\tilde{\delta},R} (t^*)^{-1}\left(1+\LL\right)^{-4}\mathscr{D}. \]
	\end{lemma}
	\begin{proof}We let all constants in this section depend on $R$.
		
		Set
		\[H_{1,\pm} \doteq \frac{e^{i\omega s}u_H\left(m,s\right)}{C_W\left(m,\LL\right)}\int_s^{+\infty} w_{1,\pm}\left(m,s'\right)H\left(m,s'\right)\, ds',\qquad H_{2,\pm} \doteq \frac{e^{i\omega s}w_{1,\pm}\left(m,s\right)}{C_W\left(m,\LL\right)}\int_{-\infty}^su_H\left(m,s'\right) H\left(m,s'\right)\, ds'.\]
		In view of Theorem~\ref{TPsection.mainprop}, there exists a constant $D$ so that we have the estimates
		\begin{equation}\label{estiamtesh1plus}
			\left|H_{1,+}\right| + \left|H_{2,+}\right| \lesssim e^{D\LL}\mathscr{D},
		\end{equation}
		\begin{equation}\label{estimateh1minus}
			\left|H_{1,-}\right| + \left|H_{2,-}\right| \lesssim \left(1+\LL\right)^{-4}\mathscr{D}.
		\end{equation}
		
		We now proceed to re-write~\eqref{thisisumain} as a suitable sum of exponentials. We start with the identity
		\[\frac{1}{1+ \Gamma(m) e^{2i\kkq}} = \sum_{q=0}^{\infty} \left(-\Gamma(m)\right)^qe^{2i\kkq q},\]
		and also write
		\[\cos\left(\kkq\right) = \frac{1}{2}\left(e^{i\kkq}+e^{-i\kkq}\right),\qquad \sin\left(\kkq\right) = \frac{1}{2i}\left(e^{i\kkq} - e^{-i\kkq}\right).\]
		We also note the identity
		\[\left(e^{2i \kkq} - 1\right)\sum_{q=0}^{\infty}\left(-\Gamma(m)\right)^qe^{2i\kkq q} = \left(\left(-\Gamma(m)\right)^{-1}-1\right)\sum_{q=1}^{\infty}\left(-\Gamma(m)\right)^qe^{2i \kkq q} - 1.\]
		After some simplification, and keeping~\eqref{estiamtesh1plus} and~\eqref{estimateh1minus} in mind, we find that we can write 
		\begin{equation}\label{intothegeometricseries}
			u_{\rm main} = A_1 + A_2\sum_{q=1}^{\infty}\left(-\Gamma(m)\right)^q e^{2i \kkq q},
		\end{equation}
		where $A_1$ and $A_2$ are constants satisfying
		\begin{equation}\label{boundA1yay}
			\left|A_1\right| \lesssim \left(1+\LL\right)^{-4}\mathscr{D},
		\end{equation}
		\begin{equation}\label{boundA2yay}
			\left|A_2\right| \lesssim \exp\left(-4\LL \log\LL + D_I\LL\right)e^{\tilde{D}\LL}\mathscr{D},
		\end{equation}
		for some constant $\tilde{D}$.
		
		Of course we have that
		\begin{equation}\label{trivialA1}
			\left|\int_{m-\LL^{-p}}^me^{-it^*\omega} A_1\, d\omega\right| \lesssim (t^*)^{-1}\left(1+\LL\right)^{-4}\mathscr{D}.
		\end{equation}
		Thus we focus on the Fourier transform of the term $A_2\sum_{q=1}^{\infty}\left(-\Gamma(m)\right)^q e^{2i \kkq q}$. Let $\hat{C}$ be a sufficiently large constant and write
		\[A_2\sum_{q=1}^{\infty}\left(-\Gamma(m)\right)^qe^{2i\kkq q} = A_2 \sum_{q=1}^{\lceil\hat{C}t\rceil}\left(-\Gamma(m)\right)^qe^{2i \kkq q} + A_2\sum_{q=\lceil\hat{C}t\rceil+1}^{\infty}\left(-\Gamma(m)\right)^qe^{2i \kkq q}.\]
		Keeping in mind that $t^* \leq e^{\tilde{\delta} \LL \log\LL}$ we have that
		\begin{align*}
			\left|A_2\int_{m-\LL^{-p}}^me^{-it^*\omega} \sum_{q=1}^{\lceil \hat{C} t^*\rceil}\left(-\Gamma(m)\right)^qe^{2i \kkq q}\, d\omega\right| &\lesssim \left|A_2\right|\left(\hat{C}t^*+1\right)
			\\ \nonumber &\lesssim \hat{C}e^{-4\LL\log\LL + D_I\LL}e^{\tilde{D}\LL}e^{\tilde{\delta}\LL\log \LL}\mathscr{D}
			\\ \nonumber &\lesssim_{\tilde{\delta}} \left(1+\LL\right)^{-4}(t^*)^{-1}\mathscr{D}.
		\end{align*}
		For the second part of the sum we recall the definition of the phase function $\Phi$ from~\eqref{thisisaphasefunction} and note that $q \geq \lceil \hat{C} t^*\rceil + 1$ implies that 
		\[t^*\partial_{\omega}\Phi\left(t^*,2q,m-\LL^{-p}\right) \lesssim -q\LL^{\frac{3p}{2}}.\]
		In particular, Lemma~\ref{precisestatphase} implies that 
		\begin{align*}
			\left|A_2\int_{m-\LL^{-p}}^me^{-it^*\omega}\sum_{q = \lceil \hat{C}t^*\rceil + 1}^{\infty}\left(-\Gamma(m)\right)^qe^{2i\kkq}\, d\omega\right| &\lesssim \exp\left(-4\LL \log\LL + D_I\LL\right)e^{\tilde{D}\LL}\mathscr{D}\sum_{q=\lceil \hat{C} t^*\rceil + 1}^{\infty}q^{-1}\left(-\Gamma(m)\right)^q
			\\ \nonumber &\lesssim \mathscr{D} \exp\left(-\LL\log \LL\right)
			\\ \nonumber &\lesssim_{\tilde{\delta}} \left(1+\LL\right)^{-4}(t^*)^{-1}\mathscr{D}.
		\end{align*}

	\end{proof}
	
	Finally, we come to $u_{\rm main}$ in the case when $t^* \geq e^{\tilde{\delta}\LL \log \LL}$.
	\begin{lemma}Let $u$ satisfy the hypothesis of Proposition~\ref{estfourtransofumainprop} define $u_{\rm main}$ by~\eqref{thisisuerror1}, and let $\delta, \tilde{\delta} > 0$ be any small constant. Then, there exists a constant $D > 0$ so that $t^* \geq e^{\tilde{\delta} \LL \log\LL}$ implies that
		\begin{align*}		
			&\sup_{r \in [r_+,R]}\left|\int_0^{\infty}e^{-it^*\omega} u_{\rm main}\, d\omega\right| \lesssim_{R,\tilde{\delta}} 
			\\ \nonumber &\qquad \left(1+L\right)^{-4}\mathscr{D} (t^*)^{-5/6+\delta}+\left(1+\LL\right)^{-4}\mathscr{D} \exp\left(-4\LL\log\LL + D_I\LL\right)e^{D\LL}(t^*)^{-5/6}\times 
			\\ \nonumber &\qquad \qquad \qquad \qquad\qquad \qquad \left|\sum_{q=1}^{\lceil (t^*)^{1/2}\rceil}e^{-it\Phi\left(t^*,2q,\omega_c\left(t^*,2q\right)\right)}\left((t^*)^{-2/3}\partial^2_{\omega\omega}\Phi\left(t^*,2q,\omega_c\left(t^*,2q\right)\right)\right)^{-1/2}\left(-\Gamma(m)\right)^q\right|,
		\end{align*}
		where $\Phi$ is the phase function from Lemma~\ref{precisestatphase}
	\end{lemma}
	\begin{proof}We let all constants in this section depend on $R$.
		
		We fist note that the steps in the proof of Lemma~\ref{mainsmalltfourier} which lead to~\eqref{intothegeometricseries},~\eqref{boundA1yay},~\eqref{boundA2yay}, and~\eqref{trivialA1} go through also in the setting of this lemma.
		
		We thus focus on estimating
		\[\left|A_2\int_{m-\LL^{-p}}^m e^{-it\omega} \sum_{q=1}^{\infty}\left(-\Gamma(m)\right)^qe^{2i\kkq}\, d\omega\right|,\]
		where
		\[\left|A_2\right| \lesssim \exp\left(-4\LL\log\LL + D_1\LL\right)e^{\tilde{D}\LL}.\]
		
		We write 
		\begin{align*}
			&\left|A_2\int_{m-\LL^{-p}}^m e^{-it\omega} \sum_{q=1}^{\infty}\left(-\Gamma(m)\right)^qe^{2i\kkq}\, d\omega\right| \lesssim 
			\underbrace{\left|A_2\int_{m-\LL^{-p}}^m e^{-it\omega} \sum_{q=1}^{\lceil (t^*)^{1/2} \rceil}\left(-\Gamma(m)\right)^qe^{2i\kkq}\, d\omega\right|}_{\doteq I}
			\\ \nonumber &\qquad +\underbrace{\left|A_2\int_{m-\LL^{-p}}^m e^{-it\omega} \sum_{q=\lceil  (t^*)^{1/2}\rceil + 1}^{\infty}\left(-\Gamma(m)\right)^qe^{2i\kkq}\, d\omega\right|}_{\doteq II}.
		\end{align*}
		
		We will now estimate $I$ and $II$ separately. We start with $I$. Directly applying Lemma~\ref{precisestatphase} leads to 
		\begin{align}\label{thisisthestimateofrIfromstatphase}
			&\left|I\right| \lesssim \left(1+\LL\right)^{-4}\mathscr{D} \exp\left(-4\LL\log\LL + D_I\LL\right)e^{D\LL}(t^*)^{-5/6}\times 
			\\ \nonumber &\qquad \qquad \qquad \qquad\qquad \left|\sum_{q=1}^{\lceil (t^*)^{1/2}\rceil}e^{-it\Phi\left(t^*,2q,\omega_c\left(t^*,2q\right)\right)}\left((t^*)^{-2/3}\partial^2_{\omega\omega}\Phi\left(t^*,2q,\omega_c\left(t^*,2q\right)\right)\right)^{-1/2}\left(-\Gamma(m)\right)^q\right|.
		\end{align}

		We now turn to the term $II$. Set
		\[F\left(\omega\right) \doteq A_2\sum_{q=\lceil (t^*)^{1/2}\rceil+1}^{\infty}\left(-\Gamma(m)\right)^qe^{2i\kkq}.\]
		We have the following estimates for $F$ and its derivatives:
		\begin{equation}\label{estimatesforFanditsderivative}
			\left|F\right| \lesssim e^{\tilde{D}\LL}\mathscr{D},\qquad \left|\partial_{\omega}F\right| \lesssim \left(m-\omega\right)^{-3/2}e^{\tilde{D}\LL}\mathscr{D}.
		\end{equation}
		Let $\hat{c} > 0$ be a small constant. Lemma~\ref{bruteforcestatphase} and the bounds~\eqref{estimatesforFanditsderivative} imply that
		\begin{equation}\label{startwithbruteforceforII}
			\left|\int_{m-\LL^{-p}}^{m-\hat{c} t^{-1/3}}e^{-it^*\omega}F(\omega)\, d\omega\right| \lesssim_{\hat{c}} \mathscr{D} (t^*)^{-5/6} e^{\tilde{D}\LL} \lesssim_{\hat{c}, \tilde{\delta},\delta} \left(1+\LL\right)^{-4}\mathscr{D} (t^*)^{-5/6+\delta}.
		\end{equation}
		
		On the other hand, if $\hat{c} > 0$ is sufficiently small, then we have that 
		\begin{equation}\label{thephaseisnottoosmallwow}
			q \geq \lceil (t^{*})^{1/2}\rceil \Rightarrow t^{*}\partial_{\omega}\Phi\left(t^*,2q,m-\hat{c}(t^{*})^{-1/3}\right) \lesssim -q (t^{*})^{1/2}.
		\end{equation}
		We then fix the value of $\hat{c}$ so that~\eqref{thephaseisnottoosmallwow} holds.
		
		In view of~\eqref{startwithbruteforceforII},~\eqref{thephaseisnottoosmallwow}, and Lemma~\ref{precisestatphase} we thus obtain that
		\begin{align*}
			\left|II\right| &\lesssim \left(1+\LL\right)^{-4}\mathscr{D} (t^*)^{-5/6+\delta} + \mathscr{D} e^{\tilde{D}\LL}\exp\left(-4\LL\log\LL + D_I\LL\right) (t^{*})^{-1/2}\sum_{q=\lceil (t^{*})^{1/2}\rceil+1}^{\infty}q^{-1}\left(-\Gamma(m)\right)^q
			\\ \nonumber &\lesssim \left(1+\LL\right)^{-4}\mathscr{D} (t^*)^{-5/6+\delta}, 
		\end{align*}
		which completes the proof.
	\end{proof}
	
	We note that by combining all of the Lemmas of this section, we have  proven Proposition~\ref{estfourtransofumainprop}.
	\subsection{Proof of the main theorem}
	In this section we will finally establish our main results.
	
	We start by estimating the exponential sum which appears in Proposition~\ref{estfourtransofumainprop}.
	\begin{lemma}\label{weusethistocontrolthesum}Let $\Phi$ be the phase function from Lemma~\ref{precisestatphase}. Define
		\begin{align*}
			Y\left(t^*,\LL\right) &\doteq \exp\left(-4\LL\log\LL + D_1\LL\right)\times 
			\\ \nonumber &\left|\sum_{q=1}^{\lceil (t^*)^{1/2}\rceil}e^{-it\Phi\left(t^*,2q,\omega_c\left(t^*,2q\right)\right)}\left((t^*)^{-2/3}\partial^2_{\omega\omega}\Phi\left(t^*,2q,\omega_c\left(t^*,2q\right)\right)\right)^{-1/2}\left(-\Gamma(m)\right)^q\right|.
		\end{align*}
		Then we have two different estimates for $Y\left(t^*,\LL\right)$, the first of which is unconditional, and the second of which depends on the assumption that $\left(3\epsilon,\frac{2}{3}+\epsilon\right)$ is an exponent pair (see Appendix~\ref{appendix.conj}).
		\begin{enumerate}
			\item There exists $\delta \in \left(0,\frac{1}{23}\right)$ so that $\left|Y\left(t^*,\LL\right)\right| \lesssim (t^*)^{\delta}$.
			\item Let $\epsilon > 0$ and assume that $\left(3\epsilon,\frac{2}{3}+\epsilon\right)$ is an exponent pair. Then, $\left|Y\left(t^*,\LL\right)\right| \lesssim (t^*)^{2\epsilon}$.
		\end{enumerate}
	\end{lemma}
	\begin{proof}For all $1 \leq k \leq \lceil (t^*)^{1/2}\rceil$, set
		\[P_k \doteq \left|\sum_{q=1}^ke^{-it\Phi\left(t^*,2q,\omega_c\left(t^*,2q\right)\right)}\right|.\]
		It will also be convenient to define
		\[b_{t^*,q} \doteq \left((t^*)^{-2/3}\partial^2_{\omega\omega}\Phi\left(t^*,2q,\omega_c\left(t^*,2q\right)\right)\right)^{-1/2}.\]
		We note that it follows from (a slight extension) of the analysis of the phase function $\Phi$ in Lemma~\ref{precisestatphase} that $1 \leq q \leq \lceil (t^*)^{1/2}\rceil$ implies that 
		\begin{equation}\label{phasecancellation}
			\left|b_{t^*,q}\right| \lesssim q^{1/3},\qquad \left|b_{t^*,q+1} - b_{t^*,q}\right| \lesssim q^{-2/3}.
		\end{equation}

		Summation by parts yields
		\begin{align}\label{thefirstsummatingbypartsfory}
			\left|Y\right| &\leq \exp\left(-4\LL\log\LL + D_1\LL\right)\sum_{q=1}^{\lceil (t^*)^{1/2}\rceil - 1}P_q\left(\left(-\Gamma(m)\right)^qb_{t^*,q}-\left(-\Gamma(m)\right)^{q+1}b_{t^*,q+1}\right)
			\\ \nonumber &\qquad + \exp\left(-4\LL\log\LL + D_1\LL\right)P_{\lceil (t^*)^{1/2}\rceil}\left(-\Gamma(m)\right)^{\lceil (t^*)^{1/2}\rceil}(t^*)^{1/6}
		\end{align}
		In view of the bound 
		\[x \geq 0 \Rightarrow xe^{-Nx} \lesssim N^{-1},\]
		and the fact that $1+\Gamma \sim \exp\left(-4\LL\log\LL + D_1\LL\right)$,
		we may thus conclude from~\eqref{thefirstsummatingbypartsfory} that
		\begin{align*}
			\left|Y\right| &\lesssim   \sup_{1 \leq k \leq \lceil (t^*)^{1/2}\rceil}\left(P_k k^{-2/3}\right). 
		\end{align*}
		
		In view of Lemma~\ref{foroscest}, the proof is then concluded by an application of Lemma~\ref{usenumbertheory} from Appendix~\ref{appendix.conj}.

	\end{proof}
	
	Finally, we are ready to prove our main Theorem~\ref{thm.intro}.
	
	\begin{proof}Let $\chi(t^*)$ be a cut-off function which is identically $0$ for $t^* \leq 0$ and identically $1$ for $t^* \geq 0$. We set $\phi_{\rm cut} \doteq \chi \phi$ and find that
		\[\left(\Box_g-m^2\right)\phi_{\rm cut} = F \doteq 2\nabla\chi\cdot\nabla\phi + \left(\Box\chi\right)\phi.\]
		We may expand $\phi_{\rm cut} = \sum_{\MM,L}\left(\phi_{\rm cut}\right)_{\MM,L}e^{i\MM\varphi}Y_{\MM,L}(\theta)$ and $F = \sum_{\MM,L}F_{\MM,L}e^{i\MM\varphi}Y_{\MM,L}(\theta)$ into their corresponding spherical harmonics.
		
		We then define
		\[u \doteq r\int_{-\infty}^{\infty}e^{it\omega}\left(\phi_{\rm cut}\right)_{\MM,L}\, dt,\qquad H \doteq r\left(1-\frac{2M}{r} + \frac{e^2}{r^2}\right)\int_{-\infty}^{\infty}e^{it\omega}F_{\MM,L}\, dt\]
		By Proposition 4.1 of~\cite{KGSchw1} we will have that $u$ will be an outgoing solution of~\eqref{eq:mainVrhs}. Moreover, we have that $u\left(s,-\omega,m,L\right) = \overline{u(\omega)}$. By Fourier inversion we thus have that
		\begin{align*}
			\left|\phi\left(t^*,r,\theta,\varphi\right)\right| &\lesssim \sup_{\MM,L}\left(1+L\right)^4\left|\phi_{\MM,l}\left(t^*,r\right)\right|
			\\ \nonumber &\lesssim \sup_{\MM,L}\left(1+L\right)^4\left|\int_{-\infty}^{\infty}e^{-it^*\omega}\left(e^{i\omega p(s)}u\right)\, d\omega\right|
			\\ \nonumber &= \sup_{\MM,L}\left(1+L\right)^4\left|\int_0^{\infty}e^{-it^*\omega}\left(e^{i\omega p(s)}u\right)\, d\omega+\overline{\int_0^{\infty}e^{-it^*\omega}\left(e^{i\omega p(s)}u\right)\, d\omega}\right|
			\\ \nonumber &\leq 2\sup_{\MM,L}\left(1+L\right)^4\left|\int_0^{\infty}e^{-it^*\omega}\left(e^{i\omega p(s)}u\right)\, d\omega\right|.
		\end{align*}
		Thus the proof follows from a combination of Proposition~\ref{estfourtransofumainprop}, Lemma~\ref{weusethistocontrolthesum}, and the fact, easily established with finite in time energy estimates, that the norm $\mathscr{D}$ from Proposition~\ref{estfourtransofumainprop} may be controlled by  $\left\vert\left\vert \phi_0\right\vert\right\vert_{H^N\left(\{t^* = 0\}\right)}+\left\vert\left\vert \phi_0\right\vert\right\vert_{H^{N-1}\left(\{t^* = 0\}\right)} $ for suitably large $N$.

	\end{proof}
	
	\appendix

	\section{Modified Bessel functions}\label{bessel.section} This section is based on classical knowledge on Bessel functions, see e.g.\ \cite{handbook,olver}.	
	First, we recall the Bessel equation  for any $\alpha \in \mathbb{C}$:
	\begin{equation}\label{bessel.eq}
		x^2 \frac{d^2 u }{dx^2}+  x \frac{du}{dx} +(x^2-\alpha^2) u =0.
	\end{equation}  A basis of solutions of \eqref{bessel.eq} is given by $J_{\alpha}(x)$,  $Y_{\alpha}(x)$  the Bessel functions of the first and second kind respectively. In what follows, we will take $\alpha = i \nu$, where $\nu \in \mathbb{R}$.
	
	For any constant $\nu \in \mathbb{R}-\{0\}$, we introduce $K_{\pm}(\nu,x)$ to be solutions of the following Bessel equation \begin{equation}\label{modified.bessel.eq}
		x^2 \frac{d^2 u }{dx^2}+  x \frac{du}{dx} -(x^2-\nu^2) u =0,
	\end{equation}  where \begin{align}
		&  K_{+}(\nu,x)=\frac{\pi}{2 \sinh(\nu \pi)} \left[ I_{i\nu}(x)+I_{-i\nu}(x)\right],  \label{K+.def}\\ &  K_{-}(\nu,x)= K_{i\nu}(x)=
		\frac{\pi}{2\sinh(\pi \nu)} \frac{I_{i\nu}(x)-I_{-i\nu}(x)}{i}\label{K-.def}
	\end{align} where    $I_{\nu}$,    $K_{\nu}(x)$ are  the modified Bessel function of the  first and second kind, respectively. Note the following asymptotics as $x\rightarrow 0$: \begin{align}
		&  K_{+}(\nu,x)=\sqrt{\frac{\pi}{\nu \sinh(\pi \nu)}} \sin \left( \nu \log(\frac{x}{2}) - \theta_{\nu} \right) +O(x^2),\\  &  K_{-}(\nu,x) = -\sqrt{\frac{\pi}{\nu \sinh(\pi \nu)}} \cos \left( \nu \log(\frac{x}{2}) - \theta_{\nu} \right) +O(x^2),
	\end{align} where $\theta_{\nu} \in [0,2\pi]= arg(\Gamma(i\nu))$; and the following asymptotics as $x\rightarrow +\infty$: \begin{align}
		&  K_{+}(\nu,x)  = \frac{1}{ \sinh(\pi \nu )} \sqrt{\frac{\pi}{2x}} \exp(x)  \left[1 \mathbf{+} \frac{\nu^2+\frac{1}{4}}{2x} + O(x^{-2}) \right],\\  &  K_{-}(\nu,x) = \sqrt{\frac{\pi}{2x}} \exp(-x) \left[1- \frac{\nu^2+\frac{1}{4}}{2x} + O(x^{-2}) \right].
	\end{align}

	\section{Volterra equations}\label{volterra.section}
	The following theorem is a mild generalization of Theorem~10.1 of Chapter 6 in~\cite{olver} to allow for a general inhomogeneous term on the right hand side. The presentation follows the Appendix of \cite{KGSchw1} and is repeated for the reader's convenience.
	\begin{thm}\label{thm:volterra}
		Let $\beta \in \R \cup\{ +\infty\}$ and $\alpha \in \R \cup \{-\infty\}$ with $\alpha < \beta$. 
		We consider the following Volterra integral equation for an unkown function $h(\xi)$:
		\begin{equation}\label{volt}
			h(\xi) = \int_{\xi}^{\beta} K(\xi, v) \psi_0(v) h(v) dv + R(\xi),
		\end{equation}
		Here, the complex-valued functions $\psi_0(v)$, $\psi_1(v)$ are continuous on  $(\alpha, \beta)$, the complex-valued function $R(\xi)$  and kernel $K(\xi, v)$ and their derivatives $\rd_\xi K$, $\rd^2_\xi K$, $\rd_{\xi}R$, $\rd^2_{\xi}R$ are continuous on $(\alpha, \beta)$. We also suppose that
		\begin{enumerate}
			\item $K(\xi, \xi) = 0$, 
			\item  When $\xi \in (\alpha, \beta)$ and $v \in [\xi, \beta)$, we have, for $j \in \{0,1,2\}$:
			\begin{equation}
				|\rd^j_\xi K(\xi, v)| \leq P_j(\xi) Q(v), 
			\end{equation}
			Here, $P_j(\xi)$ and $Q(v)$ are continuous, positive real-valued functions.
			\item There exists a positive decreasing function $\Phi(\xi)$ so that for $j \in \{0,1\}$:
			\[\left|\partial_{\xi}^jR\left(\xi\right)\right| \leq P_j(\xi)\Phi(\xi).\]
			\item When $\xi \in (\alpha, \beta)$, the following expressions converge:
			\begin{equation}
				\Psi_0(\xi) = \int_{\xi}^{\beta} |\psi_0(v)| dv, 
			\end{equation}
			and the following supremum is finite:
			\begin{equation}
				\kappa_0 =  \sup_{\xi \in (\alpha, \beta)} P_0(\xi) Q(\xi).
			\end{equation}
			
		\end{enumerate}
		Then, there exists a unique $C^2$ solution $h(\xi)$ to~\eqref{volt}, and we have, for $\xi \in (\alpha, \beta)$,
		\begin{equation}\label{finalvolterraestimate}
			\frac{\left|h(\xi)\right|}{ P_0(\xi)}, \frac{|\partial_{\xi}h(\xi)|}{P_1(\xi)} \leq  \Phi(\xi)\exp \Big(\kappa_0 \Psi_0(\xi)  \Big).
		\end{equation}
		
	\end{thm}
	In the next theorem, we consider the case when our  equation depends on an auxiliary parameter $\Par$.
	\begin{thm}\label{thm:volterraparam}We consider again the setting of Theorem~\ref{thm:volterra} and assume $P_0$, $P_1$, $\Psi_0$, $\Psi_1$, and $\Phi$ have already been fixed in that context. We then furthermore suppose that $K$, $R$, $\psi_0$, and $\psi_1$ depend additionally on a real parameter $\Par \in (\Par_0,\Par_1)$:
		\begin{equation}\label{withtheparamsvolt}
			h(\xi,\Par) = \int_{\xi}^{\beta} K(\xi, v,\Par)\psi_0(v,\Par) h(v,\Par) dv + R(\xi,\Par).
		\end{equation}
		We now assume that the assumptions of Theorem~\ref{thm:volterra} hold uniformly for $\Par \in (\Par_0,\Par_1)$ and make the additional assumptions that $K$, $\psi_0$,  and $R$ are all $C^1$ functions of $\Par$ and that there exists functions $\tilde{P}_j$ for $j \in \{0,1\}$ and $\tilde{Q}$ so that 
		\begin{enumerate}
			\item  When $\xi \in (\alpha, \beta)$ and $v \in [\xi, \beta)$, we have, for $j \in \{0,1\}$:
			\begin{equation}
				\sup_{\Par \in (\Par_0,\Par_1)}|\rd^j_\xi K(\xi, v,\Par)| \leq \tilde{P}_j(\xi) \tilde{Q}(v), 
			\end{equation}
			Here, $ \tilde{P}_j(\xi)$ and $ \tilde{Q}(v)$ are continuous, positive real-valued functions.
			\item 
			
			There exists a decreasing function $\tilde{\Phi}(\xi)$ so that for $j \in \{0,1\}$:
			\[\sup_{\Par \in (\Par_0,\Par_1)}\left[\left|\partial_\Par \partial^j_{\xi}R\right| + \left|\tilde{R}_j\left(\xi,\Par\right)\right|\right] \leq \tilde{P}_j(\xi)\tilde{\Phi}(\xi),\]
			where we define $\tilde{R}_j$ by
			\begin{align}\label{tiderj1}
				\tilde{R}_1 \doteq \exp\left(\kappa_0\Psi_0 \right)\int_{\xi}^{\beta}\Phi\left[\left(\left|\partial_{\lambda}K\right|\left|\psi_0\right| + \left|K\right|\left|\partial_{\lambda}\psi_0\right|\right)P_0 \right]\, dv,
			\end{align}
			\begin{align}\label{tilderj2}
				\tilde{R}_2 \doteq \exp\left(\kappa_0\Psi_0 \right)\int_{\xi}^{\beta}\Phi\left[\left(\left|\partial_{\lambda}\partial_{\xi}K\right|\left|\psi_0\right| + \left|\partial_{\xi}K\right|\left|\partial_{\lambda}\psi_0\right|\right)P_0 \right]\, dv
				+ \exp\left(\kappa_0\Psi_0 \right)\Phi\left(\left|\partial_{\lambda}K\right|\left(\left|\psi_0\right|P_0 \right)\right).
			\end{align}
		\end{enumerate}
		Then $h$ is continuously differentiable with respect to $\Par$ and satisfies the following estimates:
		\begin{equation}\label{finalvolterraestimatewithaparam}
			\sup_{\Par\in (\Par_0,\Par_1)}  \frac{\left|\partial_\Par h(\xi,\Par)\right|}{ \tilde{P}_0(\xi)},  \sup_{\Par\in (\Par_0,\Par_1)}\frac{|\partial^2_{\Par \xi}h(\xi,\lambda)|}{\tilde{P}_1(\xi)} \leq  \tilde{\Phi}(\xi)\exp \Big(\kappa_0 \Psi_0(\xi)\Big).
		\end{equation}
	\end{thm}
	\begin{proof}
		See \cite{KGSchw1}, Appendix B.
	\end{proof}

	\section{Bounds on sum of oscillating exponentials}\label{appendix.conj}
	\begin{defn} Let $\sigma > 0$, $T > 0$, $\delta > 0$, $N \in \mathbb{N}$, and $J \in \mathbb{N}$. We say that a smooth function $f:[N,2N]\rightarrow \R$ is in the class $\mathcal{F}\left(\sigma,T,\delta,N,J\right)$ if for all $0\leq j\leq J$: \begin{equation}\label{f.condition}
			\bigl| f^{(j+1)}(x)-  T  \frac{d^j}{dx^j}[ x^{-\sigma}] \bigr|  \leq  \delta \cdot T \bigl| \frac{d^j}{dx^j}[ x^{-\sigma}] \bigr| .
		\end{equation}
		We say the numbers $(p,q)$ with $p, q \geq 0$ are exponent pairs if for any $\sigma > 0$ there exists $\delta > 0$ and $J \in \mathbb{N}$ such that every $N,N' \in \mathbb{N}$ with $N' \in [N,2N]$ and $T \geq N^{\sigma}$, $f \in \mathcal{F}\left(\sigma,T,\delta,N,J\right)$ implies \begin{equation}\label{exp.eq}
			\bigl|	\sum_{n=N}^{N'} \exp( i f(n)) \bigr| \ls (\frac{T}{N^{\sigma}})^p N^{q}. 
		\end{equation} 
	\end{defn}
	We recall the celebrated conjecture Exponent Pair Conjecture:

	\begin{conjecture}[Exponent Pair Conjecture, \cite{NT3,Montgomery}]\label{exp.conj}
		For any $\epsilon>0$, $(\ep, \frac{1}{2}+\ep)$ is an exponent pair.
	\end{conjecture}

	The next two lemmas contain our key exponential sum estimate.
	
	\begin{lemma}\label{usenumbertheory}

		Let $0\leq \ep< \frac{1}{3}$.	Assume that $(3\ep, \frac{2}{3}+\ep)$ is an exponent pair (this is in particular satisfied if  the exponent pair conjecture is true). For some $D\neq 0$, let $\psi_k(t)=D \cdot k^{\frac{2}{3}}\cdot  t^{\frac{1}{3}}  + O(1)$ which satisfies \eqref{f.condition} for $T= t^{\frac{1}{3}}$ and $\sigma=1/3$. Then for any $ N \ls t^{1/2}$: \begin{equation}
			\bigl|	\sum_{k=0}^{N}  \exp( i \psi_k(t))\bigr| \ls \log(t)t^{\epsilon}\cdot N^{2/3}.
		\end{equation}
	\end{lemma}
	\begin{proof}
		This is an immediate application of Conjecture~\ref{exp.conj}  and a dyadic decomposition of the interval $[0,N]$.
	\end{proof}
	Now we will be using some of the known exponent pairs  to get some moderate growth in $t$. While we have not tried to optimize the choice of exponents, we refer to the survey \cite{NT5} for more details on known exponent pairs.
	\begin{prop}  For some $D\neq 0$, let $\psi_k(t)=D \cdot k^{\frac{2}{3}}\cdot  t^{\frac{1}{3}}  + O(1)$ which satisfies \eqref{f.condition} for $T= t^{\frac{1}{3}}$ and $\sigma=1/3$. For any $ N \ls t^{1/2}$, we have for any $t\geq 1$, and any $\ep>0$
		\begin{equation}
			N^{-2/3}		\bigl|	 \sum_{k=1}^{N} \exp( i \psi_k(t)) \bigr| \ls t^{\frac{959}{22224} +\ep} \ls t^{1/23}.
		\end{equation}
	\end{prop}  
	\begin{proof}
		First, for some $0< \alpha <1/2$,  we split into two cases: \begin{enumerate}
			\item $N> \lceil t^{\alpha}\rceil$. We split the sum as such $$  \sum_{k=1}^{N} \exp(i \psi_k(t) )= \sum_{k=1}^{\lceil t^{\alpha}\rceil} \exp( i \psi_k(t)) + \sum_{k=\lceil t^{\alpha}\rceil+1}^{N} \exp( i \psi_k(t)) .$$
			
			For the first sum, a trivial bound is given by $|\sum_{k=1}^{\lceil t^{\alpha}\rceil} \exp( i \psi_k(t))| \ls t^{\alpha}$, but one can actually improve over this bound noting that $(p,q)= (\frac{1}{915}+ \ep, 1-\frac{10}{915}+\ep) $ is a exponent pair for any $\ep>0$ (see \cite{NT1}: note that strictly speaking, \cite{NT1} only proves \eqref{exp.eq} up to an extra $N/T$ term in the right-hand-side. However, in view of the fact that $N\ls T^{3/2}$ in our case,  the $N/T \lesssim N^{1/2}$ term is controlled by the main term), from which we get, for any $\ep>0$: \begin{equation}\label{osc.exp1}
				N^{-2/3}|\sum_{k=1}^{\lceil t^{\alpha}\rceil} \exp( i \psi_k(t))|\ls t^{ \frac{1+ 884\alpha}{2745}  + \ep}.
			\end{equation}
			
			Now, for the second sum, we use that is an exponent pair $(p,q)= (13/84+\ep,13/84+\frac{1}{2}+\ep)$ for any $\ep>0$  (see \cite{NT2}) from which we get
			\begin{equation}\label{osc.exp2}
				N^{-2/3}	|\sum^{N}_{k=\lceil t^{\alpha}\rceil+1} \exp( i \psi_k(t))|\ls t^{ \frac{13}{252}- \frac{4\alpha}{63} + \ep}.
			\end{equation}
			
			To optimize over the value of $\alpha$, we equate the two expressions from \eqref{osc.exp1} and  \eqref{osc.exp2} to obtain  $$t^{ \frac{13}{252}- \frac{4\alpha}{63} + \ep}=  t^{ \frac{1+ 884 \alpha}{2745}  + \ep},$$ the result of which is \begin{equation}\label{alpha.app}
				\alpha= \frac{3937}{29632}.
			\end{equation} The conclusion is that 
			\begin{equation}\label{osc.exp.F}
				N^{-2/3}	|\sum^{N}_{k=\lceil t^{\alpha}\rceil} \exp( i k^{2/3} t^{1/3})|\ls t^{ \frac{959}{22224} + \ep}< t^{1/23}.
			\end{equation}
			
			\item  $N\leq \lceil t^{\alpha}\rceil$: we can use the same  technique as for the first sum above, exploiting the fact that  $(p,q)= (\frac{1}{915}+ \ep, 1-\frac{10}{915}+\ep) $ is a exponent pair for any $\ep>0$  get, for any $\ep>0$: \begin{equation}\label{osc.exp3}
				N^{-2/3}	|\sum_{k=1}^{N} \exp( i \psi_k(t))|\ls t^{ \frac{1+ 884 \alpha}{2745}  + \ep},
			\end{equation} with $\alpha=\frac{3937}{29408}$ as above.
		\end{enumerate} 
	\end{proof}
	
	\bibliographystyle{plain}
	\bibliography{bibliography.bib}

\begin{thebibliography}{100}

\bibitem{handbook}
Milton Abramowitz and Irene~A. Stegun.
\newblock {\em Handbook of mathematical functions with formulas, graphs, and
  mathematical tables}, volume No. 55 of {\em National Bureau of Standards
  Applied Mathematics Series}.
\newblock U. S. Government Printing Office, Washington, DC, 1964.

\bibitem{Agmon}
Shmuel Agmon.
\newblock {\em Lectures on exponential decay of solutions of second-order
  elliptic equations: bounds on eigenfunctions of {$N$}-body {S}chr\"{o}dinger
  operators}, volume~29 of {\em Mathematical Notes}.
\newblock Princeton University Press, Princeton, NJ; University of Tokyo Press,
  Tokyo, 1982.

\bibitem{Maxwell5}
Lars Andersson, Thomas B\"{a}ckdahl, and Pieter Blue.
\newblock Decay of solutions to the {M}axwell equation on the {S}chwarzschild
  background.
\newblock {\em Classical Quantum Gravity}, 33(8):085010, 20, 2016.

\bibitem{Blue}
Lars Andersson, Thomas B\"ackdahl, Pieter Blue, and Siyuan Ma.
\newblock Stability for linearized gravity on the {K}err spacetime.
\newblock Arxiv preprint: \url{https://www.arxiv.org/abs/1903.03859}, 2019.

\bibitem{hidden}
Lars Andersson and Pieter Blue.
\newblock Hidden symmetries and decay for the wave equation on the {K}err
  spacetime.
\newblock {\em Ann. of Math. (2)}, 182(3):787--853, 2015.

\bibitem{AAG1}
Yannis Angelopoulos, Stefanos Aretakis, and Dejan Gajic.
\newblock Late-time asymptotics for the wave equation on spherically symmetric,
  stationary spacetimes.
\newblock {\em Advances in Mathematics}, 323:529--621, 2018.

\bibitem{AAG2}
Yannis Angelopoulos, Stefanos Aretakis, and Dejan Gajic.
\newblock Price's law and precise late-time asymptotics for subextremal
  {R}eissner-{N}ordstr\"om black holes.
\newblock ArXiv preprint: \url{https://www.arxiv.org/abs/2102.11888}, 2021.

\bibitem{lindelof2}
Ralf~J. Backlund.
\newblock Über die beziehung zwischen anwachsen und nullstellen der
  zeta-funktion.
\newblock {\em Ofversigt Finska Vetensk. Soc.}, 61, 1918-1919.

\bibitem{super1}
James~M. Bardeen, William~H. Press, and Saul~A Teukolsky.
\newblock {Rotating black holes: Locally nonrotating frames, energy extraction,
  and scalar synchrotron radiation}.
\newblock {\em Astrophys. J.}, 178:347, 1972.

\bibitem{Benomio}
Gabriele Benomio.
\newblock The stable trapping phenomenon for black strings and black rings and
  its obstructions on the decay of linear waves.
\newblock {\em Anal. PDE}, 14(8):2427--2496, 2021.

\bibitem{Boson}
Piotr Bizo\'{n} and Arthur Wasserman.
\newblock On existence of mini-boson stars.
\newblock {\em Comm. Math. Phys.}, 215(2):357--373, 2000.

\bibitem{BlueMaxwell}
Pieter Blue.
\newblock Decay of the {M}axwell field on the {S}chwarzschild manifold.
\newblock {\em J. Hyperbolic Differ. Equ.}, 5(4):807--856, 2008.

\bibitem{BlueSoffer1}
Pieter Blue and Avy Soffer.
\newblock Phase space analysis on some black hole manifolds.
\newblock {\em J. Funct. Anal.}, 256(1):1--90, 2009.

\bibitem{NT2}
Jean Bourgain.
\newblock Decoupling, exponential sums and the {R}iemann zeta function.
\newblock {\em J. Amer. Math. Soc.}, 30(1):205--224, 2017.

\bibitem{Burko}
Lior Burko.
\newblock Abstracts of plenary talks and contributed papers.
\newblock {\em 15th International Conference on General Relativity and
  Gravitation, Pune, 143 (unpublished)}, 1997.

\bibitem{BurkoKhanna}
Lior Burko and Gaurav Khanna.
\newblock Universality of massive scalar field late-time tails in black-hole
  spacetimes.
\newblock {\em Phys. Rev. D (3)}, 70(4):044018, 8, 2004.

\bibitem{Burq}
Nicolas Burq.
\newblock D\'{e}croissance de l'\'{e}nergie locale de l'\'{e}quation des ondes
  pour le probl\`eme ext\'{e}rieur et absence de r\'{e}sonance au voisinage du
  r\'{e}el.
\newblock {\em Acta Math.}, 180(1):1--29, 1998.

\bibitem{cardosopopov}
Fernando Cardoso and Georgi Popov.
\newblock Quasimodes with exponentially small errors associated with elliptic
  periodic rays.
\newblock {\em Asymptot. Anal.}, 30(3-4):217--247, 2002.

\bibitem{chodosh-sr2}
Otis Chodosh and Yakov Shlapentokh-Rothman.
\newblock Time-periodic {E}instein-{K}lein-{G}ordon bifurcations of {K}err.
\newblock {\em Comm. Math. Phys.}, 356(3):1155--1250, 2017.

\bibitem{chodosh-sr1}
Otis Chodosh and Yakov Shlapentokh-Rothman.
\newblock Stationary axisymmetric black holes with matter.
\newblock {\em Comm. Anal. Geom.}, 29(1):19--76, 2021.

\bibitem{Schlag_exp}
Ovidiu Costin, Roland Donninger, Wilhelm Schlag, and Saleh Tanveer.
\newblock Semiclassical low energy scattering for one-dimensional
  {S}chr\"{o}dinger operators with exponentially decaying potentials.
\newblock {\em Ann. Henri Poincar\'{e}}, 13(6):1371--1426, 2012.

\bibitem{TeukolskyDHR}
Mihalis Dafermos, Gustav Holzegel, and Igor Rodnianski.
\newblock Boundedness and decay for the {T}eukolsky equation on {K}err
  spacetimes {I}: {T}he case {$|a|\ll M$}.
\newblock {\em Ann. PDE}, 5(1):Paper No. 2, 118, 2019.

\bibitem{MihalisStabExt}
Mihalis Dafermos, Gustav Holzegel, and Igor Rodnianski.
\newblock The linear stability of the {S}chwarzschild solution to gravitational
  perturbations.
\newblock {\em Acta Math.}, 222(1):1--214, 2019.

\bibitem{blackbox}
Mihalis Dafermos, Gustav Holzegel, and Igor Rodnianski.
\newblock Quasilinear wave equations on asymptotically flat spacetimes with
  applications to {K}err black holes.
\newblock ArXiv preprint: \url{https://arxiv.org/abs/2212.14093}, 2022.

\bibitem{SchwarzschildStab}
Mihalis Dafermos, Gustav Holzegel, Igor Rodnianski, and Martin Taylor.
\newblock The non-linear stability of the {S}chwarzschild family of black
  holes.
\newblock ArXiv preprint: \url{https://www.arxiv.org/abs/2104.08222}, 2021.

\bibitem{PriceLaw}
Mihalis Dafermos and Igor Rodnianski.
\newblock A proof of {P}rice's law for the collapse of a self-gravitating
  scalar field.
\newblock {\em Invent. Math.}, 162(2):381--457, 2005.

\bibitem{Red}
Mihalis Dafermos and Igor Rodnianski.
\newblock The red-shift effect and radiation decay on black hole spacetimes.
\newblock {\em Comm. Pure Appl. Math.}, 62(7):859--919, 2009.

\bibitem{rp}
Mihalis Dafermos and Igor Rodnianski.
\newblock A new physical-space approach to decay for the wave equation with
  applications to black hole spacetimes.
\newblock In {\em X{VI}th {I}nternational {C}ongress on {M}athematical
  {P}hysics}, pages 421--432. World Sci. Publ., Hackensack, NJ, 2010.

\bibitem{claylecturenotes}
Mihalis Dafermos and Igor Rodnianski.
\newblock Lectures on black holes and linear waves.
\newblock In {\em Evolution equations}, volume~17 of {\em Clay Math. Proc.},
  pages 97--205. Amer. Math. Soc., Providence, RI, 2013.

\bibitem{KerrDaf}
Mihalis Dafermos, Igor Rodnianski, and Yakov Shlapentokh-Rothman.
\newblock Decay for solutions of the wave equation on {K}err exterior
  spacetimes {III}: {T}he full subextremal case {$|a|<M$}.
\newblock {\em Ann. of Math. (2)}, 183(3):787--913, 2016.

\bibitem{Detweiler}
Steven Detweiler.
\newblock {K}lein-{G}ordon equation and rotating black holes.
\newblock {\em Phys. Rev. D}, 22:2323--2326, Nov 1980.

\bibitem{Santos.unpublished}
Oscar Dias, Jorge Santos, Harvey Reall, and Bogdan Ganchev.
\newblock In preparation; announced in \cite{SantosWCC}.

\bibitem{super2}
Sam~R. Dolan.
\newblock Instability of the massive {K}lein-{G}ordon field on the {K}err
  spacetime.
\newblock {\em Phys. Rev. D}, 76:084001, 2007.

\bibitem{Schlag2}
Roland Donninger, Wilhelm Schlag, and Avy Soffer.
\newblock On pointwise decay of linear waves on a schwarzschild black hole
  background.
\newblock {\em Commun. Math. Phys.}, 309:51–86, 2012.

\bibitem{Schlag1}
Ronald Donninger, Wilhelm Schlag, and Avy Soffer.
\newblock A proof of {P}rice's law on {S}chwarzschild black hole manifolds for
  all angular momenta.
\newblock {\em Advances in Mathematics}, 226(1):484--540, 2011.

\bibitem{res_final}
Semyon Dyatlov and Maciej Zworski.
\newblock {\em Mathematical theory of scattering resonances}, volume 200 of
  {\em Graduate Studies in Mathematics}.
\newblock American Mathematical Society, Providence, RI, 2019.

\bibitem{blackrings}
Roberto Emparan and Harvey~S. Reall.
\newblock Black rings.
\newblock {\em Classical Quantum Gravity}, 23(20):R169--R197, 2006.

\bibitem{SantosWCC}
Felicity~C. Eperon, Bogdan Ganchev, and Jorge~E. Santos.
\newblock Plausible scenario for a generic violation of the weak cosmic
  censorship conjecture in asymptotically flat four dimensions.
\newblock {\em Phys. Rev. D}, 101(4):041502(R), 6, 2020.

\bibitem{quantformath}
Ludwig~D. Faddeev and Oleg~A. Yakubovski\u{\i}.
\newblock {\em Lectures on quantum mechanics for mathematics students},
  volume~47 of {\em Student Mathematical Library}.
\newblock American Mathematical Society, Providence, RI, 2009.

\bibitem{massivedirac2}
Felix Finster, Niky Kamran, Joel Smoller, and Shing-Tung Yau.
\newblock Decay rates and probability estimates for massive {D}irac particles
  in the {K}err-{N}ewman black hole geometry.
\newblock {\em Comm. Math. Phys.}, 230(2):201--244, 2002.

\bibitem{hair.unstable}
Bogdan Ganchev and Jorge~E. Santos.
\newblock Scalar hairy black holes in four dimensions are unstable.
\newblock {\em Phys. Rev. Lett.}, 120:171101, Apr 2018.

\bibitem{Gannot}
Oran Gannot.
\newblock Quasinormal modes for {S}chwarzschild-{A}d{S} black holes:
  exponential convergence to the real axis.
\newblock {\em Comm. Math. Phys.}, 330(2):771--799, 2014.

\bibitem{quasimode1}
Oran Gannot.
\newblock From quasimodes to resonances: exponentially decaying perturbations.
\newblock {\em Pacific J. Math.}, 277(1):77--97, 2015.

\bibitem{Giorgi.RN}
Elena Giorgi.
\newblock The linear stability of {R}eissner-{N}ordstr\"{o}m spacetime: the
  full subextremal range {$|Q|<M$}.
\newblock {\em Comm. Math. Phys.}, 380(3):1313--1360, 2020.

\bibitem{stabilitykerrformalism}
Elena Giorgi, Sergiu Klainerman, and J\'er\'emie Szeftel.
\newblock A general formalism for the stability of {K}err.
\newblock ArXiv preprint: \url{https://www.arxiv.org/abs/2002.02740}, 2020.

\bibitem{GKS}
Elena Giorgi, Sergiu Klainerman, and J\'er\'emie Szeftel.
\newblock Wave equations estimates and the nonlinear stability of slowly
  rotating {K}err black holes.
\newblock ArXiv preprint: \url{https://www.arxiv.org/abs/2205.14808}, 2022.

\bibitem{GrahamK}
Sidney~W. Graham and Grigori Kolesnik.
\newblock {\em van der {C}orput's method of exponential sums}, volume 126 of
  {\em London Mathematical Society Lecture Note Series}.
\newblock Cambridge University Press, Cambridge, 1991.

\bibitem{strings}
Ruth Gregory and Raymond Laflamme.
\newblock Black strings and {$p$}-branes are unstable.
\newblock {\em Phys. Rev. Lett.}, 70(19):2837--2840, 1993.

\bibitem{numerics.Kerr}
Carlos A.~R. Herdeiro and Eugen Radu.
\newblock Kerr black holes with scalar hair.
\newblock {\em Phys. Rev. Lett.}, 112:221101, Jun 2014.

\bibitem{Hintz}
Peter Hintz.
\newblock A sharp version of {P}rice’s law for wave decay on asymptotically
  flat spacetimes.
\newblock {\em Commun. Math. Phys.}, 389:491–542, 2022.

\bibitem{HodPiran.mass}
Shahar Hod and Tsvi Piran.
\newblock Late-time tails in gravitational collapse of a self-interacting
  (massive) scalar-field and decay of a self-interacting scalar hair.
\newblock {\em Phys. Rev. D}, 58:044018, 1998.

\bibitem{decayads}
Gustav Holzegel and Jacques Smulevici.
\newblock Decay properties of {K}lein-{G}ordon fields on {K}err-{A}d{S}
  spacetimes.
\newblock {\em Comm. Pure Appl. Math.}, 66(11):1751--1802, 2013.

\bibitem{SchwAds}
Gustav Holzegel and Jacques Smulevici.
\newblock Stability of {S}chwarzschild-{A}d{S} for the spherically symmetric
  {E}instein-{K}lein-{G}ordon system.
\newblock {\em Comm. Math. Phys.}, 317(1):205--251, 2013.

\bibitem{quasimodeads}
Gustav Holzegel and Jacques Smulevici.
\newblock Quasimodes and a lower bound on the uniform energy decay rate for
  {K}err-{A}d{S} spacetimes.
\newblock {\em Anal. PDE}, 7(5):1057--1090, 2014.

\bibitem{NT3}
Henryk Iwaniec and Emmanuel Kowalski.
\newblock {\em Analytic number theory}, volume~53 of {\em American Mathematical
  Society Colloquium Publications}.
\newblock American Mathematical Society, Providence, RI, 2004.

\bibitem{massiveDirac}
Jiliang Jing.
\newblock Late-time evolution of charged massive dirac fields in the
  {R}eissner-nordstr\"om black-hole background.
\newblock {\em Phys. Rev. D}, 72:027501, Jul 2005.

\bibitem{boson3}
David~J. Kaup.
\newblock Klein-{G}ordon {G}eon.
\newblock {\em Phys. Rev.}, 172:1331--1342, 1968.

\bibitem{KS.polarized}
Sergiu Klainerman and J\'{e}r\'{e}mie Szeftel.
\newblock {\em Global nonlinear stability of {S}chwarzschild spacetime under
  polarized perturbations}, volume 210 of {\em Annals of Mathematics Studies}.
\newblock Princeton University Press, Princeton, NJ, 2020.

\bibitem{GCM}
Sergiu Klainerman and J\'{e}r\'{e}mie Szeftel.
\newblock Construction of {GCM} spheres in perturbations of {K}err.
\newblock {\em Ann. PDE}, 8(2):Paper No. 17, 153, 2022.

\bibitem{effectiveuniform}
Sergiu Klainerman and J\'{e}r\'{e}mie Szeftel.
\newblock Effective results on uniformization and intrinsic {GCM} spheres in
  perturbations of {K}err.
\newblock {\em Ann. PDE}, 8(2):Paper No. 18, 89, 2022.
\newblock With an appendix by Camillo De Lellis.

\bibitem{klainerman2021kerr}
Sergiu Klainerman and J\'{e}r\'{e}mie Szeftel.
\newblock Kerr stability for small angular momentum.
\newblock {\em Pure Appl. Math. Q.}, 19(3):791--1678, 2023.

\bibitem{QNM4}
Roman Konoplya.
\newblock Massive charged scalar field in a {R}eissner–{N}ordstr\"om black
  hole background: quasinormal ringing.
\newblock {\em Physics Letters B}, 550(1):117--120, 2002.

\bibitem{KonoplyaZhidenko}
Roman Konoplya and Alexander Zhidenko.
\newblock Massive charged scalar field in the {K}err-{N}ewman background:
  quasinormal modes, late-time tails and stability.
\newblock {\em Phys. Rev. D}, 88:024054, 2013.

\bibitem{Konoplya.Zhidenko.num}
Roman Konoplya, Alexander Zhidenko, and Carlos Molina.
\newblock Late time tails of the massive vector field in a black hole
  background.
\newblock {\em Phys. Rev. D}, 75:084004, 2007.

\bibitem{Proca2}
Roman Konoplya, Alexander Zhidenko, and Carlos Molina.
\newblock Late time tails of the massive vector field in a black hole
  background.
\newblock {\em Phys. Rev. D}, 75:084004, Apr 2007.

\bibitem{KoyamaTomimatsu}
Hiroko Koyama and Akira Tomimatsu.
\newblock Asymptotic power-law tails of massive scalar fields in a
  {R}eissner--{N}ordstr\"om background.
\newblock {\em Phys. Rev. D}, 63:064032, 2001.

\bibitem{KoyamaTomimatsu3}
Hiroko Koyama and Akira Tomimatsu.
\newblock Asymptotic tails of massive scalar fields in a {S}chwarzschild
  background.
\newblock {\em Phys. Rev. D}, 64:044014, 2001.

\bibitem{KoyamaTomimatsu2}
Hiroko Koyama and Akira Tomimatsu.
\newblock Slowly decaying tails of massive scalar fields in spherically
  symmetric spacetimes.
\newblock {\em Phys. Rev. D}, 65:084031, 2002.

\bibitem{lindelof}
Ernst Lindelöf.
\newblock Quelques remarques sur la croissance de la fonction $\zeta(s)$.
\newblock {\em Bull. Sci. Math.}, 32:341–356, 1908.

\bibitem{Maspin1}
Siyuan Ma.
\newblock Uniform energy bound and {M}orawetz estimate for extreme components
  of spin fields in the exterior of a slowly rotating {K}err black hole {I}:
  {M}axwell field.
\newblock {\em Ann. Henri Poincar\'{e}}, 21(3):815--863, 2020.

\bibitem{Maspin2}
Siyuan Ma.
\newblock Uniform energy bound and {M}orawetz estimate for extreme components
  of spin fields in the exterior of a slowly rotating {K}err black hole {II}:
  {L}inearized gravity.
\newblock {\em Comm. Math. Phys.}, 377(3):2489--2551, 2020.

\bibitem{Ma4}
Siyuan Ma.
\newblock Almost {P}rice's law in {S}chwarzschild and decay estimates in {K}err
  for {M}axwell field.
\newblock {\em Journal of Differential Equations}, 339:1--89, 2022.

\bibitem{Ma3}
Siyuan Ma and Lin Zhang.
\newblock {P}rice’s law for spin fields on a {S}chwarzschild background.
\newblock {\em Ann. PDE}, 8(25), 2022.

\bibitem{Ma1}
Siyuan Ma and Lin Zhang.
\newblock Sharp decay estimates for massless {D}irac fields on a
  {S}chwarzschild background.
\newblock {\em Journal of Functional Analysis}, 282(6):109375, 2022.

\bibitem{Ma5}
Siyuan Ma and Lin Zhang.
\newblock Sharp decay for {T}eukolsky equation in {K}err spacetimes.
\newblock {\em Commun. Math. Phys.}, 2023.

\bibitem{Tataru2}
Jason Metcalfe, Daniel Tataru, and Mihai Tohaneanu.
\newblock Price's law on nonstationary space-times.
\newblock {\em Adv. Math.}, 230(3):995--1028, 2012.

\bibitem{TataruMaxwell}
Jason Metcalfe, Daniel Tataru, and Mihai Tohaneanu.
\newblock Pointwise decay for the {M}axwell field on black hole space-times.
\newblock {\em Adv. Math.}, 316:53--93, 2017.

\bibitem{millet}
Pascal Millet.
\newblock {O}ptimal decay for solutions of the {T}eukolsky equation on the
  {K}err metric for the full subextremal range $|a| < m$.
\newblock ArXiv preprint: \url{https://www.arxiv.org/abs/2302.06946}, 2023.

\bibitem{Dilaton1}
Rafa\l{} Moderski and Marek Rogatko.
\newblock Late-time evolution of a self-interacting scalar field in the
  spacetime of a dilaton black hole.
\newblock {\em Phys. Rev. D}, 64:044024, 2001.

\bibitem{Dilaton2}
Rafa\l{} Moderski and Marek Rogatko.
\newblock Evolution of a self-interacting scalar field in the spacetime of a
  higher dimensional black hole.
\newblock {\em Phys. Rev. D}, 72:044027, Aug 2005.

\bibitem{Montgomery}
Hugh~L. Montgomery.
\newblock {\em Ten lectures on the interface between analytic number theory and
  harmonic analysis}, volume~84 of {\em CBMS Regional Conference Series in
  Mathematics}.
\newblock Conference Board of the Mathematical Sciences, Washington, DC; by the
  American Mathematical Society, Providence, RI, 1994.

\bibitem{massive2}
Akira Ohashi and Masaaki Sakagami.
\newblock Massive quasi-normal mode.
\newblock {\em Classical and Quantum Gravity}, 21(16):3973, aug 2004.

\bibitem{olver}
Frank Olver.
\newblock {\em Asymptotics and special functions}.
\newblock AK Peters/CRC Press, 1997.

\bibitem{pasqualotto2019spin}
Federico Pasqualotto.
\newblock The spin {$\pm1$} {T}eukolsky equations and the {M}axwell system on
  {S}chwarzschild.
\newblock {\em Ann. Henri Poincar\'{e}}, 20(4):1263--1323, 2019.

\bibitem{KGSchw1}
Federico Pasqualotto, Yakov Shlapentokh-Rothman, and Maxime Van~de Moortel.
\newblock The asymptotics of massive fields on stationary spherically symmetric
  black holes for all angular momenta.
\newblock ArXiv preprint: \url{https://arxiv.org/abs/2303.17767}, 2023.

\bibitem{Pricepaper}
Richard~H. Price.
\newblock Nonspherical perturbations of relativistic gravitational collapse.
  {I}. {S}calar and gravitational perturbations.
\newblock {\em Phys. Rev. D (3)}, 5:2419--2438, 1972.

\bibitem{reedsimonII}
Michael Reed and Barry Simon.
\newblock {\em Methods of modern mathematical physics. {II}. {F}ourier
  analysis, self-adjointness}.
\newblock Academic Press [Harcourt Brace Jovanovich, Publishers], New
  York-London, 1975.

\bibitem{reedsimonIV}
Michael Reed and Barry Simon.
\newblock {\em Methods of modern mathematical physics. {IV}. {A}nalysis of
  operators}.
\newblock Academic Press [Harcourt Brace Jovanovich, Publishers], New
  York-London, 1978.

\bibitem{reedsimonI}
Michael Reed and Barry Simon.
\newblock {\em Methods of modern mathematical physics. {I}}.
\newblock Academic Press, Inc. [Harcourt Brace Jovanovich, Publishers], New
  York, second edition, 1980.
\newblock Functional analysis.

\bibitem{NT1}
Olivier Robert.
\newblock Quelques paires d'exposants par la m\'{e}thode de {V}inogradov.
\newblock {\em J. Th\'{e}or. Nombres Bordeaux}, 14(1):271--285, 2002.

\bibitem{QNM1}
Jo\~ao~G. Rosa and Sam~R. Dolan.
\newblock Massive vector fields on the schwarzschild spacetime: Quasinormal
  modes and bound states.
\newblock {\em Phys. Rev. D}, 85:044043, 2012.

\bibitem{boson2}
Remo Ruffini and Silvano Bonazzola.
\newblock Systems of self-gravitating particles in general relativity and the
  concept of an equation of state.
\newblock {\em Phys. Rev.}, 187:1767--1783, 1969.

\bibitem{ShenGCM}
Dawei Shen.
\newblock Construction of {GCM} hypersurfaces in perturbations of {K}err.
\newblock {\em Ann. PDE}, 9(1):Paper No. 11, 112, 2023.

\bibitem{Yakov}
Yakov Shlapentokh-Rothman.
\newblock Exponentially growing finite energy solutions for the
  {K}lein-{G}ordon equation on sub-extremal {K}err spacetimes.
\newblock {\em Comm. Math. Phys.}, 329(3):859--891, 2014.

\bibitem{SRTdC2020boundedness}
Yakov Shlapentokh-Rothman and Rita Teixeira~da Costa.
\newblock Boundedness and decay for the {T}eukolsky equation on {K}err in the
  full subextremal range $|a|<{M}$: frequency space analysis.
\newblock ArXiv preprint: \url{https://www.arxiv.org/abs/2007.07211}, 2020.

\bibitem{SRTdC2023boundedness2}
Yakov Shlapentokh-Rothman and Rita Teixeira~da Costa.
\newblock Boundedness and decay for the {T}eukolsky equation on {K}err in the
  full subextremal range $|a|<{M}$: physical space analysis.
\newblock ArXiv preprint: \url{https://www.arxiv.org/abs/2302.08916}, 2023.

\bibitem{massive1}
Liliana~E. Simone and Clifford~M. Will.
\newblock Massive scalar quasi-normal modes of schwarzschild and kerr black
  holes.
\newblock {\em Classical and Quantum Gravity}, 9(4):963, 1992.

\bibitem{quasimode3}
Plamen Stefanov.
\newblock Quasimodes and resonances: sharp lower bounds.
\newblock {\em Duke Math. J.}, 99(1):75--92, 1999.

\bibitem{quasimode4}
Plamen Stefanov and Georgi Vodev.
\newblock Distribution of resonances for the {N}eumann problem in linear
  elasticity outside a strictly convex body.
\newblock {\em Duke Math. J.}, 78(3):677--714, 1995.

\bibitem{quasimode5}
Plamen Stefanov and Georgi Vodev.
\newblock Neumann resonances in linear elasticity for an arbitrary body.
\newblock {\em Comm. Math. Phys.}, 176(3):645--659, 1996.

\bibitem{BigStein}
Elias~M. Stein.
\newblock {\em Harmonic analysis: real-variable methods, orthogonality, and
  oscillatory integrals}, volume~43 of {\em Princeton Mathematical Series}.
\newblock Princeton University Press, Princeton, NJ, 1993.

\bibitem{Maxwell3}
Jacob Sterbenz and Daniel Tataru.
\newblock Local energy decay for {M}axwell fields {P}art {I}: {S}pherically
  symmetric black-hole backgrounds.
\newblock {\em Int. Math. Res. Not. IMRN}, (11):3298--3342, 2015.

\bibitem{Ethan1}
Ethan Sussman.
\newblock {H}ydrogen-like {S}chr\"odinger operators at low energies.
\newblock ArXiv preprint: \url{https://www.arxiv.org/abs/2204.08355}, 2022.

\bibitem{Ethan2}
Ethan Sussman.
\newblock Massive waves gravitationally bound to static bodies.
\newblock Online preprint: \url{https://math.mit.edu/~ethanws/KGSchw.pdf},
  2022.

\bibitem{quasimode2}
Siu-Hung Tang and Maciej Zworski.
\newblock From quasimodes to resonances.
\newblock {\em Math. Res. Lett.}, 5(3):261--272, 1998.

\bibitem{Tataru}
Daniel Tataru.
\newblock Local decay of waves on asymptotically flat stationary space-times.
\newblock {\em Amer. J. Math.}, 130(3):571--634, 2008.

\bibitem{NT5}
Timothy Trudgian and Andrew Yang.
\newblock On optimal exponent pairs.
\newblock ArXiv preprint: \url{https://arxiv.org/abs/2303.17767}, 2023.

\bibitem{massive0}
Helvi Witek, Vitor Cardoso, Akihiro Ishibashi, and Ulrich Sperhake.
\newblock Superradiant instabilities in astrophysical systems.
\newblock {\em Phys. Rev. D}, 87:043513, 2013.

\bibitem{QNM5}
Lihui Xue, Bin Wang, and Ru-Keng Su.
\newblock Numerical simulation of the massive scalar field evolution in the
  {R}eissner-{N}ordstr\"om black hole background.
\newblock {\em Phys. Rev. D}, 66:024032, Jul 2002.

\bibitem{Superradiance}
Theodoros~J.M Zouros and Douglas~M Eardley.
\newblock Instabilities of massive scalar perturbations of a rotating black
  hole.
\newblock {\em Annals of Physics}, 118(1):139--155, 1979.

\end{thebibliography}
	
\end{document}